\documentclass{article}

\usepackage[english]{babel}
\usepackage[T1]{fontenc}

\usepackage[letterpaper,top=3cm,bottom=3cm,left=2cm,right=2cm,marginparwidth=2cm]{geometry}

\usepackage{amsmath}
\usepackage{graphicx}
\usepackage{amsfonts}
\usepackage{amsthm}
\usepackage{comment}
\usepackage{slashed}
\usepackage[colorinlistoftodos]{todonotes}
\usepackage[colorlinks=true, allcolors=blue]{hyperref}
\usepackage{amssymb}
\usepackage{float}
\usepackage{tikzfill}
\usepackage{lipsum}
\allowdisplaybreaks
\newtheorem{lemma}{Lemma}[section]
\newtheorem{theorem}{Theorem}[section]
\newtheorem*{theorem*}{Theorem}
\newtheorem{definition}{Definition}[section]
\newtheorem{remark}{Remark}[section]
\newtheorem{corollary}{Corollary}[section]
\newtheorem{proposition}{Proposition}[section]
\newtheorem*{proposition*}{Proposition}
\newcommand{\tr}{\text{tr}}
\numberwithin{equation}{section}

\title{Nonlinear Scattering Theory for Asymptotically de Sitter \\ Vacuum Solutions in All Even Spatial Dimensions}
\author{Serban Cicortas\footnote{Princeton University, Department of Mathematics, Fine Hall, Washington Road, Princeton, NJ 08544, USA}}

\begin{document}

\maketitle

\begin{abstract}
    The purpose of this paper is to establish a definitive quantitative nonlinear scattering theory for asymptotically de Sitter solutions of the Einstein vacuum equations in $(n+1)$ dimensions with $n\geq4$ even, which are determined by small scattering data at the spacelike asymptotic boundaries $\mathcal{I}^-$ and $\mathcal{I}^+.$ The case of even spatial dimension $n$ poses significant challenges compared to its odd counterpart and was left open by the previous works in the literature. Here, scattering theory is understood to mean existence and uniqueness of scattering states, asymptotic completeness, and the existence of an invertible scattering map with quantitative control on its norm. The existence and uniqueness of scattering states imply that for any small asymptotic data there exists a unique global solution to the Einstein equations, which remains close to the de Sitter metric. Asymptotic completeness is the converse statement, showing that any such solution induces asymptotic data at $\mathcal{I}^-$ and at $\mathcal{I}^+.$ For sufficiently small asymptotic data, we construct the scattering map $\mathcal{S}$ taking data at $\mathcal{I}^-$ to data at $\mathcal{I}^+,$ and we show that the map $\mathcal{S}$ is locally invertible and locally Lipschitz at the de Sitter data, with respect to a Sobolev-type norm.

    The scattering map result is sharp and avoids any "derivative loss", in the sense that we measure the smallness of asymptotic data at $\mathcal{I}^-$ and $\mathcal{I}^+$ using the same Sobolev norm. The proof of the sharp result requires a detailed analysis of the Einstein equations involving a geometric Littlewood-Paley decomposition of the solution, carried out in our companion paper \cite{Cwave}.
\end{abstract}



\section{Introduction}

In this work, we aim to complete the understanding of the nonlinear scattering theory for $(n+1)$-dimensional asymptotically de Sitter vacuum solutions determined by small scattering data on a suitably defined asymptotic boundary $\mathcal{I}^{\pm}.$ The case $n=3$ was proved by Friedrich in \cite{Friedrich1, Friedrich2}, while the case of all $n\geq3$ odd was proved by Anderson in \cite{Anderson}. In this paper and our companion paper \cite{Cwave} we treat the case of even spatial dimension $n\geq4$, which contains significant new challenges and was left open by the previous works in the literature. 

\paragraph{Asymptotically de Sitter vacuum solutions.} For any $n\geq 3$, we consider the $(n+1)$-dimensional Einstein vacuum equations with positive cosmological constant $\Lambda=\frac{n(n-1)}{2}:$
\begin{equation}\label{EVE}
    Ric_{\mu\nu}-\frac{1}{2}R\widetilde{g}_{\mu\nu}+\Lambda \widetilde{g}_{\mu\nu}=0.
\end{equation}
The ground state solution of (\ref{EVE}) is given by the $(n+1)$-dimensional de Sitter space $\big(\mathbb{R}\times S^n,\widetilde{g}_{dS}\big)$:
\begin{equation}\label{de Sitter metric}
    \widetilde{g}_{dS}=-dT^2+\cosh^2(T)\cdot\slashed{g}_{S^n},
\end{equation}
where $\slashed{g}_{S^n}$ denotes the standard round metric on $S^n.$ The solution $\big(\mathbb{R}\times S^n,\widetilde{g}_{dS}\big)$ represents the higher dimensional generalization of the metric introduced in \cite{deSitter}. We denote past infinity $\{T\rightarrow-\infty\}$ by $\mathcal{I}^-,$ and future infinity $\{T\rightarrow\infty\}$ by $\mathcal{I}^+.$  Both $\mathcal{I}^-$ and $\mathcal{I}^+$ can be identified with $S^n$ and can be understood as asymptotic boundaries of the spacetime.
\begin{figure}
    \centering
    \begin{tikzpicture}[scale=1.3]
    \draw (-0.2,0) .. controls (1.3,1.5) and (1.3,2.5) .. (-0.2,4);
    \draw (4.2,0) .. controls (2.7,1.5) and (2.7,2.5) .. (4.2,4);
    \draw (0.92,2) .. controls (1.42,1.5) and (2.58,1.5) .. (3.08,2);
    \draw[dashed] (0.92,2) .. controls (1.42,2.3) and (2.58,2.3) .. (3.08,2);
    \draw[latex-] (3.12,2) -- (5,2);
    \draw (5,2)  node[anchor=west] {$\{T=0\}\cong S^n$};
    \draw[dashed] (-0.2,0) .. controls (0.8,-0.7) and (3.2,-0.7) .. (4.2,0);
    \draw[dashed] (-0.2,0) .. controls (0.8,0.5) and (3.2,0.5) .. (4.2,0);
    \draw[latex-] (4.24,0) -- (5,0);
    \draw (5,0)  node[anchor=west] {$\mathcal{I}^-=\{T\rightarrow-\infty\}\cong S^n$};
    \draw[dashed] (-0.2,4) .. controls (0.8,3.5) and (3.2,3.5) .. (4.2,4);
    \draw[dashed] (-0.2,4) .. controls (0.8,4.7) and (3.2,4.7) .. (4.2,4);
    \draw[latex-] (4.24,4) -- (5,4);
    \draw (5,4)  node[anchor=west] {$\mathcal{I}^+=\{T\rightarrow\infty\}\cong S^n$};
    \end{tikzpicture}
    \caption{Diagram of the $(n+1)$-dimensional de Sitter space $\big(\mathbb{R}\times S^n,\widetilde{g}_{dS}\big)$.}
    \label{fig:main-thm}
\end{figure}
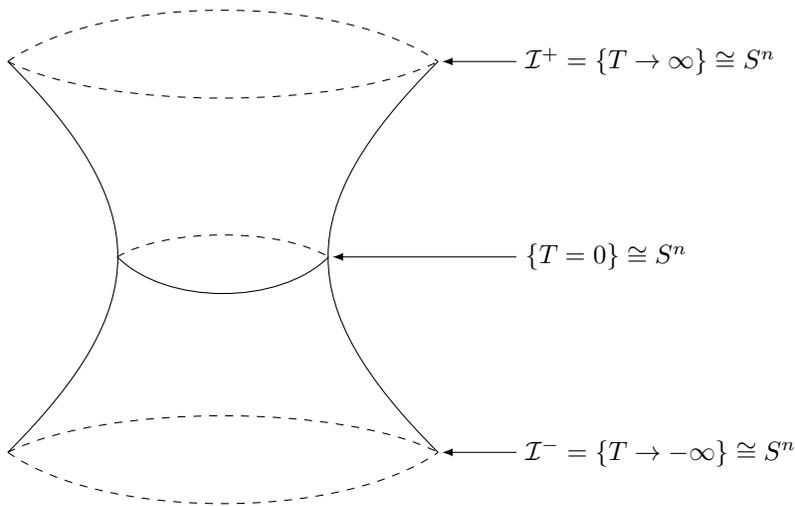

Due to the hyperbolic nature of the Einstein vacuum equations \eqref{EVE}, we study dynamical solutions of (\ref{EVE}) obtained by solving an initial value problem. We review briefly the standard setting of Cauchy initial data prescribed on a spacelike hypersurface. The data consist of a Riemannian manifold $\big(S^n,g_0\big)$ and a symmetric 2-tensor $k_0$, which satisfy certain constraint equations. In the spacetime obtained by solving (\ref{EVE}) locally with the given initial data, $\big(S^n,g_0\big)$ embeds as the initial data hypersurface with second fundamental form given by $k_0$. Moreover, the constraint equations are given by the Gauss and Codazzi equations for $g_0$ and $k_0.$

In the current work, we will focus instead on solutions of (\ref{EVE}) arising from scattering data prescribed at infinity (at $\mathcal{I}^-$ or $\mathcal{I}^+$). In this setting, scattering data consist of a Riemannian manifold $\big(S^n,\slashed{g}_0\big)$ and a symmetric traceless 2-tensor $\check{h}$ on $S^n$, which satisfies an additional constraint called the straightness condition. We briefly describe the interpretation of scattering data and the differences to the setting of Cauchy initial data. We consider the case of data prescribed at $\mathcal{I}^-$ for the purpose of exposition. Similarly to the case of Cauchy initial data, $\slashed{g}_0$ represents a suitable limit at $\mathcal{I}^-$ of the Riemannian metrics induced by the spacetime metric on the spheres $S^n$ at early times. However, it turns out that a similar attempt to consider the limit of the second fundamental form of the spheres $S^n$ at early times yields a tensor that is completely determined by $\slashed{g}_0$. Instead, $\check{h}$ represents a higher order term present in the expansion of the spacetime metric at $\mathcal{I}^-$, as we shall explain in detail below.

Solutions of \eqref{EVE} determined by scattering data at infinity are called asymptotically de Sitter spaces (\cite{Anderson}). The local well-posedness theory for scattering data at infinity is proved in \cite{selfsimilarvacuum} (see also \cite{Fefferman-Graham, ambient-metric, hintz}), with asymptotic data given by $\slashed{g}_0$ and $\check{h}$ as above. We also assume that the data are smooth for simplicity. For any such data, we obtain in a neighborhood of $\mathcal{I}^-$ a unique solution of the form:
\begin{equation}\label{general form for g tilde}
    \widetilde{g}=-dT^2+e^{-2T}\widetilde{\slashed{g}}_{AB}(T,\theta^1,\ldots,\theta^n)d\theta^Ad\theta^B,
\end{equation}
where we denote by $\{\theta^A\}$ coordinates associated to an arbitrary chart on $S^n.$ As before, in \eqref{general form for g tilde} we present the case of scattering data at past infinity $\mathcal{I}^-$ and we note that the same definitions apply for scattering data at $\mathcal{I}^+,$ upon replacing $T$ by $-T.$ 

\paragraph{The expansion at $\mathcal{I}^-$ in the even $n$ case.}
The relation between the solution $\widetilde{g}$ and the scattering data is seen in the expansion satisfied by the rescaled metric $\widetilde{\slashed{g}}$ induced on $\{T\}\times S^n$ as $T\rightarrow-\infty$. For all $n\geq4$ \underline{even} and for all $T<0$ small enough, we have in a Lie-propagated frame:
\begin{equation}\label{expansion in terms of T intro}
    \widetilde{\slashed{g}}=\slashed{g}_0+\frac{e^{2T}}{2^2}\slashed{g}_1+\ldots+\frac{e^{(n-2)T}}{2^{n-2}(n/2-1)!}\slashed{g}_{n/2-1}+\frac{2Te^{nT}}{2^n(n/2)!}\mathcal{O}+\frac{e^{nT}}{2^n(n/2)!}\check{k}+O\big(Te^{(n+2)T}\big),
\end{equation}
where the tensors $\slashed{g}_1,\ldots,\slashed{g}_{n/2-1},\mathcal{O},\tr\check{k}$ are determined by $\slashed{g}_0$ via certain compatibility relations, $\check{h}$ is the trace free part of $\check{k},$ and the higher order terms in the expansion are determined by $\slashed{g}_0$ and $\check{h}$. 

The main challenge present in the case of even spatial dimension $n\geq4$ compared to its odd counterpart can be seen in the expansion \eqref{expansion in terms of T intro}. This represents the Fefferman-Graham expansion introduced in \cite{Fefferman-Graham}, also playing a fundamental role in \cite{ambient-metric} and \cite{selfsimilarvacuum}, as we explain in Section~\ref{previous results}. We point out that for $n\geq4$ even the expansion is not smooth at $\mathcal{I}^-$ in terms of the Fefferman-Graham coordinate $\tau=e^{T}/2$ due to the presence of the term $\mathcal{O}$, called the obstruction tensor. On the other hand, the obstruction tensor vanishes identically in the odd $n$ case leading to a smooth expansion of the solution in $\tau$ at $\mathcal{I}^-$.

From the dynamical point of view, the de Sitter metric (\ref{de Sitter metric}) is the unique solution of (\ref{EVE}) with scattering data at $\mathcal{I}^-$ given by $\big(\slashed{g}_0\big)_{dS}=\frac{1}{4}\slashed{g}_{S^n}$ and $\check{h}_{dS}=0.$ It is also smooth at $\mathcal{I}^-$, since the obstruction tensor $\mathcal{O}_{dS}$ vanishes, as can be seen from its Fefferman-Graham expansion:
\[\widetilde{\slashed{g}}_{dS}=\frac{1}{4}\slashed{g}_{S^n}+\frac{e^{2T}}{2}\slashed{g}_{S^n}+\frac{e^{4T}}{4}\slashed{g}_{S^n}.\]

\paragraph{The main result: a complete scattering theory for $n\geq4$ even.} We study the global in time behavior of asymptotically de Sitter solutions of (\ref{EVE}) for $n\geq4$ even, determined by scattering data close to the data of de Sitter space with respect to some suitable norm. Our main result is establishing a complete scattering theory for such solutions. We notice that the corresponding small data result in the setting of Cauchy initial data consists of proving global existence and orbital stability, see \cite{Ringstrom1}. The additional difficulties that we encounter in the case of the scattering problem are the need to evolve the data from past infinity as opposed to from a Cauchy hypersurface, and the need to recover the scattering data at future infinity, which requires sharp control of the solution at higher order despite the lack of smoothness in time caused by the obstruction tensor $\mathcal{O}$.

In order to state the main result, we briefly introduce the notions of asymptotic initial data set, asymptotic initial data norm, and asymptotically de Sitter vacuum solutions determined by small data. We refer the reader to Remark \ref{remark about data sets} and Definition \ref{asymptotic data set definition} for the precise definitions. 

Given smooth scattering data $\big(\slashed{g}_0,\check{h}\big),$ we define the corresponding \textit{asymptotic initial data set} $\widetilde{\Sigma}\big(\slashed{g}_0,\check{h}\big)$ to be the collection of tensors on $S^n$ consisting of: the metric $\slashed{g}_0;$ the tensors $\slashed{g}_1,\ldots,\slashed{g}_{n/2-1},\tr\check{k}$ defined by $\slashed{g}_0$ using the compatibility relations, together with certain angular derivatives of these tensors; the obstruction tensor $\mathcal{O}$ defined by $\slashed{g}_0$ using the compatibility relations; and the renormalized tensor $\check{\mathfrak{h}}=\check{h}-2\big(\log\nabla\big)\mathcal{O},$ where the operator $\log\nabla$ is defined using the geometric Littlewood-Paley decomposition in Section~\ref{LP Section}. We point out that the surprising need to renormalize $\check{h}$ is related to the lack of smoothness of the expansion (\ref{expansion in terms of T intro}), and again poses difficulties in our problem. 

We also define the \textit{asymptotic initial data norm}, which measures closeness to the de Sitter data. For any tensor $\phi\in\widetilde{\Sigma}\big(\slashed{g}_0,\check{h}\big)$, we denote by $\phi_*=\phi-\phi_{dS}$ the tensor obtained as the difference of $\phi$ and its de Sitter value. For any $M>0$, we define the \textit{asymptotic initial data norm of order $M$} by:
\[\Big\|\widetilde{\Sigma}\big(\slashed{g}_0,\check{h}\big)\Big\|_M^2=\sum_{\phi\in\widetilde{\Sigma}(\slashed{g}_0,\check{h})}\big\|\phi_*\big\|_{H^{M+1}(S^n)}^2,\]
where $H^{M+1}(S^n)$ represents the Sobolev norm on $S^n$ with respect to the metric $\slashed{g}_0$.

We denote by $\widetilde{\Sigma}_{dS}=\widetilde{\Sigma}\big(\frac{1}{4}\slashed{g}_{S^n},0\big)$ the initial data set corresponding to de Sitter space. For any $\epsilon>0$, we define the set of smooth $\epsilon$-small asymptotic data of order $M$ by:
\[B_{\epsilon}^M\big(\widetilde{\Sigma}_{dS}\big)=\Big\{\widetilde{\Sigma}\big(\slashed{g}_0,\check{h}\big):\ \big\|\widetilde{\Sigma}\big(\slashed{g}_0,\check{h}\big)\big\|_M<\epsilon\Big\}.\]
We define \textit{asymptotically de Sitter vacuum solutions determined by small data} to be the solutions of (\ref{EVE}) with $\epsilon$-small asymptotic data of order $M$.

Using these definitions, we state the main result proved in this work and our companion paper \cite{Cwave}:
\begin{theorem}\label{main theorem of the paper}
    For any even integer $n\geq4,$ we have a complete scattering theory for asymptotically de Sitter vacuum solutions determined by small data. For any $M>0$ large enough there exists $\epsilon_0>0$ small enough, such that for any $0<\epsilon\leq\epsilon_0$ we have:
    \begin{enumerate}
    \item\textbf{Existence and uniqueness of scattering states:} for any $\epsilon$-small asymptotic data of order $M$ at $\mathcal{I}^-$ or $\mathcal{I}^+$ given by $\widetilde{\Sigma}\big(\slashed{g}_0,\check{h}\big)\in B_{\epsilon}^M\big(\widetilde{\Sigma}_{dS}\big)$, there exists a unique smooth global solution $\big(\widetilde{\mathcal{M}},\widetilde{g}\big)$ of the form (\ref{general form for g tilde}) to the Einstein vacuum equations (\ref{EVE}) which remains quantitatively close to the de Sitter metric and can be represented by a diagram similar to Figure~\ref{fig:main-thm}; 
    \item\textbf{Asymptotic completeness:} any smooth solution of the Einstein vacuum equations (\ref{EVE}) of the form (\ref{general form for g tilde}), which is quantitatively close to the de Sitter metric at a finite time $T,$ exists globally and induces scattering data $\big(\slashed{g}_0,\check{h}\big)$ at $\mathcal{I}^-$ and $\big(\underline{\slashed{g}_0},\underline{\check{h}}\big)$ at $\mathcal{I}^+;$
    \item\textbf{Existence of a scattering map with quantitative estimates:} there exists a constant $C_M>0$ independent of $\epsilon,$ such that we have a well-defined scattering map taking asymptotic data at $\mathcal{I}^-$ to asymptotic data at $\mathcal{I}^+:$
    \begin{equation}\label{definition of scattering map intro}
        \mathcal{S}:B_{\epsilon}^M\big(\widetilde{\Sigma}_{dS}\big)\rightarrow B_{C_M\epsilon}^M\big(\widetilde{\Sigma}_{dS}\big),\ \mathcal{S}\Big(\widetilde{\Sigma}\big(\slashed{g}_0,\check{h}\big)\Big)=\widetilde{\Sigma}\big(\underline{\slashed{g}_0},\underline{\check{h}}\big).
    \end{equation}
    The scattering map is locally invertible and locally Lipschitz at $\widetilde{\Sigma}_{dS},$ in the sense that it satisfies the quantitative estimate:
    \begin{equation}\label{main estimate for scattering map intro}
        \Big\|\mathcal{S}\big(\widetilde{\Sigma}(\slashed{g}_0,\check{h})\big)\Big\|_M\leq C_M\Big\|\widetilde{\Sigma}\big(\slashed{g}_0,\check{h}\big)\Big\|_M.
    \end{equation}
\end{enumerate}
\end{theorem}

\begin{remark}
    The result for the scattering map $\mathcal{S}$ is sharp and avoids any "derivative loss", in the sense that in \eqref{definition of scattering map intro} and \eqref{main estimate for scattering map intro} we use the same Sobolev-type asymptotic initial data norm of order $M$ to measure the smallness of asymptotic data at $\mathcal{I}^-$ and $\mathcal{I}^+$.
\end{remark}

In the remainder of the introduction, we flesh out the previous discussion with more details. In Section~\ref{previous results} we discuss some relevant previous results. In Section~\ref{ambient metric section} we introduce the ambient metric formulation of the problem and restate our main result in an equivalent form. In Section~\ref{proof outline section} we discuss the ideas of the proof in some detail and explain how the results in our companion paper \cite{Cwave} are used in the proof of the sharp result for the scattering map. Finally, in Section~\ref{paper outline section} we outline the structure of the rest of the paper.

\subsection{Previous Results}\label{previous results}

We present some previous results relevant for the scattering theory of asymptotically de Sitter vacuum solutions.

\subsubsection{The Stability of de Sitter Space}
Friedrich proved in \cite{Friedrich1, Friedrich2} that $(3+1)$-dimensional de Sitter space is non-linearly stable to small perturbations of the asymptotic data at $\mathcal{I}^-$. The proof uses the key fact that in $(3+1)$ dimensions de Sitter space has a smooth conformal compactification. By use of the conformal method, the study of global stability is reduced to a finite in time problem for the conformal equations, which can be written as a symmetric hyperbolic system. Additionally, this method also gives a scattering theory between asymptotic data at $\mathcal{I}^-$ and asymptotic data at $\mathcal{I}^+,$ which represent two regular spacelike hypersurfaces in the conformal spacetime.

In the case of the Einstein equations coupled to a non-linear scalar field, which is a generalization of (\ref{EVE}), Ringström proved stability in all dimensions in \cite{Ringstrom1} for small Cauchy initial data on a finite time spacelike hypersurface. This proof is robust in order to treat such general equations, but it does not give a description of the induced scattering data at infinity.

\subsubsection{The Fefferman-Graham Expansion} The starting point in the theory of local well-posedness with scattering data at $\mathcal{I}^-$ for all $n\geq3$ is given by the work of Fefferman-Graham \cite{Fefferman-Graham, ambient-metric}. 

To construct conformal invariants for an $n$-dimensional Riemannian manifold $\big(\mathcal{S},\slashed{g}_0\big),$ Fefferman and Graham first consider the corresponding ambient metric. We briefly introduce the ambient metric construction here and we discuss it in detail in Section~\ref{ambient metric section}. For any $\slashed{g}_0$ and any symmetric traceless 2-tensor $\check{h},$ the ambient metric is an $(n+2)$-dimensional self-similar vacuum metric given by a formal power series expansion determined by $\big(\slashed{g}_0,\check{h}\big)$. The conformal invariants of $\big(\mathcal{S},\slashed{g}_0\big)$ are then obtained using the classification of local pseudo-Riemannian invariants of the ambient metric. Under the additional assumption of straightness on $\check{h},$ which determines the divergence of $\check{h}$ in terms of $\slashed{g}_0$, the ambient metric is straight and can be quotient out by the action of the scaling vector field to obtain formal asymptotically de Sitter solutions of (\ref{EVE}).

We illustrate the Fefferman-Graham expansion of formal asymptotically de Sitter vacuum solutions. The scattering data are given by a Riemannian metric $\big(S^n,\slashed{g}_0\big)$ and a symmetric traceless straight 2-tensor $\check{h}$, which determine each term in the expansion. For $n\geq3$ odd, the expansion at $\mathcal{I}^-$ is smooth in terms of $e^T$:
\begin{equation}\label{FG expansion odd intro}
    \widetilde{\slashed{g}}=\slashed{g}_0+\frac{e^{2T}}{2^2}\slashed{g}_1+\ldots+\frac{e^{(n-1)T}}{2^{n-1}((n-1)/2)!}\slashed{g}_{(n-1)/2}+\frac{e^{nT}}{2^n}\check{k}+O\big(e^{(n+1)T}\big).
\end{equation}
In the case of $n\geq4$ even, we have the expansion at $\mathcal{I}^-$:
\begin{equation}\label{FG expansion even intro}
    \widetilde{\slashed{g}}=\slashed{g}_0+\frac{e^{2T}}{2^2}\slashed{g}_1+\ldots+\frac{e^{(n-2)T}}{2^{n-2}(n/2-1)!}\slashed{g}_{n/2-1}+\frac{2Te^{nT}}{2^n(n/2)!}\mathcal{O}+\frac{e^{nT}}{2^n(n/2)!}\check{k}+O\big(Te^{(n+2)T}\big).
\end{equation}
The compatibility relations are obtained by taking the limit of \eqref{EVE} at $\mathcal{I}^-$ at each order. The terms of order less than $e^{nT}$ are determined by $\slashed{g}_0.$ We also have that $\text{tr}\check{k}$ is determined by $\slashed{g}_0,$ and that the trace-free part of $\check{k}$ is $\check{h}.$ Finally, all the higher order terms in the expansion are determined by $\slashed{g}_0$ and $\check{h}.$

\subsubsection{The Local Well-posedness Theory with Scattering Data} The above expansions \eqref{FG expansion odd intro} and \eqref{FG expansion even intro} were computed formally in the smooth category in \cite{Fefferman-Graham, ambient-metric}, and convergence was only proven in the case of analytic scattering data. The rigorous proof of the Fefferman-Graham expansion in the smooth case was done in \cite{selfsimilarvacuum} in a more general context (and revisited in \cite{hintz}). Restricted to our situation, the results of \cite{selfsimilarvacuum} imply the following local well-posedness result with scattering data:
\begin{theorem}[\cite{selfsimilarvacuum}]\label{RSR theorem}
    For any $n\geq3$ and any smooth straight scattering data $\big(\slashed{g}_0,\check{h}\big)$, there exists a unique solution of (\ref{EVE}) of the form (\ref{general form for g tilde}) in a neighborhood of $\mathcal{I}^-$ which satisfies the above expansions.
\end{theorem}

The local well-posedness result of \cite{selfsimilarvacuum} is a fundamental ingredient needed to study the long time behavior of asymptotically de Sitter solutions of (\ref{EVE}). For simplicity, we only stated how the results of \cite{selfsimilarvacuum} apply in our situation of straight ambient metrics. However, the results of \cite{selfsimilarvacuum} hold in the very general context of "proto-ambient metrics", which only require the Fefferman-Graham expansion to hold up to the term containing $\check{k}.$

\subsubsection{A Scattering Theory in the Odd Spatial Dimension Case} The results of \cite{Friedrich1, Friedrich2} were further generalized in \cite{Anderson} for all $(n+1)$-dimensional de Sitter spaces with $n\geq3$ odd. While the conformal method does not apply in higher dimensions, there is nevertheless the simplification of having the expansion (\ref{FG expansion odd intro}) which is smooth in terms of $e^T.$ Moreover, in this case the Einstein equations (\ref{EVE}) can be replaced by the equation $\mathcal{O}=0$, which is conformally invariant and leads to a hyperbolic system in a suitable gauge. We notice that in $(3+1)$ dimensions the obstruction tensor $\mathcal{O}$ coincides with the Bach tensor, so the approach of \cite{Anderson} is to replace (\ref{EVE}) by the Bach equations. Using these ingredients, \cite{Anderson} generalizes the conformal method proof to obtain stability for all $n\geq3$ odd, which also gives a scattering theory between asymptotic data at $\mathcal{I}^-$ and asymptotic data at $\mathcal{I}^+$.

\subsubsection{The Wave Equation on de Sitter Space} The simplest model problem needed in order to understand the scattering of asymptotically de Sitter vacuum solutions is the linear wave equation on a fixed de Sitter background:
\begin{equation}\label{wave equation on dS background}
    \square_{dS}\widetilde{\phi}=0.
\end{equation}
The scattering problem in the more general case of the Klein-Gordon equation was addressed in \cite{microlocal}. Given a certain relation between the Klein-Gordon mass and the spatial dimension, which guarantees a smooth expansion at infinity similar to (\ref{FG expansion odd intro}), \cite{microlocal} provides a detailed description of the scattering map as a Fourier integral operator. However, in the case of $n\geq4$ even and vanishing Klein-Gordon mass, the solution satisfies an expansion at infinity similar to (\ref{FG expansion even intro}). The results of \cite{microlocal} prove that the scattering map is an isomorphism on $C^{\infty}$. 

In the case of (\ref{wave equation on dS background}) with $n\geq4$ even, we used a different approach in \cite{linearwave} to construct the scattering map as a Banach space isomorphism for asymptotic initial data $\big(\phi_0,\mathfrak{h}\big)\in H^{M+n}(S^n)\times H^M(S^n),$ for any $M\geq1$. We notice that $\phi_0$ plays a similar role to $\slashed{g}_0,$ and $\mathfrak{h}$ is again obtained by renormalizing a higher order term in the Fefferman-Graham expansion using $\log\nabla$ of the analogue of the obstruction tensor. Based on these similarities, \cite{linearwave} will provide the guideline for studying the scattering of asymptotically de Sitter vacuum solutions for all $n\geq4$ even in the present work. It turns out that the methods used in \cite{linearwave} are indeed robust and can be adapted in the current setting.

\subsubsection{Other Scattering Results} We note that the study of scattering theory in the context of black holes is a current area of research. We refer the reader to \cite{DHR}, \cite{DRSR}, \cite{fred}, \cite{hamed}, \cite{schdesitter}, and references therein.

\subsection{The Ambient Metric Formulation}\label{ambient metric section}

The ambient metric construction provides an embedding of solutions of (\ref{EVE}) of the form (\ref{general form for g tilde}) into $(n+2)$-dimensional self-similar vacuum spacetimes. The simplest example for this correspondence is that of de Sitter space $\big(\mathbb{R}\times S^n,\widetilde{g}_{dS}\big)$. The associated straight ambient metric is the $\{u<0,v>0\}\times S^n\subset\mathbb{R}^{n+2}$ region of Minkowski space with the Minkowski metric $m$, where $u$ and $v$ are the standard double null coordinates.

The embedding allows us to prove a scattering result at the level of the corresponding ambient metric instead. In the $n\geq4$ even case the solution is not smooth at infinity, so we can interpret this construction as a compactification that allows us to reduce a global problem with data at infinity to a finite problem with singular data. Another advantage of this setting is that the ambient spacetime has a natural double null foliation and we can use the approach developed in \cite{selfsimilarvacuum} and \cite{nakedsing}. Moreover, the explicit embedding provides additional structure on the ambient metric, as can be seen in the definition below.

\begin{definition}\label{ambient metric definition}
    Let $I\subset\mathbb{R}$ be an open interval, set $\widetilde{\mathcal{M}}=I\times S^n,$ and let $\big(\widetilde{\mathcal{M}},\widetilde{g}\big)$ be a solution of (\ref{EVE}) of the form:
    \[\widetilde{g}=-dT^2+e^{-2T}\widetilde{\slashed{g}}_{AB}(T,\theta^1,\ldots,\theta^n)d\theta^Ad\theta^B.\]
    We define the corresponding straight ambient metric to be $\big(\mathcal{M},g\big)$, where $\mathcal{M}=(-\infty,0)\times\widetilde{\mathcal{M}}$ and:
    \[g=ds^2+s^2\big(-dT^2+e^{-2T}\widetilde{\slashed{g}}_{AB}(T,\theta^1,\ldots,\theta^n)d\theta^Ad\theta^B\big).\]
    The spacetime $\big(\mathcal{M},g\big)$ is an $(n+2)$-dimensional straight self-similar vacuum spacetime, satisfying:
    \begin{equation}\label{actual vacuum equations}
        Ric(g)=0,\ \mathcal{L}_Sg=2g,\ S=s\partial_s.
    \end{equation}
    The term straight refers to the special form of the metric $g,$ obtained as a cone metric from $\widetilde{g}$ (see \cite{cone}).
    
    We define the double null coordinates $u<0,\ v>0$ by:
    \[e^T=2\sqrt{-\frac{v}{u}},\ s=-2\sqrt{-uv}.\]
    In double null coordinates, the straight ambient metric has the form:
    \[\mathcal{M}=\bigg\{u<0,\ v>0,\ \log2\sqrt{-\frac{v}{u}}\in I\bigg\}\times S^n\]
    \[g=-2(du\otimes dv+dv\otimes du)+\slashed{g}_{AB}(u,v,\theta^1,\ldots,\theta^n)d\theta^Ad\theta^B,\]
    where $\slashed{g}_{AB}(u,v,\theta^1,\ldots,\theta^n)=u^2\widetilde{\slashed{g}}_{AB}\big(T(u,v),\theta^1,\ldots,\theta^n\big).$ Moreover, the scaling vector field is $S=u\partial_u+v\partial_v.$

    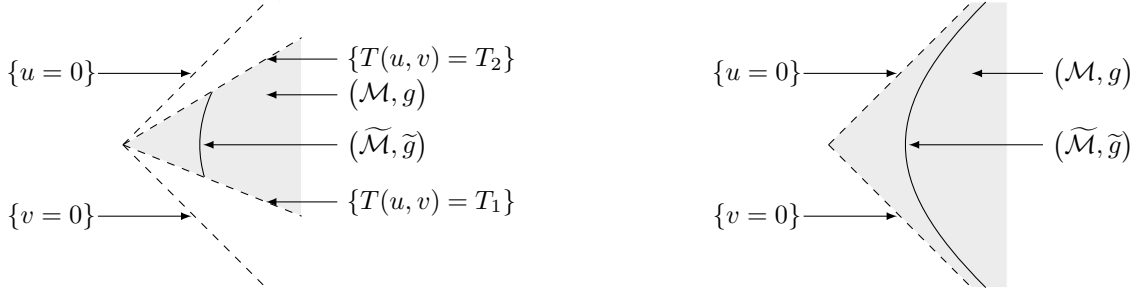
\begin{figure}[H]
    \centering
    \begin{minipage}{0.5\textwidth}
        \centering
        \begin{tikzpicture}[scale=0.95]
    \path[draw=white,fill plain picture={%
    \filldraw[fill=gray!15](0,0) circle (3);
    }]
    (0,2) -- (2.5,3.5) -- (2.5,1);
    \draw (2.2,0) .. controls (0.7,1.5) and (0.7,2.5) .. (2.2,4);
    \path[draw=white,fill plain picture={%
    \filldraw[fill=white](0,0) circle (3);
    }]
    (0,2) -- (2.5,3.5) -- (2.5,4) -- (2,4) -- (0,2);
    \path[draw=white,fill plain picture={%
    \filldraw[fill=white](0,0) circle (3);
    }]
    (0,2) -- (2.5,1) -- (2.5,0) -- (2,0) -- (0,2);
    \draw[-latex] (-0.35,3) -- (0.96,3);
    \draw (-1,3) node {$\{u=0\}$};
    \draw[dashed] (0,2) -- (2,4);
    \draw[dashed] (0,2) -- (2,0);
    \draw[-latex] (-0.35,1) -- (0.96,1);
    \draw (-1,1) node {$\{v=0\}$};
    \draw[latex-] (1.12,2) -- (3,2);
    \draw (3,2)  node[anchor=west] {$\big(\widetilde{\mathcal{M}},\widetilde{g}\big)$};
    \draw[latex-] (2,2.7) -- (3,2.7);
    \draw (3,2.7)  node[anchor=west] {$\big(\mathcal{M},g\big)$};
    \draw[latex-] (2,3.2) -- (3,3.2);
    \draw (3,3.2)  node[anchor=west] {$\{T(u,v)=T_2\}$};
    \draw[latex-] (2,1.2) -- (3,1.2);
    \draw (3,1.2)  node[anchor=west] {$\{T(u,v)=T_1\}$};
    \draw[dashed] (0,2) -- (2.5,3.5);
    \draw[dashed] (0,2) -- (2.5,1);
    \end{tikzpicture}
    \end{minipage}\hfill
    \begin{minipage}{0.5\textwidth}
        \centering
        \begin{tikzpicture}[scale=0.95]
    \path[draw=white,fill plain picture={%
    \filldraw[fill=gray!15](0,0) circle (3);
    }]
    (0,2) -- (2,0) -- (2.5,0) -- (2.5,4) -- (2,4);
    \draw[-latex] (-0.35,3) -- (0.96,3);
    \draw (-1,3) node {$\{u=0\}$};
    \draw[dashed] (0,2) -- (2,4);
    \draw[dashed] (0,2) -- (2,0);
    \draw[-latex] (-0.35,1) -- (0.96,1);
    \draw (-1,1) node {$\{v=0\}$};
    \draw (2.2,0) .. controls (0.7,1.5) and (0.7,2.5) .. (2.2,4);
    \draw[latex-] (1.12,2) -- (3,2);
    \draw (3,2)  node[anchor=west] {$\big(\widetilde{\mathcal{M}},\widetilde{g}\big)$};
    \draw[latex-] (2,3) -- (3,3);
    \draw (3,3)  node[anchor=west] {$\big(\mathcal{M},g\big)$};
    \end{tikzpicture}
    \end{minipage}
    \caption{Embedding of $\big(\widetilde{\mathcal{M}},\widetilde{g}\big)$ in the ambient spacetime $\big(\mathcal{M},g\big)$ in the cases $I=(T_1,T_2)$ and $I=\mathbb{R}$.}
    \label{fig:ambient}
\end{figure}

    In general, we refer to any spacetime $\big(\mathcal{M},g\big)$ satisfying the above properties as a straight self-similar vacuum spacetime in double null coordinates. Given any such spacetime, we can quotient by the scaling vector field in order to obtain the corresponding $(n+1)$-dimensional spacetime $\big(\widetilde{\mathcal{M}},\widetilde{g}\big)$ of the form (\ref{general form for g tilde}) which solves (\ref{EVE}).
\end{definition}

The main point of the ambient metric construction is that establishing a complete scattering theory for solutions of (\ref{EVE}) with scattering data close to the data of de Sitter space is equivalent to proving a scattering theory for straight self-similar vacuum spacetimes with scattering data on $\{v = 0\}$ or $\{u = 0\}$ close to the data of Minkowski space. This follows since we can identify the $\mathcal{I}^-$ of $\big(\widetilde{\mathcal{M}},\widetilde{g}\big)$ with the quotient of $\{v = 0\}$ by the action of the scaling vector field, and similarly we can identify $\mathcal{I}^+$ with the quotient of $\{u=0\}$ by $S$. We refer the reader to \cite{cone} and \cite{selfsimilarvacuum} for a general proof of the correspondence. 

\textbf{Notation convention.} We use the tilde superscript notation in the original $(n+1)$-dimensional formulation (for example $\widetilde{\mathcal{M}},\widetilde{g},\widetilde{\slashed{g}}$), and we drop the tilde superscript in the ambient metric formulation (for example $\mathcal{M},g,\slashed{g}$).

\paragraph{The local well-posedness theory.} Before stating our main result in the ambient metric formulation, we outline the local well-posedness theory in the current situation. By self-similarity, to obtain scattering data on $\{v=0\}$ it suffices to specify data on the sphere $\{u=-1,v=0\}\times S^n.$ Thus, the notion of scattering data in the ambient metric formulation is the same as in the original $(n+1)$-dimensional formulation. The results of \cite{selfsimilarvacuum} imply that for any smooth straight scattering data $\big(\slashed{g}_0,\check{h}\big)$ at $\{u=-1,v=0\}\times S^n,$ there exists a unique straight self-similar vacuum spacetime in double null coordinates defined in a neighborhood of $\{v = 0\}$:
\[\bigg\{u<0,\ v>0,\ 0<-\frac{v}{u}<\underline{v}\bigg\}\times S^n\]
for some small $\underline{v}>0$ depending on the size of the initial data. Moreover, the solution satisfies the expansion:
\[u^{-2}\slashed{g}=\slashed{g}_0+v/|u|\slashed{g}_1+\ldots+\frac{(v/|u|)^{\frac{n-2}{2}}}{(n/2-1)!}\slashed{g}_{n/2-1}+\frac{(v/|u|)^{\frac{n}{2}}\log\big(4v/|u|\big)}{(n/2)!}\mathcal{O}+\frac{(v/|u|)^{\frac{n}{2}}}{(n/2)!}\check{k}+O\Big((v/|u|)^{\frac{n+2}{2}}\log\big(v/|u|\big)\Big)\]
for the same 2-tensors $\slashed{g}_1,\ldots,\slashed{g}_{n/2-1},\mathcal{O},\check{k}$ as above, determined by $\slashed{g}_0$ and $\check{h}.$ The same result holds for data at $\{u = 0\},$ upon replacing $(u,v)$ by $(-v,-u).$

\paragraph{The main result in the ambient metric formulation.}
We briefly explain the corresponding notion of an asymptotic initial data set in the current setting, and refer the reader to Definition \ref{asymptotic data set definition} for the precise definition. 

In what follows, we assume some familiarity with the double null formalism introduced in detail in Section~\ref{set up section}. We denote the Ricci coefficients schematically by $\psi$ and the curvature components by $\Psi.$ We consider the case of scattering data at $\{u=-1,v=0\}\times S^n,$ and note that the case of scattering data at $\{u=0,v=1\}\times S^n$ is defined similarly by replacing $(u,v)$ with $(-v,-u).$ Given smooth scattering data $\big(\slashed{g}_0,\check{h}\big),$ we define the \textit{asymptotic initial data set} $\Sigma\big(\slashed{g}_0,\check{h}\big)$ to be the collection of tensors on $\{u=-1,v=0\}\times S^n$ consisting of: the metric $\slashed{g}_0;$ the double null quantities $\psi$ and $\Psi,$ together with certain angular and $\nabla_{\partial_v}$ derivatives of these tensors, which can be computed by the compatibility relations in terms of $\slashed{g}_0$ (as in \cite{selfsimilarvacuum}, the specification of these tensors is equivalent to the specification of $\slashed{g}_1,\ldots,\slashed{g}_{n/2-1},\tr\check{k}$); the obstruction tensor $\mathcal{O}$; and the renormalized tensor $\check{\mathfrak{h}}=\check{h}-2\big(\log\nabla\big)\mathcal{O}$. Next, we define the \textit{asymptotic initial data norm}, measuring closeness to the Minkowski data. For any tensor $\phi\in\Sigma\big(\slashed{g}_0,\check{h}\big)$, we denote by $\phi^*=\phi-\phi_{\mathrm{Minkowski}}$ the tensor obtained as the difference of $\phi$ and its Minkowski value. As before, we define the \textit{asymptotic initial data norm of order $M$} by:
\[\Big\|\Sigma\big(\slashed{g}_0,\check{h}\big)\Big\|_M^2=\sum_{\phi\in\Sigma(\slashed{g}_0,\check{h})}\big\|\phi^*\big\|_{H^{M+1}(S^n)}^2.\]
We denote by $\Sigma_{\mathrm{Minkowski}}=\Sigma\big(\frac{1}{4}\slashed{g}_{S^n},0\big)$ the initial data set corresponding to Minkowksi space. For any $\epsilon>0$, we define the set of smooth $\epsilon$-small asymptotic data of order $M$ by:
\[B_{\epsilon}^M\big(\Sigma_{\mathrm{Minkowski}}\big)=\Big\{\Sigma\big(\slashed{g}_0,\check{h}\big):\ \big\|\Sigma\big(\slashed{g}_0,\check{h}\big)\big\|_M<\epsilon\Big\}.\]

\begin{remark}\label{remark about data sets}
    In order to make the previous definition of $\widetilde{\Sigma}\big(\slashed{g}_0,\check{h}\big)$ precise, we require that the norms $\big\|\Sigma\big(\slashed{g}_0,\check{h}\big)\big\|_M$ and $\big\|\widetilde{\Sigma}\big(\slashed{g}_0,\check{h}\big)\big\|_M$ are equivalent, where $\Sigma\big(\slashed{g}_0,\check{h}\big)$ is given as in Definition \ref{asymptotic data set definition}. This determines the exact components that are contained in the set $\widetilde{\Sigma}\big(\slashed{g}_0,\check{h}\big)$.
\end{remark}
Using the ambient metric construction, we can restate Theorem \ref{main theorem of the paper} in the following equivalent formulation: 

\begin{theorem}\label{main theorem of the paper ambient}
    For any even integer $n\geq4,$ we have a complete scattering theory for straight self-similar vacuum spacetimes determined by small data.  For any $M>0$ large enough there exists $\epsilon_0>0$ small enough, such that for $0<\epsilon\leq\epsilon_0$ we have:
    \begin{enumerate}
    \item\textbf{Existence and uniqueness of scattering states:} for any smooth $\epsilon$-small asymptotic data of order $M$ at $\{v=0\}$ or $\{u=0\}$ given by $\Sigma\big(\slashed{g}_0,\check{h}\big)\in B_{\epsilon}^M\big(\Sigma_{\mathrm{Minkowski}}\big)$, there exists a unique smooth straight self-similar vacuum solution $\big(\mathcal{M},g\big)$ in double null coordinates defined globally in $\{u<0,\ v>0\}\times S^n,$ which remains quantitatively close to Minkowski space, in the sense of Propositions \ref{region I bounds proposition}, \ref{region II bounds proposition},  and \ref{region III bounds proposition};
    \item\textbf{Asymptotic completeness:} any smooth straight self-similar vacuum spacetime in double null coordinates which is quantitatively close in the sense of Remark \ref{remark about cauchy data} to Minkowski space on a spacelike hypersurface $\{v=c|u|\}$, can be extended to the region $\{u<0,\ v>0\}\times S^n$ and induces smooth scattering data $\big(\slashed{g}_0,\check{h}\big)$ at $\{v=0\}$ and $\big(\underline{\slashed{g}_0},\underline{\check{h}}\big)$ at $\{u=0\}$. Moreover, the solution extends to $\big((-\infty,0]\times[0,\infty)\backslash\{(0,0)\}\big)\times S^n$ as a weak solution of the Einstein vacuum equations \eqref{actual vacuum equations};
    \item\textbf{Existence of a scattering map with quantitative estimates:} there exists a constant $C_M>0$ independent of $\epsilon,$ such that we have a well-defined scattering map taking the asymptotic data at $\{v=0\}$ to asymptotic data at $\{u=0\}$:
    \begin{equation}\label{definition of scattering map intro ambient}
        \mathcal{S}:B_{\epsilon}^M\big(\Sigma_{\mathrm{Minkowski}}\big)\rightarrow B_{C_M\epsilon}^M\big(\Sigma_{\mathrm{Minkowski}}\big),\ \mathcal{S}\Big(\Sigma\big(\slashed{g}_0,\check{h}\big)\Big)=\Sigma\big(\underline{\slashed{g}_0},\underline{\check{h}}\big).
    \end{equation}
    The scattering map is locally invertible and locally Lipschitz  at $\Sigma_{\mathrm{Minkowski}},$ in the sense that it satisfies the quantitative estimates:
    \begin{equation}\label{main estimate for scattering map intro ambient}
        \Big\|\mathcal{S}\big(\Sigma(\slashed{g}_0,\check{h})\big)\Big\|_M\leq C_M\Big\|\Sigma\big(\slashed{g}_0,\check{h}\big)\Big\|_M.
    \end{equation}
\end{enumerate}
\end{theorem}

\begin{remark}
    As in Theorem \ref{main theorem of the paper}, the scattering map result is sharp and avoids any "derivative loss", since in \eqref{definition of scattering map intro ambient} and \eqref{main estimate for scattering map intro ambient} we use the same Sobolev-type norm to measure the smallness of asymptotic data at $\mathcal{I}^-$ and $\mathcal{I}^+$.
\end{remark}

\begin{remark}
    The ambient metric formulation is convenient in order to establish the existence, uniqueness of scattering states, and asymptotic completeness. However, one could work directly at the level of asymptotically de Sitter spaces of the form \eqref{general form for g tilde}. In this case it is convenient to use the time coordinate $\tau=e^T/2$. This approach is present in Section~\ref{model forward direction section}, Section~\ref{model backward direction section}, and \cite{Cwave}, which represent the main step in proving the sharp estimate \eqref{main estimate for scattering map intro ambient}.
\end{remark}

\subsection{Outline of the Proof}\label{proof outline section}

We present the main steps in the proof of Theorem \ref{main theorem of the paper ambient}. We assume familiarity with the basics of the double null formalism, and we refer the reader unfamiliar with these notions to read Section~\ref{set up section} for a detailed introduction. 

We recall that in this formalism, the Einstein vacuum equations \eqref{actual vacuum equations} can be written as a system of equations for the Ricci coefficients 
denoted by $\psi$ and the curvature components denoted by $\Psi$. The system has the following schematic form (see \cite{selfsimilarvacuum}):
\begin{equation}\label{general form for EVE in double null}
    \begin{cases}
        \nabla_3\psi=\Psi+\psi\cdot\psi,\ \nabla_4\psi=\Psi+\psi\cdot\psi,\\
        \nabla_3\Psi_1=\mathcal{D}\Psi_2+\psi\cdot\Psi,\ \nabla_4\Psi_2=-\mathcal{D}^*\Psi_1+\psi\cdot\Psi
    \end{cases}
\end{equation}
where $\mathcal{D}, \mathcal{D}^*$ are adjoint differential operators on $S^n,$ and $\nabla_3,\nabla_4$ are covariant derivatives in the $e_3=\partial_u, e_4=\partial_v$ directions. We denote by $\nabla$ the projection to the tangent space of $S^n$ of the covariant derivative in any direction tangent to $S^n$, and we refer to this differential operator as an angular derivative. The operators $\nabla_3,\nabla_4$, and $\nabla$ will be used below as commutators to obtain systems of equations with a similar form to \eqref{general form for EVE in double null}. We also point out that the system \eqref{general form for EVE in double null} has some simplifications, due to the special straight structure of the metric $g$ which has constant lapse and vanishing shift vector.

\subsubsection{Existence and Uniqueness of Scattering States}\label{stability intro section} The first statement of Theorem \ref{main theorem of the paper ambient} consists of global existence and quantitative estimates of the solution in the $(n+2)$-dimensional region $\{u<0,\ v>0\}\times S^n,$ given small scattering data at $\{v=0\}.$ We prove this in Theorems~\ref{stability of de sitter theorem in section} and \ref{propagation of regularity theorem}. In the original $(n+1)$-dimensional formulation, this result represents the global stability of de Sitter space with small scattering data at $\mathcal{I}^-.$ We point out that the proof of \cite{Ringstrom1} does not apply in the case of scattering data; additionally, we prefer to prove the needed stability result in the ambient metric setting, in order to obtain the estimates required for the rest of our proof.

We remark that the stability result that we prove at this stage is not optimal in terms of the smallness assumed on the initial data. For our purposes, we notice that $\Sigma\big(\slashed{g}_0,\check{h}\big)\in B_{\epsilon}^M\big(\Sigma_{\mathrm{Minkowski}}\big)$ implies:
\begin{equation}\label{smallness on initial data introduction}
    \big\|\slashed{g}_0^*\big\|_{\mathring{H}^{M}(S^n)}+\big\|\mathcal{O}\big\|_{H^{M}(S^n)}+\big\|\check{h}\big\|_{H^{M}(S^n)}\leq\epsilon,
\end{equation}
where $\mathring{H}^{M}(S^n)$ is the Sobolev space with respect to $\slashed{g}_{S^n}$ and $H^{M}(S^n)$ is the Sobolev space with respect to $\slashed{g}_0.$ For this part of the argument we use the smallness condition (\ref{smallness on initial data introduction}) instead of $\Sigma\big(\slashed{g}_0,\check{h}\big)\in B_{\epsilon}^M\big(\Sigma_{\mathrm{Minkowski}}\big)$. However, we point out that in order to prove the sharp estimate for the scattering map (\ref{main estimate for scattering map intro ambient}), we will need a more detailed analysis which makes use of the exact structure of $\Sigma\big(\slashed{g}_0,\check{h}\big).$

The strategy of the proof follows similar steps to \cite{nakedsing}. In Section~\ref{stability dS section}, we carry out a bootstrap argument and construct the solution in the following regions one at a time, for $\underline{v}>0$ sufficiently small:
\[I=\bigg\{0\leq\frac{v}{|u|}\leq\underline{v}\bigg\}\times S^n,\ II=\bigg\{\underline{v}\leq\frac{v}{|u|}\leq\underline{v}^{-1}\bigg\}\times S^n,\ III=\bigg\{\underline{v}^{-1}\leq\frac{v}{|u|}\bigg\}\times S^n\]
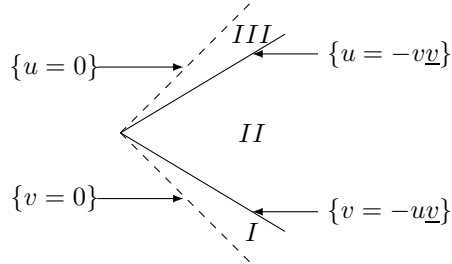
\begin{figure}[H]
\begin{center}
\begin{tikzpicture}[scale=0.875]
    \draw[-latex] (-0.35,3) -- (0.96,3);
    \draw (-1,3) node {$\{u=0\}$};
    \draw[dashed] (0,2) -- (2,4);
    \draw[dashed] (0,2) -- (2,0);
    \draw[-latex] (-0.35,1) -- (0.96,1);
    \draw (-1,1) node {$\{v=0\}$};
    \draw (2,3.5) node {$III$};
    \draw (2,2) node {$II$};
    \draw (2,0.5) node {$I$};
    \draw[latex-] (2,3.2) -- (3,3.2);
    \draw (3,3.2)  node[anchor=west] {$\{u=-v\underline{v}\}$};
    \draw[latex-] (2,0.8) -- (3,0.8);
    \draw (3,0.8)  node[anchor=west] {$\{v=-u\underline{v}\}$};
    \draw (0,2) -- (2.5,3.5);
    \draw (0,2) -- (2.5,0.5);
    \end{tikzpicture}
\end{center}
\caption{The decomposition into the regions $I,\ II,$ and $III$}
\label{fig:regions}
\end{figure}
The bounds that we prove for the solution need to be consistent with self-similarity, as in \cite{selfsimilarvacuum}. In region~I we have the scaling for the Ricci coefficients $|\psi^*|\sim\epsilon|u|^{-1}$ and for the curvature components $|\Psi|\sim\epsilon|u|^{-2}.$ Moreover, each $\nabla,\ \nabla_3,$ or $\nabla_4$ derivative that we apply to the double null unknowns raises their homogeneity by one, implying schematic self-similar bounds of the form:
\begin{equation}\label{self-similar bounds}
    \big|\nabla^i\nabla_4^j\nabla_3^k\psi^*\big|\lesssim_{i,j,k}\epsilon|u|^{-1-i-j-k},\ \big|\nabla^i\nabla_4^j\nabla_3^k\Psi\big|\lesssim_{i,j,k}\epsilon|u|^{-2-i-j-k}.
\end{equation}
The expected bounds in regions~II and III are similar for most double null quantities, replacing $|u|$ with $v.$

An essential aspect of the problem is that we can take at most $\frac{n-4}{2}$ $\nabla_4$ derivatives of the double null quantities in a neighborhood of $\{v=0\},$ and similarly for $\{u=0\}.$ This results from the presence of the obstruction term $\mathcal{O}$ in the Fefferman-Graham expansion of the solution near $\{v=0\}.$ We denote by $\Psi^G$ the curvature components different than $\alpha,$ and similarly by $\underline{\Psi}^G$ the curvature components different than $\underline{\alpha}.$ We have that all double null quantities $\nabla^i\nabla_4^j\nabla_3^k\psi$ and $\nabla^i\nabla_4^j\nabla_3^k\Psi^G$ extend to $\{v=0\}$ for $j\leq\frac{n-4}{2}$. However, $\alpha$ is mildly singular at $\{v=0\}$:
\begin{equation}\label{expansion for alpha introduction}
    |u|^{\frac{n-4}{2}}\nabla_4^{\frac{n-4}{2}}\alpha=\log\big(v/|u|\big)\mathcal{O}+h+O\Big(v/|u|\log\big(v/|u|\big)\Big),
\end{equation}
where $h$ is obtained from $\check{h}$ by subtracting a linear factor of $\mathcal{O}.$ In order to address this issue, when proving self-similar bounds for $\alpha$ in region~I we need to subtract off the singular term above.

Taking into consideration the necessary renormalization, the argument of \cite{selfsimilarvacuum} applies to prove existence and self-similar bounds in region~I, where we allow up to $N$ angular derivatives on the double null quantities, and $M=N+O(n),\ M>N.$ Similarly, the argument of \cite{nakedsing} applies to prove existence and self-similar bounds in region~II, again allowing up to $N$ angular derivatives on the double null quantities.

The main part of the proof of the first statement in Theorem \ref{main theorem of the paper ambient} involves showing existence and self-similar bounds in region~III in Section~\ref{region III section}. The difficult aspect is that we expect $\underline{\alpha}$ to be singular at $\{u=0\}$, as implied by the local well-posedness theory. Unlike in region~I, we do not determine a priori the singular part of $\underline{\alpha}$ given by the obstruction tensor, so we cannot subtract it off. We also notice that unlike the approach in \cite{nakedsing}, we cannot work with the reduced Bianchi system, which would remove $\underline{\alpha}$, as this does not work in the higher dimensional setting. Our solution is to propagate estimates for $\underline{\alpha}$ consistent with it blowing up at $\{u=0\}$, aided by the fact that in the straight higher dimensional case the singular behavior of $\underline{\alpha}$ is more mild than in \cite{nakedsing}. For the remaining double null quantities, we expect to prove regular self-similar bounds, similar to region~I.

We briefly explain how to prove energy estimates for the curvature components $\Psi$, as part of the bootstrap argument in Section~\ref{region III section}. We recall that the curvature components can be grouped into the Bianchi pairs $(\alpha,\nu),(\nu,R),(R,\underline{\nu}),(\underline{\nu},\underline{\alpha}).$ The Bianchi pair $(\underline{\nu},\underline{\alpha})$ satisfies the schematic equations:
\begin{align*}
    &\nabla_3\underline{\nu}_{ABC}=-2\nabla_{[A}\underline{\alpha}_{B]C}+\ldots\\
    &\nabla_4\underline{\alpha}_{AB}+\frac{n}{2v}\underline{\alpha}_{AB}=-\nabla^C\underline{\nu}_{C(AB)}+\ldots
\end{align*}
As in \cite{nakedsing}, for $0<q\ll p\ll1$ we conjugate the equations with $w=v^{\frac{3}{2}-p}|u|^{p-q}:$
\begin{align*}
    &\nabla_3w\underline{\nu}_{ABC}+\frac{p-q}{|u|}w\underline{\nu}_{ABC}=-2\nabla_{[A}w\underline{\alpha}_{B]C}+\ldots\\
    &\nabla_4w\underline{\alpha}_{AB}+\bigg(\frac{n-3}{2}+p\bigg)\frac{w}{v}\underline{\alpha}_{AB}=-\nabla^Cw\underline{\nu}_{C(AB)}+\ldots
\end{align*}
The energy estimates are obtained by contracting the above equations with $w\underline{\nu}$ and $w\underline{\alpha}$, integrating by parts, and multiplying by $|u|^{2q}$. The lower order terms imply the presence of bulk terms with favorable sign in the estimates. These bulk terms are even better in the case of $\underline{\nu},$ since $|u|/v$ is small in region~III. We notice that the same argument also applies when commuting with angular derivatives and up to $\frac{n-4}{2}$ $\nabla_3$ derivatives, which ensures that the lower order terms imply good bulk terms. Commuting with a high number of angular derivatives $\nabla^i$ simplifies our treatment of the error terms on the right hand side. The weight $w$ implies that the best estimate that we can prove for $\underline{\alpha}$ is:
\[\big\|\nabla^i\nabla_3^j\underline{\alpha}\big\|_{L^{2}(S^n)}\lesssim\epsilon^{1-2\delta}|u|^{-p}\cdot|v|^{-2-i-j+p},\]
where $j\leq\frac{n-4}{2},\ i<N,$ $\delta>0$ is a small constant, and the implicit constant in the above inequality is independent of $\epsilon$ and $\underline{v}.$ This bound is consistent with the singular behavior of $\alpha$ at $\{u=0\}$. The good bulk term obtained for $\underline{\nu}$ allows us to control the bulk term in the energy estimates for the Bianchi pair $\big(R,\underline{\nu}\big),$ and we similarly obtain energy estimates for all the curvature components. Moreover, the stronger control that we obtain for the bulk terms in the case of the quantities $\underline{\Psi}^G$ allows us to prove regular self-similar bounds. The simple transport structure of the equations for the Ricci coefficients also implies their respective self-similar bounds. We point out that we close the bootstrap assumptions using the smallness of $\underline{v}.$

Finally, we prove a standard propagation of regularity result in Theorem~\ref{propagation of regularity theorem} showing that if the scattering data is also smooth, the global solution obtained is smooth. 

\subsubsection{Asymptotic Completeness}\label{asymptotic completeness intro section} We consider a smooth straight self-similar vacuum spacetime in double null coordinates, which is quantitatively close to Minkowski space on a spacelike hypersurface $\{v=c|u|\}\times S^n$ with $\underline{v}\leq c\leq\underline{v}^{-1},$ in the sense that it satisfies the above self-similar bounds (\ref{self-similar bounds}) with $j,k\leq\frac{n-4}{2}$ and $i\leq N,$ where $M=N+O(n)$, $M>N$ for some large enough $N.$ In Theorem~\ref{asymptotic completeness in section theorem} we prove the second statement of Theorem \ref{main theorem of the paper ambient}, showing that the spacetime can be extended globally to $\{u<0,v>0\}\times S^n,$ and that it induces scattering data at $\{v=0\}$ and $\{u=0\}.$ The first part follows from Section~\ref{stability intro section}, as we remark that the analysis of \cite{nakedsing} in region~II applies in this setting as well, and we can repeat our analysis in region~III. We notice that in region~I the spacetime will satisfy similar bounds to the ones in region~III, as we have no information about the scattering data at this point.

In Section~\ref{asymptotic completeness section} we prove the existence of induced smooth scattering data $\big(\underline{\slashed{g}_0},\underline{\check{h}}\big)$ at $\{u=0\}$, as the case of $\{v=0\}$ is analogous. The strategy is to compute the terms in the expansion of $\slashed{g}$ at $\{u=0\}$ up to order $n/2.$

We first prove that certain regular quantities can be extended to $\{u=0\}$ and satisfy compatibility relations. These consist of up to $\frac{n-6}{2}$ $\nabla_3$ derivatives of $\underline{\alpha}$, up to $\frac{n-4}{2}$ $\nabla_3$ derivatives of $\underline{\Psi}^G$, up to $\frac{n-4}{2}$ $\nabla_3$ derivatives of $\psi$, and up to $\frac{n-2}{2}$ $\mathcal{L}_3$ derivatives of $\slashed{g}$, together with at most $N$ angular derivatives of these tensors. All these quantities satisfy a $\nabla_3$ equation, where we control the right hand side using the bounds proved in region~III. We prove that these tensors are in $W_u^{1,1}\big([-1,0]\big)L^2(S^n)$, so they can be extended to $\{u=0\}$. In particular, we compute $\underline{\slashed{g}_0}$ the induced metric on $\{u=0,v=1\}\times S^n.$ Evaluating the above $\nabla_3$ equations at $\{u=0\}$ implies that the above regular quantities are determined in terms of $\underline{\slashed{g}_0}$ by the compatibility relations of \cite{selfsimilarvacuum}. Equivalently, we obtain that the first $\frac{n-2}{2}$ terms in the expansion of $\slashed{g}$ are determined by $\underline{\slashed{g}_0}$ via the compatibility relations.

The next step is to compute the singular component of $\nabla_3^{\frac{n-4}{2}}\underline{\alpha}$, in order to obtain the obstruction tensor $\underline{\mathcal{O}}$ induced on $\{u=0\}.$ Using self-similarity, we can write the $\nabla_4$ Bianchi equation for $\nabla_3^{\frac{n-4}{2}}\underline{\alpha}$ schematically as:
\[\partial_u\big(v^{\frac{n-4}{2}}\nabla_3^{\frac{n-4}{2}}\underline{\alpha}\big)=-\frac{1}{u}\mathcal{E}_1+\mathcal{E}_2=\frac{1}{u}\underline{\mathcal{O}}+\frac{1}{|u|}\Big(\mathcal{E}_1-\mathcal{E}_1|_{u=0}\Big)+\mathcal{E}_2,\]
where we defined $\underline{\mathcal{O}}=-\mathcal{E}_1|_{u=0}$ which is independent of $v$ and can be computed in terms of the regular quantities at $\{u=0\}.$ We obtain that $\underline{\mathcal{O}}$ can be computed in terms of $\underline{\slashed{g}_0},$ and we prove it satisfies the compatibility relation of \cite{selfsimilarvacuum} which implies that it represents the obstruction tensor of $\underline{\slashed{g}_0}.$ 

The final step is to compute the induced tensor $\underline{\check{h}}$ on $\{u=0\}.$ The error term in the above equation is in $L^1_u\big([-\underline{v}v,0]\big)L^2(S^n),$ so we can integrate the equation to get:
\[v^{\frac{n-4}{2}}\nabla_3^{\frac{n-4}{2}}\underline{\alpha}-\underline{\mathcal{O}}\log\big(|u|/v\big)\in W^{1,1}_u\big([-\underline{v}v,0]\big)L^2(S^n).\]
We define the symmetric traceless 2-tensor $\underline{h}$ which is independent of $u$ and $v$ to be the limit at $\{u=0\}$ of the above expression. We obtain the expansion for $\nabla_3^{\frac{n-4}{2}}\underline{\alpha}:$
\[v^{\frac{n-4}{2}}\nabla_3^{\frac{n-4}{2}}\underline{\alpha}=\log\big(|u|/v\big)\underline{\mathcal{O}}+\underline{h}+O\big(|u|^{1-p}/v^{1-p}\big).\]
As before, $\underline{\check{h}}$ is obtained from $\underline{h}$ by adding a certain linear factor of $\underline{\mathcal{O}}$.

The proof of asymptotic completeness is concluded by applying Theorem \ref{RSR theorem} of \cite{selfsimilarvacuum}, in order to show that $\big(\underline{\slashed{g}_0},\underline{\check{h}}\big)$ represents the induced scattering data at $\{u=0\}$. This result also implies that $\underline{h}$ satisfies the straightness condition, since the spacetime $\big(\mathcal{M},g\big)$ is straight. Finally, we remark that the spacetime $\big(\mathcal{M},g\big)$ is smooth, but due to the mild singular behavior at $\{u=0\}$ and $\{v=0\}$ it extends to $\big((-\infty,0]\times[0,\infty)\backslash\{(0,0)\}\big)\times S^n$ as a weak solution of the Einstein vacuum equations \eqref{actual vacuum equations}, similarly to the solutions of \cite{selfsimilarvacuum}.

\subsubsection{The Scattering Map}\label{scattering map section intro} The third statement of Theorem \ref{main theorem of the paper ambient} consists of constructing the scattering map $\mathcal{S}$ according to (\ref{definition of scattering map intro ambient}), which satisfies the sharp estimate (\ref{main estimate for scattering map intro ambient}). Establishing the sharp result for the scattering map represents the most challenging part of our work. At top order, the proof relies on estimates proved in our companion paper \cite{Cwave} for the two model systems of wave equations introduced below. In the present paper, we already illustrate this part of the argument for a toy problem that contains the main difficulties. This introduces the reader to all the main ideas and facilitates the understanding of the proofs for the full model systems of wave equations, which are treated in \cite{Cwave}. The reader might wish to return to this section for assistance while reading the proof in Sections~\ref{model systemmm}-\ref{scattering map section}.

We first explain the preliminary scattering result obtained from the proofs of existence, uniqueness of scattering states, and asymptotic completeness in Theorems~\ref{stability of de sitter theorem in section} and \ref{asymptotic completeness in section theorem}. For smooth scattering data at $\{v=0\}$ which satisfies the smallness condition:
\[\big\|\slashed{g}^*_0\big\|_{\mathring{H}^{M}(S^n)}+\big\|\mathcal{O}\big\|_{H^{M}(S^n)}+\big\|\check{h}\big\|_{H^{M}(S^n)}\leq\epsilon,\]
we obtain a smooth straight self-similar vacuum spacetime $\big(\mathcal{M},g\big)$ in double null coordinates defined in the region $\{u<0,\ v>0\}\times S^n$. This induces smooth scattering data at $\{u=0\}$ satisfying the smallness condition:
\[\big\|\underline{\slashed{g}_0}^*\big\|_{\mathring{H}^{N}(S^n)}+\big\|\underline{\mathcal{O}}\big\|_{H^{N}(S^n)}+\big\|\underline{\check{h}}\big\|_{H^{N}(S^n)}\leq C\epsilon^{1-2\delta},\]
where $M=N+O(n),\ M>N$ for large enough $N,$ $\delta>0$ is a small constant, and $C>0$ is a constant independent of $\epsilon$. This confirms our previous claim in Section~\ref{stability intro section} that the stability result proved initially is not optimal in terms of the smallness assumptions on the initial data \eqref{smallness on initial data introduction}. In particular, the above result cannot give sharp estimates for the scattering map at this stage, since we only get control of the $H^N$ norm of the solution at $\{u=0\},$ despite starting with bounds on the $H^M$ norm of the solution at $\{v=0\},$ with $M>N.$ This issue is a fundamental feature of the problem, already present at the level of the wave equation (\ref{wave equation on dS background}) which was analysed in \cite{linearwave}. 

In order to prove a sharp scattering result, we must construct a notion of an asymptotic initial data set $\Sigma\big(\slashed{g}_0,\check{h}\big)$ and asymptotic initial data norm $\|\cdot\|_M$, which in the small data case allow us to prove the estimate:
\begin{equation}\label{sharp scattering estimate intro}
    \Big\|\Sigma\big(\underline{\slashed{g}_0},\underline{\check{h}}\big)\Big\|_M\leq C_M\Big\|\Sigma\big(\slashed{g}_0,\check{h}\big)\Big\|_M.
\end{equation}
This requires a detailed analysis of the problem which exploits the structure of the solution, and ultimately relies on replacing $h$ with the renormalized tensor $\mathfrak{h}=h-2\big(\log\nabla\big)\mathcal{O}$. 

Once we prove that for any $\Sigma\big(\slashed{g}_0,\check{h}\big)\in B_{\epsilon}^M\big(\Sigma_{\mathrm{Minkowski}}\big)$ the estimate (\ref{sharp scattering estimate intro}) holds, it is straightforward in Section~\ref{scattering map section} to construct the scattering map $\mathcal{S}$ satisfying (\ref{definition of scattering map intro ambient}) and (\ref{main estimate for scattering map intro ambient}), by also using the existence, uniqueness of scattering states, and asymptotic completeness. We outline the proof of (\ref{sharp scattering estimate intro}) for the rest of the section.

We introduce the norm $\Xi_M$ in Section~\ref{forward direction full system section}, representing the energy of the solution on $\{u=-1,v=1\}\times S^n.$ One remarkable aspect is that the norm $\Xi_M$ has improved angular control on the solution compared to the asymptotic data norm, by gaining half of a derivative. We have schematically that:
\begin{align*}
    &\Big\|\Sigma\big(\slashed{g}_0,\check{h}\big)\Big\|_M^2=\big\|\mathcal{O}\big\|_{H^{M+1}(S_{-1,0})}^2+\big\|\mathfrak{h}\big\|_{H^{M+1}(S_{-1,0})}^2+\big\|\nabla_4^{\frac{n-4}{2}}\Psi^G\big\|_{H^{M+1}(S_{-1,0})}^2+\ldots\\
    &\Xi_M^2=\sum_{i+j=0}^{\frac{n-4}{2}}\big\|\nabla^M\nabla_3^i\nabla_4^j\Psi\big\|^2_{H^{3/2}(S_{-1,1})}+\sum_{i+j=0}^{\frac{n-2}{2}}\big\|\nabla^M\nabla_3^i\nabla_4^j\Psi\big\|^2_{H^{1/2}(S_{-1,1})}+\ldots
\end{align*}
In order to prove (\ref{sharp scattering estimate intro}), it suffices to show that:
\begin{equation}\label{equivalence of asympt data and finite time intro}
    \Big\|\Sigma\big(\slashed{g}_0,\check{h}\big)\Big\|_M\lesssim\Xi_M\lesssim\Big\|\Sigma\big(\slashed{g}_0,\check{h}\big)\Big\|_M,
\end{equation}
where the implicit constant depends on $M$ but is independent on $\epsilon.$ Once we establish \eqref{equivalence of asympt data and finite time intro}, we complete the proof of (\ref{sharp scattering estimate intro}) in Section~\ref{scattering map section}, since by changing $(u,v)$ to $(-v,-u)$ we also obtain:
\[\Big\|\Sigma\big(\underline{\slashed{g}_0},\underline{\check{h}}\big)\Big\|_M\lesssim\Xi_M\lesssim\Big\|\Sigma\big(\underline{\slashed{g}_0},\underline{\check{h}}\big)\Big\|_M.\]

\begin{remark}
    The strategy used in the proof of the sharp scattering result is similar to our approach in \cite{linearwave} for the wave equation (\ref{wave equation on dS background}). Based on the analogy between the scalar field $\phi$ and $\slashed{g},$ one could expect that a notion of asymptotic initial data could consist only of $\slashed{g}_0$ and $\mathfrak{h}$ which satisfy the smallness assumption:
    \begin{equation}\label{potential smallness condition}
        \big\|\slashed{g}^*_0\big\|_{\mathring{H}^{M+n+1}(S^n)}+\big\|\mathfrak{h}\big\|_{H^{M+1}(S^n)}\leq\epsilon.
    \end{equation}
    Using the compatibility relations, this condition implies indeed that $\Sigma\big(\slashed{g}_0,\check{h}\big)\in B_{C_M\epsilon}^M\big(\Sigma_{\mathrm{Minkowski}}\big).$ In the case of (\ref{wave equation on dS background}) we also have that $\mathcal{O}\sim\Delta^{n/2}\phi_0+\ldots,$ which recovers the estimate for $\phi_0$ at top order. However, in the current situation the obstruction tensor does not satisfy the needed ellipticity property, so the smallness of $\Sigma\big(\slashed{g}_0,\check{h}\big)$ does not imply (\ref{potential smallness condition}). According to \cite{Fefferman-Graham}, we have at top order that $\mathcal{O}\sim\Delta^{n/2-2}B,$ where $B$ is the Bach tensor of $\slashed{g}_0.$ This operator is elliptic under a conformal change of the metric, see \cite{BachFlat, ellipticO}, but in our case we cannot control the conformal factor. Consequently, the smallness condition (\ref{potential smallness condition}) cannot be used to prove a sharp scattering result.

\end{remark}

\paragraph{Estimates from $\{v=0\}$ to $\{v=-u\}.$} We prove that $\Xi_M\lesssim\big\|\Sigma(\slashed{g}_0,\check{h})\big\|_M$ in Theorem~\ref{main theorem forward direction full system}, establishing the first inequality in \eqref{equivalence of asympt data and finite time intro}. According to Section~\ref{model systemmm}, we can rewrite the system of Bianchi equations restricted to $\{u=-1\}$ as a system of wave equations for any $0\leq m\leq M,$ $0\leq l\leq\frac{n}{2}-2:$
\begin{equation}\label{model wave system 1 intro}
    \begin{cases}
        v\nabla_4^2\nabla^m\nabla_4^l\alpha+\Big(3+l-\frac{n}{2}\Big)\nabla_4\nabla^m\nabla_4^l\alpha-\Delta\nabla^m\nabla_4^l\alpha=\psi\nabla^{m+1}\nabla_4^l\Psi+Err_{ml}^{\Psi} \\
        v\nabla_4^2\nabla^m\nabla_4^l\Psi^G+\Big(3+l-\frac{n}{2}\Big)\nabla_4\nabla^m\nabla_4^l\Psi^G-\Delta\nabla^m\nabla_4^l\Psi^G=\psi\nabla^{m+1}\nabla_4^l\Psi+Err_{ml}^{\Psi}
    \end{cases}
\end{equation}
Moreover, the solutions satisfy the expansions at $\{v=0\}$:
\begin{align*}
    &\nabla_4^l\Psi^G=\big(\nabla_4^l\Psi^G\big)\big|_{(-1,0)}+O(v),\ \nabla_4^l\alpha=\big(\nabla_4^l\alpha\big)\big|_{(-1,0)}+O\big(v|\log v|^2\big)\text{ for }l\leq\frac{n-6}{2},\\
    &\nabla_4^{\frac{n-4}{2}}\Psi^G=\big(\nabla_4^{\frac{n-4}{2}}\Psi^G\big)\big|_{(-1,0)}+O\big(v|\log v|^2\big),\ \nabla_4^{\frac{n-4}{2}}\alpha=\mathcal{O}\log v+h+O\big(v|\log v|^2\big).
\end{align*}
We prove the main estimates for the system (\ref{model wave system 1 intro}) in Theorem~\ref{main theorem forward direction full system}. The desired inequality will follow, since the initial data energy is controlled by $\big\|\Sigma(\slashed{g}_0,\check{h})\big\|_M,$ whereas the energy at $(-1,1)$ controls $\Xi_M.$ The top order estimates require control of the quantity:
\[\mathcal{T}=v^2\big\|\nabla_{4}\nabla^M\nabla_4^{\frac{n-4}{2}}\Psi\big\|^2_{H^{1/2}}+v\big\|\nabla^M\nabla_4^{\frac{n-4}{2}}\Psi\big\|^2_{H^{3/2}}+\big\|\nabla_4^{\frac{n-4}{2}}\Psi^G\big\|^2_{H^{M+1}}.\]
Bounding this energy represents the fundamental part of the proof, as this captures the need to renormalize $h$ and it implies the improvement in the number of angular derivatives controlled. To prove this in Section~\ref{forward direction full system section}, we treat the system (\ref{model wave system 1 intro}) as a linear system on the background obtained by restricting the metric $g$ to the null cone $\{u=-1\},$ with a general inhomogeneous term in place of $Err_{ml}^{\Psi}.$ We refer to this as the \textbf{first model system}, introduced in Section~\ref{model systemmm}, and we explain below how we prove estimates for it in Section \ref{model forward direction section} and \cite{Cwave}. Once we bound the top order quantity $\mathcal{T}$ in terms of the initial data energy and the error terms, the remaining bounds follow using more standard energy estimates for the system (\ref{model wave system 1 intro}). As before, the presence of the nonlinear error terms $Err_{ml}^{\Psi}$ does not create significant difficulties since we commuted with a high number of angular derivatives, so these terms are essentially linear.

\paragraph{Estimates from $\{v=-u\}$ to $\{v=0\}.$} We prove that $\big\|\Sigma(\slashed{g}_0,\check{h})\big\|_M\lesssim\Xi_M$ in Theorem~\ref{main theorem backward direction full system}, establishing the second inequality in \eqref{equivalence of asympt data and finite time intro}. According to Section~\ref{model systemmm}, we can also rewrite the system of Bianchi equations restricted to $\{u=-1\}$ for any $0\leq m\leq M,$ $0\leq l\leq\frac{n}{2}-2$ as:
\begin{equation}\label{model wave system 2 intro}
    \begin{cases}
        v\nabla_4^2\nabla^m\nabla_4^l\alpha+\Big(3+l-\frac{n}{2}\Big)\nabla_4\nabla^m\nabla_4^l\alpha-\Delta\nabla^m\nabla_4^l\alpha=\psi\nabla^{m+1}\nabla_4^l\Psi+Err_{ml}^{\Psi} \\
        v\nabla_4^2\nabla^m\nabla_4^l\Psi^G+\Big(2+l-\frac{n}{2}\Big)\nabla_4\nabla^m\nabla_4^l\Psi^G-\Delta\nabla^m\nabla_4^l\Psi^G=\sum_{\Psi^G_0}\psi\nabla^{m+1}\nabla_4^l\Psi^G_0+Err_{ml}^{\Psi}.
    \end{cases}
\end{equation}
Once again, the solutions satisfy the above expansions at $\{v=0\}$. In Section~\ref{backward direction full system section}, we prove estimates with initial data at $(-1,1)$ controlled by $\Xi_M,$ and the energy at $(-1,0)$ controlling $\big\|\Sigma(\slashed{g}_0,\check{h})\big\|_M.$ At top order, we bound:
\[v^2\big\|\nabla_{4}\nabla^M\nabla_4^{\frac{n-4}{2}}\alpha\big\|^2_{H^{1/2}}+v\big\|\nabla^M\nabla_4^{\frac{n-4}{2}}\alpha\big\|^2_{H^{3/2}}+\big\|\nabla^M\nabla_4^{\frac{n-4}{2}}\Psi^G\big\|^2_{H^{3/2}}+v\big\|\nabla_4\nabla^M\nabla_4^{\frac{n-4}{2}}\Psi^G\big\|^2_{H^{1/2}}+\ldots\]
and the asymptotic quantities:
\[\big\|\mathcal{O}\big\|_{H^{M+1}(S_{-1,0})}^2+\big\|\mathfrak{h}\big\|_{H^{M+1}(S_{-1,0})}^2.\]
This step represents the essential part of the proof, since as above once we control the top order terms we can also obtain bounds for the remaining terms in $\big\|\Sigma(\slashed{g}_0,\check{h})\big\|_M$ and estimate the nonlinear error terms. The strategy is again to treat the system (\ref{model wave system 2 intro}) as a linear system on the background obtained by restricting $g$ to $\{u=-1\},$ with a general inhomogeneous term. We refer to this as the \textbf{second model system}, introduced in Section~\ref{model systemmm}, and we explain below how we prove estimates for it in Section \ref{model backward direction section} and \cite{Cwave}. 

\paragraph{Geometric Littlewood-Paley projections.} The analysis of the model systems needed for the top order estimates above requires the use of Littlewood-Paley projections. These provide a robust way of constructing frequency dependent multipliers and defining fractional derivatives, including the $\log\nabla$ operator present in the definition of $\mathfrak{h}.$ In dealing with the model systems we intend to use the same approach as in \cite{linearwave}. The new difficulty is that the metric $\slashed{g}_{v}$ induced by the background on the spheres $S_v=\{u=-1\}\times\{v\}\times S^n$ has a nontrivial time dependence, compared to the case of de Sitter space. This determines us to use the geometric Littlewood-Paley theory of \cite{geometricLP}, defined using the heat equation in Section~\ref{LP Section}. The LP projections used have standard properties, with additional difficulties arising from the fact that they are time dependent and do not satisfy exact orthogonality. Thus, the projections are only "almost orthogonal", and we have for any two families of LP projections $P_k, \widetilde{P}_k$:
\[\big\| P_k\widetilde{P}_{k'}F\big\|_{L^2}\lesssim2^{-4|k-k'|}\cdot\big\|F\big\|_{L^2}.\]
We also use a series of results from \cite[Section~2]{Cwave}, which employ the methods of \cite{geometricLP} to quantify in more detail the error terms caused by the time dependence of the metric and the almost orthogonality of the projections. For example, for a horizontal tensor $F$ we have the bound:
\begin{equation}\label{useful LP bound intro}
    \big\|[\nabla_4,P_k]F\big\|_{L^2}\lesssim\big\|\underline{\widetilde{P}}_kF\big\|_{L^2}+2^{-k}\big\|F\big\|_{L^2},
\end{equation}
where $\underline{\widetilde{P}}_k$ is a projection operator defined in Section~\ref{LP Section}. We notice that this bound is summable in $k.$

Additionally, it is essential that we use the following refined Poincaré inequality for any $k\geq0,\delta>0$:
\begin{equation}\label{poincare inequality intro}
    \big\|P_kF\big\|_{L^2}^2\lesssim\frac{1}{\delta}2^{-2k}\big\|\nabla P_kF\big\|_{L^2}^2+\delta\sum_{0\leq l<k}2^{-9k+7l}\big\|\nabla P_lF\big\|_{L^2}^2+\delta^{-1}2^{-4k}\big\|F\big\|_{L^2}^2.
\end{equation}
We contrast this with the weaker Poincaré inequality $\big\|P_kF\big\|_{L^2}\lesssim2^{-k}\big\|\nabla \widetilde{P}_kF\big\|_{L^2}$, where the presence of different projection operators is caused by the almost orthogonality of the projections. On the other hand, we notice that for \eqref{poincare inequality intro} the projection operators on the RHS have the same symbol as the one on the LHS, and all
the frequencies higher than $k$ are contained in the last term, which is lower order.

\paragraph{The first model system.} In Section \ref{model systemmm}, we write the system (\ref{model wave system 1 intro}) as a linear system on the background obtained by restricting the metric $g$ to the null cone $\{u=-1\},$ with a general inhomogeneous term. We use the notation $\Phi_0=\nabla_4^{\frac{n-4}{2}}\alpha$ and $\Phi_i=\nabla_4^{\frac{n-4}{2}}\Psi^G,$ and we obtain with respect to the new time variable $\tau=\sqrt{v}$ the system:
\begin{align}
    &\nabla_{\tau}\big(\nabla_{\tau}\nabla^m\Phi_0\big)+\frac{1}{\tau}\nabla_{\tau}\nabla^m\Phi_0-4\Delta\nabla^m\Phi_0=\psi\nabla^{m+1}\Phi+F_{m}^{0}\notag \\
    &\nabla_{\tau}\big(\nabla_{\tau}\nabla^m\Phi_i\big)+\frac{1}{\tau}\nabla_{\tau}\nabla^m\Phi_i-4\Delta\nabla^m\Phi_i=\psi\nabla^{m+1}\Phi+F_{m}^{i}\notag \\
    &\Phi_0=2\mathcal{O}\log\tau+h+O\big(\tau^2|\log\tau|^2\big),\ \Phi_i=\Phi_i^0+O\big(\tau^2|\log\tau|^2\big).\notag
\end{align}
where the covariant angular derivatives are with respect to the metric $\slashed{g}_{\tau}:=\slashed{g}_{u=-1,v=\tau^2}$ induced on $S_{\tau}$. 

\begin{remark}\label{remark about model systems intro}
    We notice that in terms of the metric $\widetilde{\slashed{g}}$ from \eqref{general form for g tilde} we have $\slashed{g}_{\tau}=\widetilde{\slashed{g}}(\log(2\tau)).$ For any asymptotically de Sitter space of the form \eqref{general form for g tilde} we consider the new time coordinate $\tau=e^T/2$. We point out that one can also recover the first model system by commuting the Einstein equations \eqref{EVE} $n/2$ times with the vectorfield $\frac{1}{2\tau}\partial_{\tau}.$ A similar approach also holds for the second model system below. We further explain this perspective in \cite{Cwave}.
\end{remark}

The estimates needed at top order for the system (\ref{model wave system 1 intro}) are proved at the level of the first model system in Section \ref{model forward direction section} and \cite{Cwave}. We decompose $\Phi_0$ into its singular and regular components, similarly to \cite{linearwave}. In the present paper we illustrate in Section \ref{model forward direction section} how to prove the top order estimates in Theorem~\ref{forward direction main result theorem} for the singular component of $\Phi_0$, which decouples from the rest of the system. The regular component is better behaved at $\tau=0,$ and can be treated similarly to the tensors $\Phi_i.$ We refer the reader to \cite[Section~3]{Cwave} for a complete proof of Theorem~\ref{forward direction main result theorem general}, which deals with the full system. However, we point out that the main difficulties are already present in the analysis of the singular component of $\Phi_0,$ so the proof of Theorem~\ref{forward direction main result theorem} assists in the understanding of Theorem~\ref{forward direction main result theorem general} in \cite{Cwave}. 

We define for each $m\leq M:$
\[\nabla^m\Phi_0=\big(\nabla^m\Phi_0\big)_Y+\big(\nabla^m\Phi_0\big)_J,\]
where we define the singular component $\big(\nabla^m\Phi_0\big)_Y$ to be the horizontal tensor that solves the linear equation:
\begin{align*}
    &\nabla_{\tau}\big(\nabla_{\tau}\big(\nabla^m\Phi_0\big)_Y\big)+\frac{1}{\tau}\nabla_{\tau}\big(\nabla^m\Phi_0\big)_Y-4\Delta\big(\nabla^m\Phi_0\big)_Y=\psi\nabla\big(\nabla^m\Phi_0\big)_Y\\
    &\big(\nabla^m\Phi_0\big)_Y({\tau})=2\nabla^m\mathcal{O}\log({\tau})+2(\log\nabla)\nabla^m\mathcal{O}+ O\big({\tau}^2|\log({\tau})|^2\big),
\end{align*}
and we also define $(\log\nabla)\nabla^m\mathcal{O}=\sum_{k\geq0}P_k^2\nabla^m\mathcal{O}\cdot\log2^k,\ \mathfrak{h}_m=\nabla^mh-2(\log\nabla)\nabla^m\mathcal{O}$. We remark that the regular component satisfies a similar equation to that of $\Phi_0$ and it has the expansion:
\[\big(\nabla^m\Phi_0\big)_J({\tau})=\mathfrak{h}_m+ O\big({\tau}^2|\log({\tau})|^2\big),\ \nabla_{\tau}\big(\nabla^m\Phi_0\big)_J({\tau})=O\big({\tau}|\log({\tau})|^2\big).\]
The notation for the regular and singular components is based on the similarities to the first and second Bessel functions $J_0,\ Y_0,$ as in the case of \cite{linearwave}. The need to renormalize the asymptotic data $h$ to $\mathfrak{h}$ follows from the analysis of the singular component.

In Theorem \ref{forward direction main result theorem}, we prove the following estimate for $\tau\in(0,1]$, with an implicit constant depending only on $M$:
\begin{equation}\label{main estimate alpha Y top order intro}
    \tau^2\big\|\nabla_{\tau}\big(\nabla^M\Phi_0\big)_Y\big\|^2_{H^{1/2}}+\tau^2\big\|\nabla\big(\nabla^M\Phi_0\big)_Y\big\|^2_{H^{1/2}}\lesssim\big\|\mathcal{O}\big\|^2_{H^{M+1}},
\end{equation}
\begin{equation}\label{practical estimate alpha Y top order intro}
    \sum_{m=0}^M\big\|\big(\nabla^m\Phi_0\big)_Y\big\|^2_{H^{1}}\lesssim\big(1+|\log\tau|^2\big) \big\|\mathcal{O}\big\|^2_{H^{M+1}}.
\end{equation}
We first obtain lower order estimates using standard energy estimates in Section~\ref{forward lower order estimates section}, in order to prove (\ref{practical estimate alpha Y top order intro}) for all $m<M.$ In order to obtain sharp estimates at top order, we must use the structure in the expansion of $\big(\nabla^M\Phi_0\big)_Y$, which can only be seen at the level of each LP projection. We have schematically for every $k\geq0:$
\[P_k\big(\nabla^M\Phi_0\big)_Y({\tau})=2P_k\nabla^M\mathcal{O}\log({2^k\tau})+l.o.t.+O\big({\tau}^2|\log({\tau})|^2\big).\]
As in the case of the linear wave equation on de Sitter space studied in \cite{linearwave}, we prove that $P_k\big(\nabla^M\Phi_0\big)_Y$ satisfies similar asymptotics to the second Bessel function $Y_0$ in terms of the new time variable $t=2^k\tau$. A quantitative version of this statement is proved in Section~\ref{forward top order estimates section} using suitable energy estimates in the low frequency regime $\tau\leq2^{-k-1}$, with data given by the asymptotic initial data, and the high frequency regime $\tau\in[2^{-k-1},1],$ with data at $\tau=2^{-k-1}$ given by the solution in the low frequency regime. We remark that the asymptotic behavior and the frequency dependent time of transition between the two regimes are responsible for the improvement in regularity: at $\tau=1$ we control $M+3/2$ derivatives of the solution in terms of $M+1$ derivatives of the asymptotic data.

The main new difficulties compared to \cite{linearwave} arise from the fact that the geometric LP projections are time dependent and do not satisfy exact orthogonality, as explained above. Using bounds such as (\ref{useful LP bound intro}) implies the presence of different projection operators in the estimates. In the low frequency regime in Section \ref{forward low freq estimates section}, we can mostly avoid this issue using the structure of the error terms and the lower order estimates from Section~\ref{forward lower order estimates section}. However, in the high frequency regime in Section \ref{forward high freq estimates section}, this issue creates commutation terms that cannot be bounded at the level of each LP projection. As a result, we must carefully use the structure of the error terms and sum the estimates obtained for each LP projection before being able to close our estimates in Section \ref{forward main result estimates section}. 

\paragraph{The second model system.} Similarly to the above case, in Section \ref{model systemmm} we write the system (\ref{model wave system 2 intro}) with respect to the new time variable $\tau=\sqrt{v}$ to obtain the second model system:
\begin{align}
    &\nabla_{\tau}\big(\nabla_{\tau}\nabla^m\Phi_0\big)+\frac{1}{\tau}\nabla_{\tau}\nabla^m\Phi_0-4\Delta\nabla^m\Phi_0=\psi\nabla^{m+1}\Phi+F_{m}^{0}\notag \\
    &\nabla_{\tau}\big(\nabla_{\tau}\nabla^m\Phi_i\big)-\frac{1}{\tau}\nabla_{\tau}\nabla^m\Phi_i-4\Delta\nabla^m\Phi_i=\sum_{j\neq0}\psi\nabla^{m+1}\Phi_j+F_{m}^{i}\notag \\
    &\Phi_0=2\mathcal{O}\log\tau+h+O\big(\tau^2|\log\tau|^2\big),\ \Phi_i=\Phi_i^0+O\big(\tau^2|\log\tau|^2\big).\notag
\end{align}
The estimates needed at top order for the system (\ref{model wave system 2 intro}) are proved at the level of the above system in \cite{Cwave}. The essence of the argument is dealing with the singular quantity $\Phi_0,$ since the regular quantities $\Phi_i$ satisfy better equations and can be bounded in a straightforward way. In the present paper, we illustrate how to prove these estimates in Theorem~\ref{backward direction main result theorem} in Section \ref{model backward direction section} for a toy problem which models the singular top order quantity $\nabla^M\Phi_0,$ and captures the main difficulties. We refer the reader to \cite[Section~4]{Cwave} for a complete proof of Theorems~\ref{backward direction main result theorem general} and \ref{asymptotic data estimates theorem general}, which deal with the full second model system. We point out that Theorem~\ref{backward direction main result theorem} is not used in the proof of Theorems~\ref{backward direction main result theorem general} and \ref{asymptotic data estimates theorem general}, but it introduces all the main ideas and provides the guideline that we then follow in \cite[Section~4]{Cwave}.

We assume that the smooth horizontal tensor $\xi$ defined on $\{u=-1\}\times\{\tau\in(0,1)\}\times S^n$ satisfies the equations:
\begin{align}\label{backward direction linear wave equation main intro}
    &\nabla_{\tau}\big(\nabla_{\tau}\xi\big)+\frac{1}{\tau}\nabla_{\tau}\xi-4\Delta\xi=\psi\nabla\xi,\\
    &\xi=2\nabla^M\mathcal{O}\log\tau+2\big(\log\nabla\big)\nabla^M\mathcal{O}+\mathfrak{h}_M+O\big(\tau^2|\log\tau|^2\big).\notag
\end{align}
The asymptotic expansion of $\xi$ at $\tau=0$ is the same as that of $\nabla^M\Phi_0,$ but in equation (\ref{backward direction linear wave equation main intro}) we only kept the terms depending on $\nabla\xi$ on the right hand side for simplicity. 

In Theorem \ref{backward direction main result theorem}, we prove the estimates for all $\tau\in(0,1]$, with an implicit constant depending only on $M$:
\begin{equation}\label{backward direction main estimate one intro}
    \tau\big\|\xi\big\|_{H^{1/2}}^2+\tau^2\big\|\xi\big\|_{H^{3/2}}^2+\tau^2\big\|\nabla_{\tau}\xi\big\|_{H^{1/2}}^2+\int_{\tau}^1\tau'\big\|\xi\big\|_{H^{1}}^2d\tau'\lesssim \bigg(\big\|\xi\big\|_{H^{3/2}}^2+\big\|\nabla_{\tau}\xi\big\|_{H^{1/2}}^2\bigg)\bigg|_{\tau=1},
\end{equation}
\begin{equation}\label{backward direction asymptotic quantity estimate intro}
    \big\|\nabla^M\mathcal{O}\big\|_{H^{1}}^2+\sum_{k\geq 0}2^{2k}\big\|P_k\mathfrak{h}_M\big\|_{L^2}^2\lesssim \bigg(\big\|\xi\big\|_{H^{3/2}}^2+\big\|\nabla_{\tau}\xi\big\|_{H^{1/2}}^2\bigg)\bigg|_{\tau=1}.
\end{equation}

In order to prove optimal estimates for $\xi$, we split our analysis into the low frequency and high frequency regime. For some large constant $X=2^{x+1}$, we split the frequencies into the low frequency regime $k<x$ for all $\tau\in[0,1]$, the low frequency regime $k\geq x$ for $\tau\in[0,X2^{-k-1}],$ and the high frequency regime $k\geq x$ for $\tau\in[X2^{-k-1},1]$.

The low frequency regime with $k<x$ is dealt with using a standard preliminary estimate. In Section \ref{backward low freq estimates section}, we prove suitable energy estimates in the low frequency regime $k\geq x,\ \tau\in[0,X2^{-k-1}]$, with data at $\tau=X2^{-k-1}$ given by the solution in the high frequency regime. In Section \ref{high frequency estimates section}, we also prove estimates in the high frequency regime $\tau\in[X2^{-k-1},1]$ with data at $\tau=1,$ similarly to the estimates for the first model system.

The significant challenge that we must overcome is the presence of a top order bulk term with an unfavorable sign in the high frequency estimate of Proposition~\ref{high frequency backward estimate}:
\[\int_{\tau}^1\frac{2^k}{(\tau')^2}\big\|P_k\xi\big\|_{L^2}^2d\tau'.\]
Because the geometric LP projections do not satisfy exact orthogonality properties, this term can only be bounded using the refined Poincaré inequality (\ref{poincare inequality intro}). This introduces both low frequency regime and high frequency regime error terms on the RHS. As a consequence, we obtain a sum of error terms that we bound using the discrete Gronwall inequality and a novel discrete-continuous Gronwall-like inequality. 

Similarly to case of the first model system, we also have commutation error terms arising from the bound (\ref{useful LP bound intro}), which cause the presence of different projection operators in the estimates. Such terms can be bounded only once we sum the estimates obtained for each LP projection in Section~\ref{backward main result estimates section}. 

Finally, obtaining the sharp estimate (\ref{backward direction main estimate one intro}) for the solution on $\tau\in(0,1]$ allows us to conclude by proving the estimate (\ref{backward direction asymptotic quantity estimate intro}) for the asymptotic quantities. The estimates are carried out in Section~\ref{backward asymptotic quantities estimates section} at the level of the equations satisfied by the quantities $\overline{\xi}_k=P_k\xi-\tau\log(2^k\tau)P_k\nabla_{\tau}\xi$, where we decompose the error terms into the low frequency and high frequency regime components.

\subsection{Outline of the Paper}\label{paper outline section}
We outline the structure of the paper. In Section~\ref{set up section} we introduce the double null formalism adapted to the setting of straight self-similar spacetimes. In Section~\ref{stability dS section} we prove the existence and uniqueness of scattering states, establishing the first statement of Theorem \ref{main theorem of the paper ambient}. In Section~\ref{asymptotic completeness section} we prove asymptotic completeness, obtaining the second statement of Theorem \ref{main theorem of the paper ambient}. In Section~\ref{model systemmm} we introduce the model systems necessary for the top order estimates of the scattering map. In Section~\ref{LP Section} we introduce the geometric Littlewood-Paley projections used in the analysis of the model systems. In Section~\ref{model forward direction section} we state the main result of \cite[Theorem~1.1]{Cwave} for the first model system and present the proof in the case of the singular component of $\Phi_0.$ We use these results in Section~\ref{forward direction full system section} to prove sharp estimates from $\{v=0\}$ to $\{v=-u\}.$ In Section~\ref{model backward direction section} we state the main result of \cite[Theorem~1.2]{Cwave} for the second model system and illustrate the proof for a toy problem. We use these results in Section~\ref{backward direction full system section} to prove sharp estimates from to $\{v=-u\}$ to $\{v=0\}$. Finally, in Section~\ref{scattering map section} we combine our results to conclude the proof of the third statement of Theorem \ref{main theorem of the paper ambient}.

\textbf{Acknowledgements.} The author would like to acknowledge Igor Rodnianski for his valuable guidance in the process of writing this paper. The author would also like to thank Mihalis Dafermos, Yakov Shlapentokh-Rothman, and Warren Li for the very helpful discussions.

\section{Set Up}\label{set up section}
The purpose of this section is to introduce the double null formalism adapted to the setting of straight self-similar spacetimes. We also introduce the commutation formulas and the error term notation that we use later.

\subsection{Double Null Gauge}

We introduce a double null gauge on the $(n+2)$-dimensional manifold $\big(\mathcal{M},g\big)$ following the work of \cite[Section~3]{selfsimilarvacuum}. In this section we consider a general such foliation, in order to define the relevant quantities and write down the system of Einstein vacuum equations \eqref{actual vacuum equations} in double null gauge. In the next section we will use the additional assumptions of self-similarity and straightness, which simplify our equations. We assume for the purpose of this section that $g$ is smooth, and we later define the notion of regular vacuum solution in Definition \ref{weak solution definition}. Our introduction of the double null gauge will be brief, and we encourage the reader to consult \cite[Section~3]{selfsimilarvacuum} for complete statements and proofs.

We assume our background differentiable manifold to be $\big((-\infty,0]\times[0,\infty)\backslash\{(0,0)\}\big)\times S^n,$ where the coordinates $(u,v)$ parameterize $(-\infty,0]\times[0,\infty)\backslash\{(0,0)\}.$ We consider the metric $g$ in double null gauge:
\[g=-2\Omega^2(du\otimes dv+dv\otimes du)+\slashed{g}_{AB}(d\theta^A-b^Adu)\otimes(d\theta^B-b^Bdu),\]
where $\{\theta^A\}$ represent local coordinates on $S^n.$ We denote $S_{u,v}=\{u\}\times\{v\}\times S^n.$ We define the normalized frame:
\[e_3=\Omega^{-1}(\partial_u+b^A\partial_{A}),\ e_4=\Omega^{-1}\partial_v,\ g(e_3,e_4)=-2.\]
We denote by $D$ the Levi-Civita connection of $\big(\mathcal{M},g\big)$. Following \cite[Section~3]{selfsimilarvacuum} and the references therein, we introduce the notion of horizontal tensors on $S_{u,v}$, with entries in the tangent space of $S_{u,v}$, and we denote by $\nabla_A,\ \nabla_3,\ \nabla_4$ the projections of $D_A,\ D_3,\ D_4$ to the tangent space of $S_{u,v}$, for any vector $e_A$ tangent to $S_{u,v}$.

We define the Ricci coefficients denoted schematically by $\psi\in\{\chi,\underline{\chi},\eta,\underline{\eta},\omega,\underline{\omega},\zeta\}$ as:
\begin{align*}
    &\chi_{AB}=g(D_Ae_4,e_B),\ \underline{\chi}_{AB}=g(D_Ae_3,e_B),\ \eta_A=-\frac{1}{2}g(D_3e_A,e_4),\ \underline{\eta}_A=-\frac{1}{2}g(D_4e_A,e_3)\\
    &\omega=-\frac{1}{4}g(D_4e_3,e_4),\ \underline{\omega}=-\frac{1}{4}g(D_3e_4,e_3),\ \zeta_A=\frac{1}{2}g(D_Ae_4,e_3).
\end{align*}
We decompose $\chi$ and $\underline{\chi}$ as:
\[\chi_{AB}=\hat{\chi}_{AB}+\frac{1}{n}\tr\chi\slashed{g}_{AB},\ \underline{\chi}_{AB}=\underline{\hat{\chi}}_{AB}+\frac{1}{n}\tr\underline{\chi}\slashed{g}_{AB}.\]
As proved in \cite[Section~3]{selfsimilarvacuum}, we have the Ricci formulas:
\begin{align*}
    &D_4e_4=-2\omega e_4,\ D_4e_3=2\omega e_3+2\underline{\eta}^Ae_A,\ D_4e_A=\underline{\eta}_Ae_4+\nabla_4e_A\\
    &D_3e_3=-2\underline{\omega}e_3,\ D_3e_4=2\underline{\omega}e_4+2\eta^Ae_A,\ D_3e_A=\eta_Ae_3+\nabla_3e_A\\
    &D_Ae_4=-\zeta_Ae_4+\chi_A^Be_B,\ D_Ae_3=\zeta_Ae_3+\underline{\chi}_A^Be_B,\ D_Ae_B=\frac{1}{2}\underline{\chi}e_4+\frac{1}{2}\chi e_3+\nabla_Ae_B.
\end{align*}
We also have the metric equations:
\begin{align*}
    &\mathcal{L}_4\slashed{g}_{AB}=2\chi_{AB},\ \mathcal{L}_3\slashed{g}_{AB}=2\underline{\chi}_{AB},\ \omega=-\frac{1}{2}\nabla_4\log\Omega,\ \underline{\omega}=-\frac{1}{2}\nabla_3\log\Omega\\
    &\zeta_A=-\frac{1}{4}\Omega^{-1}\slashed{g}_{AB}e_4(b^B),\ \eta_A=\zeta_A+\nabla_A\log\Omega,\ \underline{\eta}_A=-\zeta_A+\nabla_A\log\Omega.
\end{align*}

We define the curvature components denoted schematically by $\Psi\in\{\alpha,\underline{\alpha},\beta,\underline{\beta},\nu,\underline{\nu},\sigma,\rho,\tau\}:$
\begin{align*}
    &\alpha_{AB}=R_{A4B4},\ \underline{\alpha}_{AB}=R_{A3B3},\ \beta_A=\frac{1}{2}R_{A434},\ \underline{\beta}_A=\frac{1}{2}R_{A334},\ \nu_{ABC}=R_{ABC4},\ \underline{\nu}_{ABC}=R_{ABC3}\\
    &\sigma_{AB}=\frac{1}{2}R_{3A4B}-\frac{1}{2}R_{3B4A},\ \tau_{AB}=\frac{1}{2}R_{3A4B}+\frac{1}{2}R_{3B4A},\ \rho=\frac{1}{4}R_{4343}
\end{align*}
As proved in \cite[Section~3]{selfsimilarvacuum}, we have the formulas:
\begin{align*}
    &\tr\alpha=Ric_{44},\ \tr\underline{\alpha}=Ric_{33},\ \tr\tau=Ric_{34}-2\rho\\
    &\tau_{AB}=\slashed{g}^{CD}R_{CADB}-Ric_{AB},\ \beta_A=\nu_{AB}{}^B+Ric_{A4},\ \underline{\beta}_A=-\underline{\nu}_{AB}{}^B-Ric_{A3}.
\end{align*}

The Einstein equations are equivalent to a set of null structure equations and constraint equations involving the Ricci coefficients and the curvature components. We refer the reader to \cite[Section~3]{selfsimilarvacuum} for the derivation of the equations.
\begin{proposition}\label{null structure equations proposition}
    We have the following null structure equations for the vacuum spacetime $\big(\mathcal{M},g\big)$:
    \begin{align*}
        \nabla_4\tr\chi+\frac{1}{n}\big(\tr\chi\big)^2&=-|\hat{\chi}|^2-2\omega\tr\chi\\
        \nabla_4\hat{\chi}_{AB}+\frac{2}{n}\tr\chi\hat{\chi}_{AB}&=-\alpha_{AB}-2\omega\hat{\chi}_{AB}+\hat{\chi}\cdot\hat{\chi}\\
        \nabla_3\tr\underline{\chi}+\frac{1}{n}\big(\tr\underline{\chi}\big)^2&=-|\hat{\underline{\chi}}|^2-2\underline{\omega}\tr\underline{\chi}\\
        \nabla_3\hat{\underline{\chi}}_{AB}+\frac{2}{n}\tr\underline{\chi}\hat{\underline{\chi}}_{AB}&=-\underline{\alpha}_{AB}-2\underline{\omega}\hat{\underline{\chi}}_{AB}+\hat{\underline{\chi}}\cdot\hat{\underline{\chi}}\\
        \nabla_3\hat{\chi}_{AB}+\frac{1}{n}\tr\underline{\chi}\hat{\chi}_{AB}&=-\hat{\tau}_{AB}+2\underline{\omega}\hat{\chi}_{AB}+\big(\nabla\hat{\otimes}\eta\big)_{AB}-\frac{1}{n}\tr\chi\underline{\hat{\chi}}_{AB}+\eta\cdot\eta+\hat{\chi}\cdot\hat{\underline{\chi}}\\
        \nabla_3\tr\chi+\frac{1}{n}\tr\underline{\chi}\tr\chi&=2\rho+2\underline{\omega}\tr\chi+2div(\eta)+2|\eta|^2+\hat{\chi}\cdot\hat{\underline{\chi}}\\
        \nabla_4\hat{\underline{\chi}}_{AB}+\frac{1}{n}\tr\chi\hat{\underline{\chi}}_{AB}&=-\hat{\tau}_{AB}+2\omega\hat{\underline{\chi}}_{AB}+\big(\nabla\hat{\otimes}\underline{\eta}\big)_{AB}-\frac{1}{n}\tr\underline{\chi}\hat{\chi}_{AB}+\underline{\eta}\cdot\underline{\eta}+\hat{\chi}\cdot\hat{\underline{\chi}}\\
        \nabla_4\tr\underline{\chi}+\frac{1}{n}\tr\underline{\chi}\tr\chi&=2\rho+2\omega\tr\underline{\chi}+2div(\underline{\eta})+2|\underline{\eta}|^2+\hat{\chi}\cdot\hat{\underline{\chi}}\\
        \nabla_4\eta&=-\beta+\chi\cdot(\eta-\underline{\eta})\\
        \nabla_3\underline{\eta}&=\underline{\beta}+\underline{\chi}\cdot(\eta-\underline{\eta})\\
        \nabla_4\underline{\omega}&=\frac{1}{2}\rho+\frac{1}{4}|\underline{\eta}|^2-\frac{1}{4}|\eta|^2+2\omega\underline{\omega}+3|\zeta|^2-|\nabla\log\Omega|^2\\
        \nabla_3\omega&=\frac{1}{2}\rho-\frac{1}{4}|\underline{\eta}|^2+\frac{1}{4}|\eta|^2+2\omega\underline{\omega}+3|\zeta|^2-|\nabla\log\Omega|^2.
    \end{align*}
    where $\psi\cdot\psi$ is a schematic notation for certain contractions Ricci coefficient terms.
\end{proposition}
\begin{proposition}\label{constraint equations proposition}
    We have the following constraint equations for the vacuum spacetime $\big(\mathcal{M},g\big)$:
    \begin{align*}
        \slashed{Riem}_{ABCD}&=R_{ABCD}+\frac{1}{2}\Big(\underline{\chi}_{BC}\chi_{AD}+\underline{\chi}_{AD}\chi_{BC}-\underline{\chi}_{AC}\chi_{BD}-\underline{\chi}_{BD}\chi_{AC}\Big)\\
        \slashed{Ric}_{AB}&=\tau_{AB}-\frac{1}{2}\tr\chi\underline{\chi}_{AB}-\frac{1}{2}\tr\underline{\chi}\chi_{AB}+\chi^C_{(A}\underline{\chi}_{B)}C\\
        \slashed{R}&=-2\rho+\frac{1-n}{n}\tr\chi\tr\underline{\chi}+\hat{\chi}\cdot\hat{\underline{\chi}}\\
        \nabla_A\chi_{BC}-\nabla_B\chi_{AC}&=\nu_{ABC}+\chi_{AC}\zeta_B-\chi_{BC}\zeta_A\\
        \nabla_A\underline{\chi}_{BC}-\nabla_B\underline{\chi}_{AC}&=\underline{\nu}_{ABC}-\underline{\chi}_{AC}\zeta_B+\underline{\chi}_{BC}\zeta_A\\
        \nabla^A\chi_{AB}-\nabla_B\tr\chi&=-\beta_B+\tr\chi\zeta_B-\zeta^A\chi_{AB}\\
        \nabla^A\underline{\chi}_{AB}-\nabla_B\tr\underline{\chi}&=\underline{\beta}_B-\tr\underline{\chi}\zeta_B+\zeta^A\underline{\chi}_{AB}\\
        \nabla_A\eta_B-\nabla_B\eta_A&=-\nabla_A\underline{\eta}_B+\underline{\nabla}_B\eta_A=\sigma_{AB}+\frac{1}{2}\bigg(\underline{\hat{\chi}}_A^C\hat{\chi}_{CB}-\underline{\hat{\chi}}_B^C\hat{\chi}_{CA}\bigg)\\
        \nabla^AR_{ABCD}&=2\nabla_{[C}\tau_{D]B}+\chi\cdot\underline{\nu}+\underline{\chi}\cdot\nu+\zeta\cdot\chi\cdot\underline{\chi}.
    \end{align*}
\end{proposition}

\subsection{Self-Similar Straight Vacuum Spacetimes}

The $(n+2)$-dimensional vacuum spacetimes that we study are also straight and self-similar, according to Definition~\ref{ambient metric definition}. In this section we derive the consequences of these properties, which simplify the previous null structure equations and constraint equations significantly. We also write down the system of Bianchi equations in our case. Finally, we introduce the notion of regular solutions to our system, following \cite[Section~3]{selfsimilarvacuum}.

We assume that $\big(\mathcal{M},g\big)$ is self-similar, so for $S=u\partial_u+v\partial_v$ we have:
\[\mathcal{L}_Sg=2g.\]
We also assume that $\big(\mathcal{M},g\big)$ is a straight spacetime, satisfying:
\[\Omega^2=1,\ b=0.\]
Thus, the metric is given in double null gauge by:
\[g=-2(du\otimes dv+dv\otimes du)+\slashed{g}_{AB}d\theta^A\otimes d\theta^B.\]
Moreover, the frame $\{e_A,\ e_3,\ e_4\}$ is integrable, since:
\[e_3=\partial_u,\ e_4=\partial_v,\ e_A=\partial_{\theta^A}.\]

As a consequence, we obtain that the only nontrivial Ricci coefficients are $\psi\in\{\tr\chi,\hat{\chi},\tr\underline{\chi},\underline{\hat{\chi}}\}$. Similarly, the only nontrivial curvature components are $\Psi\in\{\alpha,\nu,\tau,R,\underline{\nu},\underline{\alpha}\}.$ We prove these in the following result:
\begin{lemma}
    The Ricci coefficients and curvature components satisfy:
    \begin{align}\label{equation for chi in self similar straight}
        &u\underline{\chi}_{AB}+v\chi_{AB}=\slashed{g}_{AB}\\
        &\eta=\underline{\eta}=\zeta=0,\ \omega=\underline{\omega}=0\notag\\
        &\sigma=0,\ \rho=0,\ \beta=\underline{\beta}=0.\notag
    \end{align}
\end{lemma}
\begin{proof}
    Equation (\ref{equation for chi in self similar straight}) is proved in \cite[Appendix~B]{selfsimilarvacuum}. The metric equations, together with $\Omega=1,\ b=0,$ imply the vanishing of the Ricci coefficients $\eta,\underline{\eta},\zeta,\omega,\underline{\omega}.$ The null structure equations for $\eta,\underline{\eta},$ and $\omega$ imply $\beta=\underline{\beta}=0,$ and $\rho=0.$ Finally, the constraint equation for $\nabla_A\eta_B-\nabla_B\eta_A$ implies that
    \[\sigma_{AB}=\frac{1}{2}\Big(-\underline{\hat{\chi}}_A^C\hat{\chi}_{CB}+\underline{\hat{\chi}}_B^C\hat{\chi}_{CA}\Big).\]
    However, we also have from (\ref{equation for chi in self similar straight}) that $-u\underline{\hat{\chi}}_{AB}=v\hat{\chi}_{AB},$ which then gives $\sigma=0.$
\end{proof}
\begin{remark}
    We can now simplify the null structure equations and constraint equations using the vanishing of the above Ricci coefficients and curvature components. In particular, the last four null structure equations are trivial. Moreover, we notice that all the terms containing angular derivatives in the null structure equations vanish. As a result, when treating these equations as a system of transport equations we avoid the complications regarding loss of angular derivatives which one usually faces when considering the system for a general metric in double null gauge.
\end{remark}

We note some further consequences:
\begin{lemma}The frame $\{e_A,e_3,e_4\}$ satisfies:
    \[\nabla_4e_A=\frac{1}{v}e_A-\frac{u}{v}\underline{\chi}_A^Be_B,\ \nabla_3e_A=\underline{\chi}_A^Be_B.\]
\end{lemma}
\begin{proof}
    Using the above lemma in the Ricci formulas, we have:
    \[\nabla_4e_A=D_4e_A=D_Ae_4=\chi_{A}^Be_B=\frac{1}{v}e_A-\frac{u}{v}\underline{\chi}_A^Be_B,\ \nabla_3e_A=D_Ae_3=\underline{\chi}_A^Be_B.\]\end{proof}
\begin{lemma}
    For any curvature component $\Psi$ we have that $\nabla_S\Psi=-2\Psi.$
\end{lemma}
\begin{proof}
    For $\Psi\in\{\alpha,\tau,\underline{\alpha}\},$ we have that $\mathcal{L}_S\Psi_{AB}=0$ implies that $u\partial_u(\Psi_{AB})+v\partial_v(\Psi_{AB})=0$ in the canonical coordinate frame. The previous lemma then implies that $\nabla_S\Psi_{AB}=-2\Psi_{AB}.$ The proof is similar in the case when $\Psi\in\{\nu,\underline{\nu}\},$ for which we have $u\partial_u(\Psi_{ABC})+v\partial_v(\Psi_{ABC})=\Psi_{ABC}$ and for $\Psi=R,$ for which we have $u\partial_u(R_{ABCD})+v\partial_v(R_{ABCD})=2R_{ABCD}.$\end{proof}

We introduce the notion of signature of \cite[Section~3]{selfsimilarvacuum}, which will facilitate the schematic representation of certain error terms in our equations:
\begin{definition}
    For any $\phi\in\{\psi,\Psi\}$ we define the signature:
    \[s(\phi)=N_3(\phi)+\frac{1}{2}N_A(\phi)-1,\]
    where $N_3$ represents the number of $e_3$ vectors used in the definition of $\phi,$ and $N_A$ represents the number of $e_A$ vectors used in the definition of $\phi$.
\end{definition}
We state the following result of \cite[Section~3]{selfsimilarvacuum} in our simplified setting:
\begin{lemma}
    The signature of the nontrivial Ricci coefficients and curvature components is:
    \[s(\chi)=0,\ s(\underline{\chi})=1,\ s(\alpha)=0,\ s(\nu)=\frac{1}{2},\ s(\tau)=s(R)=1\ s(\underline{\nu})=\frac{3}{2},\ s(\underline{\alpha})=2.\]
    Moreover, we have that for any horizontal tensors $\phi,\phi_1,\phi_2:$
    \[s(\nabla_3\phi)=s(\phi)+1,\ s(\nabla_A\phi)=s(\phi)+\frac{1}{2},\ s(\nabla_4\phi)=s(\phi),\ s(\phi_1\phi_2)=s(\phi_1)+s(\phi_2).\]
    Thus, the signature is preserved by covariant differentiation.
\end{lemma}
Using this notion of signature, we introduce the notation of \cite[Section~3]{selfsimilarvacuum} for the error terms that we expect on the right hand side of the Bianchi equations. We point out that in our case $\zeta=0$ simplifies the structure of these terms. For any $s\in\{0,\frac{1}{2},\ldots,\frac{5}{2}\}$ we have:
\[\mathcal{E}_s^{(3)}=\sum_{s_1+s_2=s,\ s_1\neq1}\psi_{s_1}\Psi_{s_2},\ \mathcal{E}_s^{(4)}=\sum_{s_1+s_2=s}\psi_{s_1}\Psi_{s_2}.\]
We write down the Bianchi equations here and refer the reader to \cite[Section~3]{selfsimilarvacuum} for proof:
\begin{proposition}\label{Bianchi equation proposition}
    We have the following Bianchi equations for the straight self-similar vacuum spacetime $\big(\mathcal{M},g\big)$:
    \begin{align*}
        \nabla_3\alpha_{AB}+\frac{1}{2}\tr\underline{\chi}\alpha_{AB}&=-\nabla^C\nu_{C(AB)}+\mathcal{E}_1^{(3)}\\
        \nabla_4\nu_{ABC}&=-2\nabla_{[A}\alpha_{B]C}+\mathcal{E}_{1/2}^{(4)}\\
        \nabla_3\nu_{ABC}+\frac{2}{n}\tr\underline{\chi}\nu_{ABC}&=-2\nabla_{[A}\tau_{B]C}+2\underline{\hat{\chi}}_{[A}^D\nu_{|D|B]C}+\mathcal{E}_{3/2}^{(3)}\\
        \nabla_4R_{ABCD}&=-2\nabla_{[A}\nu_{|CD|B]}+\mathcal{E}_{1}^{(4)}\\
        \nabla_3R_{ABCD}+\frac{2}{n}\tr\underline{\chi}R_{ABCD}&=-2\nabla_{[A}\underline{\nu}_{|CD|B]}+\underline{\chi}_{A[D}\tau_{C]B}+\underline{\chi}_{B[C}\tau_{D]A}+2\underline{\hat{\chi}}_{[A}^ER_{B]ECD}+\mathcal{E}_{2}^{(3)}\\
        \nabla_4\underline{\nu}_{ABC}&=-2\nabla_{[A}\tau_{B]C}+\mathcal{E}_{3/2}^{(4)}\\
        \nabla_3\underline{\nu}_{ABC}+\frac{3}{n}\tr\underline{\chi}\underline{\nu}_{ABC}&=-2\nabla_{[A}\underline{\alpha}_{B]C}+2\underline{\hat{\chi}}_{[A}^D\underline{\nu}_{B]DC}+2\underline{\hat{\chi}}_{[A}^D\underline{\nu}_{|CD|B]}+\mathcal{E}_{5/2}^{(3)}\\
        \nabla_4\underline{\alpha}_{AB}&=-\nabla^C\underline{\nu}_{C(AB)}+\mathcal{E}_{2}^{(4)}.
    \end{align*}
\end{proposition}
We notice that the equations for $\tau$ can be derived using $\tau_{AB}=\slashed{g}^{CD}R_{CADB}.$ Moreover, the Bianchi equations for the vanishing curvature components $\{\beta,\underline{\beta},\sigma\}$ imply the following additional constraint equations:
\begin{proposition}\label{addtional constraints proposition}
    We have the following constraint equations for the straight self-similar vacuum spacetime $\big(\mathcal{M},g\big)$:
    \[\nabla^B\alpha_{AB}=\mathcal{E}_{1/2}^{(4)},\ \nabla^C\nu_{ABC}=\mathcal{E}_{1}^{(4)},\ \nabla^C\underline{\nu}_{ABC}=\mathcal{E}_{2}^{(4)},\ \nabla^B\underline{\alpha}_{AB}=\mathcal{E}_{5/2}^{(4)}.\]
\end{proposition}

We conclude this section by defining the notion of a regular solution to the Einstein vacuum equations \eqref{actual vacuum equations}, according to \cite{selfsimilarvacuum}. We refer to the collection of horizontal tensors $\slashed{g},\ \chi,\ \underline{\chi},\ \alpha,\ \nu,\ \tau,\ R,\ \underline{\nu},\ \underline{\alpha}$ introduced above as the set of double null unknowns.
\begin{definition}\label{weak solution definition}
    The straight self-similar  metric $g$ defined on a subset of the background differentiable manifold $\big((-\infty,0]\times[0,\infty)\backslash\{(0,0)\}\big)\times S^n,$ is a regular solution of the Einstein vacuum equations \eqref{actual vacuum equations} if:
    \begin{itemize}
        \item For $n>4$ even, the metric:
        \[g=-2(du\otimes dv+dv\otimes du)+\slashed{g}_{AB}d\theta^A\otimes d\theta^B\]
        is a classical $C^2$ solution of the Einstein vacuum equations \eqref{actual vacuum equations}. Equivalently, the double null unknowns are classical solutions to the metric equations, the system of null structure equations in Proposition \ref{null structure equations proposition}, the Bianchi equations in Proposition \ref{Bianchi equation proposition},  and the system of constraint equations in Propositions \ref{constraint equations proposition} and \ref{addtional constraints proposition}.
        \item For $n=4$ the metric $g$ defined as above is a classical $C^2$ solution of the Einstein vacuum equations \eqref{actual vacuum equations} for $v>0$ and $u<0.$ Moreover, the corresponding double null unknowns are classical solutions to the metric equations, the system of constraint equations in Proposition \ref{constraint equations proposition} and weak solutions to the system of equations in Propositions \ref{null structure equations proposition}, \ref{Bianchi equation proposition}, and \ref{addtional constraints proposition}.
    \end{itemize}
\end{definition}

The solutions that we construct in this paper will be smooth in the region $\{v>0,\ u<0\}$ and extend as regular solutions to $\big((-\infty,0]\times[0,\infty)\backslash\{(0,0)\}\big)\times S^n.$ Moreover, the solutions will also satisfy the Fefferman-Graham expansions at $v=0$ and $u=0$ up to order $\frac{n}{2}.$ In the case of $n=4$, all the double null unknowns extend continuously to $v=0$ and $u=0,$ with the exception of $\alpha$ which has a $\log v$ singularity at $v=0$ and $\underline{\alpha}$ which has a $\log u$ singularity at $u=0$. These mild singularities allow us to conclude that for $n=4$ a regular solution $g$ solves the Einstein vacuum equations \eqref{actual vacuum equations} weakly in $L^2.$

\subsection{Commutation Formulas and Error Terms Notation}
\begin{lemma}\label{commutation with nabla lemma}We have the commutation formulas:
    \begin{align*}
        &\big[\nabla_4,\nabla^i\big]\phi+\frac{i}{n}\tr\chi\nabla^i\phi=\sum_{j=1}^i\nabla^j\chi\cdot\nabla^{i-j}\phi+\hat{\chi}\cdot\nabla^i\phi,\\
        &\big[\nabla_3,\nabla^i\big]\phi+\frac{i}{n}\tr\underline{\chi}\nabla^i\phi=\sum_{j=1}^i\nabla^j\underline{\chi}\cdot\nabla^{i-j}\phi+\hat{\underline{\chi}}\cdot\nabla^i\phi,
    \end{align*}
    where $\nabla^j$ is a schematic notation for all the possible combinations of $j$ angular derivatives.
\end{lemma}
\begin{proof}
    We prove the first statement, since the second one statement is similar. A standard computation, see for example \cite{luk}, together with the additional vanishing of certain Ricci coefficients, implies that:
    \[\big[\nabla_4,\nabla_A\big]\phi=\big[D_4,D_A\big]\phi-\chi_A^B\nabla_B\phi=-\chi_A^B\nabla_B\phi+\nu\cdot\phi=-\chi_A^B\nabla_B\phi+\nabla\chi\cdot\phi.\]
    We can rewrite this as $\big[\nabla_4,\nabla_A\big]\phi+\frac{1}{n}\tr\chi\nabla_A\phi=\nabla\chi\cdot\phi+\hat{\chi}\cdot\nabla\phi,$ and conclude by induction.
\end{proof}
\begin{lemma}\label{commutation with nabla 4 lemma}
    We have the commutation formulas:
    \begin{align*}
        &\big[\nabla_4^l,\nabla\big]\phi=\sum_{i+j+k=l-1}\nabla_4^i\chi^{k+1}\cdot\nabla\nabla_4^{j}\phi+\sum_{i+j+k=l-1}\nabla\nabla_4^i\chi^{k+1}\cdot\nabla_4^{j}\phi,\\
        &\big[\nabla_3^l,\nabla\big]\phi=\sum_{i+j+k=l-1}\nabla_3^i\underline{\chi}^{k+1}\cdot\nabla\nabla_3^{j}\phi+\sum_{i+j+k=l-1}\nabla\nabla_3^i\underline{\chi}^{k+1}\cdot\nabla_3^{j}\phi,
    \end{align*}
    where $\nabla_4^i\chi^{k+1}$ is a schematic notation for all possible ways of distributing $i$ derivatives in the $e_4$ direction over a product of $k+1$ $\chi$ terms, and similarly in the $e_3$ direction.
\end{lemma}
\begin{proof}
    As before, we have that $\big[\nabla_4,\nabla\big]\phi=\nabla\chi\cdot\phi+\chi\cdot\nabla\phi$ and we conclude by induction.
\end{proof}

We recall the notation for differences of the Ricci coefficients and their Minkowski values, as in the introduction:
\[\psi^*=\psi-\psi_{\mathrm{Minkowski}}.\]
Similarly, we recall the notation for certain "good" curvature components:
\[\Psi^G\in\{\nu,\tau,R,\underline{\nu},\underline{\alpha}\},\ \underline{\Psi}^G\in\{\alpha,\nu,\tau,R,\underline{\nu}\}.\]

We adapt the notation of \cite[Section~5]{selfsimilarvacuum} for error terms by defining:
\begin{definition}
    For any $m+l\leq p$, we introduce the schematic notation:
    \begin{align*}
        &\mathcal{F}_{mlp}(\Psi)=\sum_{\substack{i+j+k\leq p \\ i\leq l,k\leq m}}\nabla^k\nabla_4^i\big(\psi^{j+1}\Psi\big),\ \mathcal{F}_{mlp}(\Psi^G)=\sum_{\substack{i+j+k\leq p \\ i\leq l,k\leq m}}\nabla^k\nabla_4^i\big(\psi^{j+1}\Psi^G\big)\\
        &\underline{\mathcal{F}}_{mlp}(\Psi)=\sum_{\substack{i+j+k\leq p \\ i\leq l,k\leq m}}\nabla^k\nabla_3^i\big(\psi^{j+1}\Psi\big),\ \underline{\mathcal{F}}_{mlp}(\underline{\Psi}^G)=\sum_{\substack{i+j+k\leq p \\ i\leq l,k\leq m}}\nabla^k\nabla_3^i\big(\psi^{j+1}\underline{\Psi}^G\big),
    \end{align*}
    where the terms $\nabla^k\nabla_4^i\big(\psi^{j+1}\Psi\big)$ denote the sum of all the possible products obtained when distributing the $\nabla^k\nabla_4^i$ derivatives. We also define $\mathcal{F}'(\Psi)$ and $\underline{\mathcal{F}}'(\Psi)$ as above, in the case when at least one of the Ricci coefficients is $\psi^*.$ 
    
    We identify the top order term in $\mathcal{F}_{mlp}(\Psi)$ as being $\psi\nabla^m\nabla_4^l\Psi$. We write:
\[\mathcal{F}_{mlp}(\Psi)=\psi\nabla^m\nabla_4^l\Psi+\mathcal{F}_{mlp}^{lot}(\Psi),\]
with the understanding that when we expand $\mathcal{F}_{mlp}^{lot}(\Psi)$ using the product rule it does not contain any top order terms. Similarly, we also define $\mathcal{F}_{mlp}^{lot}(\Psi^G)$, $\underline{\mathcal{F}}_{mlp}^{lot}(\Psi)$, and $\underline{\mathcal{F}}_{mlp}^{lot}(\underline{\Psi}^G).$
\end{definition}
Using the above commutation formulas, we can prove by induction that:
\begin{lemma}
    The error terms $\mathcal{F}_{mlp}$ and $\underline{\mathcal{F}}_{mlp}$ satisfy:
    \[\nabla^i\nabla_4^j\mathcal{F}_{mlp}(\Psi)=\mathcal{F}_{(m+i)(l+j)(p+i+j)}(\Psi),\ \nabla^i\nabla_3^j\underline{\mathcal{F}}_{mlp}(\Psi)=\underline{\mathcal{F}}_{(m+i)(l+j)(p+i+j)}(\Psi).\]
    Similar results hold for $\mathcal{F}_{mlp}(\Psi^G)$, $\underline{\mathcal{F}}_{mlp}(\underline{\Psi}^G),$ $\mathcal{F}_{mlp}'(\Psi)$ and $\underline{\mathcal{F}}_{mlp}'(\Psi).$
\end{lemma}

We can restate the commutation lemmas as follows:
\begin{lemma}\label{commutation lemma with F notation} We have the commutation formulas:
    \begin{align*}
        &\big[\nabla^m,\nabla_4\big]\phi=\mathcal{F}_{(m)(0)(m)}(\phi),\ \big[\nabla^m,\nabla_4^2\big]\phi=\mathcal{F}_{(m)(1)(m+1)}(\phi)\\
        &\big[\nabla_4^l,\nabla\big]\phi=\mathcal{F}_{(1)(l-1)(l)}(\phi),\ \big[\nabla_4^l,\Delta\big]\phi=\mathcal{F}_{(2)(l-1)(l+1)}(\phi).
    \end{align*}
    The same result holds if we replace $e_4$ by $e_3$ and $\mathcal{F}$ by $\underline{\mathcal{F}}.$
\end{lemma}
\begin{proof}
    The first three formulas follows directly from Lemmas \ref{commutation with nabla lemma} and \ref{commutation with nabla 4 lemma} using the notation introduced above. Applying the third formula, we also get the last formula.
\end{proof}

\section{Existence and Uniqueness of Scattering States}\label{stability dS section}

The main result of this section is the proof of the first statement of Theorem \ref{main theorem of the paper ambient}, establishing global existence and quantitative estimates for the solution in the region $\{u < 0,\ v > 0\}$, given small scattering data at $\{v = 0\}$. In the original $(n+1)$-dimensional formulation, this represents the proof of the first statement of Theorem \ref{main theorem of the paper}, by showing the stability of de Sitter space with scattering data at $\mathcal{I}^-$. We prove the following result:

\begin{theorem}\label{stability of de sitter theorem in section}
    For any $N>0$ large enough there exist $\epsilon_0>0$ small enough, and a universal constant $c_0>0$, such that for any $\epsilon\leq\epsilon_0$, if the straight initial data $\big(\slashed{g}_0,\check{h}\big)$ satisfies the smallness assumption:
    \begin{equation}\label{smallness assumption g}
    \sum_{i=0}^N\big\|\mathcal{L}_{\theta}^i\slashed{g}_0^*\big\|_{L^2(S^n)}+\big\|h\big\|_{H^N(S^n)}<\epsilon,
\end{equation}
then the corresponding straight self-similar vacuum spacetime $\big(\mathcal{M},g\big)$ with data at $\{u=-1,\ v=0\}$ given by $\big(\slashed{g}_0,\check{h}\big)$ exists globally on $\{u<0,v>0\},$ extends to $\{v=0\}$ as a regular solution, and satisfies quantitative estimates with regularity $N'=N-c_0n$, as made precise in Propositions \ref{region I bounds proposition}, \ref{region II bounds proposition}, and \ref{region III bounds proposition} (with $N_1,N_2,N_3=N'$).

Moreover, there exists a small constant $\underline{v}>0,$ with $\epsilon\ll\underline{v}\ll1,$ such that the same future global stability result holds in the case of small Cauchy initial data on a spacelike hypersurface $\{v=-cu\}$ with $\underline{v}\leq c\leq\underline{v}^{-1},$ as made precise in Remark \ref{remark about cauchy data}.
\end{theorem}
\begin{remark}
    In (\ref{smallness assumption g}) we denote by $\mathcal{L}_{\theta}^i$ all the possible combinations of $i$ Lie angular derivatives in a coordinate patch, and we sum over a family of coordinate patches that covers all of $S^n.$ 
\end{remark}
\begin{remark}
    We notice it is equivalent to specify initial data $\big(\slashed{g}_0,h\big)$, where $h$ is the term in the expansion (\ref{expansion for alpha introduction}) of $\alpha,$ since $h$ is obtained from $\check{h}$ by subtracting a linear factor of $\mathcal{O},$ which can be computed using $\frac{n}{2}$ derivatives of $\slashed{g}_0.$ Unless using a checked quantity, we shall always refer to this notion of initial data.
\end{remark}
\begin{remark}      
    We recall according to the introduction that the assumption on the asymptotic data of $\epsilon$-smallness of order $M$ implies (\ref{smallness on initial data introduction}). This also implies (\ref{smallness assumption g}) with $N=M$ and replacing $\epsilon$ by $C\epsilon$ (the proof follows from (\ref{Lie derivative in terms of covariant derivative})).
\end{remark}

In addition, we also prove a propagation of regularity result:
\begin{theorem}\label{propagation of regularity theorem}
    Consider a global straight self-similar vacuum spacetime $\big(\mathcal{M},g\big)$ which satisfies the hypothesis of Theorem \ref{stability of de sitter theorem in section}. If the initial data $\big(\slashed{g}_0,h\big)$ is smooth, the spacetime $\big(\mathcal{M},g\big)$ is also smooth.
\end{theorem}

\begin{remark}
    The proof of the first statement of Theorem \ref{main theorem of the paper ambient} follows from Theorem \ref{stability of de sitter theorem in section}, Theorem \ref{propagation of regularity theorem}, and the above remarks. The proof of Theorem \ref{stability of de sitter theorem in section} follows from Propositions \ref{region I bounds proposition}, \ref{region II bounds proposition}, and \ref{region III bounds proposition}. Finally, we prove Theorem~\ref{propagation of regularity theorem} at the end of this section.
\end{remark}

As explained in Section \ref{stability intro section} of the introduction, the proof follows the steps of \cite{nakedsing}. For small $\underline{v}>0$ with $\epsilon\ll\underline{v}\ll1,$ we consider the following regions of the spacetime:
\[I=\bigg\{0\leq\frac{v}{|u|}\leq\underline{v}\bigg\},\ II=\bigg\{\underline{v}\leq\frac{v}{|u|}\leq\underline{v}^{-1}\bigg\},\ III=\bigg\{\underline{v}^{-1}\leq\frac{v}{|u|}\bigg\}.\]
We also set $\delta>0$ to be a small constant, and consider $\epsilon>0$ small enough such that $C(\underline{v})\leq\epsilon^{-\delta}$, where $C(\underline{v})$ is a constant determined in the proof.

\textbf{Notation.} We make the convention that in Section~\ref{stability dS section} and Section~\ref{asymptotic completeness section} we write $A\lesssim B$ for any quantities $A,B>0,$ if there is a constant $C>0$ depending only on $N$ such that $A\leq C B.$

\textbf{Absolute value convention.} We make the convention for the rest of the paper that for any horizontal tensor $\phi$ defined on $S_{u,v}$, its absolute value $|\phi|$ is defined with respect to the metric $\slashed{g}(u,v)$ induced on $S_{u,v}$.

\textbf{Integration convention.} We make the convention that in Section~\ref{stability dS section} and Section~\ref{asymptotic completeness section} the volume forms used are $d\mathring{Vol},$ $d\mathring{Vol}\ du$, $d\mathring{Vol}\ dv$, and $d\mathring{Vol}\ dudv,$ as needed in each context, where $d\mathring{Vol}$ represents the volume form with respect to the round metric $\big(\slashed{g}_0\big)_{\mathrm{Minkowski}}=\frac{1}{4}\slashed{g}_{S^n}.$

\textbf{Error terms convention.} We make the convention for the rest of the paper that if an index denoting the order of some derivative is negative, then that term is empty. For example $\mathcal{F}_{(m)(-1)(p)}:=0$, and so on.

\subsection{Regions I and II}
We use the argument of \cite{selfsimilarvacuum} to obtain existence in the first region. Moreover, we notice that with a few modifications one can repeat the argument in the small data case in order to prove:
\begin{proposition}\label{region I bounds proposition}
If the straight initial data $\big(\slashed{g}_0,h\big)$ satisfies the smallness assumption (\ref{smallness assumption g}), then there exists a regular straight self-similar vacuum solution in region~I with the given initial data. Setting $N_1=N-\lfloor c_0/4\rfloor n,$ we have the bounds in region~I for all $0\leq i\leq N_1,$ $0\leq j\leq\frac{n-4}{2},\ 0\leq k\leq\frac{n}{2}$:
\begin{align*}
    \Big\|\nabla^i\nabla_4^j\nabla_3^k\Big(\alpha-v^{\frac{n-4}{2}}u^{\frac{4-n}{2}}\log(-v/u)\mathcal{O}/\big((n-4)/2\big)!\Big)\Big\|_{L^{\infty}(S_{u,v})}&\lesssim\epsilon|u|^{-2-i-j-k}\\
    \big\|\nabla^i\nabla_4^j\nabla_3^k\Psi^G\big\|_{L^{\infty}(S_{u,v})}&\lesssim\epsilon|u|^{-2-i-j-k}\\
    \big\|\nabla^i\nabla_4^j\nabla_3^k\psi^*\big\|_{L^{\infty}(S_{u,v})}&\lesssim\epsilon|u|^{-1-i-j-k}\\
    \big\|\nabla^i\nabla_4^{\frac{n-2}{2}}\nabla_3^k\psi^*\big\|_{L^{\infty}(S_{u,v})}&\lesssim\epsilon\cdot\log\bigg(\frac{|u|}{v}\bigg)\cdot|u|^{-1-i-{\frac{n-2}{2}}-k}\\
    \big\|\mathcal{L}_{\theta}^i\slashed{g}^*\big\|_{L^{\infty}(S_{u,v})}&\lesssim\epsilon|u|^{-i}.
\end{align*}
\end{proposition}
\begin{remark}
    The proof of \cite{selfsimilarvacuum} implies the need to prove estimates for the solution with angular regularity at most $N_1,$ which is sufficiently small compared to the regularity of the initial data $N.$ In this section we do not attempt to optimize the universal constant $c_0,$ and for our purposes it suffices to take $c_0=100.$
\end{remark}

To obtain estimates in region~II, one can adapt the argument of \cite[Section~7]{nakedsing} to the case of $n\geq4$ and obtain the following bounds:
\begin{proposition}\label{region II bounds proposition}
The solution of Proposition \ref{region I bounds proposition} can be extended uniquely as a regular straight self-similar vacuum solution in region~II. Setting $N_2=N-2\lfloor c_0/4\rfloor n,$ we have the bounds in region~II for all $0\leq i\leq N_2,$ and all $0\leq j,k\leq\frac{n-4}{2}$:
\begin{align*}
    \big\|\nabla^i\nabla_4^j\nabla_3^k\Psi\big\|_{L^{\infty}(S_{u,v})}&\lesssim\epsilon^{1-\delta}|v|^{-2-i-j-k}\\
    \big\|\nabla^i\nabla_4^j\nabla_3^k\psi^*\big\|_{L^{\infty}(S_{u,v})}&\lesssim\epsilon^{1-\delta} |v|^{-1-i-j-k}\\
    \big\|\nabla^i\nabla_4^{\frac{n-2}{2}}\nabla_3^k\psi^*\big\|_{L^{\infty}(S_{u,v})}&\lesssim\epsilon^{1-\delta} |v|^{-1-i-\frac{n-2}{2}-k}\\
    \big\|\mathcal{L}_{\theta}^i\slashed{g}^*\big\|_{L^{\infty}(S_{u,v})}&\lesssim\epsilon^{1-\delta}|v|^{-i}.
\end{align*}
\end{proposition}

\begin{remark}\label{remark about cauchy data}
    The same result holds in the case of small Cauchy initial data on a spacelike hypersurface $\{v=-cu\},$ with $\underline{v}\leq c\leq\underline{v}^{-1}.$ We consider a set of double null unknowns $\slashed{g},\psi,\Psi$ which satisfy the constraint equations in Propositions \ref{constraint equations proposition}, \ref{addtional constraints proposition}, and the following smallness conditions for all $0\leq i\leq N_1,$ and all $0\leq j,k\leq\frac{n-2}{2}$:
    \begin{align*}
        \big\|\nabla^i\nabla_4^j\nabla_3^k\Psi\big\|_{L^{\infty}(S_{-cv,v})}&\lesssim\epsilon|v|^{-2-i-j-k}\\
        \big\|\nabla^i\nabla_4^j\nabla_3^k\psi^*\big\|_{L^{\infty}(S_{-cv,v})}&\lesssim\epsilon|v|^{-1-i-j-k}\\
        \big\|\mathcal{L}_{\theta}^i\slashed{g}^*\big\|_{L^{\infty}(S_{-cv,v})}&\lesssim\epsilon|v|^{-i}.
    \end{align*}
    This is the precise notion in which the spacetime $\big(\mathcal{M},g\big)$ determined by the initial data is close to Minkowski space on the spacelike hypersurface $\{v=-cu\},$ as referred to in Theorem \ref{main theorem of the paper ambient} and Theorem \ref{stability of de sitter theorem in section}.
    
    The argument of \cite[Section~7]{nakedsing} also applies in this case, giving the same conclusion as Proposition \ref{region II bounds proposition}. We point out that we do not obtain at this stage the bounds of Proposition \ref{region I bounds proposition} in region~I. Thus, for small Cauchy initial data we expect to prove in region~I similar bounds to region~III, by repeating the argument in the following section in the reverse time direction (or simply replacing $(u,v)$ by $(-v,-u)$).
\end{remark}

\subsection{Region III}\label{region III section}

We prove existence of the solution and self-similar bounds in region~III. By self-similarity, we can restrict to $v<\underline{v}$.

\begin{proposition}\label{region III bounds proposition}
The solution of Proposition \ref{region II bounds proposition} can be extended uniquely as a regular straight self-similar vacuum solution in region~III. Setting $N_3=N-3\lfloor c_0/4\rfloor n,$ and taking $p>0$ to be a small constant, we have the bounds in region~III with $\{v\leq\underline{v}\}$, for all $0\leq i\leq N_3,$ and all $0\leq j\leq\frac{n-4}{2}$:
\begin{align*}
    \big\|\nabla^i\nabla_3^j\underline{\alpha}\big\|_{L^{2}(S_{u,v})}&\lesssim\epsilon^{1-2\delta}|u|^{-p}\cdot|v|^{-2-i-j+p}\\
    \big\|\nabla^i\nabla_3^j\underline{\Psi}^G\big\|_{L^{2}(S_{u,v})}&\lesssim\epsilon^{1-2\delta}|v|^{-2-i-j}\\
    \big\|\nabla^i\nabla_3^j\psi^*\big\|_{L^{2}(S_{u,v})}&\lesssim\epsilon^{1-2\delta}|v|^{-1-i-j}\\
    \big\|\nabla^i\nabla_3^{\frac{n-2}{2}}\psi^*\big\|_{L^{2}(S_{u,v})}&\lesssim\epsilon^{1-2\delta}|u|^{-p}\cdot|v|^{-1-i-\frac{n-2}{2}+p},\text{ for }i\leq N_3-1\\
    \big\|\mathcal{L}_{\theta}^i\slashed{g}^*\big\|_{L^{2}(S_{u,v})}&\lesssim\epsilon^{1-2\delta}|v|^{-i}.
\end{align*}
In addition, we have control of more detailed norms as proved in Propositions \ref{high regularity curvature proposition},\ref{low regularity curvature proposition},\ref{Ricci coefficients proposition}, and \ref{metric coefficients proposition}.
\end{proposition}

\subsubsection{Norms}\label{stability norms section}
For the fixed $N_3>0$ defined above, we introduce the following sets of indices for high and low regularity norms:
\begin{align*}
    H&=\bigg\{(m,l):\ 0\leq l\leq\frac{n-4}{2},\ 0\leq m\leq N_3+\frac{n-4}{2}-l\bigg\}\\
    L&=\bigg\{(m,l):\ 0\leq l\leq\frac{n-4}{2},\ 0\leq m\leq N_3+\frac{n-6}{2}-l\bigg\}.
\end{align*}
We also define the characteristic triangles:
\[P_{\widetilde{u},\widetilde{v}}=\Big\{(u,v):\ -\underline{v}\widetilde{v}\leq u\leq\widetilde{u},\ -\underline{v}^{-1}\widetilde{u}\leq v\leq\widetilde{v}\Big\}.\]
We define the high regularity curvature norms in region~III with $v<\underline{v}$, where $w_{ml}(u,v)=v^{\frac{3}{2}+m+l-p}|u|^{p-q}$ for some constants $0<q\ll p\ll1,$ and $(m,l)\in H.$ We recall our above conventions for absolute value and integration.
\[\big\|\underline{\alpha}\big\|^2_{\mathcal{C}_{m,l}}(u,v)=|u|^{2q}\int_{-v\underline{v}}^u\int_{S^n}w_{ml}^2\big|\nabla^m\nabla_3^l\underline{\alpha}\big|^2d\mathring{Vol}d\hat{u}+|u|^{2q}\int_{-v\underline{v}}^u\int_{-\frac{\hat{u}}{\underline{v}}}^v\int_{S^n}\frac{w_{ml}^2}{\hat{v}}\big|\nabla^m\nabla_3^l\underline{\alpha}\big|^2d\mathring{Vol}d\hat{v}d\hat{u}.\]
For any $\underline{\Psi}^G\neq\alpha,$ we define:
\[\big\|\underline{\Psi}^G\big\|^2_{\mathcal{C}_{m,l}}(u,v)=|u|^{2q}\int_{-v\underline{v}}^u\int_{S^n}w_{ml}^2\big|\nabla^m\nabla_3^l\underline{\Psi}^G\big|^2d\mathring{Vol}d\hat{u}+|u|^{2q}\int_{-\frac{u}{\underline{v}}}^v\int_{S^n}w_{ml}^2\big|\nabla^m\nabla_3^l\underline{\Psi}^G\big|^2d\mathring{Vol}d\hat{v}+\]\[+|u|^{2q}\int_{-v\underline{v}}^u\int_{-\frac{\hat{u}}{\underline{v}}}^v\int_{S^n}\frac{w_{ml}^2}{|\hat{u}|}\big|\nabla^m\nabla_3^l\underline{\Psi}^G\big|^2d\mathring{Vol}d\hat{v}d\hat{u}.\]
For $\alpha,$ we define:
\[\big\|\alpha\big\|^2_{\mathcal{C}_{m,l}}(u,v)=|u|^{2q}\int_{-\frac{u}{\underline{v}}}^v\int_{S^n}w_{ml}^2\big|\nabla^m\nabla_3^l\alpha\big|^2d\mathring{Vol}d\hat{v}+|u|^{2q}\int_{-v\underline{v}}^u\int_{-\frac{\hat{u}}{\underline{v}}}^v\int_{S^n}\frac{w_{ml}^2}{|\hat{u}|}\big|\nabla^m\nabla_3^l\alpha\big|^2d\mathring{Vol}d\hat{v}d\hat{u}.\]
We define the total high regularity curvature norm as:
\[\mathcal{C}_{\widetilde{u},\widetilde{v}}=\sup_{(u,v)\in P_{\widetilde{u},\widetilde{v}}}\sum_{(m,l)\in H}\bigg(\big\|\underline{\alpha}\big\|^2_{\mathcal{C}_{m,l}}(u,v)+\big\|\underline{\Psi}^G\big\|^2_{\mathcal{C}_{m,l}}(u,v)\bigg).\]
We define the low regularity curvature norms in region~III with $v<\underline{v}$ for  ${(m,l)\in L}$:
\begin{align*}
    \big\|\underline{\Psi}^G\big\|^2_{\mathcal{L}_{m,l}}(u,v)&=\int_{S^n}v^{4+2m+2l}\big|\nabla^m\nabla_3^l\underline{\Psi}^G\big|^2d\mathring{Vol}\\
    \big\|\underline{\alpha}\big\|^2_{\mathcal{L}_{m,l}}(u,v)&=\int_{S^n}v^{4+2m+2l-2p}|u|^{2p}\big|\nabla^m\nabla_3^l\underline{\alpha}\big|^2d\mathring{Vol}.
\end{align*}
We define the total low regularity curvature norm as:
\[\mathcal{L}_{\widetilde{u},\widetilde{v}}=\sup_{(u,v)\in P_{\widetilde{u},\widetilde{v}}}\sum_{(m,l)\in L}\bigg(\big\|\underline{\alpha}\big\|^2_{\mathcal{L}_{m,l}}(u,v)+\big\|\underline{\Psi}^G\big\|^2_{\mathcal{L}_{m,l}}(u,v)\bigg).\]
We define the norms for Ricci coefficients in region~III with $v<\underline{v}$ for  ${(m,l)\in H}$:
\[\big\|\psi\big\|^2_{\mathcal{R}_{m,l}}(u,v)=\int_{S^n}v^{2+2m+2l}\big|\nabla^m\nabla_3^l\psi^*\big|^2d\mathring{Vol}.\]
We define the total Ricci coefficients norm as:
\[\mathcal{R}_{\widetilde{u},\widetilde{v}}=\sup_{(u,v)\in P_{\widetilde{u},\widetilde{v}}}\sum_{(m,l)\in H}\big\|\psi\big\|^2_{\mathcal{R}_{m,l}}(u,v).\]
Finally, we define the norms for the metric coefficients for ${(m,0)\in H}$:
\[\big\|\slashed{g}\big\|^2_{\mathcal{M}_{m,0}}(u,v)=\int_{S^n}v^{2m}\big|\mathcal{L}_{\theta}^m\slashed{g}^*\big|^2d\mathring{Vol}.\]
We define the total metric coefficients norm as:
\[\mathcal{M}_{\widetilde{u},\widetilde{v}}=\sup_{(u,v)\in P_{\widetilde{u},\widetilde{v}}}\sum_{(m,0)\in H}\big\|\slashed{g}\big\|^2_{\mathcal{M}_{m,0}}(u,v).\]

\subsubsection{Bootstrap Assumptions}

The global existence result and quantitative estimates in Proposition \ref{region III bounds proposition} are proved by using the following bootstrap result and standard local existence arguments:
\begin{proposition}\label{bootstrap argument proposition}
    We denote $\epsilon'=\epsilon^{1-2\delta}.$ Let $\big(\mathcal{M},g\big)$ be a spacetime obtained in Proposition \ref{region II bounds proposition}, which exists in the characteristic triangle $P_{\widetilde{u},\widetilde{v}}$ contained in region~III. We assume that the spacetime satisfies the bootstrap assumption:
    \[\mathcal{C}_{\widetilde{u},\widetilde{v}}+\mathcal{L}_{\widetilde{u},\widetilde{v}}+\mathcal{R}_{\widetilde{u},\widetilde{v}}+\mathcal{M}_{\widetilde{u},\widetilde{v}}\leq2A\epsilon'^2.\]
    We prove that we actually have:
    \[\mathcal{C}_{\widetilde{u},\widetilde{v}}+\mathcal{L}_{\widetilde{u},\widetilde{v}}+\mathcal{R}_{\widetilde{u},\widetilde{v}}+\mathcal{M}_{\widetilde{u},\widetilde{v}}\leq A\epsilon'^2.\]
\end{proposition}
We remark that the bootstrap assumptions hold initially on $\{(u,v):\ u=-\underline{v}v,\ 0<v\leq\underline{v}\}$. Thus, for all $0<v\leq\underline{v}$ we have by Proposition \ref{region II bounds proposition}:
\[\mathcal{C}_{-\underline{v}v,v}+\mathcal{L}_{-\underline{v}v,v}+\mathcal{R}_{-\underline{v}v,v}+\mathcal{M}_{-\underline{v}v,v}\lesssim\epsilon'^2.\]

We note some consequences of the bootstrap assumption, which also rely on the boundedness of the quantities below on the corresponding region of Minkowski space. We have that for all $(u,v)\in P_{\widetilde{u},\widetilde{v}}$ and ${(m,l)\in L}$:
\begin{align*}
    &\int_{S^n}v^{2+2m+2l}\big|\nabla^m\nabla_3^l\psi\big|^2d\mathring{Vol}\lesssim1\\
    &\int_{S^n}v^{2m}\big|\mathcal{L}_{\theta}^m\slashed{g}\big|^2d\mathring{Vol}\lesssim1.
\end{align*}

\subsubsection{Sobolev Spaces and Sobolev Inequalities}
We use the definition of \cite[Section~6]{nakedsing} for the weighted $\widetilde{H}^m(S_{u,v})$ Sobolev spaces:
\[\big\|\phi\big\|_{\widetilde{H}^m(S_{u,v})}=\sum_{i=0}^m(v-u)^{i}\bigg(\int_{S^n}\big|\nabla^i\phi\big|^2d\mathring{Vol}\bigg)^{\frac{1}{2}}.\]
Using the results of \cite[Section~6]{nakedsing}, we have that the bootstrap assumptions imply that for any $(m,0)\in H$ we have the Sobolev inequalities:
\begin{align*}
    \big\|\phi\big\|_{L^{\infty}(S_{u,v})}&\lesssim\big\|\phi\big\|_{\widetilde{H}^n(S_{u,v})}\\
    \big\|\phi\cdot\psi\big\|_{\widetilde{H}^m(S_{u,v})}&\lesssim\big\|\phi\big\|_{\widetilde{H}^m(S_{u,v})}\big\|\psi\big\|_{\widetilde{H}^m(S_{u,v})}.
\end{align*}

\subsubsection{Estimates for High Regularity Curvature Components}

In this section we improve the bootstrap assumption on the high regularity curvature components by proving:
\begin{proposition}\label{high regularity curvature proposition}
    There exists a constant $C\ll A$, such that we have the improved estimate:
    \[\mathcal{C}_{\widetilde{u},\widetilde{v}}\leq C\epsilon'^2.\]
\end{proposition}

\textit{Step 1. The Bianchi pair $(\underline{\nu},\underline{\alpha}).$} Using signature considerations, we can write the equations for $(\underline{\nu},\underline{\alpha})$ as:
\begin{align*}
    \nabla_3\underline{\nu}_{ABC}&=-2\nabla_{[A}\underline{\alpha}_{B]C}+\psi\underline{\Psi}^G\\
    \nabla_4\underline{\alpha}_{AB}+\frac{1}{2}\tr\chi\underline{\alpha}_{AB}&=-\nabla^C\underline{\nu}_{C(AB)}+\psi\underline{\Psi}^G
\end{align*}

\textit{Step 1a. The system of commuted equations.} For any ${(m,l)\in H}$, we commute the equations with $\nabla^m\nabla_3^l:$
\[\nabla_3\nabla^m\nabla_3^l\underline{\nu}_{ABC}=-2\nabla_{[A}\nabla^m\nabla_3^l\underline{\alpha}_{B]C}+\underline{\mathcal{F}}_{(m+1)(l-1)(m+l)}(\Psi)+\underline{\mathcal{F}}_{(m)(l)(m+l)}(\underline{\Psi}^G)+\sum_{i+2j=m-1}\nabla^i\big(\slashed{Riem}^{j+1}\nabla_3^l\underline{\alpha}\big)\]
\[\nabla_4\nabla^m\nabla_3^l\underline{\alpha}_{AB}+\bigg(\frac{1}{2}+\frac{m}{n}\bigg)\tr\chi\nabla^m\nabla_3^l\underline{\alpha}_{AB}=-\nabla^C\nabla^m\nabla_3^l\underline{\nu}_{C(AB)}+\underline{\mathcal{F}}'_{(m)(l)(m+l)}(\Psi)+\underline{\mathcal{F}}_{(m+1)(l-1)(m+l)}(\underline{\Psi}^G)+\]\[+\underline{\mathcal{F}}_{(m)(l)(m+l)}(\underline{\Psi}^G)+\sum_{i+2j=m-1}\nabla^i\big(\slashed{Riem}^{j+1}\nabla_3^l\underline{\nu}\big)+\sum_{i=0}^{l-1}\nabla_3^{l-i}\tr\chi\nabla^m\nabla_3^{i}\underline{\alpha}\]
and we used the fact that $\nabla\psi=\nabla\psi^*.$ We also recall that $\tr\chi=\tr\chi^*+n/(v-u)$ and $\tr\underline{\chi}=\tr\underline{\chi}^*-n/(v-u)$. The last term in the second equation above can be written as:
\[\sum_{i=0}^{l-1}\nabla_3^{l-i}\tr\chi\nabla^m\nabla_3^{i}\underline{\alpha}=\underline{\mathcal{F}}'_{(m)(l)(m+l)}(\Psi)+O\bigg(\sum_{i=0}^{l-1}\frac{\big|\nabla^m\nabla_3^{i}\underline{\alpha}\big|}{(v-u)^{l-i+1}}\bigg).\]
We conjugate the equations with $w_{ml}=v^{\frac{3}{2}+m+l-p}|u|^{p-q}:$
\[\nabla_3w_{ml}\nabla^m\nabla_3^l\underline{\nu}_{ABC}+\frac{p-q}{|u|}w_{ml}\nabla^m\nabla_3^l\underline{\nu}_{ABC}=-2\nabla_{[A}w_{ml}\nabla^m\nabla_3^l\underline{\alpha}_{B]C}+w_{ml}Err_{ml}^{\underline{\nu}}\]
\[\nabla_4w_{ml}\nabla^m\nabla_3^l\underline{\alpha}_{AB}+\bigg(\frac{n-3}{2}-l+p\bigg)\frac{w_{ml}}{v}\nabla^m\nabla_3^l\underline{\alpha}_{AB}=-\nabla^Cw_{ml}\nabla^m\nabla_3^l\underline{\nu}_{C(AB)}+\]\[+w_{ml}Err_{ml}^{\underline{\alpha}}+O\bigg(\underline{v}\cdot\frac{w_{ml}}{v}\cdot\big|\nabla^m\nabla_3^l\underline{\alpha}\big|\bigg)\]
where we have the error terms:
\begin{align*}
    Err_{ml}^{\underline{\nu}}&=\underline{\mathcal{F}}_{(m)(l)(m+l)}(\Psi)+\underline{\mathcal{F}}_{(m+1)(l-1)(m+l)}(\Psi)+\sum_{i+2j=m-1}\nabla^i\big(\slashed{Riem}^{j+1}\nabla_3^l\underline{\alpha}\big)\\
    Err_{ml}^{\underline{\alpha}}&=\underline{\mathcal{F}}_{(m)(l)(m+l)}(\underline{\Psi}^G)+\underline{\mathcal{F}}_{(m+1)(l-1)(m+l)}(\underline{\Psi}^G)+\underline{\mathcal{F}}'_{(m)(l)(m+l)}(\Psi)\\
    &+\sum_{i+2j=m-1}\nabla^i\big(\slashed{Riem}^{j+1}\nabla_3^l\underline{\nu}\big)+O\bigg(\sum_{i=0}^{l-1}\frac{\big|\nabla^m\nabla_3^{i}\underline{\alpha}\big|}{(v-u)^{l-i+1}}\bigg)
\end{align*}

\textit{Step 1b. The energy estimates for $(\underline{\nu},\underline{\alpha}).$} For simplicity, we denote $\mathcal{D}=w_{ml}\nabla^m\nabla_3^l$. We notice that integration by parts and the bootstrap assumption give:
\begin{align*}
    &\int_{S^n}\Big(-\nabla_{[A}\mathcal{D}\underline{\alpha}_{B]C}\cdot\mathcal{D}\underline{\nu}^{ABC}-\mathcal{D}\underline{\alpha}^{AB}\cdot\nabla^C\mathcal{D}\underline{\nu}_{C(AB)}\Big)d\mathring{Vol}=\\
    &=\int_{S^n}\nabla_A\mathcal{D}\underline{\alpha}_{BC}\cdot\mathcal{D}\underline{\nu}^{A[BC]}d\mathring{Vol}+O\bigg(\int_{S^n}\frac{1}{v}\big|\mathcal{D}\underline{\alpha}\big|\cdot\big|\mathcal{D}\underline{\nu}\big|d\mathring{Vol}\bigg)=O\bigg(v^{-1}\int_{S^n}\big|\mathcal{D}\underline{\alpha}\big|\cdot\big|\mathcal{D}\underline{\nu}\big|d\mathring{Vol}\bigg)
\end{align*}
We use this identity in order to prove energy estimates for the Bianchi pair $(\underline{\nu},\underline{\alpha}).$ We contract the equation for $\mathcal{D}\underline{\nu}$ with $\frac{1}{2}\mathcal{D}\underline{\alpha},$ we contract the equation for $\mathcal{D}\underline{\alpha}$ with $\mathcal{D}\underline{\nu},$ and add the resulting equations. We integrate by parts, then multiply everything by $|u|^{2q}.$ Finally, we use Cauchy-Schwarz, the positive sign bulk terms, and the fact that $|u|/v\leq\underline{v}\ll1$ in order to absorb some of the error terms. We obtain the energy estimate in region~III for $v<\underline{v}$:
\[|u|^{2q}\int_{-v\underline{v}}^u\int_{S^n}w_{ml}^2\big|\nabla^m\nabla_3^l\underline{\alpha}\big|^2d\mathring{Vol}d\hat{u}+|u|^{2q}\int_{-v\underline{v}}^u\int_{-\frac{\hat{u}}{\underline{v}}}^v\int_{S^n}\frac{w_{ml}^2}{\hat{v}}\big|\nabla^m\nabla_3^l\underline{\alpha}\big|^2d\mathring{Vol}d\hat{v}d\hat{u}+\]
\[+|u|^{2q}\int_{-\frac{u}{\underline{v}}}^v\int_{S^n}w_{ml}^2\big|\nabla^m\nabla_3^l\underline{\nu}\big|^2d\mathring{Vol}d\hat{v}+|u|^{2q}\int_{-v\underline{v}}^u\int_{-\frac{\hat{u}}{\underline{v}}}^v\int_{S^n}\frac{w_{ml}^2}{|\hat{u}|}\big|\nabla^m\nabla_3^l\underline{\nu}\big|^2d\mathring{Vol}d\hat{v}d\hat{u}\lesssim\]\[\lesssim \underline{v}^{2p-2q}|u|^{2q}\int_{-\frac{u}{\underline{v}}}^v\int_{S^n}\hat{v}^{3+2m+2l-2q}\Big(\big|\nabla^m\nabla_3^l\underline{\alpha}\big|^2+\big|\nabla^m\nabla_3^l\underline{\nu}\big|^2\Big)(-\underline{v}\hat{v},\hat{v})d\mathring{Vol}d\hat{v} +\]\[+|u|^{2q}\int_{-v\underline{v}}^u\int_{-\frac{\hat{u}}{\underline{v}}}^v\int_{S^n}\Big(\hat{v}w_{ml}^2\big|Err_{ml}^{\underline{\alpha}}\big|^2+|\hat{u}|w_{ml}^2\big|Err_{ml}^{\underline{\nu}}\big|^2\Big)d\mathring{Vol}d\hat{v}d\hat{u}\]
We bound the data term on $\{u=-\underline{v}v\}$ using Proposition \ref{region II bounds proposition} by:
\[\underline{v}^{2p-2q}|u|^{2q}\int_{-\frac{u}{\underline{v}}}^v\int_{S^n}\hat{v}^{3+2m+2l-2q}\Big(\big|\nabla^m\nabla_3^l\underline{\alpha}\big|^2+\big|\nabla^m\nabla_3^l\underline{\nu}\big|^2\Big)(-\underline{v}\hat{v},\hat{v})\lesssim\underline{v}^{2p-2q}|u|^{2q}\int_{-\frac{u}{\underline{v}}}^v(\epsilon')^2\hat{v}^{-1-2q}d\hat{v}\lesssim\underline{v}^{2p}(\epsilon')^2\]
Here we notice the importance of the factor $v^{-q}$ in the definition of $w_{ml}$, needed to avoid logarithmic degeneracy. 

\textit{Step 1c. Bounding the error terms.} We bound the error terms one by one. We remark that:
\[\sum_{(m,l)\in H}w_{ml}^2\big|\underline{\mathcal{F}}_{(m+1)(l-1)(m+l)}(\Psi)\big|^2\lesssim\sum_{(m,l)\in H}w_{ml}^2\big|\underline{\mathcal{F}}_{(m)(l)(m+l)}(\Psi)\big|^2,\]
so bounding the first terms in $Err_{ml}^{\underline{\alpha}}$ and $Err_{ml}^{\underline{\nu}}$ will also imply control of the second terms once summing. For any $0\leq i\leq l-1,$ we have $(m,i)\in L$ and we get the bound:
\[|u|^{2q}\int_{-v\underline{v}}^u\int_{-\frac{\hat{u}}{\underline{v}}}^v\int_{S^n}\hat{v}w_{ml}^2\frac{\big|\nabla^m\nabla_3^{i}\underline{\alpha}\big|^2}{(\hat{v}-\hat{u})^{2l-2i+2}}\lesssim A\epsilon'^2\cdot|u|^{2q}\int_{-v\underline{v}}^u\int_{-\frac{\hat{u}}{\underline{v}}}^vw_{ml}^2|\hat{u}|^{-2p}\hat{v}^{-5-2m-2l+2p}\lesssim A\underline{v}\cdot\epsilon'^2\]
Next, we have the bound using the Sobolev inequalities:
\[|u|^{2q}\int_{-v\underline{v}}^u\int_{-\frac{\hat{u}}{\underline{v}}}^v\int_{S^n}\hat{v}w_{ml}^2\big|\underline{\mathcal{F}}_{(m)(l)(m+l)}(\underline{\Psi}^G)\big|^2\lesssim\sum_{\substack{i+j+k\leq m+l \\ i\leq l,k\leq m}}|u|^{2q}\int_{-v\underline{v}}^u\int_{-\frac{\hat{u}}{\underline{v}}}^v\int_{S^n}\hat{v}w_{ml}^2\big|\nabla^k\nabla_3^i\big(\psi^{j+1}\underline{\Psi}^G\big)\big|^2\]
\[\lesssim\sum_{\substack{|i|+j+|k|\leq m+l \\ |i|\leq l,|k|\leq m}}|u|^{2q}\int_{-v\underline{v}}^u\int_{-\frac{\hat{u}}{\underline{v}}}^v\int_{S^n}\frac{|\hat{u}|}{\hat{v}}\cdot\hat{v}^{3+2i_0+2k_0-2p}|\hat{u}|^{2p-2q-1}\Big|\nabla^{k_0}\nabla_3^{i_0}\underline{\Psi}^G\Big|\cdot\prod_{a=1}^{j+1}\hat{v}^{2+2i_a+2k_a}\Big|\nabla^{k_a}\nabla_3^{i_a}\psi\Big|\]
\[\lesssim|u|^{2q}\int_{-v\underline{v}}^u\int_{-\frac{\hat{u}}{\underline{v}}}^v\frac{|\hat{u}|}{\hat{v}}\sum_{(k_0,i_0)\in H}\hat{v}^{3+2i_0-2p}|\hat{u}|^{2p-2q-1}\Big\|\nabla_3^{i_0}\underline{\Psi}^G\Big\|_{\widetilde{H}^{k_0}(S_{\hat{u},\hat{v}})}^2\cdot\prod_{a=1}^{j+1}\sum_{(k_a,i_a)\in H}\hat{v}^{2+2i_a}\Big\|\nabla_3^{i_a}\psi\Big\|_{\widetilde{H}^{k_a}(S_{\hat{u},\hat{v}})}^2\]
\[\lesssim\underline{v}\cdot\sum_{(k,i)\in H}\big\|\underline{\Psi}^G\big\|^2_{\mathcal{C}_{k,i}}\lesssim\underline{v}\cdot A\epsilon'^2\]
Following the same steps, we also have that:
\[|u|^{2q}\int_{-v\underline{v}}^u\int_{-\frac{\hat{u}}{\underline{v}}}^v\int_{S^n}|\hat{u}|w_{ml}^2\big|\underline{\mathcal{F}}_{(m)(l)(m+l)}(\Psi)\big|^2d\mathring{Vol}d\hat{v}d\hat{u}\lesssim\]
\[\lesssim|u|^{2q}\int_{-v\underline{v}}^u\int_{-\frac{\hat{u}}{\underline{v}}}^v\frac{|\hat{u}|}{\hat{v}}\sum_{(k_0,i_0)\in H}\hat{v}^{2+2i_0-2p}|\hat{u}|^{2p-2q}\Big\|\nabla_3^{i_0}\Psi\Big\|_{\widetilde{H}^{k_0}(S_{\hat{u},\hat{v}})}^2\cdot\prod_{a=1}^{j+1}\sum_{(k_a,i_a)\in H}\hat{v}^{2+2i_a}\Big\|\nabla_3^{i_a}\psi\Big\|_{\widetilde{H}^{k_a}(S_{\hat{u},\hat{v}})}^2\]
\[\lesssim\underline{v}\cdot\sum_{(k,i)\in H}\big\|\Psi\big\|^2_{\mathcal{C}_{k,i}}\lesssim\underline{v}\cdot A\epsilon'^2\]
Similarly to the first error term, we have the estimate:
\[|u|^{2q}\int_{-v\underline{v}}^u\int_{-\frac{\hat{u}}{\underline{v}}}^v\int_{S^n}|\hat{v}|w_{ml}^2\big|\underline{\mathcal{F}}'_{(m)(l)(m+l)}(\Psi)\big|^2d\mathring{Vol}d\hat{v}d\hat{u}\lesssim\]
\[\lesssim|u|^{2q}\int_{-v\underline{v}}^u\int_{-\frac{\hat{u}}{\underline{v}}}^v\sum_{(k_0,i_0)}\hat{v}^{2+2i_0-2p}|\hat{u}|^{2p-2q}\Big\|\nabla_3^{i_0}\Psi\Big\|_{\widetilde{H}^{k_0}}^2\sum_{(k_1,i_1)}\hat{v}^{2+2i_1}\Big\|\nabla_3^{i_1}\psi^*\Big\|_{\widetilde{H}^{k_1}}^2\prod_{a=2}^{j+1}\sum_{(k_a,i_a)}\hat{v}^{2+2i_a}\Big\|\nabla_3^{i_a}\psi\Big\|_{\widetilde{H}^{k_a}}^2\]\[\lesssim\mathcal{R}_{u,v}\cdot\sum_{(k,i)\in H}\big\|\Psi\big\|^2_{\mathcal{C}_{k,i}}\lesssim A^2\epsilon'^4\]
where the indices in the above sums satisfy $(k_0,i_0),(k_1,i_1),(k_a,i_a)\in H.$ Next, we use the fact that $\slashed{Riem}=R+\psi\psi$ in order to bound:
\[|u|^{2q}\int_{-v\underline{v}}^u\int_{-\frac{\hat{u}}{\underline{v}}}^v\int_{S^n}|\hat{v}|w_{ml}^2\bigg|\sum_{k+2j=m-1}\nabla^k\big(\slashed{Riem}^{j+1}\nabla_3^l\underline{\Psi}^G\big)\bigg|^2d\mathring{Vol}d\hat{v}d\hat{u}\lesssim\]
\[\lesssim|u|^{2q}\int_{-v\underline{v}}^u\int_{-\frac{\hat{u}}{\underline{v}}}^v\sum_{(k_0,l)\in H}\hat{v}^{2+2l-2p}|\hat{u}|^{2p-2q}\Big\|\nabla_3^{l}\underline{\Psi}^G\Big\|_{\widetilde{H}^{k_0}(S_{\hat{u},\hat{v}})}^2\prod_{a=1}^{j+1}\bigg(\sum_{(k_a,0)\in L}\hat{v}^{4}\Big\|\slashed{Riem}\Big\|_{\widetilde{H}^{k_a}(S_{\hat{u},\hat{v}})}^2\bigg)d\hat{v}d\hat{u}\]
\[\lesssim\underline{v}\cdot\sum_{(k_0,l)\in H}\big\|\underline{\Psi}^G\big\|^2_{\mathcal{C}_{k_0,l}}\lesssim\underline{v}\cdot A\epsilon'^2\]
As before, a very simple modification of this argument allows us to also bound the corresponding term with $\slashed{Riem}$ in $Err_{ml}^{\underline{\nu}}.$ This completes bounding the error terms for our first energy estimate. In particular, we improved the bootstrap assumption for $\big\|\underline{\alpha}\big\|^2_{\mathcal{C}_{m,l}}$ and the last two terms in $\big\|\underline{\nu}\big\|^2_{\mathcal{C}_{m,l}}$. We point out that we already have good control of the bulk term for $\underline{\nu}$, which will help us in the next energy estimate.

\textit{Step 2. The Bianchi pair $(R,\underline{\nu}).$} By signature considerations, the equations can be written as:
\[\nabla_3R_{ABCD}=-2\nabla_{[A}\underline{\nu}_{|CD|B]}+\psi\Psi,\ \nabla_4\underline{\nu}_{ABC}=-2\nabla_{[A}\tau_{B]C}+\psi\underline{\Psi}^G\]

\textit{Step 2a. The system of commuted equations.} For any ${(m,l)\in H}$, we commute the equations with $\nabla^m\nabla_3^l:$
\[\nabla_3\nabla^m\nabla_3^lR_{ABCD}=-2\nabla_{[A}\nabla^m\nabla_3^l\underline{\nu}_{|CD|B]}+\underline{\mathcal{F}}_{(m+1)(l-1)(m+l)}(\Psi)+\underline{\mathcal{F}}_{(m)(l)(m+l)}(\Psi)+\sum_{i+2j=m-1}\nabla^i\big(\slashed{Riem}^{j+1}\nabla_3^l\underline{\Psi}^G\big)\]
\[\nabla_4\nabla^m\nabla_3^l\underline{\nu}_{ABC}=-2\nabla_{[A}\nabla^m\nabla_3^l\tau_{B]C}+\underline{\mathcal{F}}_{(m+1)(l-1)(m+l)}(\underline{\Psi}^G)+\underline{\mathcal{F}}_{(m)(l)(m+l)}(\underline{\Psi}^G)+\sum_{i+2j=m-1}\nabla^i\big(\slashed{Riem}^{j+1}\nabla_3^l\underline{\Psi}^G\big)\]
We conjugate the equations with $w_{ml}=v^{\frac{3}{2}+m+l-p}|u|^{p-q}:$
\[\nabla_3w_{ml}\nabla^m\nabla_3^lR_{ABCD}+\frac{p-q}{|u|}w_{ml}\nabla^m\nabla_3^lR_{ABCD}=-2\nabla_{[A}w_{ml}\nabla^m\nabla_3^l\underline{\nu}_{|CD|B]}+w_{ml}Err_{ml}^{R}\]
\[\nabla_4w_{ml}\nabla^m\nabla_3^l\underline{\nu}_{ABC}=-2\nabla_{[A}w_{ml}\nabla^m\nabla_3^l\tau_{B]C}+w_{ml}\underline{Err}_{ml}^{\underline{\nu}}+O\bigg(\frac{w_{ml}}{v}\cdot\big|\nabla^m\nabla_3^l\underline{\nu}\big|\bigg)\]
where we have the error terms:
\begin{align*}
    Err_{ml}^{R}&=\underline{\mathcal{F}}_{(m)(l)(m+l)}(\Psi)+\underline{\mathcal{F}}_{(m+1)(l-1)(m+l)}(\Psi)+\sum_{i+2j=m-1}\nabla^i\big(\slashed{Riem}^{j+1}\nabla_3^l\underline{\Psi}^G\big)\\
    \underline{Err}_{ml}^{\underline{\nu}}&=\underline{\mathcal{F}}_{(m)(l)(m+l)}(\underline{\Psi}^G)+\underline{\mathcal{F}}_{(m+1)(l-1)(m+l)}(\underline{\Psi}^G)+\sum_{i+2j=m-1}\nabla^i\big(\slashed{Riem}^{j+1}\nabla_3^l\underline{\Psi}^G\big)
\end{align*}

\textit{Step 2b. The energy estimates for $(R,\underline{\nu}).$} We denote $\mathcal{D}=w_{ml}\nabla^m\nabla_3^l$. We notice that integration by parts and the constraint equations give:
\begin{align*}
    &\int_{S^n}\Big(-2\nabla_{[A}\mathcal{D}\underline{\nu}_{|CD|B]}\cdot\mathcal{D}R^{ABCD}-4\mathcal{D}\underline{\nu}^{ABC}\cdot\nabla_{[A}\mathcal{D}\tau_{B]C}\Big)d\mathring{Vol}=\\
    &=\int_{S^n}2\mathcal{D}\underline{\nu}^{ABC}\cdot\Big(\nabla^D\mathcal{D}R_{DCAB}-2\nabla_{[A}\mathcal{D}\tau_{B]C}\Big)d\mathring{Vol}+O\bigg(v^{-1}\int_{S^n}\big|\mathcal{D}R\big|\cdot\big|\mathcal{D}\underline{\nu}\big|d\mathring{Vol}\bigg)\\
    &=O\bigg(\int_{S^n}\big|\mathcal{D}\underline{\nu}\big|\cdot\Big(v^{-1}\big|\mathcal{D}R\big|+w_{ml}\big|\underline{Err}_{ml}^{\underline{\nu}}\big|\Big)d\mathring{Vol}\bigg)
\end{align*}
Proceeding as before, we obtain the energy estimate in region~III for $v<\underline{v}$:
\[|u|^{2q}\int_{-v\underline{v}}^u\int_{S^n}w_{ml}^2\big|\nabla^m\nabla_3^l\underline{\nu}\big|^2d\mathring{Vol}d\hat{u}+\int_{-\frac{u}{\underline{v}}}^v\int_{S^n}w_{ml}^2\big|\nabla^m\nabla_3^lR\big|^2d\mathring{Vol}d\hat{v}+|u|^{2q}\int_{-v\underline{v}}^u\int_{-\frac{\hat{u}}{\underline{v}}}^v\int_{S^n}\frac{w_{ml}^2}{|\hat{u}|}\big|\nabla^m\nabla_3^lR\big|^2\lesssim\]
\[\lesssim\underline{v}^{2p-2q}|u|^{2q}\int_{-\frac{u}{\underline{v}}}^v\int_{S^n}\hat{v}^{3+2m+2l-2q}\Big(\big|\nabla^m\nabla_3^lR\big|^2+\big|\nabla^m\nabla_3^l\underline{\nu}\big|^2\Big)(-\underline{v}\hat{v},\hat{v})d\mathring{Vol}d\hat{v}+\]
\[+ |u|^{2q}\int_{-v\underline{v}}^u\int_{-\frac{\hat{u}}{\underline{v}}}^v\int_{S^n}\bigg(|\hat{u}|w_{ml}^2\big|Err_{ml}^{R}\big|^2+w_{ml}^2\big|\nabla^m\nabla_3^l\underline{\nu}\big|\cdot\big|\underline{Err}_{ml}^{\underline{\nu}}\big|+\frac{w_{ml}^2}{\hat{v}}\big|\nabla^m\nabla_3^l\underline{\nu}\big|^2\bigg)d\mathring{Vol}d\hat{v}d\hat{u}\]
\[\lesssim\underline{v}^{2p}\epsilon'^2+\underline{v}\cdot A\epsilon'^2+|u|^{2q}\int_{-v\underline{v}}^u\int_{-\frac{\hat{u}}{\underline{v}}}^v\int_{S^n}\Big(|\hat{u}|w_{ml}^2\big|Err_{ml}^{R}\big|^2+\hat{v}w_{ml}^2\big|\underline{Err}_{ml}^{\underline{\nu}}\big|^2\Big)d\mathring{Vol}d\hat{v}d\hat{u},\]
where we used the fact that we already control the bulk term for $\underline{\nu}$ from the previous energy estimate, and we bounded the data term as before. Moreover, we remark that each term of $\hat{v}w_{ml}^2\big|\underline{Err}_{ml}^{\underline{\nu}}\big|^2$ is contained in the error terms of $\hat{v}w_{ml}^2\big|Err_{ml}^{\underline{\alpha}}\big|^2,$ and similarly each term of $|\hat{u}|w_{ml}^2\big|Err_{ml}^{R}\big|^2$ is contained in the error terms of $|\hat{u}|w_{ml}^2\big|Err_{ml}^{\underline{\nu}}\big|^2.$ Therefore, the bounds on the error terms in the previous energy estimate also allow us to bound the right hand side of the above estimate by $\underline{v}\cdot A\epsilon'^2+\epsilon'^2$. This improves the bootstrap assumption for $\big\|\underline{\nu}\big\|^2_{\mathcal{C}_{m,l}}$ and the last two terms in $\big\|R\big\|^2_{\mathcal{C}_{m,l}}$.

The same argument applies in the case of the Bianchi pairs $(R,\nu)$ and $(\nu,\alpha),$ where at each step we use the control of the bulk terms from the previous step and proceed as above. As a result, we can improve the bootstrap assumption on the high regularity curvature norm and prove that:
\[\mathcal{C}_{\widetilde{u},\widetilde{v}}\lesssim\epsilon'^2+\underline{v}\cdot A\epsilon'^2.\]

\subsubsection{Estimates for Low Regularity Curvature Components}

In this section we improve the bootstrap assumption on the low regularity curvature components by proving:
\begin{proposition}\label{low regularity curvature proposition}
    There exists a constant $C\ll A$, such that we have the improved estimate:
    \[\mathcal{L}_{\widetilde{u},\widetilde{v}}\leq C\epsilon'^2.\]
\end{proposition}

\textit{Step 1. Improving the bounds for $\underline{\alpha}$.} For any ${(m,l)\in L}$, we consider the equation satisfied by $\underline{\alpha}$ conjugated by $x_{ml}=v^{2+m+l-p}|u|^{p}:$
\[\nabla_4x_{ml}\nabla^m\nabla_3^l\underline{\alpha}_{AB}+\bigg(\frac{n-4}{2}-l+p\bigg)\frac{x_{ml}}{v}\nabla^m\nabla_3^l\underline{\alpha}_{AB}=\]\[=x_{ml}\nabla^{m+1}\nabla_3^l\underline{\nu}+x_{ml}\widetilde{Err}_{ml}^{\underline{\alpha}}+O\bigg(\sum_{i=0}^{l-1}\frac{x_{mi}}{v}\cdot\big|\nabla^m\nabla_3^{i}\underline{\alpha}\big|+\underline{v}\cdot\frac{x_{ml}}{v}\cdot\big|\nabla^m\nabla_3^{l}\underline{\alpha}\big|\bigg)\]
where the implicit constant in the last term is independent of $p,$ and the error term is defined by:
\[\widetilde{Err}_{ml}^{\underline{\alpha}}=\underline{\mathcal{F}}_{(m)(l)(m+l)}(\underline{\Psi}^G)+\underline{\mathcal{F}}_{(m+1)(l-1)(m+l)}(\underline{\Psi}^G)+\underline{\mathcal{F}}'_{(m)(l)(m+l)}(\Psi)+\sum_{i+2j=m-1}\nabla^i\big(\slashed{Riem}^{j+1}\nabla_3^l\underline{\nu}\big)\]

\textit{Step 1a. The energy estimates for $\underline{\alpha}$.} We contract the above equation by $x_{ml}\nabla^m\nabla_3^l\underline{\alpha},$ and integrate in $v.$ We use the good bulk term obtained and Cauchy-Schwarz, together with a brief induction on $l$ argument needed to bound the second to last term. Summing for all $(m,l)\in L,$ we obtain the estimate in region~III for $v<\underline{v}$:
\[\sum_{(m,l)\in L}\int_{S^n}x_{ml}^2\big|\nabla^m\nabla_3^l\underline{\alpha}\big|^2(u,v)d\mathring{Vol}+\sum_{(m,l)\in L}\int_{-\frac{u}{\underline{v}}}^v\int_{S^n}\frac{x_{ml}^2}{\hat{v}}\big|\nabla^m\nabla_3^l\underline{\alpha}\big|^2(u,\hat{v})d\mathring{Vol}d\hat{v}\lesssim\]
\[\lesssim\sum_{(m,l)\in L}\int_{S^n}x_{ml}^2\big|\nabla^m\nabla_3^l\underline{\alpha}\big|^2(u,-\underline{v}^{-1}u)+\sum_{(m,l)\in L}\int_{-\frac{u}{\underline{v}}}^v\int_{S^n}\hat{v}x_{ml}^2\big|\nabla^{m+1}\nabla_3^l\underline{\nu}\big|^2+\sum_{(m,l)\in L}\int_{-\frac{u}{\underline{v}}}^v\int_{S^n} \hat{v}x_{ml}^2\big|\widetilde{Err}_{ml}^{\underline{\alpha}}\big|^2\]
The data term is bounded using Proposition \ref{region II bounds proposition} by:
\[\sum_{(m,l)\in L}\int_{S^n}x_{ml}^2\big|\nabla^m\nabla_3^l\underline{\alpha}\big|^2(u,-\underline{v}^{-1}u)d\mathring{Vol}\lesssim \underline{v}^{2p}\epsilon'^2.\]
Since $vx_{ml}^2=|u|^{2q}w_{(m+1)l}^2,$ we notice that the second term is bounded by $\big\|\underline{\nu}\big\|_{\mathcal{C}_{(m+1)l}}^2,$ which is controlled by $\underline{v}\cdot A\epsilon'^2$ according to Proposition \ref{high regularity curvature proposition}. We obtain the estimate:
\[\sum_{(m,l)\in L}\int_{S^n}x_{ml}^2\big|\nabla^m\nabla_3^l\underline{\alpha}\big|^2(u,v)d\mathring{Vol}+\sum_{(m,l)\in L}\int_{-\frac{u}{\underline{v}}}^v\int_{S^n}\frac{x_{ml}^2}{\hat{v}}\big|\nabla^m\nabla_3^l\underline{\alpha}\big|^2(u,\hat{v})d\mathring{Vol}d\hat{v}\lesssim\]
\[\lesssim\epsilon'^2+\underline{v}\cdot A\epsilon'^2+\sum_{(m,l)\in L}\int_{-\frac{u}{\underline{v}}}^v\int_{S^n}\bigg(\hat{v}x_{ml}^2\big|\underline{\mathcal{F}}_{(m)(l)(m+l)}(\underline{\Psi}^G)\big|^2+\hat{v}x_{ml}^2\big|\underline{\mathcal{F}}_{(m+1)(l-1)(m+l)}(\underline{\Psi}^G)\big|^2\bigg)d\mathring{Vol}d\hat{v}+\]\[+\sum_{(m,l)\in L}\int_{-\frac{u}{\underline{v}}}^v\int_{S^n} \hat{v}x_{ml}^2\big|\underline{\mathcal{F}}'_{(m)(l)(m+l)}(\Psi)\big|^2d\mathring{Vol}d\hat{v}+\sum_{(m,l)\in L}\sum_{i+2j=m-1}\int_{-\frac{u}{\underline{v}}}^v\int_{S^n}\hat{v}x_{ml}^2\Big|\nabla^i\big(\slashed{Riem}^{j+1}\nabla_3^l\underline{\nu}\big)\Big|^2d\mathring{Vol}d\hat{v}\]

\textit{Step 1b. Bounding the error terms.} As before, we remark that:
\[\sum_{(m,l)\in L}x_{ml}^2\big|\underline{\mathcal{F}}_{(m+1)(l-1)(m+l)}(\underline{\Psi}^G)\big|^2\lesssim\sum_{(m,l)\in L}x_{ml}^2\big|\underline{\mathcal{F}}_{(m)(l)(m+l)}(\underline{\Psi}^G)\big|^2\]
so bounding the first term in $\widetilde{Err}_{ml}^{\underline{\alpha}}$ will also imply control of the second term. We bound the first term:
\[|u|^{2q}\int_{-\frac{u}{\underline{v}}}^v\int_{S^n}\hat{v}^2w_{ml}^2\big|\underline{\mathcal{F}}_{(m)(l)(m+l)}(\underline{\Psi}^G)\big|^2\lesssim\sum_{\substack{i+j+k\leq m+l \\ i\leq l,k\leq m}}|u|^{2q}\int_{-\frac{u}{\underline{v}}}^v\int_{S^n}\hat{v}^2w_{ml}^2\big|\nabla^k\nabla_3^i\big(\psi^{j+1}\underline{\Psi}^G\big)\big|^2\]
\[\lesssim|u|^{2q}\int_{-\frac{u}{\underline{v}}}^v\sum_{(k_0,i_0)\in L}\hat{v}^{3+2i_0-2p}|u|^{2p-2q}\Big\|\nabla_3^{i_0}\underline{\Psi}^G\Big\|_{\widetilde{H}^{k_0}(S_{u,\hat{v}})}^2\cdot\prod_{a=1}^{j+1}\sum_{(k_a,i_a)\in L}\hat{v}^{2+2i_a}\Big\|\nabla_3^{i_a}\psi\Big\|_{\widetilde{H}^{k_a}(S_{u,\hat{v}})}^2\]
\[\lesssim\sum_{(k,i)\in L}\big\|\underline{\Psi}^G\big\|^2_{\mathcal{C}_{k,i}}\lesssim\underline{v}\cdot A\epsilon'^2,\]
where we used the improved estimates from Proposition \ref{high regularity curvature proposition}. Next, we have:
\[|u|^{2q}\int_{-\frac{u}{\underline{v}}}^v\int_{S^n}\hat{v}^2w_{ml}^2\big|\underline{\mathcal{F}}'_{(m)(l)(m+l)}(\Psi)\big|^2d\mathring{Vol}d\hat{v}\lesssim\]\[|u|^{2q}\int_{-\frac{u}{\underline{v}}}^v\bigg(\sum_{(k_0,i_0)\in L}\hat{v}^{3+2i_0-2p}|u|^{2p-2q}\Big\|\nabla_3^{i_0}\Psi\Big\|_{\widetilde{H}^{k_0}}^2\sum_{(k_1,i_1)\in L}\hat{v}^{2+2i_1}\Big\|\nabla_3^{i_1}\psi^*\Big\|_{\widetilde{H}^{k_1}}^2\prod_{a=2}^{j+1}\sum_{(k_a,i_a)\in L}\hat{v}^{2+2i_a}\Big\|\nabla_3^{i_a}\psi\Big\|_{\widetilde{H}^{k_a}}^2\bigg)\]
\[\lesssim\mathcal{R}_{u,v}\cdot\sum_{(k,i)\in L}\big\|\underline{\Psi}^G\big\|^2_{\mathcal{C}_{k,i}}+A\epsilon'^2\sum_{(k,i)\in L}|u|^{2q}\int_{-\frac{u}{\underline{v}}}^v\hat{v}^{3+2i-2p}|u|^{2p-2q}\Big\|\nabla_3^{i}\underline{\alpha}\Big\|_{\widetilde{H}^k(S_{u,\hat{v}})}^2d\hat{v}\]
The second term can be absorbed on the left hand side of the estimate using the good bulk term. We bound the last error term using $\slashed{Riem}=R+\psi\psi$:
\[\sum_{i+2j=m-1}|u|^{2q}\int_{-\frac{u}{\underline{v}}}^v\int_{S^n}\hat{v}^2w_{ml}^2\Big|\nabla^i\big(\slashed{Riem}^{j+1}\nabla_3^l\underline{\nu}\big)\Big|^2d\mathring{Vol}d\hat{v}\lesssim\]
\[\lesssim|u|^{2q}\int_{-\frac{\hat{u}}{\underline{v}}}^v\sum_{(k_0,l)\in L}\hat{v}^{3+2l-2p}|\hat{u}|^{2p-2q}\Big\|\nabla_3^{l}\underline{\Psi}^G\Big\|_{\widetilde{H}^{k_0}}^2\prod_{a=1}^{j+1}\bigg(\sum_{(k_a,0)\in L}\hat{v}^{4}\Big\|\slashed{Riem}\Big\|_{\widetilde{H}^{k_a}}^2\bigg)\lesssim\sum_{(k,l)\in L}\big\|\underline{\Psi}^G\big\|^2_{\mathcal{C}_{k,l}}\lesssim\underline{v}A\epsilon'^2\]
Thus, we improved the bootstrap assumption for $\big\|\underline{\alpha}\big\|^2_{\mathcal{L}_{m,l}}$ for any $(m,l)\in L$. 

\textit{Step 2. Improving the bounds for $\underline{\Psi}^G$.} The curvature components $\underline{\Psi}^G$ satisfy the schematic equation:
\[\nabla_3v^{2+m+l}\nabla^m\nabla_3^l\underline{\Psi}^G=v^{2+m+l}Err_{ml},\]
where we have the error term:
\[Err_{ml}=\nabla^{m+1}\nabla_3^l\Psi^G+\underline{\mathcal{F}}_{(m)(l)(m+l)}(\Psi)+\underline{\mathcal{F}}_{(m+1)(l-1)(m+l)}(\Psi)\]
We use \cite[Lemma~9.6]{nakedsing} and the bound in Proposition \ref{region II bounds proposition} to obtain the estimate:
\[\sup_{(u,v)\in P_{\widetilde{u},\widetilde{v}}}\int_{S^n}v^{4+2m+2l}\big|\nabla^m\nabla_3^l\underline{\Psi}^G\big|^2d\mathring{Vol}\lesssim \epsilon'^2+\sup_{(u,v)\in P_{\widetilde{u},\widetilde{v}}}\bigg(\int_{-v\underline{v}}^uv^{2+m+l}\bigg(\int_{S^n}\big|Err_{ml}\big|^2d\mathring{Vol}\bigg)^{\frac{1}{2}}d\hat{u}\bigg)^2\]
\[\lesssim \epsilon'^2+\sup_{(u,v)\in P_{\widetilde{u},\widetilde{v}}}\underline{v}^{1-100p}\sup_{\mathring{u}\in[-v\underline{v},u]}|\mathring{u}|^{100q}\int_{-v\underline{v}}^{\mathring{u}}\int_{S^n}v^{5+2m+2l-100p}|\hat{u}|^{100p-100q}\big|Err_{ml}\big|^2d\mathring{Vol}d\hat{u}\]
Since $\Psi^G\neq\alpha,$ the first error term is bounded by:
\[\sup_{(u,v)\in P_{\widetilde{u},\widetilde{v}}}\underline{v}^{1-100p}\sup_{\mathring{u}\in[-v\underline{v},u]}|\mathring{u}|^{100q}\int_{-v\underline{v}}^{\mathring{u}}\int_{S^n}v^{5+2m+2l-100p}|\hat{u}|^{100p-100q}\big|\nabla^{m+1}\nabla_3^l\Psi^G\big|^2d\mathring{Vol}d\hat{u}\lesssim\underline{v}^{1-100p}\mathcal{C}_{\widetilde{u},\widetilde{v}}\]
Arguing as before, bounding the second error term will also imply bounds for the third error term once summing. We have the bound:
\[\underline{v}^{1-100p}\sup_{\mathring{u}\in[-v\underline{v},u]}|\mathring{u}|^{100q}\int_{-v\underline{v}}^{\mathring{u}}\int_{S^n}v^{5+2m+2l-100p}|\hat{u}|^{100p-100q}\big|\underline{\mathcal{F}}_{(m)(l)(m+l)}(\Psi)\big|^2d\mathring{Vol}d\hat{u}\]
\[\lesssim\underline{v}^{1-100p}\sup_{\mathring{u}\in[-v\underline{v},u]}|\mathring{u}|^{100q}\int_{-v\underline{v}}^{\mathring{u}}\sum_{(k_0,i_0)\in L}v^{3+2i_0-100p}|\hat{u}|^{100p-100q}\Big\|\nabla_3^{i_0}\Psi\Big\|_{\widetilde{H}^{k_0}}^2\prod_{a=1}^{j+1}\sum_{(k_a,i_a)\in L}\hat{v}^{2+2i_a}\Big\|\nabla_3^{i_a}\psi\Big\|_{\widetilde{H}^{k_a}}^2\]
\[\lesssim\mathcal{L}_{u,v}\cdot\underline{v}^{1-100p}\sup_{\mathring{u}\in[-v\underline{v},u]}|\mathring{u}|^{100q}\int_{-v\underline{v}}^{\mathring{u}}v^{-1-98p}|\hat{u}|^{98p-100q}d\hat{u}\lesssim\underline{v}\cdot\mathcal{L}_{u,v}\lesssim\underline{v}\cdot A\epsilon'^2\]
As a result, we improved the bootstrap assumption on the low regularity curvature norm and proved that:
\[\mathcal{L}_{\widetilde{u},\widetilde{v}}\lesssim\epsilon'^2+\underline{v}\cdot A\epsilon'^2.\]

\subsubsection{Estimates for Ricci Coefficients}

In this section we improve the bootstrap assumption on the Ricci coefficients by proving:
\begin{proposition}\label{Ricci coefficients proposition}
    There exists a constant $C\ll A,$ such that we have the improved estimate:
    \[\mathcal{R}_{\widetilde{u},\widetilde{v}}\leq C\epsilon'^2.\]
\end{proposition}

We notice that the Ricci coefficients satisfy schematic equations:
\[\nabla_3\big(\hat{\chi},\underline{\hat{\chi}}\big)=\Psi^G+\psi\psi^*,\ \nabla_3\big(\tr\chi^*,\tr\underline{\chi}^*\big)=\psi\psi^*.\]
We allow the curvature term on the right hand side in order to treat all the equations at the same time. For any ${(m,l)\in H},$ we commute the equations with $\nabla^m\nabla_3^l:$
\[\nabla_3v^{1+m+l}\nabla^m\nabla_3^l\psi^*=v^{1+m+l}Err_{ml},\]
where we have the error term:
\[Err_{ml}=\nabla^{m}\nabla_3^l\Psi^G+\underline{\mathcal{F}}_{(m)(l)(m+l)}(\psi^*),\]
and we define $\underline{\mathcal{F}}_{(m)(l)(m+l)}(\psi^*)$ just as $\underline{\mathcal{F}}_{(m)(l)(m+l)}(\Psi),$ but replacing $\Psi$ with $\psi^*.$ We integrate the above equation and apply \cite[Lemma~9.6]{nakedsing} as before to obtain:
\[\sup_{(u,v)\in P_{\widetilde{u},\widetilde{v}}}\int_{S^n}v^{2+2m+2l}\big|\nabla^m\nabla_3^l\psi^*\big|^2d\mathring{Vol}\lesssim \epsilon'^2+\sup_{(u,v)\in P_{\widetilde{u},\widetilde{v}}}\bigg(\int_{-v\underline{v}}^uv^{1+m+l}\bigg(\int_{S^n}\big|Err_{ml}\big|^2d\mathring{Vol}\bigg)^{\frac{1}{2}}d\hat{u}\bigg)^2\]
\[\lesssim \epsilon'^2+\sup_{(u,v)\in P_{\widetilde{u},\widetilde{v}}}\underline{v}^{1-100p}\sup_{\mathring{u}\in[-v\underline{v},u]}|\mathring{u}|^{100q}\int_{-v\underline{v}}^{\mathring{u}}\int_{S^n}v^{3+2m+2l-100p}|\hat{u}|^{100p-100q}\big|Err_{ml}\big|^2d\mathring{Vol}d\hat{u}\]
Since $\Psi^G\neq\alpha,$ the first error term is bounded by:
\[\sup_{(u,v)\in P_{\widetilde{u},\widetilde{v}}}\underline{v}^{1-100p}\sup_{\mathring{u}\in[-v\underline{v},u]}|\mathring{u}|^{100q}\int_{-v\underline{v}}^{\mathring{u}}\int_{S^n}v^{3+2m+2l-100p}|\hat{u}|^{100p-100q}\big|\nabla^{m}\nabla_3^l\Psi^G\big|^2d\mathring{Vol}d\hat{u}\lesssim\underline{v}^{1-100p}\mathcal{C}_{\widetilde{u},\widetilde{v}}\]
We have the bound for the second error term:
\[\underline{v}^{1-100p}\sup_{\mathring{u}\in[-v\underline{v},u]}|\mathring{u}|^{100q}\int_{-v\underline{v}}^{\mathring{u}}\int_{S^n}v^{3+2m+2l-100p}|\hat{u}|^{100p-100q}\big|\underline{\mathcal{F}}_{(m)(l)(m+l)}(\psi^*)\big|^2d\mathring{Vol}d\hat{u}\]
\[\lesssim\underline{v}^{1-100p}\sup_{\mathring{u}\in[-v\underline{v},u]}|\mathring{u}|^{100q}\int_{-v\underline{v}}^{\mathring{u}}\sum_{(k_0,i_0)\in H}v^{1+2i_0-100p}|\hat{u}|^{100p-100q}\Big\|\nabla_3^{i_0}\psi^*\Big\|_{\widetilde{H}^{k_0}}^2\prod_{a=1}^{j+1}\sum_{(k_a,i_a)\in H}\hat{v}^{2+2i_a}\Big\|\nabla_3^{i_a}\psi\Big\|_{\widetilde{H}^{k_a}}^2\]
\[\lesssim\mathcal{R}_{u,v}\cdot\underline{v}^{1-100p}\sup_{\mathring{u}\in[-v\underline{v},u]}|\mathring{u}|^{100q}\int_{-v\underline{v}}^{\mathring{u}}v^{-1-98p}|\hat{u}|^{98p-100q}d\hat{u}\lesssim\underline{v}\cdot 
A\epsilon'^2\]
As a result, we improved the bootstrap assumption on the Ricci coefficients norm and proved that:
\[\mathcal{R}_{\widetilde{u},\widetilde{v}}\lesssim\epsilon'^2+\underline{v}\cdot A\epsilon'^2.\]
\begin{remark}
    We can also prove the estimate on $\nabla_3^{\frac{n-2}{2}}\psi^*$ in Proposition \ref{region III bounds proposition}, even though this term was not needed in the bootstrap argument. For any $i\leq N_3-1,$ we have that $\big(i,\frac{n-4}{2}\big)\in L,$ and the equation:
    \[\nabla^i\nabla_3^{\frac{n-2}{2}}\psi^*=\nabla^i\nabla_3^{\frac{n-4}{2}}\Psi+\nabla^i\nabla_3^{\frac{n-4}{2}}\big(\psi\psi^*\big)\]
    Using the estimates in Proposition \ref{low regularity curvature proposition} and Proposition \ref{Ricci coefficients proposition}, we have that for all $i\leq N_3-1:$
    \[\big\|\nabla^i\nabla_3^{\frac{n-2}{2}}\psi^*\big\|_{L^{2}(S_{u,v})}\lesssim\epsilon'|u|^{-p}\cdot|v|^{-1-i-\frac{n-2}{2}+p}.\]
\end{remark}

\subsubsection{Estimates for Metric Coefficients}
In this section we improve the bootstrap assumption on the metric coefficients by proving:
\begin{proposition}\label{metric coefficients proposition}
    There exists a constant $C\ll A$, such that we have the improved estimate:
    \[\mathcal{M}_{\widetilde{u},\widetilde{v}}\leq C\epsilon'^2.\]
\end{proposition}

The metric equations imply that $\mathcal{L}_3\mathcal{L}_{\theta}^m\slashed{g}^*=\mathcal{L}_{\theta}^m\psi^*.$ We denote by $\Gamma$ the Christoffel symbols of the metric $\slashed{g}.$ We notice that the bootstrap assumption implies that for any $(m,0)\in H$:
\begin{equation}\label{Christoffel symbols bound}
    v^{2m}\big\|\mathcal{L}_{\theta}^{m-1}\Gamma\big\|_{L^2(S_{u,v})}^2\lesssim 1.
\end{equation}
For any horizontal tensor $\phi$ we have the formula $\mathcal{L}_{\theta}\phi=\nabla\phi+\Gamma\cdot\phi,$ which implies by induction that:
\begin{equation}\label{Lie derivative in terms of covariant derivative}
    \mathcal{L}_{\theta}^m\phi=\nabla^m\phi+\sum_{i+j+k=m-1}\mathcal{L}_{\theta}^i\big(\Gamma^{j+1}\big)\nabla^k\phi.
\end{equation}
Thus, we obtain the equation:
\[\mathcal{L}_3\mathcal{L}_{\theta}^m\slashed{g}^*=\nabla^m\psi^*+\sum_{i+j+k=m-1}\mathcal{L}_{\theta}^i\big(\Gamma^{j+1}\big)\nabla^k\psi^*=:Err_m.\]
We have the estimate for any $(m,0)\in H$:
\[\sup_{(u,v)\in P_{\widetilde{u},\widetilde{v}}}v^{2m}\big\|\mathcal{L}_{\theta}^m\slashed{g}^*\big\|_{L^2(S_{u,v})}^2\lesssim \epsilon'^2+\sup_{(u,v)\in P_{\widetilde{u},\widetilde{v}}}\bigg(\int_{-v\underline{v}}^uv^{m}\bigg(\int_{S^n}\big|Err_{m}\big|^2d\mathring{Vol}\bigg)^{\frac{1}{2}}d\hat{u}\bigg)^2\]
\[\lesssim \epsilon'^2+\sup_{(u,v)\in P_{\widetilde{u},\widetilde{v}}}\underline{v}^{1-100p}\sup_{\mathring{u}\in[-v\underline{v},u]}|\mathring{u}|^{100q}\int_{-v\underline{v}}^{\mathring{u}}\int_{S^n}v^{1+2m-100p}|\hat{u}|^{100p-100q}\big|Err_{m}\big|^2d\mathring{Vol}d\hat{u}\]
The first error term can be bounded as usual by $\underline{v}\cdot A\epsilon'^2.$ For the second term we have:
\[\sum_{i+j+k=m-1}\underline{v}^{1-100p}\sup_{\mathring{u}\in[-v\underline{v},u]}|\mathring{u}|^{100q}\int_{-v\underline{v}}^{\mathring{u}}\int_{S^n}v^{1+2k-100p}|\hat{u}|^{100p-100q}\big|\nabla^k\psi^*\big|^2v^{2+2i+2j}\big|\mathcal{L}_{\theta}^i\big(\Gamma^{j+1}\big)\big|^2d\mathring{Vol}d\hat{u}\]
\[\lesssim\underline{v}^{1-100p}\sup_{\mathring{u}\in[-v\underline{v},u]}|\mathring{u}|^{100q}\int_{-v\underline{v}}^{\mathring{u}}\sum_{(k,0)\in L}v^{1-100p}|\hat{u}|^{100p-100q}\big\|\psi^*\big\|^2_{\widetilde{H}^k(S_{\hat{u},v})}\prod_{a=1}^{j+1}\sum_{(k_a,0)\in L}\int_{S^n}v^{2+2k_a}\big|\mathcal{L}_{\theta}^{k_a}\Gamma\big|^2d\mathring{Vol}d\hat{u},\]
which is also bounded by $\underline{v}\cdot A\epsilon'^2,$ allowing us to improve the bootstrap assumption:
\[\mathcal{M}_{\widetilde{u},\widetilde{v}}\lesssim\epsilon'^2+\underline{v}\cdot A\epsilon'^2.\]

To conclude this section, we establish the proof of the propagation of regularity result in Theorem \ref{propagation of regularity theorem}:

\textit{Proof of Theorem \ref{propagation of regularity theorem}.} We explain how a slight modification of our previous arguments implies the proof of this result. We first consider smooth initial data $\big(\slashed{g}_0,h\big)$ satisfying (\ref{smallness assumption g}), and prove self-similar estimates on the double null quantities for any angular regularity. For any $K\geq N_3,\ 0<q\ll p\ll 1,\ 0<\underline{v}_K\ll1$, we denote by $\widetilde{\mathcal{C}}_{K},\ \widetilde{\mathcal{L}}_{K},\ \widetilde{\mathcal{R}}_{K},\ \widetilde{\mathcal{M}}_{K}$ the norms defined in Section~\ref{stability norms section} with $(N_3,q,p,\underline{v})$ replaced by $(K,q,p,\underline{v}_K)$, and without the $*$ in the definitions of $\widetilde{\mathcal{R}}_{K}$ and $\widetilde{\mathcal{M}}_{K}.$ We prove by induction that for any $K\geq N_3,$ there exist $0<\underline{v}_K\ll1$ small enough and $C(K)>0$ large enough such that for any $(\Tilde{u},\Tilde{v})\in\{0\leq v\leq\underline{v}_K,\ -v\underline{v}_K\leq u\leq0\}:$
\[\widetilde{\mathcal{C}}_{K}(\Tilde{u},\Tilde{v})+\widetilde{\mathcal{L}}_{K}(\Tilde{u},\Tilde{v})+\widetilde{\mathcal{R}}_{K}(\Tilde{u},\Tilde{v})+\widetilde{\mathcal{M}}_{K}(\Tilde{u},\Tilde{v})\leq C(K).\]
The base case $K=N_3$ was proved in Proposition \ref{region III bounds proposition}. We assume that the induction hypothesis holds up to $K-1$ and prove it for $K.$ The strategy is to split the spacetime into regions $I_K,\ II_K,\ III_K$, defined analogously to $I,\ II,\ III$ but with $\underline{v}$ replaced by $\underline{v}_K.$

In the region $I_K=\big\{-1\leq u\leq0,\ 0\leq v\leq-u\underline{v}_K\big\}$ we have propagation of regularity for the local existence result of \cite{selfsimilarvacuum}. For $0<\underline{v}_K\ll1$ small enough, we apply the regular estimates of \cite{selfsimilarvacuum} and we obtain that the estimates of Proposition \ref{region I bounds proposition} hold up to $K+4c_0n$ angular regularity, with the $\epsilon$ on the right hand side replaced by a constant $C_K^I>0$ that is not assumed to be small.

In the region $II_K=\big\{-1\leq u\leq0,\ -u\underline{v}_K\leq v\leq-u/\underline{v}_K\big\}$, the propagation of regularity follows by the argument of \cite[Section~7]{nakedsing}. The idea is to conjugate the equations by $\exp(D_K\cdot v/u),$ for some large constant $D_K.$ This argument is equivalent to using the Gronwall lemma. We obtain that the estimates of Proposition \ref{region II bounds proposition} hold up to $K+3c_0n$ angular regularity, once again with the $\epsilon$ on the right hand side replaced by a constant $C_K^{II}>0$ that is not assumed to be small.

As a result, there exists a constant $C'_K>0$ such that on $\big\{v=-u/\underline{v}_K\big\}$ we have:
\[\widetilde{\mathcal{C}}_{K}(u,-u/\underline{v}_K)+\widetilde{\mathcal{L}}_{K}(u,-u/\underline{v}_K)+\widetilde{\mathcal{R}}_{K}(u,-u/\underline{v}_K)+\widetilde{\mathcal{M}}_{K}(u,-u/\underline{v}_K)\leq C'_K.\]
Moreover, we can assume that $C'_K\gg C(K-1).$ For some $A>0$ large, depending on $K$, we make the bootstrap assumption in the region $III_K=\big\{-1\leq u\leq0,\ -u/\underline{v}_K\leq v\leq\underline{v}_K\big\}$:
\[\widetilde{\mathcal{C}}_{K}(\Tilde{u},\Tilde{v})+\widetilde{\mathcal{L}}_{K}(\Tilde{u},\Tilde{v})+\widetilde{\mathcal{R}}_{K}(\Tilde{u},\Tilde{v})+\widetilde{\mathcal{M}}_{K}(\Tilde{u},\Tilde{v})\leq 2AC'_K.\]
We briefly explain how to improve this bootstrap assumption in order to prove that:
\[\widetilde{\mathcal{C}}_{K}(\Tilde{u},\Tilde{v})+\widetilde{\mathcal{L}}_{K}(\Tilde{u},\Tilde{v})+\widetilde{\mathcal{R}}_{K}(\Tilde{u},\Tilde{v})+\widetilde{\mathcal{M}}_{K}(\Tilde{u},\Tilde{v})\leq AC'_K.\]
We repeat the proof of Proposition \ref{bootstrap argument proposition}, with angular regularity given by $K.$ Since all the quantities with angular regularity up to $K-1$ are already bounded by the induction hypothesis, the equations are linear in the top order terms that need to be bounded. Thus, each error term is a product of factors where at most one such factor is bounded using the bootstrap assumption, while the remaining factors are bounded using the induction hypothesis. In particular, at each step in the proof of Proposition \ref{bootstrap argument proposition} where we bounded an error term by $\lesssim\underline{v}^{\frac{1}{2}-100p}\cdot A\cdot\epsilon'^2,$ we would now have the bound $\leq C\big(C(K-1)\big)+C\big(C(K-1)\big)\underline{v}_K^{\frac{1}{2}-100p}\cdot A\cdot C'_K.$ We also notice that we had $O(n)$ top order error terms bounded by $\lesssim A^2\cdot\epsilon'^4,$ which would now be bounded by $\leq C\big(C(K-1)\big)+\epsilon'^2\cdot A\cdot C'_K\cdot O(n).$ Therefore, for $0<\underline{v}_K\ll1$ small enough we can improve the bootstrap assumption as desired. This establishes the induction hypothesis for $K,$ with $C(K)=AC'_K.$ Since the above bounds hold for all $K\geq N_3,$ we obtain that the spacetime $\big(\mathcal{M},g\big)$ is smooth.
\qed

\section{Asymptotic Completeness}\label{asymptotic completeness section}

In this section we prove the second statement of Theorem \ref{main theorem of the paper ambient}, establishing asymptotic completeness. Given suitably small Cauchy initial data on a spacelike hypersurface in the sense of Remark~\ref{remark about cauchy data}, Theorem~\ref{stability of de sitter theorem in section} implies global existence in the region $\{u < 0,\ v > 0\}$. We prove the existence of induced smooth scattering data at $\{u=0\}$ and $\{v=0\}.$ In the original $(n+1)$-dimensional formulation, this represents the proof of the second statement of Theorem~\ref{main theorem of the paper}.

\begin{theorem}\label{asymptotic completeness in section theorem}
    Consider a global straight self-similar vacuum spacetime $\big(\mathcal{M},g\big)$ which satisfies the conclusion of Theorem~\ref{stability of de sitter theorem in section}. There exists induced straight data $\big(\underline{\slashed{g}_{0}},\underline{h}\big)$ at $(u,v)=(0,1)$, such that $\big(\mathcal{M},g\big)$ is the unique straight self-similar vacuum spacetime determined by this initial data, and moreover we have the estimates for $0\leq i\leq N'$:
    \begin{align*}
        \big\|\mathcal{L}_{\theta}^i\big(\underline{\slashed{g}_{0}}-\slashed{g}_{S^n}\big)\big\|_{L^2(S^n)}&\lesssim\epsilon'\\
        \big\|\nabla^i\underline{\mathcal{O}}\big\|_{L^2(S^n)}&\lesssim\epsilon'\\
        \big\|\nabla^i\underline{h}\big\|_{L^2(S^n)}&\lesssim\epsilon',
    \end{align*}
    where $N'=N-c_0n$ as in Theorem~\ref{stability of de sitter theorem in section}, $\epsilon'=\epsilon^{1-2\delta},$ and $\underline{\mathcal{O}}$ is the obstruction tensor of the metric $\underline{\slashed{g}_{0}}$.

    Moreover, if the spacetime $\big(\mathcal{M},g\big)$ is smooth the induced data $\big(\underline{\slashed{g}_{0}},\underline{h}\big)$ is also smooth.
\end{theorem}

\begin{remark}
    The proof of the second statement of Theorem~\ref{main theorem of the paper ambient} follows from the second part of Theorem~\ref{stability of de sitter theorem in section}, Theorem~\ref{asymptotic completeness in section theorem}, and Remark~\ref{remark about cauchy data}. We prove Theorem~\ref{asymptotic completeness in section theorem} at the end of this section using Propositions~\ref{asymptotic completeness preliminary proposition} and \ref{asymptotic completeness preliminary proposition 2}.
\end{remark}

We follow the strategy outlined in Section~\ref{asymptotic completeness intro section} of the introduction. We use the quantitative estimates on the global solution obtained in Theorem~\ref{stability of de sitter theorem in section} in order to recover the induced scattering data at $\{u=0\}$ and $\{v=0\}.$ By time orientation reversal symmetry, it suffices to prove the existence of induced asymptotic data at $\{u=0\}$. 

We first prove that certain regular quantities can be extended to $\{u=0\}.$ These correspond to determining the first $\frac{n-2}{2}$ terms in the expansion of $\slashed{g}$ at $\{u=0\}.$
\begin{proposition}\label{asymptotic completeness preliminary proposition}
    We have the following continuous extensions to $\{u=0\}:$
    \begin{enumerate}
        \item For any $0\leq l\leq\frac{n-4}{2}$ and $0\leq m\leq N'+\frac{n-6}{2}-l$, we have:
        \[\nabla^m\nabla_3^l\underline{\Psi}^G\in W^{1,1}_u\big([-1,0]\big)L^2(S^n).\]
        Moreover, we have the self-similarity relations on $\{u=0\}:$
        \begin{align*}
            \nabla^m\nabla_3^l\big(\alpha_{AB},\tau_{AB}\big)(0,\lambda v)&=\lambda^{-l}\nabla^m\nabla_3^l\big(\alpha_{AB},\tau_{AB}\big)(0,v)\\
            \nabla^m\nabla_3^l\big(\nu_{ABC},\underline{\nu}_{ABC}\big)(0,\lambda v)&=\lambda^{1-l}\nabla^m\nabla_3^l\big(\nu_{ABC},\underline{\nu}_{ABC}\big)(0,v)\\
            \nabla^m\nabla_3^lR_{ABCD}(0,\lambda v)&=\lambda^{2-l}\nabla^m\nabla_3^lR_{ABCD}(0,v),
        \end{align*}
        and the self-similar bounds:
        \[\big\|v^{2+m+l}\nabla^m\nabla_3^l\underline{\Psi}^G\big\|_{L^2(S^n)}\big|_{u=0}\lesssim\epsilon'.\]
        \item For any $0\leq l\leq\frac{n-6}{2}$ and $0\leq m\leq N'+\frac{n-8}{2}-l$, we have:
        \[\nabla^m\nabla_3^l\underline{\alpha}\in W^{1,1}_u\big([-1,0]\big)L^2(S^n).\]
        Moreover, we have the self-similarity relation on $\{u=0\}:$
        \[\nabla^m\nabla_3^l\underline{\alpha}_{AB}(0,\lambda v)=\lambda^{-l}\nabla^m\nabla_3^l\underline{\alpha}_{AB}(0,v),\]
        and the self-similar bounds:
        \[\big\|v^{2+m+l}\nabla^m\nabla_3^l\underline{\alpha}\big\|_{L^2(S^n)}\big|_{u=0}\lesssim\epsilon'\]
        \begin{equation}\label{bound for regular parts of alpha at I^+}
            v^{2+m}\Big\|v^l\nabla^m\nabla_3^l\underline{\alpha}-\big(v^l\nabla^m\nabla_3^l\underline{\alpha}\big)\big|_{u=0}\Big\|_{L^2(S^n)}\lesssim\epsilon'\cdot\bigg|\frac{u}{v}\bigg|^{1-p}.
        \end{equation}
        \item For any $0\leq l\leq\frac{n-4}{2}$ and $0\leq m\leq N'+\frac{n-4}{2}-l$, we have:
        \[\nabla^m\nabla_3^l\psi\in W^{1,1}_u\big([-1,0]\big)L^2(S^n).\]
        Moreover, we have the self-similarity relations on $\{u=0\}:$
        \[\nabla^m\nabla_3^l\big(\chi_{AB},\underline{\chi}_{AB}\big)(0,\lambda v)=\lambda^{1-l}\nabla^m\nabla_3^l\big(\chi_{AB},\underline{\chi}_{AB}\big)(0,v),\]
        and the self-similar bound:
        \[\big\|v^{1+m+l}\nabla^m\nabla_3^l\psi^*\big\|_{L^2(S^n)}\big|_{u=0}\lesssim\epsilon'.\]
        \item For any $0\leq l\leq\frac{n-2}{2}$ and $0\leq m\leq N'+\frac{n-4}{2}-l$, we have:
        \[v^{m+l}\mathcal{L}_{\theta}^m\mathcal{L}_3^{l}\slashed{g}\in W^{1,1}_u\big([-1,0]\big)L^2(S^n).\]
        Moreover, we have the self-similarity relation on $\{u=0\}:$
        \[\mathcal{L}_{\theta}^m\mathcal{L}_3^{l}\slashed{g}_{AB}(0,\lambda v)=\lambda^{2-l}\mathcal{L}_{\theta}^m\mathcal{L}_3^{l}\slashed{g}_{AB}(0,v),\]
        and the self-similar bound:
        \[\big\|v^{m+l}\mathcal{L}_{\theta}^m\mathcal{L}_3^{l}\slashed{g}^*\big\|_{L^2(S^n)}\big|_{u=0}\lesssim\epsilon'.\]
    \end{enumerate}
\end{proposition}
\begin{proof}
    We first restrict to the region $\{v\leq\underline{v}\}.$ Thus, we can use the estimates in Proposition \ref{region III bounds proposition}. We recall that in the proof of Proposition \ref{low regularity curvature proposition} we had for any $(m,l)\in L$:
    \[\nabla_3v^{2+m+l}\nabla^m\nabla_3^l\underline{\Psi}^G=v^{2+m+l}Err_{ml},\]
    where the error term satisfies:
    \[\big\|v^{2+m+l}Err_{ml}\big\|_{L^1_u([-v\underline{v},0))L^2(S^n)}\lesssim\epsilon'\]
    We also have that $\partial_u\big(v^{2+m+l}\nabla^m\nabla_3^l\underline{\Psi}^G\big)=\nabla_3v^{2+m+l}\nabla^m\nabla_3^l\underline{\Psi}^G+v^{2+m+l}\underline{\chi}\cdot\nabla^m\nabla_3^l\underline{\Psi}^G$. From the proof of Proposition~\ref{Ricci coefficients proposition}, we have that:
    \[\big\|v^{2+m+l}\underline{\chi}\cdot\nabla^m\nabla_3^l\underline{\Psi}^G\big\|_{L^1_u([-v\underline{v},0))L^2(S^n)}\lesssim\epsilon'.\]
    As a result, we obtain that $v^{2+m+l}\nabla^m\nabla_3^l\underline{\Psi}^G\in W^{1,1}_u\big([-v\underline{v},0]\big)L^2(S^n)$ and that:
    \[\big\|v^{2+m+l}\nabla^m\nabla_3^l\underline{\Psi}^G\big\|_{L^2(S^n)}\big|_{u=0}\lesssim\epsilon'.\]
    Using self-similarity, we can extend these results from the region $\{v\leq\underline{v}\}$ to all $v>0.$ Moreover, by continuity we obtain that the self-similarity relations also hold along $u=0.$
    
    In order to prove the second statement, we notice that for any $0\leq l\leq\frac{n-6}{2}$ and $0\leq m\leq N'+\frac{n-8}{2}-l:$
    \[\partial_u\big(v^{2+m+l}\nabla^m\nabla_3^l\underline{\alpha}\big)=v\psi\cdot v^{1+m+l}\nabla^m\nabla_3^l\underline{\alpha}+v^{2+m+l}\nabla^m\nabla_3^{l+1}\underline{\alpha}+\sum_{i=0}^mv^{1+i}\nabla^i\psi\cdot v^{1+m-i+l}\nabla^{m-i}\nabla_3^l\underline{\alpha}\]
    Since $(m,l+1)\in L,$ we can bound the above terms using our previous estimates and we obtain:
    \[\Big\|\partial_u\big(v^{2+m+l}\nabla^m\nabla_3^l\underline{\alpha}\big)\Big\|_{L^2(S^n)}\lesssim\epsilon'\cdot|u|^{-p}v^{-1+p}.\]
    This implies that $v^{2+m+l}\nabla^m\nabla_3^l\underline{\alpha}\in W^{1,1}_u\big([-v\underline{v},0]\big)L^2(S^n)$, that (\ref{bound for regular parts of alpha at I^+}) holds, and that:
    \[\big\|v^{2+m+l}\nabla^m\nabla_3^l\underline{\alpha}\big\|_{L^2(S^n)}\big|_{u=0}\lesssim\epsilon'.\]
    As before, the self-similarity relation extends to $u=0$ by continuity.

    The proof of the third statement follows exactly as the first statement, but using the corresponding estimates from the proof of Proposition \ref{Ricci coefficients proposition}.

    We prove the forth statement in the case of $\mathcal{L}_{\theta}^m\mathcal{L}_3^{l+1}\slashed{g},$ since the case of no $\mathcal{L}_3$ derivatives is similar. We compute the commuted equation for any $(m,l)\in H$ using \eqref{Lie derivative in terms of covariant derivative} and its analogue in the $e_3$ direction:
    \[\mathcal{L}_3\big(\mathcal{L}_{\theta}^m\mathcal{L}_3^{l+1}\slashed{g}^*\big)=2\mathcal{L}_{\theta}^m\mathcal{L}_3^{l+1}\underline{\chi}^*=\mathcal{L}_{\theta}^m\mathcal{L}_3^{l}\big(\underline{\alpha}+\psi\psi^*\big)=\mathcal{L}_{\theta}^m\big(\nabla_3^l\underline{\alpha}+\underline{\mathcal{F}}_{(0)(l-1)(l-1)}(\underline{\alpha})+\underline{\mathcal{F}}_{(0)(l)(l)}(\psi^*)\big)\]
    \[=\nabla^m\nabla_3^l\underline{\alpha}+\sum_{i+j+k=m-1}\mathcal{L}_{\theta}^i\big(\Gamma^{j+1}\big)\nabla^k\nabla_3^l\underline{\alpha}+\underline{\mathcal{F}}_{(m)(l-1)(m+l-1)}(\underline{\alpha})+\sum_{i+j+k=m-1}\mathcal{L}_{\theta}^i\big(\Gamma^{j+1}\big)\underline{\mathcal{F}}_{(k)(l-1)(k+l-1)}(\underline{\alpha})+\]\[+\underline{\mathcal{F}}_{(m)(l)(m+l)}(\psi^*)+\sum_{i+j+k=m-1}\mathcal{L}_{\theta}^i\big(\Gamma^{j+1}\big)\underline{\mathcal{F}}_{(k)(l)(k+l)}(\psi^*)\]
    We denote the RHS by $Err_{ml}.$ We have the bound:
    \[\big\|v^{1+m+l}Err_{ml}\big\|_{L^1_u([-v\underline{v},u))L^2(S^n)}^2\lesssim\underline{v}^{1-100p}\sup_{\mathring{u}\in[-v\underline{v},u]}|\mathring{u}|^{100q}\int_{-v\underline{v}}^{\mathring{u}}\int_{S^n}v^{3+2m+2l-100p}|\hat{u}|^{100p-100q}\big|Err_{ml}\big|^2d\mathring{Vol}d\hat{u}\]
    The first and fifth error terms were already bounded in the proof of Proposition \ref{Ricci coefficients proposition}. The third error term was already bounded in the the proof of Proposition \ref{low regularity curvature proposition}. As in the proof of Proposition \ref{metric coefficients proposition}, for the second error term we have the bound:
    \[\sum_{i+j+k=m-1}\underline{v}^{1-100p}\sup_{\mathring{u}\in[-v\underline{v},u]}|\mathring{u}|^{100q}\int_{-v\underline{v}}^{\mathring{u}}\int_{S^n}v^{3+2k+2l-100p}|\hat{u}|^{100p-100q}\big|\nabla^k\nabla_3^l\underline{\alpha}\big|^2v^{2+2i+2j}\big|\mathcal{L}_{\theta}^i\big(\Gamma^{j+1}\big)\big|^2d\mathring{Vol}d\hat{u}\]
    \[\lesssim\underline{v}^{1-100p}\sup_{\mathring{u}\in[-v\underline{v},u]}|\mathring{u}|^{100q}\int_{-v\underline{v}}^{\mathring{u}}\sum_{(k,l)\in L}v^{3+2l-100p}|\hat{u}|^{100p-100q}\big\|\nabla_3^l\underline{\alpha}\big\|^2_{\widetilde{H}^k(S^n)}\prod_{a=1}^{j+1}\sum_{(k_a,0)\in L}\int_{S^n} v^{2+2k_a}\big|\mathcal{L}_{\theta}^{k_a}\Gamma\big|^2d\mathring{Vol}d\hat{u}\]
    which is bounded using the previous section and (\ref{Christoffel symbols bound}). Similarly, the forth term is bounded by:
    \[\sum_{(k,l)\in L}\underline{v}^{1-100p}\sup_{\mathring{u}\in[-v\underline{v},u]}|\mathring{u}|^{100q}\int_{-v\underline{v}}^{\mathring{u}}\int_{S^n}v^{3+2k+2l-100p}|\hat{u}|^{100p-100q}\big|\underline{\mathcal{F}}_{(k)(l-1)(k+l-1)}(\underline{\alpha})\big|^2d\mathring{Vol}d\hat{u},\]
    which is bounded as the third error term. Finally, the last error term can be similarly reduced to the fifth error term. As a result, we obtain that $v^{1+m+l}\mathcal{L}_{\theta}^m\mathcal{L}_3^{l+1}\slashed{g}\in W^{1,1}_u\big([-v\underline{v},0]\big)L^2(S^n)$ and that:
    \[\big\|v^{1+m+l}\mathcal{L}_{\theta}^m\mathcal{L}_3^{l+1}\slashed{g}^*\big\|_{L^2(S^n)}\big|_{u=0}\lesssim\epsilon'.\]
    Using self-similarity, we extend these results from the region $\{v\leq\underline{v}\}$ to all $v>0.$ By continuity, we also obtain that the self-similarity relations hold along $u=0.$
\end{proof}

Next, we compute the induced obstruction tensor $\underline{\mathcal{O}}$ and the remaining component of the scattering data $\underline{h}:$
\begin{proposition}\label{asymptotic completeness preliminary proposition 2}
    There exist symmetric traceless 2-tensors $\underline{\mathcal{O}}$ and $\underline{h},$ which are independent of $u$ and $v$ such that for any $0\leq m\leq N'-3$ and $\{-\underline{v}v\leq u\leq0\}:$
    \begin{equation}\label{bound for obstruction tensor on I^+}
        \big\|v^{2+m}\nabla^m\underline{\mathcal{O}}\big\|_{L^2(S^n)}\lesssim\epsilon'
    \end{equation}
    \begin{equation}\label{bound for h on I^+}
        \big\|v^{2+m}\nabla^m\underline{h}\big\|_{L^2(S^n)}\lesssim\epsilon'
    \end{equation}
    \begin{equation}\label{bound in expansion of alpha on I^+}
        v^{2+m}\bigg\|v^{\frac{n-4}{2}}\nabla^m\nabla_3^{\frac{n-4}{2}}\underline{\alpha}-\nabla^m\underline{\mathcal{O}}\log\bigg|\frac{u}{v}\bigg|-\nabla^m\underline{h}\bigg\|_{L^2(S^n)}\lesssim\epsilon'\cdot\bigg|\frac{u}{v}\bigg|^{1-p}.
    \end{equation}
\end{proposition}
\begin{proof}
\textit{Step 1. Computing the induced obstruction tensor $\underline{\mathcal{O}}$.} As usual, we first restrict to the region $\{v\leq\underline{v}\}$ to prove the desired result, then we use self-similarity to extend to all $v>0.$ We use self-similarity to rewrite the Bianchi equation for $\underline{\alpha}$ as:
\[-u\nabla_3\underline{\alpha}+\frac{n-4}{2}\underline{\alpha}-\frac{u}{2}\tr\underline{\chi}\underline{\alpha}=v\nabla\underline{\Psi}^G+v\psi\underline{\Psi}^G\]
We set $l=\frac{n-4}{2}$ for the remaining of the proof. We commute the equation to obtain:
\[|u|\nabla_3\nabla_3^l\underline{\alpha}=v\nabla\nabla_3^l\underline{\Psi}^G+v\cdot\underline{\mathcal{F}}_{(0)(l)(l)}(\underline{\Psi}^G)+v\cdot\underline{\mathcal{F}}_{(1)(l-1)(l)}(\underline{\Psi}^G)+\underline{\mathcal{F}}_{(0)(l-1)(l-1)}(\Psi)+|u|\underline{\mathcal{F}}_{(0)(l)(l)}(\Psi)\]
We notice that in the proof of Proposition \ref{asymptotic completeness preliminary proposition}, we also obtain the bound:
\[\Big\|uv^{1+m+l}\nabla^m\nabla_3^l\underline{\alpha}\Big\|_{L^2(S^n)}\lesssim\epsilon'\cdot|u|^{1-p}v^{-1+p}.\]
Thus, Proposition \ref{asymptotic completeness preliminary proposition} implies that each term on the RHS of the above equation is in $C^0_u\big([-\underline{v}v,0]\big)\widetilde{H}^m(S^n).$ As a result, $-uv^l\nabla_3\nabla_3^l\underline{\alpha}$ can be extended to $\{u=0\}$ as a symmetric traceless 2-tensor which is independent of the $v$ coordinate. We define:
\[\underline{\mathcal{O}}_{AB}=\big(uv^l\nabla_3\nabla_3^l\underline{\alpha}_{AB}\big)\big|_{u=0}\in\widetilde{H}^m(S^n).\]
For any $0\leq m\leq N'-3,$ we have the equation:
\[\nabla_3\nabla^m\nabla_3^l\underline{\alpha}=\frac{v}{|u|}\nabla^{m+1}\nabla_3^l\underline{\Psi}^G+\frac{v}{|u|}\cdot\underline{\mathcal{F}}_{(m)(l)(m+l)}(\underline{\Psi}^G)+\frac{v}{|u|}\cdot\underline{\mathcal{F}}_{(m+1)(l-1)(m+l)}(\underline{\Psi}^G)+\]\[+\frac{1}{|u|}\cdot\underline{\mathcal{F}}_{(m)(l-1)(m+l-1)}(\Psi)+\underline{\mathcal{F}}_{(m)(l)(m+l)}(\Psi)\]
Since $\lim_{u\rightarrow0^-}u[\nabla_3,\nabla^m]\nabla_3^l\underline{\alpha}=0,$ we also obtain that:
\[\nabla^m\underline{\mathcal{O}}_{AB}=\big(uv^l\nabla_3\nabla^m\nabla_3^l\underline{\alpha}_{AB}\big)\big|_{u=0}\]
Thus, we have the schematic equation:
\[-\nabla^m\underline{\mathcal{O}}=v^{l+1}\nabla^{m+1}\nabla_3^l\underline{\Psi}^G\big|_{u=0}+v^{l+1}\underline{\mathcal{F}}_{(m)(l)(m+l)}(\underline{\Psi}^G)\big|_{u=0}+v^{l+1}\underline{\mathcal{F}}_{(m+1)(l-1)(m+l)}(\underline{\Psi}^G)\big|_{u=0}+\]\[+v^{l}\underline{\mathcal{F}}_{(m)(l-1)(m+l-1)}(\Psi)\big|_{u=0}\]
We use the bounds in Proposition \ref{asymptotic completeness preliminary proposition}, together with the bounds in the proofs of Propositions \ref{low regularity curvature proposition} and \ref{Ricci coefficients proposition}, to control the RHS and get:
\[\big\|v^{2+m}\nabla^m\underline{\mathcal{O}}\big\|_{L^2(S^n)}\lesssim\underline{v}^{\frac{1}{2}-p}\cdot\epsilon'.\]

\textit{Step 2. Computing the induced component of the scattering data $\underline{h}.$} Extending $\nabla^m\underline{\mathcal{O}}$ to be independent of $u$ and $v,$ we can write the equation for $\nabla^m\nabla_3^l\underline{\alpha}$ as:
\begin{equation}\label{equation du underline alpha to see O}
    \partial_u\big(v^l\nabla^m\nabla_3^l\underline{\alpha}\big)=\frac{1}{|u|}\mathcal{E}_1+\mathcal{E}_2=\frac{1}{u}\nabla^m\underline{\mathcal{O}}+\frac{1}{|u|}\Big(\mathcal{E}_1-\mathcal{E}_1\big|_{u=0}\Big)+\mathcal{E}_2
\end{equation}
\[\mathcal{E}_1=v^{l+1}\nabla^{m+1}\nabla_3^l\underline{\Psi}^G+v^{l+1}\underline{\mathcal{F}}_{(m)(l)(m+l)}(\underline{\Psi}^G)+v^{l+1}\underline{\mathcal{F}}_{(m+1)(l-1)(m+l)}(\underline{\Psi}^G)+v^{l}\underline{\mathcal{F}}_{(m)(l-1)(m+l-1)}(\Psi)\]
\[\mathcal{E}_2=v^l\underline{\mathcal{F}}_{(m)(l)(m+l)}(\Psi)\]
The key part of our proof is to establish the claim that:
\begin{equation}\label{asympt completeness claim}
    \mathcal{E}:=\frac{1}{|u|}\Big(\mathcal{E}_1-\mathcal{E}_1\big|_{u=0}\Big)+\mathcal{E}_2\in L^1_u\big([-\underline{v}v,0]\big)L^2(S^n).
\end{equation}

\textit{Step 2a. The proof of claim \eqref{asympt completeness claim}.} In the proof of Proposition \ref{low regularity curvature proposition} we already bounded:
\[\int_{-\underline{v}v}^0\big\|\mathcal{E}_2\big\|_{L^2(S^n)}d\hat{u}\lesssim v^{-2-m}\underline{v}^{\frac{1}{2}}\epsilon'.\]
We also have the bound:
\[\big\|\mathcal{E}_1-\mathcal{E}_1\big|_{u=0}\big\|_{L^2(S_{u,v})}\lesssim\int_u^0v^{l+1}\big\|\nabla_3\nabla^{m+1}\nabla_3^l\underline{\Psi}^G\big\|_{L^2(S^n)}+v^{l+1}\big\|\underline{\mathcal{F}}_{(m)(l+1)(m+l+1)}(\underline{\Psi}^G)\big\|_{L^2(S^n)}d\hat{u}+\]\[+\int_u^0v^{l+1}\big\|\underline{\mathcal{F}}_{(m+1)(l)(m+l+1)}(\underline{\Psi}^G)\big\|_{L^2(S^n)}+v^{l}\big\|\underline{\mathcal{F}}_{(m)(l)(m+l)}(\Psi)\big\|_{L^2(S^n)}d\hat{u}\]
Using the estimates proved in region~III, we have the bound for any $(i,j)\in L$:
\begin{equation}\label{bound for F error in L2}
    \big\|\underline{\mathcal{F}}_{(i)(j)(i+j)}(\Psi)\big\|_{L^2(S^n)}\lesssim v^{-3-i-j+p}|u|^{-p}\epsilon'.
\end{equation}
Since $(m,l),(m+1,l)\in L,$ we have the bound:
\[\int_u^0v^{l+1}\big\|\underline{\mathcal{F}}_{(m+1)(l)(m+l+1)}(\underline{\Psi}^G)\big\|_{L^2(S^n)}+v^{l}\big\|\underline{\mathcal{F}}_{(m)(l)(m+l)}(\Psi)\big\|_{L^2(S^n)}d\hat{u}\lesssim\epsilon'\cdot v^{-2-m}\bigg|\frac{u}{v}\bigg|^{1-p}\]
We have the schematic equation as in the proof of Proposition \ref{low regularity curvature proposition}:
\[\nabla_3v^{1+l}\nabla^{m+1}\nabla_3^l\underline{\Psi}^G=v^{1+l}\nabla^{m+2}\nabla_3^l\Psi+v^{1+l}\underline{\mathcal{F}}_{(m+1)(l)(m+l+1)}(\Psi)+v^{1+l}\underline{\mathcal{F}}_{(m+2)(l-1)(m+l+1)}(\Psi).\]
Since $(m+2,l)\in L,$ we use Proposition \ref{region III bounds proposition} and (\ref{bound for F error in L2}) to get:
\[\int_u^0v^{l+1}\big\|\nabla_3\nabla^{m+1}\nabla_3^l\underline{\Psi}^G\big\|_{L^2(S^n)}d\hat{u}\lesssim\epsilon'\cdot v^{-2-m}\bigg|\frac{u}{v}\bigg|^{1-p}\]
Using the Bianchi equations we can rewrite the error term:
\[\underline{\mathcal{F}}_{(m)(l+1)(m+l+1)}(\underline{\Psi}^G)=\underline{\mathcal{F}}_{(m+1)(l)(m+l+1)}(\Psi)+\sum_{\substack{i+j+k\leq m+l+1 \\ i\leq l+1,k\leq m}}\nabla^k\big(\underline{\Psi}^G\nabla_3^i\psi^{j+1}\big)\]
The estimates in Proposition \ref{low regularity curvature proposition} and Proposition \ref{Ricci coefficients proposition} imply the bound:
\[\sum_{\substack{i+j+|k|\leq m+l+1 \\ i\leq l,|k|\leq m}}\Big\|v^{2+k_1}\nabla^{k_1}\underline{\Psi}^G\cdot v^{i+j+k_2+1}\nabla^{k_2}\nabla_3^i\psi^{j+1}\Big\|_{L^2(S^n)}\lesssim\epsilon'\]
Similarly, we use the schematic equations for $\nabla_3\psi$ from Proposition \ref{Ricci coefficients proposition} to get:
\[\sum_{j+|k|\leq m}\Big\|v^{2+k_1}\nabla^{k_1}\underline{\Psi}^G\cdot v^{l+1+j+k_2+1}\nabla^{k_2}\nabla_3^{l+1}\psi^{j+1}\Big\|_{L^2(S^n)}\lesssim\epsilon'+\epsilon'\sum_{j+k\leq m}\Big\| v^{l+k+2}\nabla^{k}\nabla_3^{l+1}\psi\Big\|_{L^2(S^n)}\lesssim\]\[\lesssim\epsilon'+\epsilon'\sum_{k\leq m}\Big\| v^{l+k+2}\nabla^{k}\nabla_3^l\big(\Psi+\psi\psi\big)\Big\|_{L^2(S^n)}\lesssim\epsilon'|u|^{-p}v^p\]
As a result, we obtain that the last remaining error term satisfies:
\[\int_u^0v^{l+1}\big\|\underline{\mathcal{F}}_{(m)(l+1)(m+l+1)}(\underline{\Psi}^G)\big\|_{L^2(S^n)}d\hat{u}\lesssim\epsilon'\cdot v^{-2-m}\bigg|\frac{u}{v}\bigg|^{1-p}\]
We proved that:
\[\frac{1}{|u|}\cdot\big\|\mathcal{E}_1-\mathcal{E}_1\big|_{u=0}\big\|_{L^2(S_{u,v})}\lesssim\epsilon'\cdot v^{-3-m}\bigg|\frac{u}{v}\bigg|^{-p},\]
which proves our claim \eqref{asympt completeness claim} that $\mathcal{E}\in L^1_u\big([-\underline{v}v,0]\big)L^2(S^n).$ 

\textit{Step 2b. The proof of \eqref{bound for h on I^+} and \eqref{bound in expansion of alpha on I^+}.} We integrate (\ref{equation du underline alpha to see O}) from $-\underline{v}v$ to $u$:
\[v^l\nabla^m\nabla_3^l\underline{\alpha}-\nabla^m\underline{\mathcal{O}}\log\bigg|\frac{u}{v}\bigg|=v^l\nabla^m\nabla_3^l\underline{\alpha}\big|_{u=-\underline{v}v}-\nabla^m\underline{\mathcal{O}}\log\underline{v}+\int_{-\underline{v}v}^u\mathcal{E}d\hat{u}.\]
In particular, we obtain that:
\[v^l\nabla^m\nabla_3^l\underline{\alpha}-\nabla^m\underline{\mathcal{O}}\log\bigg|\frac{u}{v}\bigg|\in W^{1,1}_u\big([-\underline{v}v,0]\big)L^2(S^n).\]
We can define the symmetric traceless two tensor $\underline{h}$ which is independent of $u$ and $v$ by:
\[\nabla^m\underline{h}:=\lim_{u\rightarrow0^-}\bigg(v^l\nabla^m\nabla_3^l\underline{\alpha}-\nabla^m\underline{\mathcal{O}}\log\bigg|\frac{u}{v}\bigg|\bigg)\]
The above estimates imply that:
\[\big\|v^{2+m}\nabla^m\underline{h}\big\|_{L^2(S^n)}\lesssim\epsilon'.\]
Finally, we have that:
\[v^{2+m}\bigg\|v^l\nabla^m\nabla_3^l\underline{\alpha}-\nabla^m\underline{\mathcal{O}}\log\bigg|\frac{u}{v}\bigg|-\nabla^m\underline{h}\bigg\|_{L^2(S^n)}\lesssim\int_u^0v^{2+m}\big\|\mathcal{E}\big\|_{L^2(S^n)}d\hat{u}\lesssim\epsilon'\cdot\bigg|\frac{u}{v}\bigg|^{1-p}.\]
\end{proof}

We can use the results proved so far in order to complete the proof of Theorem \ref{asymptotic completeness in section theorem}:

\textit{Proof of Theorem \ref{asymptotic completeness in section theorem}.} Based on Proposition \ref{asymptotic completeness preliminary proposition} and \ref{asymptotic completeness preliminary proposition 2}, we define $\slashed{g}$ and $\underline{h}$ on $\{u=0\}.$ We prove that the spacetime $\big(\mathcal{M},g\big)$ satisfies the required conditions needed in order to apply the main result of \cite{selfsimilarvacuum}, which shows that the spacetime is determined by the induced asymptotic data $\underline{\slashed{g}_{0}}=\slashed{g}\big|_{(u,v)=(0,1)},\ \underline{h}.$ For $N'$ large enough, we have:
\begin{itemize}
    \item For any $0\leq l\leq\frac{n-2}{2},\ 0\leq m\leq N'$, the limit $\lim_{u\rightarrow0^-}\mathcal{L}_{\theta}^m\mathcal{L}_u^l\slashed{g}$ exists and is uniformly bounded (with appropriate self-similar $v$ weights) by Proposition \ref{asymptotic completeness preliminary proposition}. Moreover, we remark that Proposition \ref{asymptotic completeness preliminary proposition} also implies that we can extend to $\{u=0\}$ the following equations: the constraint equations, the $\nabla_4$ null structure equations, the $\nabla_4$ Bianchi equations for $\underline{\Psi}^G$ when commuted with up to $\frac{n-4}{2}$ $\nabla_3$ derivatives; the $\nabla_3$ null structure equations, the $\nabla_3$ Bianchi equations, the $\nabla_4\underline{\alpha}$ Bianchi equation when commuted with up to $\frac{n-6}{2}$ $\nabla_3$ derivatives. Using these equations on $\{u=0\}$, the argument in \cite[Proposition~4.3]{selfsimilarvacuum} implies that for $0<l\leq\frac{n-2}{2}$ the limits $\mathcal{L}_{\theta}^m\mathcal{L}_u^l\slashed{g}\big|_{(u,v)=(0,1)}$ have certain prescribed values in terms of $\underline{\slashed{g}_{0}}$ and satisfy the compatibility conditions.
    \item Proposition \ref{asymptotic completeness preliminary proposition} implies that for any $0\leq l\leq\frac{n-6}{2}$ and $0\leq m\leq N'$, we have on $\{v=1\}$:
        \[\Big\|\nabla^m\nabla_3^l\underline{\alpha}-\big(\nabla^m\nabla_3^l\underline{\alpha}\big)\big|_{u=0}\Big\|_{L^2(S^n)}\lesssim\epsilon'\cdot|u|^{1-p}.\]
    We obtain a similar result for any $0\leq l\leq\frac{n-2}{2}$ and $0\leq m\leq N',$ on $\{v=1\}$:
    \[\Big\|\mathcal{L}_{\theta}^m\mathcal{L}_3^l\hat{\slashed{g}}-\mathcal{L}_{\theta}^m\mathcal{L}_3^l\underline{\hat{\slashed{g}}_{0}}\Big\|_{L^2(S^n)}\lesssim\epsilon'\cdot|u|^{1-p}.\]
    \item In Proposition \ref{asymptotic completeness preliminary proposition 2} we define the obstruction tensor $\underline{\mathcal{O}}$ in terms of $\underline{\slashed{g}_{0}}$ (since all the curvature components and Ricci coefficients appearing in our definition of $\underline{\mathcal{O}}$ can be expressed in terms of $\underline{\slashed{g}_{0}}$). We remark that by construction $\underline{\mathcal{O}}$ satisfies the compatibility condition in \cite[Proposition~4.3]{selfsimilarvacuum}. Moreover, we proved that for any $0\leq m\leq N'$ the limit:
    \[\nabla^m\underline{h}=\lim_{u\rightarrow0^-}\bigg(v^{\frac{n-4}{2}}\nabla^m\nabla_3^{\frac{n-4}{2}}\underline{\alpha}-\nabla^m\underline{\mathcal{O}}\log\bigg|\frac{u}{v}\bigg|\bigg)\]
    exists and is uniformly bounded (with appropriate self-similar $v$ weights).
    \item Proposition \ref{asymptotic completeness preliminary proposition 2} implies that for any $0\leq m\leq N'$, we have on $\{v=1\}$:
    \[\Big\|\nabla^m\nabla_3^{\frac{n-4}{2}}\underline{\alpha}-\nabla^m\underline{\mathcal{O}}\log|u|-\nabla^m\underline{h}\Big\|_{L^2(S^n)}\lesssim\epsilon'\cdot|u|^{1-p}.\]
\end{itemize}
As a result, we obtain that on $v=1$ and $-\underline{v}\leq u\leq0,$ the metric $\slashed{g}$ induces a 1-parameter family $\hat{\slashed{g}}(u)$ of conformal classes of metrics on $S^n$ admissible relative to $\underline{\slashed{g}_{0}}$, according to the definitions of \cite[Definition~1.2]{selfsimilarvacuum}. Thus, we can apply \cite[Theorem~1.1]{selfsimilarvacuum} to obtain that $\big(\mathcal{M},g\big)$ is the unique self-similar solution with data $\big(\underline{\slashed{g}_{0}},\underline{h}\big).$ Since we already know that $\big(\mathcal{M},g\big)$ is a straight spacetime, we obtain that $\underline{h}$ must satisfy the straightness condition. Moreover, the estimates for $\underline{\slashed{g}_{0}}^*,\underline{\mathcal{O}},$ and $\underline{h}$ were already proved in Proposition \ref{asymptotic completeness preliminary proposition} and \ref{asymptotic completeness preliminary proposition 2}.

Finally, we establish the propagation of regularity statement. Continuing the argument in the proof of Theorem \ref{propagation of regularity theorem}, for any $K>N_3$ we can use the bounds previously obtained to repeat the proofs of Propositions \ref{asymptotic completeness preliminary proposition} and \ref{asymptotic completeness preliminary proposition 2}, with $(N',q,p,\underline{v})$ replaced by $(K,q,p,\underline{v}_K)$, and the $\epsilon'$ on the right hand side of the estimates replaced by $C(K).$ As a result, we obtain that  $\underline{\slashed{g}_{0}},\ \underline{\mathcal{O}},\  \underline{h}\in H^{K-3}(S^n)$ for all $K>N'$. We conclude that the straight data induced at $(u,v)=(0,1)$ given by $\big(\underline{\slashed{g}_{0}},\underline{h}\big)$ is smooth.
\qed

\section{The Model Systems}\label{model systemmm}

We derive the systems of commuted Bianchi equations along $\{u=-1\}$ required for the analysis of the scattering map. We also introduce the model systems, given by the principal part of the commuted Bianchi equations at top order with a general inhomogeneous term. The systems that we consider consist of wave equations that are singular at $\{v=0\}.$

\textbf{Notation convention.} We consider $M>N$ large enough, where $N>0$ is as in Theorem~\ref{stability of de sitter theorem in section}. Unless otherwise noted, for the rest of the paper we write $A\lesssim B$ for some quantities $A,B>0$ if there exists a constant $C>0$ depending only on $M$ such that $A\leq CB.$ 

\textbf{Integration convention.} For the remainder of the paper we make the convention that the volume form used on the sphere $S_{u,v}=\{(u,v)\}\times S^n$ with induced metric $\slashed{g}_{u,v}$ is $d\slashed{Vol}_{\slashed{g}_{u,v}},$ in order to be consistent with the notation in \cite{Cwave}. We note that this convention is different from the one in Section~\ref{stability dS section} and Section~\ref{asymptotic completeness section}.

\subsection{Bianchi Equations}

We write the Bianchi system along $\{u=-1\}$ as a system of wave equations. Using the fact that $\nabla_S\Psi=-2\Psi,$ and that $\tr\underline{\chi}=v\tr\chi-n,$ we can rewrite the equations for the Bianchi pairs on $\{u=-1\}$ as follows:
\begin{align}
    v\nabla_4\alpha_{AB}+\bigg(2-\frac{n}{2}+\frac{v}{2}\tr\chi\bigg)\alpha_{AB}&=-\nabla^C\nu_{C(AB)}+\mathcal{E}_1^{(3)}\\
    \nabla_4\nu_{ABC}&=-2\nabla_{[A}\alpha_{B]C}+\mathcal{E}_{1/2}^{(4)}\\
    v\nabla_4\nu_{ABC}+\frac{2v}{n}\tr\chi\nu_{ABC}&=-2\nabla_{[A}\tau_{B]C}+2\underline{\hat{\chi}}_{[A}^D\nu_{|D|B]C}+\mathcal{E}_{3/2}^{(3)}\\
    \nabla_4R_{ABCD}&=-2\nabla_{[A}\nu_{|CD|B]}+\mathcal{E}_{1}^{(4)}\label{R}\\
    v\nabla_4R_{ABCD}+\frac{2v}{n}\tr\chi R_{ABCD}&=-2\nabla_{[A}\underline{\nu}_{|CD|B]}+\underline{\chi}_{A[D}\tau_{C]B}+\underline{\chi}_{B[C}\tau_{D]A}+2\underline{\hat{\chi}}_{[A}^ER_{B]ECD}+\mathcal{E}_{2}^{(3)}\label{v R}\\
    \nabla_4\underline{\nu}_{ABC}&=-2\nabla_{[A}\tau_{B]C}+\mathcal{E}_{3/2}^{(4)}\label{underline nu}\\
    v\nabla_4\underline{\nu}_{ABC}+\bigg(-1+\frac{3v}{n}\tr\chi\bigg)\underline{\nu}_{ABC}&=-2\nabla_{[A}\underline{\alpha}_{B]C}+2\underline{\hat{\chi}}_{[A}^D\underline{\nu}_{B]DC}+2\underline{\hat{\chi}}_{[A}^D\underline{\nu}_{|CD|B]}+\mathcal{E}_{5/2}^{(3)}\label{v underline nu}\\
    \nabla_4\underline{\alpha}_{AB}&=-\nabla^C\underline{\nu}_{C(AB)}+\mathcal{E}_{2}^{(4)}\label{underline alpha}.
\end{align}

Using the commutation formulas in Lemma~\ref{commutation lemma with F notation}, we can rewrite the system of Bianchi equations on $\{u=-1\}$ as a system of wave equations:
\begin{proposition}
    We have the system of wave equations on $\{u=-1\}$ for any $0\leq m\leq M,$ $0\leq l\leq\frac{n}{2}-2:$
    \begin{equation}\label{model wave system 2}
    \begin{cases}
        v\nabla_4^2\nabla^m\nabla_4^l\alpha+\Big(3+l-\frac{n}{2}\Big)\nabla_4\nabla^m\nabla_4^l\alpha-\Delta\nabla^m\nabla_4^l\alpha=\psi\nabla^{m+1}\nabla_4^l\Psi+Err_{ml}^{\Psi} \\
        v\nabla_4^2\nabla^m\nabla_4^l\Psi^G+\Big(2+l-\frac{n}{2}\Big)\nabla_4\nabla^m\nabla_4^l\Psi^G-\Delta\nabla^m\nabla_4^l\Psi^G=\sum_{\Psi^G_0}\psi\nabla^{m+1}\nabla_4^l\Psi^G_0+Err_{ml}^{\Psi},
    \end{cases}
    \end{equation}
    where we have the error term notation:
    \[Err_{ml}^{\Psi}=v\mathcal{F}_{(m)(l+1)(l+m+1)}(\Psi)+v\mathcal{F}_{(m+1)(l)(l+m+1)}(\Psi)+\mathcal{F}_{(2+m)(l-1)(l+m+1)}(\Psi)+\mathcal{F}_{(m)(l)(m+l)}(\Psi)+\]\[+\mathcal{F}_{(m+1)(l)(m+l+1)}^{lot}(\Psi)+\nabla^m\nabla_4^l\big(\Psi\cdot\Psi^G\big)+\sum_{i+2j=m}\nabla^i\big(\slashed{Riem}^{j+1}\cdot\nabla_4^l\Psi\big)+\nabla^m\nabla_4^l\nabla\big(\psi\psi\big)+\nabla^m\nabla_4^l\nabla\big(\psi\psi\psi\big)\]
    and we point out that in the RHS of \eqref{model wave system 2} we sum the terms for all $\Psi.$ Similarly, using the schematic equation:
    \begin{equation}\label{schematic equation PsiG}
        \nabla_4\nabla^m\nabla_4^l\Psi^G=\psi\nabla^{m+1}\nabla_4^l\Psi+Err_{ml}^{\Psi},
    \end{equation}
    we also have the system of wave equations on $\{u=-1\}$ for any $0\leq m\leq M,$ $0\leq l\leq\frac{n}{2}-2:$
    \begin{equation}\label{model wave system 1}
    \begin{cases}
        v\nabla_4^2\nabla^m\nabla_4^l\alpha+\Big(3+l-\frac{n}{2}\Big)\nabla_4\nabla^m\nabla_4^l\alpha-\Delta\nabla^m\nabla_4^l\alpha=\psi\nabla^{m+1}\nabla_4^l\Psi+Err_{ml}^{\Psi} \\
        v\nabla_4^2\nabla^m\nabla_4^l\Psi^G+\Big(3+l-\frac{n}{2}\Big)\nabla_4\nabla^m\nabla_4^l\Psi^G-\Delta\nabla^m\nabla_4^l\Psi^G=\psi\nabla^{m+1}\nabla_4^l\Psi+Err_{ml}^{\Psi}.
    \end{cases}
    \end{equation}
\end{proposition}
\begin{proof}
    Using the equations satisfied by the Bianchi pair $(\alpha,\nu),$ we obtain:
    \[v\nabla_4^2\alpha_{AB}+\bigg(3-\frac{n}{2}\bigg)\nabla_4\alpha_{AB}=-\nabla_4\nabla^C\nu_{C(AB)}+\nabla_4\big(\psi\Psi^G\big)+\nabla_4\big(v\psi\Psi\big)=\]\[=\Delta\alpha_{AB}+\slashed{Riem}\cdot\alpha+\nabla_4\big(\psi\Psi^G\big)+\nabla_4\big(v\psi\Psi\big)+\nabla\big(\psi\Psi^G\big)=\Delta\alpha_{AB}+\psi\nabla\Psi+\Psi\Psi^G+\mathcal{F}_{101}^{lot}(\Psi)+v\mathcal{F}_{011}(\Psi),\]
    where we also used the other Bianchi equations and null structure equations on $\{u=-1\}$. We commute with $\nabla_4^l$:
    \begin{align*}
        v\nabla_4^2\nabla_4^l\alpha_{AB}+\bigg(3+l-\frac{n}{2}\bigg)\nabla_4\nabla_4^l\alpha_{AB}-\Delta\nabla_4^l\alpha_{AB}&=\psi\nabla\nabla_4^l\Psi+\mathcal{F}_{(2)(l-1)(l+1)}(\Psi)\\&+\mathcal{F}_{(1)(l)(l+1)}^{lot}(\Psi)+v\mathcal{F}_{(0)(l+1)(l+1)}(\Psi)+\nabla_4^l\big(\Psi\Psi^G\big).
    \end{align*}
    Commuting with $\nabla^m$, we obtain the equation for $\nabla^m\nabla_4^l\alpha.$ Next, we compute the wave equation satisfied by $\nu$:
    \begin{align*}
        v\nabla_4^2\nu_{ABC}&=-2\nabla_{[A}v\nabla_4\alpha_{B]C}+ v\nabla\big(\psi\Psi\big)+v\nabla_4\big(\psi\Psi^G\big)=\\
        &=\nabla_A\nabla^D\nu_{D(BC)}-\nabla_B\nabla^D\nu_{D(AC)}-\bigg(2-\frac{n}{2}\bigg)\nabla_4\nu_{ABC}+\mathcal{F}_{101}(\Psi^G)+v\mathcal{F}_{101}(\Psi)+v\mathcal{F}_{011}(\Psi).
    \end{align*}
    Using the constraint equations $\nu_{ABC}=2\nabla_{[A}\chi_{B]C}$ and $\nabla^A\chi_{AB}=\nabla_B\tr\chi,$ we can rewrite the first two terms on the RHS as $\Delta\nu_{ABC}+\nabla\big(\slashed{Riem}\cdot\chi\big).$ As a result, we get:
    \[v\nabla_4^2\nu_{ABC}=\Delta\nu_{ABC}-\bigg(2-\frac{n}{2}\bigg)\nabla_4\nu_{ABC}+\psi\nabla\Psi^G+\mathcal{F}^{lot}_{101}(\Psi)+v\mathcal{F}_{101}(\Psi)+v\mathcal{F}_{011}(\Psi)+\nabla(\psi\psi\psi).\]
    We remark that the term $\psi\nabla\alpha$ does not appear on the RHS. We commute as before to get the equation for $\nabla^m\nabla_4^l\nu,$ which does not contain $\psi\nabla^{m+1}\nabla_4^l\alpha$ on the RHS. 

    We use a similar argument for the remaining curvature components, and we briefly note here the structure of the equations that we use to rule out the dangerous term $\psi\nabla^{m+1}\nabla_4^l\alpha$ on the RHS. For $R$ we consider \eqref{v R} and note by signature considerations that $\alpha,\nu$ are absent. We differentiate the equation in $v$ and use the fact that \eqref{underline nu},\eqref{underline alpha} do not contain $\alpha,$ while \eqref{R} does not contain $\nabla\alpha.$ For $\nu$ we use the same argument, starting with \eqref{v underline nu}. We use a similar argument for $\underline{\alpha},$ starting with \eqref{underline alpha} and using the fact that $\alpha$ is absent in \eqref{v underline nu}. Thus, we proved \eqref{model wave system 2}.
    
    Finally, we note that \eqref{schematic equation PsiG} follows from the commutation formulas in Lemma~\ref{commutation lemma with F notation}. Thus, using this in \eqref{model wave system 2} we obtain \eqref{model wave system 1} as well.
\end{proof}

\subsection{Model Systems}\label{model systems section}

In this section we introduce two systems of linear wave equations on the background spacetime obtained in Theorem~\ref{stability of de sitter theorem in section}. These systems will model the linear part of the commuted Bianchi equations on $\{u=-1\}$ from the previous section, with the nonlinear part being contained in an inhomogeneous term. In our companion work \cite{Cwave}, we provide a detailed study of solutions to the model systems.

As inspired by our treatment of the linear wave equation on the background of de Sitter space in \cite{linearwave}, we consider the new time variable:
\[\tau=\sqrt{v},\ e_4=\frac{1}{2\tau}\partial_{\tau}.\]
We recall that all the tensors along $\{u=-1\}$ are expressed in a Lie propagated frame with respect to $e_4=\partial_v.$ We notice that $\mathcal{L}_4e_A=0$ is equivalent to $\mathcal{L}_{\tau}e_A=0$, so we can extend any tensors defined only at $\{\tau=0\}$ to be independent of $\tau.$ For any horizontal $k$-tensor $\Phi$ we compute that:
\begin{align}\label{formula for lie time derivative}
    \nabla_{\tau}\Phi&=\mathcal{L}_{\tau}\Phi+\tau\chi\cdot\Phi\\
    \nabla_{\tau}\nabla_{\tau}\Phi_{A_1\ldots A_k}&=\nabla_{\tau}\big(\nabla_{\tau}\Phi\big)_{A_1\ldots A_k}-\frac{1}{\tau}\nabla_{\tau}\Phi_{A_1\ldots A_k}.\notag
\end{align}
We consider the following model system:
\begin{definition}[First Model System]
    We assume that the smooth horizontal tensors $\Phi_0,\ldots,\Phi_I$ defined on the hypersurface $\{u=-1\}\times\{\tau\in(0,1)\}\times S^n$ of the spacetime $\big(\mathcal{M},g\big)$ obtained in Theorem~\ref{stability of de sitter theorem in section} satisfy the expansions for all $1\leq i\leq I$:
    \begin{align*}
        &\Phi_0=2\mathcal{O}\log\tau+h+O\big(\tau^2|\log\tau|^2\big),\ \nabla_{\tau}\Phi_0=\frac{2\mathcal{O}}{\tau}+O\big(\tau|\log\tau|^2\big)\text{ in }C^{\infty}(S^n)\\
        &\Phi_i=\Phi_i^0+O\big(\tau^2|\log\tau|^2\big),\ \nabla_{\tau}\Phi_i=O\big(\tau|\log\tau|^2\big)\text{ in }C^{\infty}(S^n)
    \end{align*}
    and the model system of wave equations on $\{u=-1\}$ for any $0\leq m\leq M$:
    \begin{equation}\label{first model system definition}
        \begin{cases}
            \nabla_{\tau}\big(\nabla_{\tau}\nabla^m\Phi_0\big)+\frac{1}{\tau}\nabla_{\tau}\nabla^m\Phi_0-4\Delta\nabla^m\Phi_0=\psi\nabla^{m+1}\Phi+F_{m}^{0} \\
            \nabla_{\tau}\big(\nabla_{\tau}\nabla^m\Phi_i\big)+\frac{1}{\tau}\nabla_{\tau}\nabla^m\Phi_i-4\Delta\nabla^m\Phi_i=\psi\nabla^{m+1}\Phi+F_{m}^{i},
        \end{cases}
    \end{equation}
    where the inhomogeneous terms satisfy $F_{m}^{0},F_{m}^{i}\in L^1_\tau\big([0,1]\big)C^{\infty}(S^n),$ and the covariant angular derivatives are with respect to the metric $\slashed{g}_{\tau}:=\slashed{g}_{u=-1,v=\tau^2}$ induced on $S_{\tau}=\{u=-1\}\times\{\tau\}\times S^n$.
\end{definition}

Based on equation (\ref{model wave system 1}), we obtain that: \[\Phi_0=\nabla_4^{\frac{n-4}{2}}\alpha,\ \Phi_i=\nabla_4^{\frac{n-4}{2}}\Psi^G,\ F_{m}^{0}=Err_{m,\frac{n-4}{2}}^{\Psi},\ F_{m}^{i}=Err_{m,\frac{n-4}{2}}^{\Psi}\]
satisfy the first model system, where the desired asymptotic expansions follow by Section~\ref{asymptotic completeness section}, or similarly by \cite{selfsimilarvacuum}. We prove estimates for this system in Section~\ref{model forward direction section} and \cite{Cwave}, and use these in Section~\ref{forward direction full system section} to obtain estimates for the commuted Bianchi system at finite times in terms of the asymptotic data at $\{v=0\}$.

Similarly, we also consider the following model system:
\begin{definition}[Second Model System]
    We assume that the smooth horizontal tensors $\Phi_0,\ldots,\Phi_I$ defined on the hypersurface $\{u=-1\}\times\{\tau\in(0,1)\}\times S^n$ of the spacetime $\big(\mathcal{M},g\big)$ obtained in Theorem \ref{stability of de sitter theorem in section} satisfy the expansions for all $1\leq i\leq I$:
    \begin{align*}
        &\Phi_0=2\mathcal{O}\log\tau+h+O\big(\tau^2|\log\tau|^2\big),\ \nabla_{\tau}\Phi_0=\frac{2\mathcal{O}}{\tau}+O\big(\tau|\log\tau|^2\big)\text{ in }C^{\infty}(S^n)\\
        &\Phi_i=\Phi_i^0+O\big(\tau^2|\log\tau|^2\big),\ \nabla_{\tau}\Phi_i=O\big(\tau|\log\tau|^2\big)\text{ in }C^{\infty}(S^n)
    \end{align*}
    and the model system of wave equations on $\{u=-1\}$ for any $0\leq m\leq M$:
    \begin{equation}\label{second model system definition}
        \begin{cases}
            \nabla_{\tau}\big(\nabla_{\tau}\nabla^m\Phi_0\big)+\frac{1}{\tau}\nabla_{\tau}\nabla^m\Phi_0-4\Delta\nabla^m\Phi_0=\psi\nabla^{m+1}\Phi+F_{m}^{0} \\
            \nabla_{\tau}\big(\nabla_{\tau}\nabla^m\Phi_i\big)-\frac{1}{\tau}\nabla_{\tau}\nabla^m\Phi_i-4\Delta\nabla^m\Phi_i=\sum_{j\neq0}\psi\nabla^{m+1}\Phi_j+F_{m}^{i},
        \end{cases}
    \end{equation}
    where the inhomogeneous terms satisfy $F_{m}^{0},F_{m}^{i}\in L^1_\tau\big([0,1]\big)C^{\infty}(S^n),$ and the covariant angular derivatives are with respect to the metric $\slashed{g}_{\tau}:=\slashed{g}_{u=-1,v=\tau^2}$ induced on $S_{\tau}=\{u=-1\}\times\{\tau\}\times S^n$.
\end{definition}

As before, equation (\ref{model wave system 2}) implies that: \[\Phi_0=\nabla_4^{\frac{n-4}{2}}\alpha,\ \Phi_i=\nabla_4^{\frac{n-4}{2}}\Psi^G,\ F_{m}^{0}=Err_{m,\frac{n-4}{2}}^{\Psi},\ F_{m}^{i}=Err_{m,\frac{n-4}{2}}^{\Psi}\]
satisfy the second model system. We prove estimates for this system in Section~\ref{model backward direction section} and \cite{Cwave}, and use these in Section~\ref{backward direction full system section} to obtain estimates for the asymptotic data at $\{v=0\},\ \{u=0\}$ to the commuted Bianchi system in terms of initial data at finite times.

\begin{remark}\label{remark about model systems in section}
    In \cite{Cwave} we give an equivalent definition of the model systems using the fact that $\slashed{g}_{\tau}=\widetilde{\slashed{g}}(\log(2\tau)),$ where $\widetilde{\slashed{g}}$ is defined in \eqref{general form for g tilde} and $\tau=e^T/2$, as in Remark~\ref{remark about model systems intro}. We also remark that the necessary assumptions on the background spacetime required in \cite{Cwave} follow by writing the estimates in Theorem~\ref{stability of de sitter theorem in section} using the $\tau$ coordinate.
\end{remark}

\section{Geometric Littlewood-Paley Theory}\label{LP Section}

An essential part of the argument in \cite{linearwave} is the use of Littlewood-Paley projections, which provide a robust way of constructing frequency dependent multipliers. We want to use the same approach for the model systems introduced in Section~\ref{model systems section}. The new difficulty is that the metric $\slashed{g}_{\tau}$ induced by the background on the spheres $S_{\tau}=\{u=-1\}\times\{\tau\}\times S^n$ has a nontrivial time dependence, compared to the case of de Sitter space. Our solution is to use the geometric Littlewood-Paley theory of \cite{geometricLP}. In this section, we state the main definitions and results of \cite{geometricLP} that we use. We also state a series of additional results needed in our situation that we proved in \cite{Cwave}.

For any tensor field $F$ on $S_{\tau}$, we denote by $U(z)F$ the solution on $[0,\infty)\times S_{\tau}$ to the heat equation:
\[\partial_zU(z)F-\Delta U(z)F=0,\ U(0)F=F,\]
where $\Delta$ is the Laplace-Beltrami operator on $\big(S_{\tau},\slashed{g}_{\tau}\big).$

For any $m\in\mathcal{M}$ smooth symbol as defined in \cite{geometricLP}, decaying at infinity, and satisfying vanishing moments properties, we set $m_k(z)=2^{2k}m(2^{2k}z).$ For any tensor field $F$ on $S_{\tau}$, we define the LP projection:
\[P_kF=\int_0^{\infty}m_k(z)U(z)Fdz.\]
We refer the reader to Theorem 5.5 in \cite{geometricLP} for the fundamental properties of these operators, similar to the standard LP projections. We use the following estimates for the  LP projections of \cite{geometricLP}:
\begin{proposition*}[{\cite[Theorem 5.5, Remark 5.6]{geometricLP}}]
    For an arbitrary LP projection, and any smooth tensor $F$ we have:
    \begin{enumerate}
        \item Bessel inequality.
        \[\sum_{k\in\mathbb{Z}}\big\|P_kF\big\|_{L^2}^2\lesssim\big\|F\big\|_{L^2}^2.\]
        \item Finite band property.
        \begin{align*}
            &\big\|\nabla P_kF\big\|_{L^2}\lesssim2^k\big\|F\big\|_{L^2},\ \big\|P_kF\big\|_{L^2}\lesssim2^{-k}\big\|\nabla F\big\|_{L^2}\\
            &\big\|\Delta P_kF\big\|_{L^2}^2\lesssim2^{2k}\big\|F\big\|_{L^2},\ \big\|P_kF\big\|_{L^2}\lesssim2^{-2k}\big\|\Delta F\big\|_{L^2}.
        \end{align*}
        \item $L^2$-almost orthogonality. For any two families of LP projections $P_k, \widetilde{P}_k$ we have:
        \[\big\| P_k\widetilde{P}_{k'}F\big\|_{L^2}\lesssim2^{-4|k-k'|}\cdot\big\|F\big\|_{L^2}.\]
    \end{enumerate}
\end{proposition*}
Our main use of the geometric LP projections is to define fractional Sobolev spaces. For the remainder of the paper we use the definition of fractional Sobolev spaces according to the following result of \cite{geometricLP}:
\begin{proposition*}[{\cite[Corollary~7.12]{geometricLP}}]
    For an arbitrary LP projection, $a\geq0$ and any smooth tensor $F,$ we have:
   \[\sum_{k\geq0}2^{2ak}\big\|P_kF\big\|^2_{L^2}\lesssim\big\|F\big\|_{H^{a}}^2.\]
   Moreover, if $\sum_kP_k^2=I$ and $a<4,$ then:
   \[\big\|F\big\|_{H^{a}}^2\lesssim\sum_{k\geq0}2^{2ak}\big\|P_kF\big\|^2_{L^2}+\big\|F\big\|_{L^2}^2.\]
\end{proposition*}
Following the ideas of \cite{geometricLP}, we prove the following result in \cite{Cwave}:
\begin{lemma}[{\cite[Lemma~2.3]{Cwave}}]\label{LP Projections lemma}
The LP projection operators satisfy the following bounds for $k\geq 0, m\leq M$:
\begin{equation}\label{LP bounds 1}
    [\nabla_t,P_k]=\frac{t}{2^{2k-1}}[\nabla_4,P_k],\text{ where } t=2^{k}\tau=2^k\sqrt{v}
\end{equation}
\begin{equation}\label{LP bounds 2}
    \big\|[\nabla^m,P_k]F\big\|_{L^2}\lesssim2^{-k} C\big(\|\slashed{Riem}\|_{H^{m-1}}\big)\cdot\big\|F\big\|_{H^{m-1}}
\end{equation}
\begin{equation}\label{LP bounds 3}
    \big\|\nabla[\nabla^m,P_k]F\big\|_{L^2}\lesssim2^{-k} C\big(\|\slashed{Riem}\|_{H^{m}}\big)\cdot\big\|F\big\|_{H^{m}}
\end{equation}
\begin{equation}\label{LP bounds 4}
    \big\|[\nabla_4,P_k]F\big\|_{L^2}\lesssim\big\|F\big\|_{L^2}
\end{equation}
\begin{equation}\label{LP bounds 5}
    \big\|\nabla[\nabla_4,P_k]F\big\|_{L^2}\lesssim\big\|F\big\|_{H^1}
\end{equation}
\begin{equation}\label{LP bounds 6}
    \big\|[P_k,G]F\big\|_{L^2}\lesssim2^{-k}\big\|G\big\|_{W^{2,\infty}}\big\|F\big\|_{L^2}
\end{equation}
\begin{equation}\label{LP bounds 7}
    \big\|\nabla[P_k,G]F\big\|_{L^2}\lesssim\big\|G\big\|_{W^{2,\infty}}\big\|F\big\|_{L^2}.
\end{equation}
\end{lemma}

We also need more refined versions of the estimates (\ref{LP bounds 4}) and (\ref{LP bounds 5}). We define the symbol $\widetilde{m}\in\mathcal{M}$ given by $\widetilde{m}(z)=zm(z),$ and we denote by $\widetilde{P}_k$ the associated projection operator. Moreover, we also introduce the projection operator $\underline{\widetilde{P}}_k$ which satisfies $\underline{\widetilde{P}}_k^2=\widetilde{P}_k.$

\begin{lemma}[{\cite[Lemma~2.4]{Cwave}}]\label{LP Projections lemma refined}
    We have the following estimates for $k\geq 0$:
    \begin{equation}\label{LP refined 1}
        [\nabla_4,P_k]F=2^{-2k}\chi\nabla^2\widetilde{P}_kF+O\Big(2^{-k}\big\|F\big\|_{L^2}\Big)
    \end{equation}
    \begin{equation}\label{LP refined 2}
        \big\|[\nabla_4,P_k]F\big\|_{L^2}\lesssim\big\|\underline{\widetilde{P}}_kF\big\|_{L^2}+2^{-k}\big\|F\big\|_{L^2}
    \end{equation}
    \begin{equation}\label{LP refined 3}
        \nabla[\nabla_4,P_k]F=2^{-2k}\chi\nabla^2\widetilde{P}_k\nabla F+O\Big(2^{-k}\big\|F\big\|_{H^1}\Big)
    \end{equation}
    \begin{equation}\label{LP refined 4}
        \big\|\nabla[\nabla_4,P_k]F\big\|_{L^2}\lesssim \big\|\underline{\widetilde{P}}_k\nabla F\big\|_{L^2}+2^{-k}\big\|F\big\|_{H^1}.
    \end{equation}
\end{lemma}

Additionally, we have an estimate where we trade $1/2$ derivatives on $F$ for $2^{k/2}$ growth:
\begin{lemma}[{\cite[Lemma~2.5]{Cwave}}]\label{fractional LP bounds}
    We have the following estimates for $k\geq 0$:
    \begin{equation}\label{fractional LP bounds 1}
    \big\|[\nabla_4,P_k]\nabla F\big\|_{L^2}+\big\|\nabla[\nabla_4,P_k] F\big\|_{L^2}\lesssim2^{k/2}\big\|F\big\|_{H^{1/2}}.
\end{equation}
\end{lemma}

For the remainder of the paper, we make the convention that the projection $P_k$ satisfies $\sum_kP_k^2=I.$ We point out that the estimates stated above are valid for any projection operator with symbol in $\mathcal{M}.$

In Section~\ref{model backward direction section} we need the following refined Poincaré inequality for LP projections:
\begin{lemma}[{\cite[Lemma~2.6]{Cwave}}] For any $k\geq0,$ and $\delta>0$, we have the inequality:
    \begin{equation}\label{refined Poincare inequality}
    \big\|P_kF\big\|_{L^2}^2\lesssim\frac{1}{\delta}2^{-2k}\big\|\nabla P_kF\big\|_{L^2}^2+\delta\sum_{0\leq l<k}2^{-9k+7l}\big\|\nabla P_lF\big\|_{L^2}^2+\delta^{-1}2^{-4k}\big\|F\big\|_{L^2}^2.
\end{equation}
\end{lemma}

Finally, we define for any smooth tensor $F$ on $S_0$ and $k\geq0$:
\[(\log\nabla)F=\sum_{l\geq0}P_l^2F\cdot\log2^l,\]
\begin{equation}\label{definition of Rk}
    R_kF=2P_k(\log\nabla)F-2\log2^k\cdot P_kF=2\sum_{l\geq0}\log2\cdot(l-k)\cdot P_kP_l^2F-2\sum_{l<0}\log2^k\cdot P_kP_l^2F.
\end{equation}
We consider the projection operator $\underline{P}_k$ which satisfies $\underline{P}_k^2=P_k.$ We have the estimates for $R_k:$
\begin{lemma}[{\cite[Lemma~2.8]{Cwave}}]\label{Rk lemma} For any smooth tensor $F$ on $S_0$, we extend $R_kF$ to be independent of $\tau.$ We also denote $t=2^k\tau.$ Then, for any $k\geq0$ we have:
    \begin{equation}\label{Rk 1}
        \big\|\Delta_{\slashed{g}_{\tau}} R_kF\big\|_{L^2}\lesssim2^{k}\big\|\underline{P}_kF\big\|_{H^1}
    \end{equation}
    \begin{equation}\label{Rk 2}
        \big\|\nabla R_kF\big\|_{L^2}\lesssim\big\|\underline{P}_kF\big\|_{H^1}
    \end{equation}
    \begin{equation}\label{Rk 3}
        2^{k}\big\|R_kF\big\|_{L^2}\lesssim \big\|\underline{P}_kF\big\|_{H^1}
    \end{equation}
    \begin{equation}\label{Rk 4}
        \big\|\nabla_tR_kF\big\|_{L^2}\lesssim2^{-3k}t\big\|\underline{P}_kF\big\|_{H^1}
    \end{equation}
    \begin{equation}\label{Rk 5}
        \big\|\nabla\nabla_tR_kF\big\|_{L^2}\lesssim2^{-2k}t\big\|\underline{P}_kF\big\|_{H^1}.
    \end{equation}
\end{lemma}

\section{Estimates for the First Model System}\label{model forward direction section}

One of the central parts of our argument is proving estimates for the first model system (\ref{first model system definition}), in terms of the asymptotic data at $\{\tau=0\}.$ This system includes the commuted Bianchi system, giving estimates on the solution at finite times in terms of the data at $\{v=0\}$.

In the present paper we illustrate how to prove the top order estimates for the singular component of $\Phi_0,$ which decouples from the rest of the system. This will serve as a guideline for the general model systems treated in \cite{Cwave}. We encourage the reader to return to Section~\ref{scattering map section intro} of the introduction for an outline of the proof. We refer the reader to \cite{Cwave} for a complete proof for the system (\ref{first model system definition}), where we prove:
\begin{theorem}[{\cite[Theorem~1.1]{Cwave}}]\label{forward direction main result theorem general}
For any $M>0$ large enough, the system (\ref{first model system definition}) satisfies the estimates for all $\tau\in(0,1]$:
\[\tau^2\big\|\nabla_{\tau}\nabla^M\Phi_0\big\|^2_{H^{1/2}}+\tau^2\big\|\nabla^M\Phi_0\big\|^2_{H^{3/2}}+\sum_{i=1}^I\Big(\tau\big\|\nabla_{\tau}\nabla^M\Phi_i\big\|^2_{H^{1/2}}+ \tau\big\|\nabla^M\Phi_i\big\|^2_{H^{3/2}}+\big\|\Phi_i\big\|^2_{H^{M+1}}\Big)\lesssim\]\[\lesssim\big\|\mathcal{O}\big\|^2_{H^{M+1}}+\big\|\mathfrak{h}\big\|^2_{H^{M+1}}+\sum_{i=1}^I\big\|\Phi_i^0\big\|^2_{H^{M+1}}+\sum_{m=0}^{M}\sum_{i=0}^I\int_0^{\tau}\big\|F_m^i\big\|_{L^2}^2+\sum_{i=0}^I\int_0^{\tau}\tau'\big\|F_M^i\big\|_{H^{1/2}}^2,\]
\[\big\|\Phi_0\big\|^2_{H^{M+1}}\lesssim\big(1+|\log\tau|^2\big) \big\|\mathcal{O}\big\|^2_{H^{M+1}}+\big\|\mathfrak{h}\big\|^2_{H^{M+1}}+\sum_{i=1}^I\big\|\Phi_i^0\big\|^2_{H^{M+1}}+\sum_{m=0}^{M}\sum_{i=0}^I\int_0^{\tau}\big\|F_m^i\big\|_{L^2}^2+\sum_{i=0}^I\int_0^{\tau}\tau'\big\|F_M^i\big\|_{H^{1/2}}^2,\]
where we define $(\log\nabla)\mathcal{O}=\sum_{k\geq0}P_k^2\mathcal{O}\cdot\log2^k,\ \mathfrak{h}=h-2(\log\nabla)\mathcal{O},$ and the Sobolev spaces are defined in Section~\ref{LP Section}. In the above estimates, the implicit constants depend on $M>N$, the bounds on $\big\|\slashed{Riem}(\slashed{g}_0)\big\|_{H^{M}},$ and the bounds satisfied by the background $\big(\mathcal{M},g\big)$ according to Theorem~\ref{stability of de sitter theorem in section}.
\end{theorem}

In order to prove this result, we must decompose $\Phi_0$ into its singular and regular components, similarly to \cite{linearwave}. We have for each $m\leq M:$
\[\nabla^m\Phi_0=\big(\nabla^m\Phi_0\big)_Y+\big(\nabla^m\Phi_0\big)_J,\]
where we define the singular component $\big(\nabla^m\Phi_0\big)_Y$ to be the horizontal tensor that solves the linear equation:
\begin{equation}\label{equation for alpha Y}
    \nabla_{\tau}\big(\nabla_{\tau}\big(\nabla^m\Phi_0\big)_Y\big)+\frac{1}{\tau}\nabla_{\tau}\big(\nabla^m\Phi_0\big)_Y-4\Delta\big(\nabla^m\Phi_0\big)_Y=\psi\nabla\big(\nabla^m\Phi_0\big)_Y
\end{equation}
\[\big(\nabla^m\Phi_0\big)_Y({\tau})=2\nabla^m\mathcal{O}\log({\tau})+2(\log\nabla)\nabla^m\mathcal{O}+ O\big({\tau}^2|\log({\tau})|^2\big),\ \nabla_{\tau}\big(\nabla^m\Phi_0\big)_Y({\tau})=\frac{2\nabla^m\mathcal{O}}{\tau}+ O\big({\tau}|\log({\tau})|^2\big)\]
and we also define as above $(\log\nabla)\nabla^m\mathcal{O}=\sum_{k\geq0}P_k^2\nabla^m\mathcal{O}\cdot\log2^k,\ \mathfrak{h}_m=\nabla^mh-2(\log\nabla)\nabla^m\mathcal{O}$. We define the regular component of $\nabla^m\Phi_0$ by $\big(\nabla^m\Phi_0\big)_J=\nabla^m\Phi_0-\big(\nabla^m\Phi_0\big)_Y.$ This satisfies the equations:
\begin{equation}\label{equation for alpha J}
    \nabla_{\tau}\big(\nabla_{\tau}\big(\nabla^m\Phi_0\big)_J\big)+\frac{1}{\tau}\nabla_{\tau}\big(\nabla^m\Phi_0\big)_J-4\Delta\big(\nabla^m\Phi_0\big)_J=\psi\nabla\big(\nabla^m\Phi_0\big)_J+\sum_{j=1}^I\psi\nabla\nabla^m\Phi_{j}+F_{m}^{0}
\end{equation}
\[\big(\nabla^m\Phi_0\big)_J({\tau})=\mathfrak{h}_m+ O\big({\tau}^2|\log({\tau})|^2\big),\ \nabla_{\tau}\big(\nabla^m\Phi_0\big)_J({\tau})=O\big({\tau}|\log({\tau})|^2\big)\text{ in }C^{\infty}(S^n).\]
The notation for the regular and singular components is based on the similarities to the first and second Bessel functions $J_0,\ Y_0,$ as in the case of \cite{linearwave}. The regular component is better behaved at $\tau=0,$ similarly to the tensors $\Phi_1,\ldots,\Phi_I.$ The need to renormalize the asymptotic data $h$ to $\mathfrak{h}$ follows from the analysis of the singular component. In \cite{Cwave}, we prove separately estimates for $\big(\Phi_0\big)_Y$ and for $\big(\Phi_0\big)_J,\Phi_1,\ldots,\Phi_I,$ by considering their behavior in the low frequency and high frequency regimes. The main difficulty is present in the analysis of $\big(\Phi_0\big)_Y$ at top order, due to its singularity at $\tau=0.$ We illustrate the key aspects of the problem by proving the result: 
\begin{theorem}\label{forward direction main result theorem} For any $M>0$ large enough, the singular component satisfies the estimates for all $\tau\in(0,1]$:
\begin{equation}\label{main estimate alpha Y top order}
    \tau^2\big\|\nabla_{\tau}\big(\nabla^M\Phi_0\big)_Y\big\|^2_{H^{1/2}}+\tau^2\big\|\nabla\big(\nabla^M\Phi_0\big)_Y\big\|^2_{H^{1/2}}\lesssim\big\|\mathcal{O}\big\|^2_{H^{M+1}},
\end{equation}
\begin{equation}\label{practical estimate alpha Y top order}
    \sum_{m=0}^M\big\|\big(\nabla^m\Phi_0\big)_Y\big\|^2_{H^{1}}\lesssim\big(1+|\log\tau|^2\big) \big\|\mathcal{O}\big\|^2_{H^{M+1}}.
\end{equation}
\end{theorem}
We point out that in this result we are using the notation convention from Section~\ref{model systemmm}, so the implicit constants depend only on $M.$ We follow the strategy outlined in Section~\ref{scattering map section intro}. We first obtain lower order estimates in Section~\ref{forward lower order estimates section} using the results of \cite{Cwave}. At top order, we consider each $P_k$ projection of the solution, and we treat separately the low frequency regime in Section~\ref{forward low freq estimates section} and the high frequency regime in Section~\ref{forward high freq estimates section}. In Section~\ref{forward main result estimates section} we combine the low frequency regime and the high frequency regime estimates to prove Theorem \ref{forward direction main result theorem}.

\subsection{Lower Order Estimates}\label{forward lower order estimates section}

In this section we briefly outline the lower order estimates for the singular component that are obtained in \cite{Cwave}. We point out that these results are not sharp, but all the quantities that we consider are controlled using $\big\|\mathcal{O}\big\|_{H^{M+1}}.$ We set $\eta=1/10.$ According to \cite[Section~3]{Cwave}, we have:
\begin{proposition*}[{\cite[Proposition~3.2, Proposition~3.5]{Cwave}}]
    For any $m<M$, the singular component satisfies the following estimates:
    \begin{align*}
        \big\|\big(\nabla^m\Phi_{0}\big)_Y\big\|^2_{H^1}&\lesssim\big(1+|\log\tau|^2\big)\big\|\mathcal{O}\big\|^2_{H^{m+1+\eta}}\\
        \big\|\nabla(\nabla^{M-1}\Phi_0)_Y\big\|^2_{H^{1/2}}&\lesssim\big(1+|\log\tau|^2\big)\big\|\mathcal{O}\big\|^2_{H^{M+1/2+\eta}}\\
        \big\|\nabla_{\tau}(\nabla^{M-1}\Phi_0)_Y\big\|^2_{H^{1/2}}&\lesssim\frac{1}{\tau^2}\big\|\mathcal{O}\big\|^2_{H^{M+1/2+\eta}}.
    \end{align*}
\end{proposition*}
The idea of the proof is to consider the decomposition $\big(\nabla^m\Phi_0\big)_Y=\big(\nabla^m\Phi_0\big)_Y^1+\big(\nabla^m\Phi_0\big)_Y^2$ for each $m<M,$ where $\big(\nabla^m\Phi_0\big)_Y^1,\ \big(\nabla^m\Phi_0\big)_Y^2$ are the solutions of (\ref{equation for alpha Y}) that satisfy the expansions:
    \[\big(\nabla^m\Phi_0\big)_Y^1({\tau})=2\nabla^m\mathcal{O}\log({\tau})+ O\big({\tau}^2|\log({\tau})|^2\big),\ \nabla_{\tau}\big(\nabla^m\Phi_0\big)_Y^1({\tau})=\frac{2\nabla^m\mathcal{O}}{\tau}+ O\big({\tau}|\log({\tau})|^2\big)\]
    \[\big(\nabla^m\Phi_0\big)_Y^2({\tau})=2(\log\nabla)\nabla^m\mathcal{O}+ O\big({\tau}^2|\log({\tau})|^2\big),\ \nabla_{\tau}\big(\nabla^m\Phi_0\big)_Y^2({\tau})=O\big({\tau}|\log({\tau})|^2\big).\]
Using this decomposition, the first statement follows using standard $\nabla_{\tau}$ multipliers in the equations satisfied by $\big(\nabla^m\Phi_0\big)_Y^1/\log\tau$ and $\big(\nabla^m\Phi_0\big)_Y^2,$ similarly to the \cite{linearwave}. We need to allow for the constant $\eta>0$ on the right hand side of the estimates, due to the inequality:
\[\big\|(\log\nabla)\nabla^m\mathcal{O}\big\|_{H^1}\lesssim\big\|\mathcal{O}\big\|_{H^{m+1+\eta}}.\]
The proof of the last two statements follows a similar strategy as above, in the case $m=M-1.$ However, to obtain the fractional estimates we need to first commute the equations with the LP operators $P_k,\ k\geq0$ of Section~\ref{LP Section}, and multiply each estimate by $2^k,$ which amounts to taking half of a derivative. The additional commutation terms obtained because of the LP projections are handled using the bounds provided in Section~\ref{LP Section}.

Additionally, we remark that the need to define the singular component of each $\nabla^m\Phi_0$ is caused by the failure of the operators $\log\nabla$ and $\nabla^m$ to commute. We will also need estimates for the commutator term $\mathcal{C}=\big(\nabla^M\Phi_0\big)_Y-\nabla\big(\nabla^{M-1}\Phi_0\big)_Y$. This satisfies a similar equation to $\big(\Phi_0\big)_J,$ and has a regular expansion:
\[\mathcal{C}=2[\log\nabla,\nabla]\nabla^{M-1}\mathcal{O}+O\big({\tau}^2|\log({\tau})|^2\big),\ \nabla_{\tau}\mathcal{C}=O\big({\tau}|\log({\tau})|^2\Big).\]
According to \cite[Section~3]{Cwave}, we have the estimate:
\begin{proposition*}[{\cite[Proposition~3.4]{Cwave}}]
    The commutator term $\mathcal{C}=\big(\nabla^M\Phi_0\big)_Y-\nabla\big(\nabla^{M-1}\Phi_0\big)_Y$ satisfies the estimate:
    \[\big\|\nabla_{\tau}\mathcal{C}\big\|^2_{L^2}+\big\|\mathcal{C}\big\|^2_{H^1}\lesssim\big\|\mathcal{O}\big\|^2_{H^{M+\eta}}+\int_0^{\tau}\big\|\tau'\nabla_{\tau}\big(\nabla^M\Phi_0\big)_Y\big\|_{L^2}^2d\tau'.\]
    Moreover, for any $k\geq0,\ \tau\in(0,2^{-k-1}]$ we have:
    \[\big\|\mathcal{C}\big\|^2_{L^2}\lesssim\big\|\mathcal{O}\big\|^2_{H^{M+\eta}}+2^{-k}\int_0^{\tau}\big\|\tau'\nabla_{\tau}\big(\nabla^M\Phi_0\big)_Y\big\|_{L^2}^2d\tau'.\]
\end{proposition*}

\subsection{Top Order Estimates}\label{forward top order estimates section}

To prove top order estimates, we need a precise understanding of the behavior of the $P_k$ projections of $\big(\nabla^M\Phi_0\big)_Y$. As in the case of the linear wave equation on de Sitter background studied in \cite{linearwave}, we need to treat differently the low frequency regime $\tau\in(0,2^{-k-1}]$ and the high frequency regime $\tau\in[2^{-k-1},1]$.

\subsubsection{Low Frequency Estimates}\label{forward low freq estimates section}

For every $k\geq0$ we have the following expansions:
\begin{align*}
    P_k\big(\nabla^M\Phi_0\big)_Y({\tau})&=2P_k\nabla^M\mathcal{O}\log({2^k\tau})+R_k\nabla^M\mathcal{O}+O\big({\tau}^2|\log({\tau})|^2\big),\\
    P_k\nabla_{2^{-k}\partial_\tau}\big(\nabla^M\Phi_0\big)_Y({\tau})&=2P_k\nabla^M\mathcal{O}\frac{1}{2^k\tau}+O\big({\tau}|\log({\tau})|^2\big),
\end{align*}
where we defined $R_k\nabla^M\mathcal{O}$ by (\ref{definition of Rk}). We point out that $\nabla^M\mathcal{O},\ P_k\nabla^M\mathcal{O},\ R_k\nabla^M\mathcal{O}$ on the right hand side are defined at $\tau=0$ and extended by Lie transport. This can be done since the difference between projecting with respect to $\slashed{g}_{0}$ or $\slashed{g}_{\tau}$ is $O(\tau^2)$ according to \cite[Lemma~2.9]{Cwave}.

We prove the main low frequency regime estimates in Propositions~\ref{singular component low frequency proposition} and \ref{singular component low frequency proposition 2}. Our goal is to prove energy estimates on $\tau\in(0,2^{-k-1}]$ for the singular component $P_k\big(\nabla^M\Phi_0\big)_Y/\log(2^k\tau),$ renormalized to account for the contribution of $R_k\nabla^M\mathcal{O}.$ According to the above expansions, the asymptotic data will be given by $2P_k\nabla^M\mathcal{O}.$

It is convenient to consider the new time variable $t=2^k\tau$. When using $\nabla_t$ below as a multiplier, we can control the error terms resulting from time derivatives of the volume form and apply Gronwall on the interval $t\in(0,2^k]$ since for any tensor $F$ we have:
\[\frac{1}{2}\frac{d}{dt}\big\|F\big\|_{L^2}^2=\int_{S_{t}}F\cdot\nabla_{t}F+O\big(2^{-2k}t\|F\|_{L^2}^2\big).\]
 
We prove a preliminary low frequency regime estimate:
\begin{proposition}\label{singular low frequency preliminary proposition} For any $k\geq0$ and $\tau\leq2^{-k-1}$, we have that $\big(\nabla^M\Phi_0\big)_Y$ satisfies the estimate:
    \[\bigg\|P_k\nabla_{\tau}\frac{\big(\nabla^M\Phi_0\big)_Y}{\log2^k\tau}-R_k\nabla^M\mathcal{O}\partial_{\tau}\bigg(\frac{1}{\log2^k\tau}\bigg)\bigg\|^2_{L^2}+\bigg\|\nabla\bigg(P_k\frac{\big(\nabla^M\Phi_0\big)_Y}{\log2^k\tau}-\frac{R_k\nabla^M\mathcal{O}}{\log 2^k\tau}\bigg)\bigg\|^2_{L^2}\lesssim\big\|\nabla P_k\nabla^M\mathcal{O}\big\|^2_{L^2}+\]\[+\big\|\underline{P}_k\nabla^M\mathcal{O}\big\|^2_{H^1}+\int_0^{\tau}\frac{(\tau')^2}{2^{k}}\bigg\|\nabla[P_k,\nabla_4]\frac{\big(\nabla^M\Phi_0\big)_Y}{\log 2^k\tau'}\bigg\|_{L^2}^2+\int_0^{\tau}\frac{(\tau')^2}{2^{k}}\bigg\|[P_k,\nabla_4]\nabla_{\tau}\frac{\big(\nabla^M\Phi_0\big)_Y}{\log 2^k\tau'}\bigg\|_{L^2}^2+\int_0^{\tau}\bigg\|\frac{\big(\nabla^M\Phi_0\big)_Y}{\log2^k\tau'}\bigg\|_{L^2}^2\]
\end{proposition}
\begin{proof}
    In terms of the time variable $t=2^k\tau\leq1/2$, the main equation is:
    \[\nabla_{t}\bigg(\nabla_t\frac{\big(\nabla^M\Phi_0\big)_Y}{\log t}\bigg)+\frac{1}{t}\bigg(1+\frac{2}{\log t}\bigg)\nabla_{t}\frac{\big(\nabla^M\Phi_0\big)_Y}{\log t}-\frac{4}{2^{2k}}\Delta\frac{\big(\nabla^M\Phi_0\big)_Y}{\log t}=\frac{1}{2^{2k}}\psi\nabla\frac{\big(\nabla^M\Phi_0\big)_Y}{\log t}\]
    We apply $P_k$ to the equation:
    \[\nabla_{t}\bigg(P_k\nabla_t\frac{\big(\nabla^M\Phi_0\big)_Y}{\log t}\bigg)+\frac{1}{t}\bigg(1+\frac{2}{\log t}\bigg)P_k\nabla_{t}\frac{\big(\nabla^M\Phi_0\big)_Y}{\log t}-\frac{4}{2^{2k}}\Delta P_k\frac{\big(\nabla^M\Phi_0\big)_Y}{\log t}=\frac{1}{2^{2k}}\psi\nabla P_k\frac{\big(\nabla^M\Phi_0\big)_Y}{\log t}+\]\[+[\nabla_t,P_k]\nabla_t\frac{\big(\nabla^M\Phi_0\big)_Y}{\log t}+\frac{1}{2^{2k}}[P_k,\nabla]\frac{\psi\big(\nabla^M\Phi_0\big)_Y}{\log t}-\frac{1}{2^{2k}}P_k\bigg(\nabla\psi\cdot\frac{\big(\nabla^M\Phi_0\big)_Y}{\log t}\bigg)+\frac{1}{2^{2k}}[\nabla P_k,\psi]\frac{\big(\nabla^M\Phi_0\big)_Y}{\log t}\]
    We introduce the notation:
    \[X=P_k\nabla_t\frac{\big(\nabla^M\Phi_0\big)_Y}{\log t}-R_k\nabla^M\mathcal{O}\partial_t\bigg(\frac{1}{\log t}\bigg),\ Y=P_k\frac{\big(\nabla^M\Phi_0\big)_Y}{\log t}-\frac{R_k\nabla^M\mathcal{O}}{\log t},\]
    \[Z=2^{-2k}\frac{4\Delta R_k\nabla^M\mathcal{O}+\psi\nabla R_k\nabla^M\mathcal{O}}{\log t}+\frac{\nabla_{t}\big(R_k\nabla^M\mathcal{O}\big)}{t|\log t|^2},\ X=\nabla_tY+[P_k,\nabla_t]\frac{\big(\nabla^M\Phi_0\big)_Y}{\log t}+\frac{\nabla_tR_k\nabla^M\mathcal{O}}{\log t}.\]
    Using this, we rewrite the equation as:
    \[\nabla_{t}X+\frac{1}{t}\bigg(1+\frac{2}{\log t}\bigg)X-\frac{4}{2^{2k}}\Delta Y=\frac{1}{2^{2k}}\psi\nabla Y+[\nabla_t,P_k]\nabla_t\frac{\big(\nabla^M\Phi_0\big)_Y}{\log t}+\]\[+\frac{1}{2^{2k}}[P_k,\nabla]\frac{\psi\big(\nabla^M\Phi_0\big)_Y}{\log t}+Z-\frac{1}{2^{2k}}P_k\bigg(\nabla\psi\cdot\frac{\big(\nabla^M\Phi_0\big)_Y}{\log t}\bigg)+\frac{1}{2^{2k}}[\nabla P_k,\psi]\frac{\big(\nabla^M\Phi_0\big)_Y}{\log t}\]
    We obtain the energy estimate by contracting with $X$ and integrating by parts:
    \[\big\|X\big\|^2_{L^2}+\frac{1}{2^{2k}}\big\|\nabla Y\big\|^2_{L^2}\lesssim\frac{1}{2^{2k}}\big\|\nabla P_k\nabla^M\mathcal{O}\big\|^2_{L^2}+\int_0^t\frac{\textbf{1}_{[1/10,1/2]}}{t'|\log t'|}\big\|X\big\|^2_{L^2}+\int_0^t\frac{1}{2^{2k}}\big\|\nabla Y\big\|_{L^2}\bigg\|\nabla[P_k,\nabla_t]\frac{\big(\nabla^M\Phi_0\big)_Y}{\log t'}\bigg\|_{L^2}\]\[+\int_0^t\frac{1}{2^{2k}}\big\|\nabla Y\big\|_{L^2}\big\|[\nabla,\nabla_t] Y\big\|_{L^2}+\int_0^t\frac{1}{2^{2k}}\big\|\nabla Y\big\|_{L^2}\bigg\|\frac{\nabla\nabla_tR_k\nabla^M\mathcal{O}}{\log t'}\bigg\|_{L^2}+\int_0^t\frac{1}{2^{2k}}\big\|X\big\|_{L^2}\big\|\nabla Y\big\|_{L^2}+\]\[+\int_0^t\big\|X\big\|_{L^2}\bigg\|[\nabla_t,P_k]\nabla_t\frac{\big(\nabla^M\Phi_0\big)_Y}{\log t'}\bigg\|_{L^2}+\int_0^t\frac{1}{2^{2k}}\big\|X\big\|_{L^2}\bigg\|[P_k,\nabla]\frac{\psi\big(\nabla^M\Phi_0\big)_Y}{\log t'}\bigg\|_{L^2}+\]\[+\int_0^t\frac{1}{2^{2k}}\big\|X\big\|_{L^2}\bigg\|P_k\bigg(\nabla\psi\cdot\frac{\big(\nabla^M\Phi_0\big)_Y}{\log t'}\bigg)\bigg\|_{L^2}+\int_0^t\frac{1}{2^{2k}}\big\|X\big\|_{L^2}\bigg\|[\nabla P_k,\psi]\frac{\big(\nabla^M\Phi_0\big)_Y}{\log t'}\bigg\|_{L^2}+\int_0^t\big\|X\big\|_{L^2}\big\|Z\big\|_{L^2}\]
    We point out that the error term $\frac{1}{t'|\log t'|}\big\|X\big\|^2_{L^2}\cdot\textbf{1}_{[1/10,1/2]}$ appears on the RHS because for $t\in[0,1/10]$ we have $1+2/\log t\gtrsim1.$ We use Gronwall for $t\in\big[0,1/2\big],$ and the bounds in Lemma \ref{LP Projections lemma}:
    \[\big\|X\big\|^2_{L^2}+\frac{1}{2^{2k}}\big\|\nabla Y\big\|^2_{L^2}\lesssim\frac{1}{2^{2k}}\big\|\nabla P_k\nabla^M\mathcal{O}\big\|^2_{L^2}+\int_0^t\frac{(t')^2}{2^{6k}}\big\|[\nabla,\nabla_4] Y\big\|_{L^2}^2+\int_0^t\frac{(t')^2}{2^{6k}}\bigg\|\nabla[P_k,\nabla_4] \frac{\big(\nabla^M\Phi_0\big)_Y}{\log t'}\bigg\|_{L^2}^2\]\[+\int_0^t\frac{(t')^2}{2^{4k}}\bigg\|[\nabla_4,P_k]\nabla_t\frac{\big(\nabla^M\Phi_0\big)_Y}{\log t'}\bigg\|_{L^2}^2+\int_0^t\frac{1}{2^{3k}}\bigg\|\frac{\big(\nabla^M\Phi_0\big)_Y}{\log t'}\bigg\|_{L^2}^2+\int_0^t\frac{1}{2^{2k}}\bigg\|\frac{\nabla\nabla_tR_k\nabla^M\mathcal{O}}{\log t'}\bigg\|_{L^2}^2+\int_0^t\big\|Z\big\|_{L^2}^2\]
    We get from Lemma \ref{LP Projections lemma} and Lemma \ref{Rk lemma} that for $\underline{P}_k^2=P_k$:
    \[\big\|X\big\|^2_{L^2}+\frac{1}{2^{2k}}\big\|\nabla Y\big\|^2_{L^2}\lesssim\frac{1}{2^{2k}}\big\|\nabla P_k\nabla^M\mathcal{O}\big\|^2_{L^2}+\frac{1}{2^{2k}}\big\|\underline{P}_k\nabla^M\mathcal{O}\big\|^2_{H^1}+\int_0^t\frac{1}{2^{3k}}\bigg\|\frac{\big(\nabla^M\Phi_0\big)_Y}{\log t}\bigg\|_{L^2}^2\]\[+\int_0^t\frac{(t')^2}{2^{6k}}\bigg\|\nabla[P_k,\nabla_4] \frac{\big(\nabla^M\Phi_0\big)_Y}{\log t}\bigg\|_{L^2}^2+\int_0^t\frac{(t')^2}{2^{4k}}\bigg\|[\nabla_4,P_k]\nabla_t\frac{\big(\nabla^M\Phi_0\big)_Y}{\log t}\bigg\|_{L^2}^2\]
    Changing coordinates back to $\tau$ we obtain the desired estimate.
\end{proof}

The right hand side of the estimate in Proposition \ref{singular low frequency preliminary proposition} cannot be bounded using the top order terms on the left hand side. However, as explained in Section~\ref{scattering map section intro} of the introduction, we can bound these terms using the lower order estimates (including the fractional estimates) and the commutator estimate stated in Section~\ref{forward lower order estimates section}. We obtain our main low frequency regime estimates in Propositions~\ref{singular component low frequency proposition} and \ref{singular component low frequency proposition 2}:

\begin{proposition}\label{singular component low frequency proposition}
    For any $k\geq0$ and $\tau\leq2^{-k-1}$, we have that $\big(\nabla^M\Phi_0\big)_Y$ satisfies the estimate:
    \[\bigg\|P_k\nabla_{\tau}\frac{\big(\nabla^M\Phi_0\big)_Y}{\log2^k\tau}-R_k\nabla^M\mathcal{O}\partial_{\tau}\bigg(\frac{1}{\log2^k\tau}\bigg)\bigg\|^2_{L^2}+\bigg\|\nabla\bigg(P_k\frac{\big(\nabla^M\Phi_0\big)_Y}{\log2^k\tau}-\frac{R_k\nabla^M\mathcal{O}}{\log 2^k\tau}\bigg)\bigg\|^2_{L^2}\lesssim\]\[\lesssim\big\|\nabla P_k\nabla^M\mathcal{O}\big\|^2_{L^2}+\big\|\underline{P}_k\nabla^M\mathcal{O}\big\|^2_{H^1}+2^{-k/2}\big\|\mathcal{O}\big\|^2_{H^{M+1/2+\eta}}+2^{-2k}\int_0^{\tau}\big\|\tau'\nabla_{\tau}\big(\nabla^M\Phi_0\big)_Y\big\|_{L^2}^2.\]
\end{proposition}
\begin{proof}
    We recall the notation $\mathcal{C}=\big(\nabla^M\Phi_0\big)_Y-\nabla\big(\nabla^{M-1}\Phi_0\big)_Y$. We can rewrite the estimate in Proposition \ref{singular low frequency preliminary proposition} as:
    \[\bigg\|P_k\nabla_{\tau}\frac{\big(\nabla^M\Phi_0\big)_Y}{\log2^k\tau}-R_k\nabla^M\mathcal{O}\partial_{\tau}\bigg(\frac{1}{\log2^k\tau}\bigg)\bigg\|^2_{L^2}+\bigg\|\nabla\bigg(P_k\frac{\big(\nabla^M\Phi_0\big)_Y}{\log2^k\tau}-\frac{R_k\nabla^M\mathcal{O}}{\log 2^k\tau}\bigg)\bigg\|^2_{L^2}\lesssim\big\|\nabla P_k\nabla^M\mathcal{O}\big\|^2_{L^2}+\]\[+\big\|\underline{P}_k\nabla^M\mathcal{O}\big\|^2_{H^1}+\int_0^{\tau}\frac{(\tau')^2}{2^{k}}\bigg\|\nabla[P_k,\nabla_4]\frac{\mathcal{C}}{\log 2^k\tau'}\bigg\|_{L^2}^2+\int_0^{\tau}\frac{(\tau')^2}{2^{k}}\bigg\|[P_k,\nabla_4]\nabla_{\tau}\frac{\mathcal{C}}{\log 2^k\tau'}\bigg\|_{L^2}^2+\int_0^{\tau}\bigg\|\frac{\mathcal{C}}{\log2^k\tau'}\bigg\|_{L^2}^2+\]
    \[+\int_0^{\tau}\frac{(\tau')^2}{2^{k}}\bigg\|[P_k,\nabla_4]\frac{\nabla\big(\nabla^{M-1}\Phi_0\big)_Y}{\log 2^k\tau'}\bigg\|_{H^1}^2+\int_0^{\tau}\frac{(\tau')^2}{2^{k}}\bigg\|[P_k,\nabla_4]\nabla\frac{\nabla_{\tau}\big(\nabla^{M-1}\Phi_0\big)_Y}{\log 2^k\tau'}\bigg\|_{L^2}^2+\]\[+\int_0^{\tau}\frac{(\tau')^2}{2^{k}}\bigg\|[P_k,\nabla_4]\frac{[\nabla,\nabla_{\tau}]\big(\nabla^{M-1}\Phi_0\big)_Y}{\log 2^k\tau'}\bigg\|_{L^2}^2+\int_0^{\tau}\bigg\|\frac{\nabla\big(\nabla^{M-1}\Phi_0\big)_Y}{\log2^k\tau'}\bigg\|_{L^2}^2\]
    We use the estimates in Lemma \ref{LP Projections lemma} for the terms with $\mathcal{C}$ and Lemma \ref{fractional LP bounds} for the terms with $\big(\nabla^{M-1}\Phi_0\big)_Y$:
    \[\bigg\|P_k\nabla_{\tau}\frac{\big(\nabla^{M}\Phi_0\big)_Y}{\log2^k\tau}-R_k\nabla^M\mathcal{O}\partial_{\tau}\bigg(\frac{1}{\log2^k\tau}\bigg)\bigg\|^2_{L^2}+\bigg\|\nabla\bigg(P_k\frac{\big(\nabla^{M}\Phi_0\big)_Y}{\log2^k\tau}-\frac{R_k\nabla^M\mathcal{O}}{\log 2^k\tau}\bigg)\bigg\|^2_{L^2}\lesssim\]\[\lesssim\big\|\nabla P_k\nabla^M\mathcal{O}\big\|^2_{L^2}+\big\|\underline{P}_k\nabla^M\mathcal{O}\big\|^2_{H^1}+\int_0^{\tau}(\tau')^2\bigg\|\frac{\nabla\big(\nabla^{M-1}\Phi_0\big)_Y}{\log 2^k\tau'}\bigg\|_{H^{1/2}}^2+\int_0^{\tau}(\tau')^2\bigg\|\frac{\nabla_{\tau}\big(\nabla^{M-1}\Phi_0\big)_Y}{\log 2^k\tau'}\bigg\|_{H^{1/2}}^2\]\[+\int_0^{\tau}\bigg\|\frac{\big(\nabla^{M-1}\Phi_0\big)_Y}{\log2^k\tau'}\bigg\|_{H^1}^2+\int_0^{\tau}\frac{(\tau')^2}{2^{k}}\bigg\|\nabla\frac{\mathcal{C}}{\log 2^k\tau'}\bigg\|_{L^2}^2+\int_0^{\tau}\frac{(\tau')^2}{2^{k}}\bigg\|\frac{\nabla_{\tau}\mathcal{C}}{\log 2^k\tau'}\bigg\|_{L^2}^2+\int_0^{\tau}\bigg\|\frac{\mathcal{C}}{\log2^k\tau'}\bigg\|_{L^2}^2\]
    We use the lower order estimates to get:
    \[\bigg\|P_k\nabla_{\tau}\frac{\big(\nabla^{M}\Phi_0\big)_Y}{\log2^k\tau}-R_k\nabla^M\mathcal{O}\partial_{\tau}\bigg(\frac{1}{\log2^k\tau}\bigg)\bigg\|^2_{L^2}+\bigg\|\nabla\bigg(P_k\frac{\big(\nabla^{M}\Phi_0\big)_Y}{\log2^k\tau}-\frac{R_k\nabla^M\mathcal{O}}{\log 2^k\tau}\bigg)\bigg\|^2_{L^2}\lesssim\big\|\nabla P_k\nabla^M\mathcal{O}\big\|^2_{L^2}+\]\[+\big\|\underline{P}_k\nabla^M\mathcal{O}\big\|^2_{H^1}+2^{-k/2}\big\|\mathcal{O}\big\|^2_{H^{M+1/2+\eta}}+\int_0^{\tau}\frac{(\tau')^2}{2^{k}}\bigg\|\nabla\frac{\mathcal{C}}{\log 2^k\tau'}\bigg\|_{L^2}^2+\int_0^{\tau}\frac{(\tau')^2}{2^{k}}\bigg\|\frac{\nabla_{\tau}\mathcal{C}}{\log 2^k\tau'}\bigg\|_{L^2}^2+\int_0^{\tau}\bigg\|\frac{\mathcal{C}}{\log2^k\tau'}\bigg\|_{L^2}^2\]    
    Finally, we use the commutator estimates to conclude:
    \[\bigg\|P_k\nabla_{\tau}\frac{\big(\nabla^{M}\Phi_0\big)_Y}{\log2^k\tau}-R_k\nabla^M\mathcal{O}\partial_{\tau}\bigg(\frac{1}{\log2^k\tau}\bigg)\bigg\|^2_{L^2}+\bigg\|\nabla\bigg(P_k\frac{\big(\nabla^{M}\Phi_0\big)_Y}{\log2^k\tau}-\frac{R_k\nabla^M\mathcal{O}}{\log 2^k\tau}\bigg)\bigg\|^2_{L^2}\lesssim\]\[\lesssim\big\|\nabla P_k\nabla^M\mathcal{O}\big\|^2_{L^2}+\big\|\underline{P}_k\nabla^M\mathcal{O}\big\|^2_{H^1}+2^{-k/2}\big\|\mathcal{O}\big\|^2_{H^{M+1/2+\eta}}+2^{-2k}\int_0^{\tau}\big\|\tau'\nabla_{\tau}\big(\nabla^{M}\Phi_0\big)_Y\big\|_{L^2}^2.\]
\end{proof}

\begin{proposition}\label{singular component low frequency proposition 2} For any $k\geq0$ and $\tau\leq2^{-k-1}$, we have that $\big(\nabla^M\Phi_0\big)_Y$ satisfies the estimate:
    \[2^{2k}\bigg\|P_k\frac{\big(\nabla^{M}\Phi_0\big)_Y}{\log2^k\tau}-\frac{R_k\nabla^M\mathcal{O}}{\log 2^k\tau}\bigg\|^2_{L^2}\lesssim\big\|\nabla P_k\nabla^M\mathcal{O}\big\|^2_{L^2}+\big\|\underline{P}_k\nabla^M\mathcal{O}\big\|^2_{H^1}+\]\[+2^{-k/2}\big\|\mathcal{O}\big\|^2_{H^{M+1/2+\eta}}+2^{-2k}\int_0^{\tau}\big\|\tau'\nabla_{\tau}\big(\nabla^{M}\Phi_{0}\big)_Y\big\|_{L^2}^2.\]
\end{proposition}
\begin{proof} We have the bound:
\[\bigg\|P_k\frac{\big(\nabla^{M}\Phi_0\big)_Y}{\log2^k\tau}-\frac{R_k\nabla^M\mathcal{O}}{\log 2^k\tau}\bigg\|^2_{L^2}\lesssim\big\|P_k\nabla^M\mathcal{O}\big\|^2_{L^2}+2^{-k}\int_0^{\tau}\bigg\|\nabla_{\tau}\bigg(P_k\frac{\big(\nabla^{M}\Phi_0\big)_Y}{\log2^k\tau'}-\frac{R_k\nabla^M\mathcal{O}}{\log 2^k\tau'}\bigg)\bigg\|^2_{L^2}+\]\[+2^{k}\int_0^{\tau}\bigg\|P_k\frac{\big(\nabla^{M}\Phi_0\big)_Y}{\log2^k\tau'}-\frac{R_k\nabla^M\mathcal{O}}{\log 2^k\tau'}\bigg\|^2_{L^2}.\]
Using Gronwall and the previous proposition, we obtain that:
\[2^{2k}\bigg\|P_k\frac{\big(\nabla^{M}\Phi_0\big)_Y}{\log2^k\tau}-\frac{R_k\nabla^M\mathcal{O}}{\log 2^k\tau}\bigg\|^2_{L^2}\lesssim2^{2k}\big\|P_k\nabla^M\mathcal{O}\big\|^2_{L^2}+\big\|\nabla P_k\nabla^M\mathcal{O}\big\|^2_{L^2}+\big\|\underline{P}_k\nabla^M\mathcal{O}\big\|^2_{H^1}+\]\[+2^{-k/2}\big\|\mathcal{O}\big\|^2_{H^{M+1/2+\eta}}+2^{-2k}\int_0^{\tau}\big\|\tau'\nabla_{\tau}\big(\nabla^{M}\Phi_0\big)_Y\big\|_{L^2}^2+2^{-k}\int_0^{\tau}\bigg\|\frac{\mathcal{C}}{\log2^k\tau}\bigg\|_{L^2}^2+2^{-k}\int_0^{\tau}\bigg\|\frac{\big(\nabla^{M-1}\Phi_{0}\big)_Y}{\log2^k\tau}\bigg\|_{H^1}^2.\]
\end{proof}

We conclude this subsection by using the above results in order to obtain bounds at $\tau=2^{-k-1}.$ These will serve as estimates on the initial data in the high frequency regime $\tau\in[2^{-k-1},1].$

\begin{corollary}For any $k\geq0$ and $\tau=2^{-k-1}$ we have the estimate:
    \[\bigg(\big\|P_k\nabla_{\tau}\big(\nabla^M\Phi_{0}\big)_Y\big\|^2_{L^2}+2^{2k}\big\|P_k\big(\nabla^M\Phi_{0}\big)_Y\big\|^2_{L^2}+\big\|\nabla P_k\big(\nabla^M\Phi_{0}\big)_Y\big\|^2_{L^2}\bigg)\bigg|_{\tau=2^{-k-1}}\lesssim\]
\[\lesssim\big\|\nabla P_k\nabla^M\mathcal{O}\big\|^2_{L^2}+\big\|\underline{P}_k\nabla^M\mathcal{O}\big\|^2_{H^1}+2^{-k/2}\big\|\mathcal{O}\big\|^2_{H^{M+1}}+2^{-2k}\int_0^{2^{-k-1}}\big\|\tau'\nabla_{\tau}\big(\nabla^{M}\Phi_{0}\big)_Y\big\|_{L^2}^2.\]
\end{corollary}
\begin{proof}The low frequency estimates imply for $\tau=2^{-k-1}:$
\[\bigg(\big\|P_k\nabla_{\tau}(\nabla^{M}\Phi_{0}\big)_Y\big\|^2_{L^2}+2^{2k}\big\|P_k(\nabla^{M}\Phi_{0}\big)_Y\big\|^2_{L^2}+\big\|\nabla P_k(\nabla^{M}\Phi_{0}\big)_Y\big\|^2_{L^2}\bigg)\bigg|_{\tau=2^{-k-1}}\lesssim\big\|\nabla P_k\nabla^M\mathcal{O}\big\|^2_{L^2}+\]\[+\big\|\underline{P}_k\nabla^M\mathcal{O}\big\|^2_{H^1}+2^{-k/2}\big\|\mathcal{O}\big\|^2_{H^{M+1/2+\eta}}+2^{2k}\big\|R_k\nabla^M\mathcal{O}\big\|^2_{L^2}+\big\|\nabla R_k\nabla^M\mathcal{O}\big\|^2_{L^2}+2^{-2k}\int_0^{2^{-k-1}}\big\|\tau'\nabla_{\tau}(\nabla^{M}\Phi_{0}\big)_Y\big\|_{L^2}^2\]
We conclude using Lemma \ref{Rk lemma}.
\end{proof}

\subsubsection{High Frequency Estimates}\label{forward high freq estimates section}

In this section, we prove a high frequency regime estimate for the singular component $\big(\nabla^M\Phi_0\big)_Y$. As in the case of \cite{linearwave}, according to the Bessel function type asymptotics, we prove an energy estimate for $\sqrt{2^k\tau}\cdot P_k\big(\nabla^M\Phi_0\big)_Y,$ which implies the $1/2$ gain of regularity at $\tau=1$ compared to $\tau=2^{-k-1.}$ 

\begin{proposition}\label{high frequency forward estimate} For any $k\geq0,\ \tau\in\big[2^{-k-1},1\big]$ we have the estimate for the singular component:
    \[\tau\big\|P_k\nabla_{\tau}\big(\nabla^M\Phi_{0}\big)_Y\big\|^2_{L^2}+\frac{1}{\tau}\big\|P_k\big(\nabla^M\Phi_{0}\big)_Y\big\|^2_{L^2}+\tau\big\|\nabla P_k\big(\nabla^M\Phi_{0}\big)_Y\big\|^2_{L^2}\lesssim\]
    \[\lesssim\frac{1}{2^k}\bigg(\big\|P_k\nabla_{\tau}\big(\nabla^M\Phi_{0}\big)_Y\big\|^2_{L^2}+2^{2k}\big\|P_k\big(\nabla^M\Phi_{0}\big)_Y\big\|^2_{L^2}+\big\|\nabla P_k\big(\nabla^M\Phi_{0}\big)_Y\big\|^2_{L^2}\bigg)\bigg|_{\tau=2^{-k-1}}\]\[+\int_{2^{-k-1}}^{\tau}\tau'\big\|\underline{\widetilde{P}}_k\big(\nabla^M\Phi_{0}\big)_Y\big\|_{L^2}^2+\int_{2^{-k-1}}^{\tau}(\tau')^3\big\|\underline{\widetilde{P}}_k\nabla\big(\nabla^M\Phi_{0}\big)_Y\big\|_{L^2}^2+\int_{2^{-k-1}}^{\tau}(\tau')^3\big\|\underline{\widetilde{P}}_k\nabla_{\tau}\big(\nabla^M\Phi_{0}\big)_Y\big\|_{L^2}^2\]\[+\int_{2^{-k-1}}^{\tau}\frac{\tau'}{2^{2k}}\big\|\big(\nabla^M\Phi_{0}\big)_Y\big\|_{H^1}^2+\int_{2^{-k-1}}^{\tau}\frac{(\tau')^3}{2^{2k}}\big\|\nabla_{\tau}\big(\nabla^M\Phi_{0}\big)_Y\big\|_{L^2}^2.\]
\end{proposition}
\begin{proof}We denote $\xi=\big(\nabla^M\Phi_{0}\big)_Y$ and we introduce the new time variable $t=2^k\tau$:
    \[\nabla_{t}\big(\nabla_{t}\xi\big)+\frac{1}{t}\nabla_{t}\xi-\frac{4}{2^{2k}}\Delta\xi=\frac{1}{2^{2k}}\psi\nabla\xi.\]
    We multiply by $\sqrt{t}$ to get:
    \[\nabla_{t}\big(\nabla_{t}(\xi\sqrt{t})\big)+\frac{1}{4t^2}\xi\sqrt{t}-\frac{4}{2^{2k}}\Delta\xi\sqrt{t}=\frac{1}{2^{2k}}\psi\nabla\xi\sqrt{t}.\]
    For any $k\geq0$, we apply $P_k$ to obtain the equation:
    \[\nabla_{t}\big(P_k\nabla_{t}(\xi\sqrt{t})\big)+\frac{1}{4t^2}P_k\xi\sqrt{t}-\frac{4}{2^{2k}}\Delta P_k\xi\sqrt{t}=\frac{P_k\big(\psi\nabla\xi\sqrt{t}\big)}{2^{2k}}+[\nabla_t,P_k]\nabla_{t}\xi\sqrt{t}.\]
    We contract each equation with $P_k\nabla_t(\xi\sqrt{t})$ and integrate by parts to obtain the energy estimate:
    \[\big\|P_k\nabla_{t}\xi\sqrt{t}\big\|^2_{L^2}+\frac{1}{t^2}\big\|P_k\xi\sqrt{t}\big\|^2_{L^2}+\frac{1}{2^{2k}}\big\|\nabla P_k\xi\sqrt{t}\big\|^2_{L^2}+\int_{1/2}^t\frac{1}{(t')^2}\big\|P_k\xi\big\|_{L^2}^2\lesssim\]\[\lesssim\bigg(\big\|P_k\nabla_{t}\xi\big\|^2_{L^2}+\big\|P_k\xi\big\|^2_{L^2}+\frac{1}{2^{2k}}\big\|\nabla P_k\xi\big\|^2_{L^2}\bigg)\bigg|_{t=1/2}+\int_{1/2}^t\int_{S^n}\frac{1}{(t')^2}\big|P_k\xi\sqrt{t}\big|\cdot\big|[P_k,\nabla_t]\xi\sqrt{t'}\big|\]\[+\int_{1/2}^t\int_{S^n}\frac{1}{2^{2k}}\big|\nabla P_k\xi\sqrt{t'}\big|\cdot\big|\nabla[P_k,\nabla_t]\xi\sqrt{t'}\big|+\int_{1/2}^t\int_{S^n}\frac{1}{2^{2k}}\big|\nabla P_k\xi\sqrt{t'}\big|\cdot\big|[\nabla,\nabla_t]P_k\xi\sqrt{t'}\big|\]\[+\int_{1/2}^t\int_{S^n}\frac{1}{2^{2k}}\big|P_k\nabla_t\xi\sqrt{t'}\big|\cdot\big|\nabla P_k\xi\sqrt{t'}\big|+\int_{1/2}^t\int_{S^n}\frac{1}{2^{2k}}\big|P_k\nabla_t\xi\sqrt{t'}\big|\cdot\big|[P_k,\nabla]\xi\sqrt{t'}\big|\]\[+\int_{1/2}^t\int_{S^n}\frac{1}{2^{2k}}\big|P_k\nabla_t\xi\sqrt{t'}\big|\cdot\big|[P_k,\psi]\nabla\xi\sqrt{t'}\big|+\int_{1/2}^t\int_{S^n}\big|P_k\nabla_t\xi\sqrt{t'}\big|\cdot\big|[\nabla_t,P_k]\nabla_{t}\xi\sqrt{t'}\big|\]
    We point out that the good bulk term simplifies our analysis significantly, unlike the case of the second model system studied in Section~\ref{model backward direction section}. We use the bounds in Lemma \ref{LP Projections lemma refined} and Gronwall for $t\in\big[1/2,2^k\big]$ to get:
    \[\big\|P_k\nabla_{t}\xi\sqrt{t}\big\|^2_{L^2}+\frac{1}{t^2}\big\|P_k\xi\sqrt{t}\big\|^2_{L^2}+\frac{1}{2^{2k}}\big\|\nabla P_k\xi\sqrt{t}\big\|^2_{L^2}\lesssim\bigg(\big\|P_k\nabla_{t}\xi\big\|^2_{L^2}+\big\|P_k\xi\big\|^2_{L^2}+\frac{1}{2^{2k}}\big\|\nabla P_k\xi\big\|^2_{L^2}\bigg)\bigg|_{t=1/2}\]\[+\int_{1/2}^t\frac{1}{2^{2k}t'}\big\|P_k\xi\sqrt{t'}\big\|_{L^2}\cdot\Big(\big\|\underline{\widetilde{P}}_k\xi\sqrt{t'}\big\|_{L^2}+2^{-k}\big\|\xi\sqrt{t'}\big\|_{L^2}\Big)+\int_{1/2}^t\frac{1}{2^{2k}}\big\|P_k\nabla_t\xi\sqrt{t'}\big\|_{L^2}\cdot\big\|\nabla P_k\xi\sqrt{t'}\big\|_{L^2}+\]\[+\int_{1/2}^t\frac{t'}{2^{4k}}\big\|\nabla P_k\xi\sqrt{t'}\big\|_{L^2}\cdot\bigg(\big\|\underline{\widetilde{P}}_k\nabla\xi\sqrt{t'}\big\|_{L^2}+2^{-k}\big\|\xi\sqrt{t'}\big\|_{H^1}\bigg)+\int_{1/2}^t\frac{1}{2^{3k}}\big\|P_k\nabla_t\xi\sqrt{t'}\big\|_{L^2}\cdot\big\|\xi\sqrt{t'}\big\|_{L^2}+\]\[+\int_{1/2}^t\frac{1}{2^{3k}}\big\|P_k\nabla_t\xi\sqrt{t'}\big\|_{L^2}\cdot\big\|\nabla\xi\sqrt{t'}\big\|_{L^2}+\int_{1/2}^t\frac{t'}{2^{2k}}\big\|P_k\nabla_t\xi\sqrt{t'}\big\|_{L^2}\Big(\big\|\underline{\widetilde{P}}_k\nabla_t\xi\sqrt{t'}\big\|_{L^2}+2^{-k}\big\|\nabla_t\xi\sqrt{t'}\big\|_{L^2}\Big)\]
    We use Gronwall again for $t\in\big[1/2,2^k\big]$ to get the estimate:
    \[t\big\|P_k\nabla_{t}\xi\big\|^2_{L^2}+\frac{1}{t}\big\|P_k\xi\big\|^2_{L^2}+\frac{t}{2^{2k}}\big\|\nabla P_k\xi\big\|^2_{L^2}\lesssim\bigg(\big\|P_k\nabla_{t}\xi\big\|^2_{L^2}+\big\|P_k\xi\big\|^2_{L^2}+\frac{1}{2^{2k}}\big\|\nabla P_k\xi\big\|^2_{L^2}\bigg)\bigg|_{t=\frac{1}{2}}+\int_{1/2}^t\frac{1}{2^{3k}}\big\|\underline{\widetilde{P}}_k\xi\sqrt{t'}\big\|_{L^2}^2\]\[+\int_{1/2}^t\frac{(t')^2}{2^{5k}}\big\|\underline{\widetilde{P}}_k\nabla\xi\sqrt{t'}\big\|_{L^2}^2+\int_{1/2}^t\frac{(t')^2}{2^{3k}}\big\|\underline{\widetilde{P}}_k\nabla_t\xi\sqrt{t'}\big\|_{L^2}^2+\int_{1/2}^t\frac{1}{2^{5k}}\big\|\xi\sqrt{t'}\big\|_{H^1}^2+\int_{1/2}^t\frac{(t')^2}{2^{5k}}\big\|\nabla_t\xi\sqrt{t'}\big\|_{L^2}^2.\]
    We change variables to $\tau$ and we obtain the conclusion.
\end{proof}

\subsubsection{Main Result for the Singular Component}\label{forward main result estimates section}
In this section we combine the low frequency regime and the high frequency regime estimates for the singular component $\big(\nabla^M\Phi_0\big)_Y$ in order to prove Theorem \ref{forward direction main result theorem}. As remarked in Section~\ref{scattering map section intro} of the introduction, due to presence of different projection operators in the estimates established above, we must sum the estimates obtained for each LP projection before being able to bound the error terms.

\textit{Proof of Theorem \ref{forward direction main result theorem}.} \textit{Step 1. The improved high frequency regime estimate.} The result in Proposition~\ref{high frequency forward estimate} implies that for $\tau\in\big[2^{-k-1},1\big]$:
\[2^k\tau\big\|P_k\nabla_{\tau}\big(\nabla^M\Phi_{0}\big)_Y\big\|^2_{L^2}+\frac{2^k}{\tau}\big\|P_k\big(\nabla^M\Phi_{0}\big)_Y\big\|^2_{L^2}+2^k\tau\big\|\nabla P_k\big(\nabla^M\Phi_{0}\big)_Y\big\|^2_{L^2}\lesssim\]
\[\lesssim\bigg(\big\|P_k\nabla_{\tau}\big(\nabla^M\Phi_{0}\big)_Y\big\|^2_{L^2}+2^{2k}\big\|P_k\big(\nabla^M\Phi_{0}\big)_Y\big\|^2_{L^2}+\big\|\nabla P_k\big(\nabla^M\Phi_{0}\big)_Y\big\|^2_{L^2}\bigg)\bigg|_{\tau=2^{-k-1}}\]\[+2^k\int_{2^{-k-1}}^{\tau}\tau'\big\|\underline{\widetilde{P}}_k\big(\nabla^M\Phi_{0}\big)_Y\big\|_{L^2}^2+2^k\int_{2^{-k-1}}^{\tau}(\tau')^3\big\|\underline{\widetilde{P}}_k\nabla\big(\nabla^M\Phi_{0}\big)_Y\big\|_{L^2}^2+2^k\int_{2^{-k-1}}^{\tau}(\tau')^3\big\|\underline{\widetilde{P}}_k\nabla_{\tau}\big(\nabla^M\Phi_{0}\big)_Y\big\|_{L^2}^2\]\[+\int_{2^{-k-1}}^{\tau}\frac{\tau'}{2^{k}}\big\|\big(\nabla^M\Phi_{0}\big)_Y\big\|_{H^1}^2+\int_{2^{-k-1}}^{\tau}\frac{(\tau')^3}{2^{k}}\big\|\nabla_{\tau}\big(\nabla^M\Phi_{0}\big)_Y\big\|_{L^2}^2.\]
Using the low frequency regime estimates in Propositions~\ref{singular component low frequency proposition} and \ref{singular component low frequency proposition 2}, we get the improved high frequency regime estimate for $\tau\in\big[2^{-k-1},1\big]$:
\[2^k\tau\big\|P_k\nabla_{\tau}\big(\nabla^M\Phi_{0}\big)_Y\big\|^2_{L^2}+\frac{2^k}{\tau}\big\|P_k\big(\nabla^M\Phi_{0}\big)_Y\big\|^2_{L^2}+2^k\tau\big\|\nabla P_k\big(\nabla^M\Phi_{0}\big)_Y\big\|^2_{L^2}\lesssim\]\[\lesssim\big\|\nabla P_k\nabla^M\mathcal{O}\big\|^2_{L^2}+\big\|\underline{P}_k\nabla^M\mathcal{O}\big\|^2_{H^1}+2^{-k/2}\big\|\mathcal{O}\big\|^2_{H^{M+1}}+2^{-2k}\int_0^{2^{-k-1}}\big\|\tau'\nabla_{\tau}\big(\nabla^{M}\Phi_{0}\big)_Y\big\|_{L^2}^2+\]\[+2^k\int_{2^{-k-1}}^{\tau}\tau'\big\|\underline{\widetilde{P}}_k\big(\nabla^M\Phi_{0}\big)_Y\big\|_{L^2}^2+2^k\int_{2^{-k-1}}^{\tau}(\tau')^3\big\|\underline{\widetilde{P}}_k\nabla\big(\nabla^M\Phi_{0}\big)_Y\big\|_{L^2}^2+2^k\int_{2^{-k-1}}^{\tau}(\tau')^3\big\|\underline{\widetilde{P}}_k\nabla_{\tau}\big(\nabla^M\Phi_{0}\big)_Y\big\|_{L^2}^2\]\[+\int_{2^{-k-1}}^{\tau}\frac{\tau'}{2^{k}}\big\|\big(\nabla^M\Phi_{0}\big)_Y\big\|_{H^1}^2+\int_{2^{-k-1}}^{\tau}\frac{(\tau')^3}{2^{k}}\big\|\nabla_{\tau}\big(\nabla^M\Phi_{0}\big)_Y\big\|_{L^2}^2.\]

\textit{Step 2. Bounding the non-negative frequencies.} We define the following energy for all $k\geq0$:
\[2^kE_k^2(\tau)=2^k\tau^2\big\|P_k\nabla_{\tau}\big(\nabla^{M}\Phi_{0}\big)_Y\big\|^2_{L^2}+2^k\tau\big\|P_k\big(\nabla^{M}\Phi_{0}\big)_Y\big\|^2_{L^2}+2^k\tau^2\big\|\nabla P_k\big(\nabla^{M}\Phi_{0}\big)_Y\big\|^2_{L^2}+\tau\big\|\nabla P_k\big(\nabla^{M}\Phi_{0}\big)_Y\big\|^2_{L^2}\]
For any $\tau\in(0,1]$ we can write:
\begin{equation}\label{split E in low and high}
    \sum_{k\geq0}2^kE_k^2(\tau)=\sum_{\tau\leq2^{-k-1}}2^kE_k^2(\tau)+\sum_{\tau>2^{-k-1}}2^kE_k^2(\tau).
\end{equation}
We use our previous estimates to bound \eqref{split E in low and high}. The high frequency regime estimate implies the bound:
\[\sum_{\tau>2^{-k-1}}2^kE_k^2\lesssim\big\|\mathcal{O}\big\|^2_{H^{M+1}}+\int_0^{\tau}(\tau')^2\big\|\nabla_{\tau}\big(\nabla^{M}\Phi_{0}\big)_Y\big\|_{H^{1/2}}^2+\int_{0}^{\tau}(\tau')^2\big\|\nabla\big(\nabla^{M}\Phi_{0}\big)_Y\big\|_{H^{1/2}}^2+\]\[+\int_{0}^{\tau}\tau'\big\|\nabla\big(\nabla^{M}\Phi_{0}\big)_Y\big\|_{L^2}^2+\int_{0}^{\tau}\tau'\big\|\big(\nabla^{M}\Phi_{0}\big)_Y\big\|_{H^{1/2}}^2.\]
The low frequency regime estimates in Propositions~\ref{singular component low frequency proposition} and \ref{singular component low frequency proposition 2} imply:
\[\sum_{\tau\leq2^{-k-1}}2^k\tau^2\big\|\nabla P_k\big(\nabla^{M}\Phi_{0}\big)_Y\big\|^2_{L^2}\lesssim\sum_{\tau\leq2^{-k-1}}\tau\big\|\nabla P_k\big(\nabla^{M}\Phi_{0}\big)_Y\big\|^2_{L^2}\lesssim\]\[\lesssim\sum_{\tau\leq2^{-k-1}}\tau\big\|\nabla R_k\nabla^M\mathcal{O}\big\|^2_{L^2}+\sum_{\tau\leq2^{-k-1}}\bigg\|\nabla\bigg(P_k\frac{\big(\nabla^{M}\Phi_{0}\big)_Y}{\log2^k\tau}-\frac{R_k\nabla^M\mathcal{O}}{\log 2^k\tau}\bigg)\bigg\|^2_{L^2}\lesssim\big\|\mathcal{O}\big\|^2_{H^{M+1}}+\int_0^{\tau}\big\|\tau'\nabla_{\tau}\big(\nabla^{M}\Phi_{0}\big)_Y\big\|_{L^2}^2\]
Similarly, we also have the bound:
\[\sum_{\tau\leq2^{-k-1}}2^k\tau^2\big\|P_k\nabla_{\tau}\big(\nabla^{M}\Phi_{0}\big)_Y\big\|^2_{L^2}\lesssim\sum_{\tau\leq2^{-k-1}}2^k\tau^2\big|\log(2^k\tau)\big|^2\cdot\bigg\|P_k\nabla_{\tau}\frac{\big(\nabla^{M}\Phi_{0}\big)_Y}{\log2^k\tau}\bigg\|^2_{L^2}+\sum_{\tau\leq2^{-k-1}}2^k\bigg\|\frac{P_k\big(\nabla^{M}\Phi_{0}\big)_Y}{\log2^k\tau}\bigg\|^2_{L^2}\]
\[\lesssim\sum_{\tau\leq2^{-k-1}}2^k\tau^2\big|\log(2^k\tau)\big|^2\cdot\bigg\|P_k\nabla_{\tau}\frac{\big(\nabla^{M}\Phi_{0}\big)_Y}{\log2^k\tau}-R_k\nabla^M\mathcal{O}\partial_{\tau}\bigg(\frac{1}{\log2^k\tau}\bigg)\bigg\|^2_{L^2}+\sum_{\tau\leq2^{-k-1}}2^k\big\|R_k\nabla^M\mathcal{O}\big\|^2_{L^2}+\]\[+\sum_{\tau\leq2^{-k-1}}2^k\bigg\|\frac{P_k\big(\nabla^{M}\Phi_{0}\big)_Y}{\log2^k\tau}\bigg\|^2_{L^2}\lesssim\big\|\mathcal{O}\big\|^2_{H^{M+1}}+\int_0^{\tau}\big\|\tau'\nabla_{\tau}\big(\nabla^{M}\Phi_{0}\big)_Y\big\|_{L^2}^2.\]
Next, we have the bound:
\[\sum_{\tau\leq2^{-k-1}}2^k\tau\big\|P_k\big(\nabla^{M}\Phi_{0}\big)_Y\big\|^2_{L^2}\lesssim\sum_{\tau\leq2^{-k-1}}2^{2k}\bigg\|P_k\frac{\big(\nabla^{M}\Phi_{0}\big)_Y}{\log2^k\tau}-\frac{R_k\nabla^M\mathcal{O}}{\log 2^k\tau}\bigg\|^2_{L^2}+\sum_{\tau\leq2^{-k-1}}2^k\big\|R_k\nabla^M\mathcal{O}\big\|^2_{L^2}\]
\[\lesssim\big\|\mathcal{O}\big\|^2_{H^{M+1}}+\int_0^{\tau}\big\|\tau'\nabla_{\tau}\big(\nabla^{M}\Phi_{0}\big)_Y\big\|_{L^2}^2.\]
As a result, we obtain the following bound for the sum in \eqref{split E in low and high}:
\[\sum_{k\geq0}2^kE_k^2\lesssim\big\|\mathcal{O}\big\|^2_{H^{M+1}}+\int_0^{\tau}(\tau')^2\big\|\nabla_{\tau}\big(\nabla^{M}\Phi_{0}\big)_Y\big\|_{H^{1/2}}^2+\]\[+\int_{0}^{\tau}(\tau')^2\big\|\nabla\big(\nabla^{M}\Phi_{0}\big)_Y\big\|_{H^{1/2}}^2+\int_{0}^{\tau}\tau'\big\|\nabla\big(\nabla^{M}\Phi_{0}\big)_Y\big\|_{L^2}^2+\int_{0}^{\tau}\tau'\big\|\big(\nabla^{M}\Phi_{0}\big)_Y\big\|_{H^{1/2}}^2.\]

\textit{Step 3. Bounding the negative frequencies.} In order to prove (\ref{main estimate alpha Y top order}), we also need to deal with the negative frequencies. According to \cite{geometricLP}, for any $k<0$ we also have $\|P_k\nabla F\|_{L^2}\lesssim2^k\|F\|_{L^2}.$ Thus, we have:
\[\sum_{k<0}2^k\tau^2\big\|P_k\nabla \big(\nabla^{M}\Phi_{0}\big)_Y\big\|^2_{L^2}+\tau\big\|P_k\nabla \big(\nabla^{M}\Phi_{0}\big)_Y\big\|^2_{L^2}\lesssim\sum_{k<0}2^k\tau\big\| \big(\nabla^{M}\Phi_{0}\big)_Y\big\|^2_{L^2}\lesssim\tau\big\|\mathcal{C}\big\|^2_{L^2}+\tau\big\| \nabla\big(\nabla^{M-1}\Phi_{0}\big)_Y\big\|^2_{L^2},\]
\[\sum_{k<0}2^k\tau\big\|P_k\big(\nabla^{M}\Phi_{0}\big)_Y\big\|^2_{L^2}\lesssim\tau\big\|\mathcal{C}\big\|^2_{L^2}+\tau\big\|\nabla\big(\nabla^{M-1}\Phi_{0}\big)_Y\big\|^2_{L^2},\]
\[\sum_{k<0}2^k\tau^2\big\|P_k\nabla_{\tau}\big(\nabla^{M}\Phi_{0}\big)_Y\big\|^2_{L^2}\lesssim\tau^2\big\|\nabla_{\tau}\mathcal{C}\big\|^2_{L^2}+\sum_{k<0}2^k\tau^2\big\|[\nabla,\nabla_{\tau}]\big(\nabla^{M-1}\Phi_{0}\big)_Y\big\|^2_{L^2}+\sum_{k<0}2^{3k}\tau^2\big\|\nabla_{\tau}\big(\nabla^{M-1}\Phi_{0}\big)_Y\big\|^2_{L^2}.\]
We proved that for the negative frequencies we have:
\[\sum_{k<0}2^k\tau^2\big\|P_k\nabla \big(\nabla^{M}\Phi_{0}\big)_Y\big\|^2_{L^2}+\tau\big\|P_k\nabla \big(\nabla^{M}\Phi_{0}\big)_Y\big\|^2_{L^2}+2^k\tau\big\|P_k\big(\nabla^{M}\Phi_{0}\big)_Y\big\|^2_{L^2}+2^k\tau^2\big\|P_k\nabla_{\tau}\big(\nabla^{M}\Phi_{0}\big)_Y\big\|^2_{L^2}\lesssim\]
\[\lesssim\tau^2\big\|\nabla_{\tau}\mathcal{C}\big\|^2_{L^2}+\tau\big\|\mathcal{C}\big\|^2_{L^2}+\tau\big\|\big(\nabla^{M-1}\Phi_{0}\big)_Y\big\|^2_{H^1}+\tau^2\big\|\nabla_{\tau}\big(\nabla^{M-1}\Phi_{0}\big)_Y\big\|^2_{L^2}\lesssim\big\|\mathcal{O}\big\|^2_{H^{M+1}}+\int_0^{\tau}\big\|\tau'\nabla_{\tau}\big(\nabla^{M}\Phi_{0}\big)_Y\big\|_{L^2}^2.\]

\textit{Step 4. The proof of (\ref{main estimate alpha Y top order}) and (\ref{practical estimate alpha Y top order}).} We obtain that:
\[\tau^2\big\|\nabla_{\tau}\big(\nabla^{M}\Phi_{0}\big)_Y\big\|^2_{H^{1/2}}+\tau\big\|\big(\nabla^{M}\Phi_{0}\big)_Y\big\|^2_{H^{1/2}}+\tau^2\big\|\nabla\big(\nabla^{M}\Phi_{0}\big)_Y\big\|^2_{H^{1/2}}+\tau\big\|\nabla\big(\nabla^{M}\Phi_{0}\big)_Y\big\|^2_{L^2}\lesssim\]\[\lesssim\tau\big\|\big(\nabla^{M}\Phi_{0}\big)_Y\big\|^2_{L^2}+\sum_{k\geq0}2^kE_k^2+\big\|\mathcal{O}\big\|^2_{H^{M+1}}+\int_0^{\tau}\big\|\tau'\nabla_{\tau}\big(\nabla^{M}\Phi_{0}\big)_Y\big\|_{L^2}^2\lesssim\big\|\mathcal{O}\big\|^2_{H^{M+1}}+\]\[+\int_0^{\tau}(\tau')^2\big\|\nabla_{\tau}\big(\nabla^{M}\Phi_{0}\big)_Y\big\|_{H^{1/2}}^2+(\tau')^2\big\|\nabla\big(\nabla^{M}\Phi_{0}\big)_Y\big\|_{H^{1/2}}^2+\tau'\big\|\nabla\big(\nabla^{M}\Phi_{0}\big)_Y\big\|_{L^2}^2+\tau'\big\|\big(\nabla^{M}\Phi_{0}\big)_Y\big\|_{H^{1/2}}^2.\]
By Gronwall, we obtain (\ref{main estimate alpha Y top order}). Next, we notice that for any $k\geq0$ using the high frequency estimate we get:
\[\sum_{\tau>2^{-k-1}}\big\|P_k\nabla\big(\nabla^M\Phi_{0}\big)_Y\big\|^2_{L^2}\lesssim\big\|\mathcal{C}\big\|^2_{L^2}+\big\|\nabla\big(\nabla^{M-1}\Phi_{0}\big)_Y\big\|^2_{L^2}+\sum_{\tau>2^{-k-1}}2^k\tau\big\|\nabla P_k\big(\nabla^M\Phi_{0}\big)_Y\big\|^2_{L^2}\lesssim\]\[\lesssim\big(1+|\log\tau|^2\big)\big\|\mathcal{O}\big\|^2_{H^{M+1}}+\int_0^{\tau}(\tau')^2\big\|\nabla_{\tau}\big(\nabla^{M}\Phi_{0}\big)_Y\big\|_{H^{1/2}}^2+\int_{0}^{\tau}(\tau')^2\big\|\nabla\big(\nabla^{M}\Phi_{0}\big)_Y\big\|_{H^{1/2}}^2+\]\[+\int_{0}^{\tau}\tau'\big\|\nabla\big(\nabla^{M}\Phi_{0}\big)_Y\big\|_{L^2}^2+\int_{0}^{\tau}\tau'\big\|\big(\nabla^{M}\Phi_{0}\big)_Y\big\|_{H^{1/2}}^2\lesssim\big(1+|\log\tau|^2\big)\big\|\mathcal{O}\big\|^2_{H^{M+1}}.\]
As before, we also have for any $k\geq0$:
\[\sum_{\tau\leq2^{-k-1}}\big\|\nabla P_k\big(\nabla^M\Phi_{0}\big)_Y\big\|^2_{L^2}\lesssim\sum_{\tau\leq2^{-k-1}}\big\|\nabla R_k\nabla^M\mathcal{O}\big\|^2_{L^2}+\sum_{\tau\leq2^{-k-1}}|\log\tau|^2\bigg\|\nabla\bigg(P_k\frac{\big(\nabla^M\Phi_{0}\big)_Y}{\log2^k\tau}-\frac{R_k\nabla^M\mathcal{O}}{\log 2^k\tau}\bigg)\bigg\|^2_{L^2}\lesssim\]\[\lesssim\big(1+|\log\tau|^2\big)\big\|\mathcal{O}\big\|^2_{H^{M+1}}.\]
We get the same bound for $\sum_{\tau\leq2^{-k-1}}\big\|P_k\nabla\big(\nabla^M\Phi_{0}\big)_Y\big\|^2_{L^2}$ by commutation. For the negative frequencies we have:
\[\sum_{k<0}\big\|P_k\nabla\big(\nabla^M\Phi_{0}\big)_Y\big\|^2_{L^2}\lesssim\big\| \big(\nabla^M\Phi_{0}\big)_Y\big\|^2_{L^2}\lesssim\big\|\mathcal{C}\big\|^2_{L^2}+\big\| \nabla\big(\nabla^{M-1}\Phi_{0}\big)_Y\big\|^2_{L^2}\lesssim\big(1+|\log\tau|^2\big)\big\|\mathcal{O}\big\|^2_{H^{M+1}}.\]
This completes the proof of (\ref{practical estimate alpha Y top order}).
\qed

\section{Estimates from $\{v=0\}$ to $\{v=-u\}$}\label{forward direction full system section}
In this section we prove optimal estimates on the smooth spacetime $(\mathcal{M},g)$ obtained in Theorem \ref{stability of de sitter theorem in section} in terms of the asymptotic data at $\{u=-1,\ v=0\}.$ Using the ambient metric construction, in the original $(n+1)$-dimensional formulation these correspond to proving estimates at finite times in terms of the asymptotic data at $\mathcal{I}^-$.

We first introduce the notion of asymptotic data set at $\{u=-1,\ v=0\}$:
\begin{definition}\label{asymptotic data set definition}
    Let $\big(\slashed{g}_0,h\big)$ be smooth straight initial data at $\{u=-1,\ v=0\}$. We define the corresponding asymptotic data set at $\{u=-1,\ v=0\}$ by:
    \[\Sigma\big(\slashed{g}_0,h\big):=\Bigg\{\slashed{g}_0,\bigg\{\nabla_4^l\psi:\ 0\leq l\leq\frac{n-4}{2}\bigg\},\bigg\{\nabla_4^l\Psi^G:\ 0\leq l\leq\frac{n-4}{2}\bigg\},\bigg\{\nabla_4^l\alpha:\ 0\leq l\leq\frac{n-6}{2}\bigg\}, \mathcal{O}, \mathfrak{h}\Bigg\},\]
    where we have that for all admissible $l$ the tensors $\nabla_4^l\psi,\nabla_4^l\Psi^G,\nabla_4^l\alpha,$ and $\mathcal{O}$ are defined in terms of $\slashed{g}_0$ by the compatibility relations in the Fefferman-Graham expansion as in \cite{selfsimilarvacuum}, and we define $\mathfrak{h}=h-2(\log\nabla)\mathcal{O}$. We also denote by $\Sigma_{\mathrm{Minkowski}}=\Sigma\big(\slashed{g}_{S^n},0\big)$ the Minkowski data set.

    For $M>0$ large enough, we define the asymptotic data norm of order $M$, measuring closeness to the Minkowski data:
    \[\Big\|\Sigma\big(\slashed{g}_0,h\big)\Big\|_M^2=\big\|\slashed{g}_0^*\big\|_{\mathring{H}^{M+1}}^2+\sum_{k=0}^1\big\|\nabla^{M+1+k}\nabla_4^{\frac{n-4}{2}-k}\psi^*\big\|^2_{H^{1/2}}+\sum_{l=0}^{\frac{n-4}{2}}\sum_{m=0}^{M+\frac{n-4}{2}-l}\big\|\nabla_4^{l}\psi^*\big\|^2_{H^{m+1}}+\]\[+\sum_{l=0}^{\frac{n-4}{2}}\sum_{m=0}^{M+\frac{n-4}{2}-l}\big\|\nabla_4^l\Psi^G\big\|_{H^{m+1}}^2+\sum_{l=0}^{\frac{n-6}{2}}\sum_{m=0}^{M+\frac{n-4}{2}-l}\big\|\nabla_4^l\alpha\big\|_{H^{m+1}}^2+\big\|\mathcal{O}\big\|_{H^{M+1}}^2+\big\|\mathfrak{h}\big\|_{H^{M+1}}^2,\]
    where $\mathring{H}^{M+1}$ is the Sobolev space with respect to $\slashed{g}_{S^n},$ and the other Sobolev spaces are defined with respect to $\slashed{g}.$
    
    For every $\epsilon>0$ small enough, we denote the set of $\epsilon-$small asymptotic data by:
    \[B_{\epsilon}^M\big(\Sigma_{\mathrm{Minkowski}}\big)=\Big\{\Sigma\big(\slashed{g}_0,h\big):\ \big(\slashed{g}_0,h\big) \text{ smooth straight initial data},\ \big\|\Sigma\big\|_M<\epsilon\Big\}.\]
\end{definition}
\begin{remark}\label{remark about def of sigma}
    For $\big(\underline{\slashed{g}_0},\underline{h}\big)$ asymptotic data at $\{u=0,v=1\},$ the precise notion of asymptotic data norm is obtained via the transformation $(u,v)\rightarrow(-v,-u)$ in the above definition. In particular, we notice that we replace $\Psi^G,\alpha,\mathcal{O},\mathfrak{h}$ by $\underline{\Psi}^G,\underline{\alpha},\underline{\mathcal{O}},\underline{\mathfrak{h}},$ and all $\nabla_4$ derivatives by $\nabla_3$ derivatives. However, in view of the compatibility relations, the quantities in $\Sigma\big(\underline{\slashed{g}_0},\underline{h}\big)$ are expressed in terms of $\big(\underline{\slashed{g}_0},\underline{h}\big)$ using the same formulas as the ones satisfied by $\Sigma\big(\slashed{g}_0,h\big)$ in terms of $\big(\slashed{g}_0,h\big)$.
\end{remark}

The main result of this section is the following estimate:
\begin{theorem}\label{main theorem forward direction full system}
    For any $M>0$ large enough and $\epsilon>0$ small enough we consider the smooth straight initial data $\big(\slashed{g}_0,h\big)$ such that $\Sigma\big(\slashed{g}_0,h\big)\in B_{\epsilon}^M\big(\Sigma_{\mathrm{Minkowski}}\big).$ The smooth spacetime $(\mathcal{M},g)$ obtained in Theorem \ref{stability of de sitter theorem in section} with asymptotic initial data given by $\big(\slashed{g}_0,h\big)$ satisfies the following estimate on $S_{(-1,1)}:$
    \begin{align*}
        \Xi_M^2:=&\sum_{i+j=0}^{\frac{n-4}{2}}\sum_{m=0}^{M+\frac{n-4}{2}-i-j}\big\|\nabla_3^i\nabla_4^j\Psi\big\|^2_{H^{m+1}}+\sum_{i+j=0}^{\frac{n-4}{2}}\big\|\nabla^M\nabla_3^i\nabla_4^j\Psi\big\|^2_{H^{3/2}}+\sum_{i+j=0}^{\frac{n-2}{2}}\big\|\nabla^M\nabla_3^i\nabla_4^j\Psi\big\|^2_{H^{1/2}}+\big\|\slashed{g}^*\big\|_{H^{M+1}}^2\\
        &+\sum_{k=0}^1\sum_{i+j=\frac{n-4}{2}-k}\big\|\nabla^{M+1+k}\nabla_3^i\nabla_4^j\psi^*\big\|^2_{H^{1/2}}+\sum_{i+j=0}^{\frac{n-4}{2}}\sum_{m=0}^{M+\frac{n-4}{2}-i-j}\big\|\nabla_3^i\nabla_4^j\psi^*\big\|^2_{H^{m+1}}\lesssim\Big\|\Sigma\big(\slashed{g}_0,h\big)\Big\|_M^2.
    \end{align*}
\end{theorem}

By self-similarity, this result follows from the corresponding estimates along $\{u=-1\},$ where we replace all $\nabla_3$ derivatives by $\nabla_4.$ We define the following norms on $\big\{u=-1,\ 0\leq v\leq1\big\}:$
\begin{itemize}
    \item Top order energy $\mathcal{T}=\mathcal{T}(-1,v)$:
    \[\mathcal{T}=v^2\big\|\nabla_{4}\nabla^M\nabla_4^{\frac{n-4}{2}}\Psi\big\|^2_{H^{1/2}}+v\big\|\nabla^M\nabla_4^{\frac{n-4}{2}}\Psi\big\|^2_{H^{3/2}}+\big\|\nabla_4^{\frac{n-4}{2}}\Psi^G\big\|^2_{H^{M+1}}\]
    \item Mildly singular top order energy $\mathcal{S}=\mathcal{S}(-1,v)$:
    \[\mathcal{S}=\big\|\nabla_4^{\frac{n-4}{2}}\alpha\big\|^2_{H^{M+1}}\]
    \item Lower order energy $\mathcal{L}=\mathcal{L}(-1,v)$:
    \[\mathcal{L}=\sum_{l=0}^{\frac{n-6}{2}}\sum_{m=0}^{M+\frac{n-4}{2}-l}v\big\|\nabla_4^{l+1}\Psi\big\|^2_{H^{m}}+\big\|\nabla_4^l\Psi\big\|^2_{H^{m+1}}\]
    \item Fractional lower order energy $\mathcal{M}_l=\mathcal{M}_l(-1,v)$ for any $0\leq l\leq\frac{n-6}{2}$:
    \[\mathcal{M}_l=\big\|\nabla^M\nabla_4^l\Psi\big\|^2_{H^{5/2}}\]
    \item Ricci coefficients norm $\mathcal{R}=\mathcal{R}(-1,v)$:
    \[\mathcal{R}=\sum_{k=0}^1\big\|\nabla^{M+1+k}\nabla_4^{\frac{n-4}{2}-k}\psi^*\big\|^2_{H^{1/2}}+\sum_{l=0}^{\frac{n-4}{2}}\sum_{m=0}^{M+\frac{n-4}{2}-l}\big\|\nabla_4^{l}\psi^*\big\|^2_{H^{m+1}}\]
    \item Lower order pointwise norm $\mathcal{P}=\mathcal{P}(-1,v)$ for $N'=\frac{M}{2}+\frac{n}{4}$:
    \[\mathcal{P}=\sum_{l=0}^{\frac{n-6}{2}}\sum_{m=0}^{N'}\big\|\nabla^m\nabla_4^{l}\Psi\big\|^2_{L^{\infty}}+\sum_{l=0}^{\frac{n-4}{2}}\sum_{m=0}^{N'}\big\|\nabla^m\nabla_4^{l}\psi^*\big\|^2_{L^{\infty}}\]
    \item Mildly singular pointwise norm $\mathcal{SP}=\mathcal{SP}(-1,v)$:
     \[\mathcal{SP}=\sum_{m=0}^{N'}\big\|\nabla^m\nabla_4^{\frac{n-4}{2}}\Psi\big\|^2_{L^{\infty}}+\sum_{m=0}^{N'}\big\|\nabla^m\nabla_4^{\frac{n-2}{2}}\psi^*\big\|^2_{L^{\infty}}\]
    \item Initial data norm $\mathcal{D}$:
    \[\mathcal{D}=\Big\|\Sigma\big(\slashed{g}_0,h\big)\Big\|_M^2\]
\end{itemize}

We complete the proof of Theorem~\ref{main theorem forward direction full system} at the end of the section. This will be straightforward once we establish the following result:
\begin{proposition}\label{main proposition forward direction full system}
    The spacetime $(\mathcal{M},g)$ satisfies the following estimates on $\big\{u=-1,\ 0\leq v\leq1\big\}:$
    \begin{align*}
        &\mathcal{T}+\mathcal{L}+\mathcal{R}\lesssim\mathcal{D},\\
        &\mathcal{S}\lesssim\big(1+|\log v|^2\big)\mathcal{D}.
    \end{align*}
\end{proposition}

We want to use the estimates of Section~\ref{stability dS section} as preliminary results. Since $\big\|\slashed{g}_0^*\big\|_{\mathring{H}^{M+1}}<\epsilon,$ we get by (\ref{Lie derivative in terms of covariant derivative}) that:
\[\sum_{m\leq M+1}\big\|\mathcal{L}_{\theta}^m\slashed{g}_0^*\big\|_{L^2(S^n)}\lesssim\epsilon.\]
Additionally, since $\big\|\mathcal{O}\big\|_{H^{M+1}}<\epsilon$ we also have that $\big\|(\log\nabla)\mathcal{O}\big\|_{H^{M}}\lesssim\epsilon.$ Thus, $\big\|\mathfrak{h}\big\|_{H^{M+1}}<\epsilon$ implies that $\big\|h\big\|_{H^{M}}\lesssim\epsilon,$ so the smallness condition (\ref{smallness assumption g}) holds. As a result, we can apply Propositions \ref{region I bounds proposition} and \ref{region II bounds proposition} for $N_1,N_2=N'=\frac{M}{2}+\frac{n}{4}$ to get that: \[\mathcal{P}\leq\epsilon,\ \mathcal{SP}\leq\epsilon\big(1+|\log(v)|^2\big)\]

The bounds for the top order energies $\mathcal{T}$ and $\mathcal{S}$ rely on the refined estimates proved in Section~\ref{model forward direction section} and \cite{Cwave}. Once we established these, we can bound the remaining norms $\mathcal{L},\mathcal{M},$ and $\mathcal{R}$ using standard estimates. We then bound the nonlinear error terms $Err^{\Psi},$ and we notice that similarly to Section~\ref{stability dS section}, these do not create significant difficulties since we commuted the equations with a high number of angular derivatives.

As a consequence of Theorem \ref{forward direction main result theorem general} for the system (\ref{first model system definition}), we obtain the following estimates for the top order energies:
\begin{corollary}
    The top order energy $\mathcal{T}$ and the mildly singular top order energy $\mathcal{S}$ satisfy the estimates for $0\leq v\leq 1:$
    \begin{align*}
        \mathcal{T}&\lesssim\mathcal{D}+\sum_{m=0}^{M}\int_0^{v}v'^{-\frac{1}{2}}\Big\|Err^{\Psi}_{m,\frac{n-4}{2}}\Big\|_{L^2}^2dv'+\int_0^{v}\Big\|Err^{\Psi}_{M,\frac{n-4}{2}}\Big\|_{H^{1/2}}^2dv'\\
        \mathcal{S}&\lesssim\big(1+|\log v|^2\big)\mathcal{D}+\sum_{m=0}^{M}\int_0^{v}v'^{-\frac{1}{2}}\Big\|Err^{\Psi}_{m,\frac{n-4}{2}}\Big\|_{L^2}^2dv'+\int_0^{v}\Big\|Err^{\Psi}_{M,\frac{n-4}{2}}\Big\|_{H^{1/2}}^2dv'
    \end{align*}
\end{corollary}
\begin{proof}
    We recall that according to Section~\ref{model systemmm} and Remark~\ref{remark about model systems in section}: \[\Phi_0=\nabla_4^{\frac{n-4}{2}}\alpha,\ \Phi_i=\nabla_4^{\frac{n-4}{2}}\Psi^G,\ F_{m}^{0}=Err_{m,\frac{n-4}{2}}^{\Psi},\ F_{m}^{i}=Err_{m,\frac{n-4}{2}}^{\Psi}\]
    satisfy the first model system as defined in (\ref{first model system definition}) and also \cite[Definition~1.1]{Cwave}. The bounds on the background $\big(\mathcal{M},g\big)$ required in \cite[Theorem~1.1]{Cwave} follow by Theorem~\ref{stability of de sitter theorem in section}. We also have that $\big\|\slashed{Riem}(\slashed{g}_0)\big\|_{H^M}\lesssim\mathcal{D}\lesssim1,$ so we can apply Theorem~\ref{forward direction main result theorem general} with an implicit constant depending only on $M$. We change coordinates to $v=\tau^2$ to obtain:
    \[v^2\big\|\nabla_{4}\nabla^M\nabla_4^{\frac{n-4}{2}}\alpha\big\|^2_{H^{1/2}}+v\big\|\nabla^M\nabla_4^{\frac{n-4}{2}}\alpha\big\|^2_{H^{3/2}}\lesssim\]\[\lesssim\big\|\mathcal{O}\big\|^2_{H^{M+1}}+\big\|\mathfrak{h}\big\|^2_{H^{M+1}}+\big\|\nabla_4^{\frac{n-4}{2}}\Psi^G_0\big\|^2_{H^{M+1}}+\sum_{m=0}^{M}\int_0^{v}v'^{-\frac{1}{2}}\Big\|Err^{\Psi}_{m,\frac{n-4}{2}}\Big\|_{L^2}^2dv'+\int_0^{v}\Big\|Err^{\Psi}_{M,\frac{n-4}{2}}\Big\|_{H^{1/2}}^2dv'.\]
    \[\big\|\nabla_4^{\frac{n-4}{2}}\alpha\big\|^2_{H^{M+1}}\lesssim\big(1+|\log v|^2\big) \big\|\mathcal{O}\big\|^2_{H^{M+1}}+\big\|\mathfrak{h}\big\|^2_{H^{M+1}}+\big\|\nabla_4^{\frac{n-4}{2}}\Psi^G_0\big\|^2_{H^{M+1}}+\]\[+\sum_{m=0}^{M}\int_0^{v}v'^{-\frac{1}{2}}\Big\|Err^{\Psi}_{m,\frac{n-4}{2}}\Big\|_{L^2}^2dv'+\int_0^{v}\Big\|Err^{\Psi}_{M,\frac{n-4}{2}}\Big\|_{H^{1/2}}^2dv'.\]
    \[v^{\frac{3}{2}}\big\|\nabla_{4}\nabla^M\nabla_4^{\frac{n-4}{2}}\Psi^G\big\|^2_{H^{1/2}}+\sqrt{v}\big\|\nabla^M\nabla_4^{\frac{n-4}{2}}\Psi^G\big\|^2_{H^{3/2}}+\big\|\nabla_4^{\frac{n-4}{2}}\Psi^G\big\|^2_{H^{M+1}}\lesssim\]\[\lesssim\big\|\mathcal{O}\big\|^2_{H^{M+1}}+\big\|\mathfrak{h}\big\|^2_{H^{M+1}}+\big\|\nabla_4^{\frac{n-4}{2}}\Psi^G_0\big\|^2_{H^{M+1}}+\sum_{m=0}^{M}\int_0^{v}v'^{-\frac{1}{2}}\Big\|Err^{\Psi}_{m,\frac{n-4}{2}}\Big\|_{L^2}^2dv'+\int_0^{v}\Big\|Err^{\Psi}_{M,\frac{n-4}{2}}\Big\|_{H^{1/2}}^2dv'.\]
    Using the previous notation for the initial data norm we obtain the conclusion.
\end{proof}

We prove the following result for the lower order energy:
\begin{proposition}
    The lower order energy $\mathcal{L}$ satisfies the estimate for $0\leq v\leq 1:$
    \[\mathcal{L}\lesssim\mathcal{D}+\epsilon\mathcal{R}+\int_0^v\big(\mathcal{R}+\mathcal{T}+\mathcal{S}\big)dv'+\sum_{l=0}^{\frac{n-6}{2}}\sum_{m=0}^{M+\frac{n-4}{2}-l}\int_0^{v}\Big\|Err^{\Psi}_{m,l}\Big\|_{L^2}^2dv'.\]
\end{proposition}
\begin{proof}
We recall that by (\ref{model wave system 1}), for every $0\leq l\leq\frac{n-6}{2}$ and $0\leq m\leq M+\frac{n-4}{2}-l$ we have that all curvature components $\Psi$ satisfy the commuted equation:
\[v\nabla_4^2\nabla^m\nabla_4^l\Psi+\bigg(3+l-\frac{n}{2}\bigg)\nabla_4\nabla^m\nabla_4^l\Psi-4\Delta\nabla^m\nabla_4^l\Psi=\sum_{\mu}\psi\nabla^{m+1}\nabla_4^l\Psi_{\mu}+Err_{ml}^{\Psi}\]
We contract each equation with $\nabla_4\nabla^m\nabla_4^l\Psi$, then sum over all $m,l$ in the admissible range and all curvature components $\Psi$ to obtain:
\[\sum_{l=0}^{\frac{n-6}{2}}\sum_{m=0}^{M+\frac{n-4}{2}-l}v\big\|\nabla_4\nabla^m\nabla_4^{l}\Psi\big\|^2_{L^2}+\big\|\nabla\nabla_4^l\Psi\big\|^2_{H^{m}}\lesssim\mathcal{D}+\sum_{l=0}^{\frac{n-6}{2}}\sum_{m=0}^{M+\frac{n-4}{2}-l}\int_0^v\big\|\nabla_4\nabla^m\nabla_4^l\Psi\big\|^2_{L^2}+\]\[+\sum_{l=0}^{\frac{n-6}{2}}\sum_{m=0}^{M+\frac{n-4}{2}-l}\int_0^v\big\|[\nabla,\nabla_4]\nabla^m\nabla_4^l\Psi\big\|^2_{L^2}+\big\|\nabla^{m+1}\nabla_4^l\Psi\big\|^2_{L^2}+\sum_{l=0}^{\frac{n-6}{2}}\sum_{m=0}^{M+\frac{n-4}{2}-l}\int_0^v\big\|Err_{ml}^{\Psi}\big\|^2_{L^2}\]
We also have the estimate:
\[\big\|\nabla_4^l\Psi\big\|^2_{L^2}\lesssim\mathcal{D}+\int_0^v\big\|\nabla_4^{l+1}\Psi\big\|^2_{L^2}\lesssim\ldots\lesssim\mathcal{D}+\int_0^v\big\|\nabla_4^{\frac{n-4}{2}}\Psi\big\|^2_{L^2}\lesssim\mathcal{D}+\int_0^v\mathcal{T}+\mathcal{S}\]
We use the commutation formulas and Gronwall to obtain:
\[\sum_{l=0}^{\frac{n-6}{2}}\sum_{m=0}^{M+\frac{n-4}{2}-l}v\big\|\nabla_4\nabla^m\nabla_4^{l}\Psi\big\|^2_{L^2}+\big\|\nabla_4^l\Psi\big\|^2_{H^{m+1}}\lesssim\]\[\lesssim\mathcal{D}+\sum_{l=0}^{\frac{n-6}{2}}\sum_{m=0}^{M+\frac{n-4}{2}-l}\int_0^v\mathcal{T}+\mathcal{S}+\big\|\nabla^m\nabla_4^{l+1}\Psi\big\|^2_{L^2}+\big\|\psi\nabla_4^l\Psi\big\|^2_{H^m}+\sum_{l=0}^{\frac{n-6}{2}}\sum_{m=0}^{M+\frac{n-4}{2}-l}\int_0^v\big\|Err_{ml}^{\Psi}\big\|^2_{L^2}\]
Once again we use Gronwall and bound the other terms using $\mathcal{T},\mathcal{S},$ and $\mathcal{R}$:
\[\sum_{l=0}^{\frac{n-6}{2}}\sum_{m=0}^{M+\frac{n-4}{2}-l}v\big\|\nabla_4\nabla^m\nabla_4^{l}\Psi\big\|^2_{L^2}+\big\|\nabla_4^l\Psi\big\|^2_{H^{m+1}}\lesssim\mathcal{D}+\int_0^v\big(\mathcal{T}+\mathcal{S}+\mathcal{R}\big)+\sum_{l=0}^{\frac{n-6}{2}}\sum_{m=0}^{M+\frac{n-4}{2}-l}\int_0^v\big\|Err_{ml}^{\Psi}\big\|^2_{L^2}\]
We use this and the lower order pointwise bound to also obtain:
\[\sum_{l=0}^{\frac{n-6}{2}}\sum_{m=0}^{M+\frac{n-4}{2}-l}v\big\|\psi\nabla_4^{l}\Psi\big\|^2_{H^m}\lesssim\sum_{m=N'}^{M+\frac{n-4}{2}}\big\|\psi^*\big\|^2_{H^m}\cdot\sum_{l=0}^{\frac{n-6}{2}}\big\|\nabla_4^{l}\Psi\big\|^2_{W^{N',\infty}}+\sum_{l=0}^{\frac{n-6}{2}}\sum_{m=0}^{M+\frac{n-4}{2}-l}\big\|\nabla_4^{l}\Psi\big\|^2_{H^m}\cdot\big\|\psi\big\|^2_{W^{N',\infty}}\]
\[\lesssim\epsilon\mathcal{R}+\mathcal{D}+\int_0^v\big(\mathcal{T}+\mathcal{S}+\mathcal{R}\big)+\sum_{l=0}^{\frac{n-6}{2}}\sum_{m=0}^{M+\frac{n-4}{2}-l}\int_0^v\big\|Err_{ml}^{\Psi}\big\|^2_{L^2}\]
We obtain the conclusion by using the commutation formulas:
\[\mathcal{L}\lesssim\sum_{l=0}^{\frac{n-6}{2}}\sum_{m=0}^{M+\frac{n-4}{2}-l}v\big\|\psi\nabla_4^{l}\Psi\big\|^2_{H^m}+\sum_{l=0}^{\frac{n-6}{2}}\sum_{m=0}^{M+\frac{n-4}{2}-l}v\big\|\nabla_4\nabla^m\nabla_4^{l}\Psi\big\|^2_{L^2}+\big\|\nabla_4^l\Psi\big\|^2_{H^{m+1}}.\]
\end{proof}

We prove the following result for the fractional lower order energy:
\begin{proposition}
    The fractional lower order energy $\mathcal{M}_l$ satisfies the estimate for any $0\leq l\leq\frac{n-6}{2}$ and $0\leq v\leq 1:$
    \[\mathcal{M}_l\lesssim\mathcal{T}+\mathcal{S}+\mathcal{L}+\mathcal{R}+\big\|Err_{M,l}^{\Psi}\big\|^2_{H^{1/2}}.\]
\end{proposition}
\begin{proof}
    We have the estimate for any $0\leq l\leq\frac{n-6}{2}$:
    \[\big\|\nabla^M\nabla_4^l\Psi\big\|^2_{H^{5/2}}\lesssim\big\|\nabla_4^l\Psi\big\|^2_{H^{M+2}}+\big\|\Delta\nabla^M\nabla_4^l\Psi\big\|^2_{H^{1/2}}\]
    \[\lesssim\mathcal{L}+\big\|v\nabla_4^2\nabla^M\nabla_4^l\Psi\big\|^2_{H^{1/2}}+\big\|\nabla_4\nabla^M\nabla_4^l\Psi\big\|^2_{H^{1/2}}+\big\|\nabla^{M+1}\nabla_4^l\Psi\big\|^2_{H^{1/2}}+\big\|Err_{M,l}^{\Psi}\big\|^2_{H^{1/2}}\]
    \[\lesssim\mathcal{L}+\big\|v\nabla_4\nabla^M\nabla_4^{l+1}\Psi\big\|^2_{H^{\frac{1}{2}}}+\big\|v\nabla_4\nabla^M\big(\psi\nabla_4^l\Psi\big)\big\|^2_{H^{\frac{1}{2}}}+\big\|\nabla^M\nabla_4^{l+1}\Psi\big\|^2_{H^{\frac{1}{2}}}+\big\|\nabla^M\big(\psi\nabla_4^l\Psi\big)\big\|^2_{H^{\frac{1}{2}}}+\big\|Err_{M,l}^{\Psi}\big\|^2_{H^{\frac{1}{2}}}\]
    \[\lesssim\mathcal{L}+\mathcal{T}+\mathcal{S}+\mathcal{R}+\big\|v\nabla_4\nabla^M\nabla_4^{l+1}\Psi\big\|^2_{H^{1/2}}+\big\|v\nabla_4\nabla^M\big(\psi\nabla_4^l\Psi\big)\big\|^2_{H^{1/2}}+\big\|Err_{M,l}^{\Psi}\big\|^2_{H^{1/2}}\]
    For $l=\frac{n-6}{2},$ we have the bound:
    \[\big\|v\nabla_4\nabla^M\nabla_4^{\frac{n-4}{2}}\Psi\big\|^2_{H^{1/2}}\lesssim\mathcal{T}.\]
    For $0\leq l\leq\frac{n-8}{2}$, we have:
    \[\big\|v\nabla_4\nabla^M\nabla_4^{l+1}\Psi\big\|^2_{H^{1/2}}\lesssim\big\|v\nabla^M\nabla_4^{l+2}\Psi\big\|^2_{H^{1/2}}+\big\|v\nabla^M\big(\psi\nabla_4^{l+1}\Psi\big)\big\|^2_{H^{1/2}}\lesssim\mathcal{T}+\mathcal{S}+\mathcal{L}+\mathcal{R}.\]
    So far, we proved that:
    \[\mathcal{M}_l\lesssim\mathcal{T}+\mathcal{S}+\mathcal{L}+\mathcal{R}+\big\|v\nabla_4\nabla^M\big(\psi\nabla_4^l\Psi\big)\big\|^2_{H^{1/2}}+\big\|Err_{M,l}^{\Psi}\big\|^2_{H^{1/2}}\]
    We can conclude since we also have the bounds using the null structure equations:
    \[\big\|v\nabla_4\nabla^M\big(\psi\nabla_4^l\Psi\big)\big\|^2_{H^{1/2}}\lesssim\big\|v\nabla^M\nabla_4\big(\psi\nabla_4^l\Psi\big)\big\|^2_{H^{1/2}}+\big\|v\nabla^M\big(\psi^2\nabla_4^l\Psi\big)\big\|^2_{H^{1/2}}\lesssim\]
    \[\lesssim\big\|v\nabla^M\big(\psi\nabla_4^{l+1}\Psi\big)\big\|^2_{H^{1/2}}+\big\|v\nabla^M\big((\Psi+\psi\psi)\nabla_4^l\Psi\big)\big\|^2_{H^{1/2}}+\mathcal{L}+\mathcal{R}\lesssim\mathcal{T}+\mathcal{S}+\mathcal{L}+\mathcal{R}.\]
\end{proof}

We prove the following result for the Ricci coefficients:
\begin{proposition}
    The Ricci coefficients norm $\mathcal{R}$ satisfies the estimate for any $0\leq v\leq 1:$
    \[\mathcal{R}\lesssim\mathcal{D}+\int_0^v\big((v')^{-1/2}\mathcal{T}+\mathcal{S}+\mathcal{L}+\mathcal{R}+(v')^{1/2}\mathcal{M}_{\frac{n-6}{2}}\big)dv'.\]
\end{proposition}
\begin{proof}
   For every $0\leq l\leq\frac{n-4}{2}$ and $0\leq m\leq M+\frac{n-4}{2}-l$, the Ricci coefficients $\psi^*$ satisfy the commuted equation:
    \[\nabla_4\nabla^{m+1}\nabla_4^l\psi^*=\nabla^{m+1}\nabla_4^l\Psi+\nabla^{m+1}\nabla_4^l\big(\psi\psi^*\big)+[\nabla_4,\nabla^{m+1}]\nabla_4^l\psi^*.\]
    We obtain the estimate:
    \[\sum_{l=0}^{\frac{n-4}{2}}\sum_{m=0}^{M+\frac{n-4}{2}-l}\big\|\nabla_4^l\psi^*\big\|^2_{H^{m+1}}\lesssim\mathcal{D}+\sum_{l=0}^{\frac{n-4}{2}}\sum_{m=0}^{M+\frac{n-4}{2}-l}\int_0^v\big\|\nabla_4^l\Psi\big\|^2_{H^{m+1}}+\big\|\nabla_4^l\big(\psi\psi^*\big)\big\|^2_{H^{m+1}}+\big\|\psi\nabla_4^l\psi^*\big\|^2_{H^{m+1}}dv'.\]\[\lesssim\mathcal{D}+\int_0^v\big(\mathcal{T}+\mathcal{S}+\mathcal{L}+\mathcal{R}\big)dv'\]
    Using the LP projections and Gronwall, we also obtain the following fractional estimate:
    \[\big\|\nabla^{M+1}\nabla_4^{\frac{n-4}{2}}\psi^*\big\|^2_{H^{1/2}}\lesssim\mathcal{D}+\int_0^v\mathcal{T}+\mathcal{S}+\mathcal{L}+\mathcal{R}+\int_0^v(v')^{1/2}\big\|\nabla_4\nabla^{M+1}\nabla_4^{\frac{n-4}{2}}\psi^*\big\|_{H^{1/2}}^2\]
    \[\lesssim\mathcal{D}+\int_0^v\mathcal{T}+\mathcal{S}+\mathcal{L}+\mathcal{R}+\sum_{m=0}^{M+1}\int_0^v\big\|\nabla^{m}\nabla_4^{\frac{n-4}{2}}\big(\psi\psi^*\big)\big\|^2_{H^{1/2}}+(v')^{1/2}\big\|\nabla^{m}\nabla_4^{\frac{n-4}{2}}\Psi\big\|_{H^{1/2}}^2\lesssim\mathcal{D}+\int_0^v\frac{\mathcal{T}}{(v')^{1/2}}+\mathcal{S}+\mathcal{L}+\mathcal{R}\]
    Similarly, we also get the estimate:
    \[\big\|\nabla^{M+2}\nabla_4^{\frac{n-6}{2}}\psi^*\big\|^2_{H^{1/2}}\lesssim\mathcal{D}+\int_0^v\mathcal{T}+\mathcal{S}+\mathcal{L}+\mathcal{R}+\int_0^v(v')^{1/2}\big\|\nabla_4\nabla^{M+2}\nabla_4^{\frac{n-6}{2}}\psi^*\big\|_{H^{1/2}}^2\]
    \[\lesssim\mathcal{D}+\int_0^v\mathcal{T}+\mathcal{S}+\mathcal{L}+\mathcal{R}+\sum_{m=0}^{M+2}\int_0^v\big\|\nabla^{m}\nabla_4^{\frac{n-6}{2}}\big(\psi\psi^*\big)\big\|^2_{H^{1/2}}+(v')^{1/2}\big\|\nabla^{m}\nabla_4^{\frac{n-6}{2}}\Psi\big\|_{H^{1/2}}^2\]
    \[\lesssim\mathcal{D}+\int_0^v\mathcal{T}+\mathcal{S}+\mathcal{L}+\mathcal{R}+(v')^{1/2}\mathcal{M}_{\frac{n-6}{2}}.\]
\end{proof}
Next, we bound the error terms that appear in the above estimates. We recall that for any $0\leq l\leq\frac{n-4}{2}$ and $0\leq m\leq M+\frac{n-4}{2}-l$ we have:
\[Err_{ml}^{\Psi}=v\mathcal{F}_{(m)(l+1)(l+m+1)}(\Psi)+v\mathcal{F}_{(m+1)(l)(l+m+1)}(\Psi)+\mathcal{F}_{(2+m)(l-1)(l+m+1)}(\Psi)+\mathcal{F}_{(m)(l)(m+l)}(\Psi)+\]\[+\mathcal{F}_{(m+1)(l)(m+l+1)}^{lot}(\Psi)+\nabla^m\nabla_4^l\big(\Psi\cdot\Psi^G\big)+\sum_{i+2j=m}\nabla^i\big(\slashed{Riem}^{j+1}\cdot\nabla_4^l\Psi\big)+\nabla^m\nabla_4^l\nabla\big(\psi\psi\big)+\nabla^m\nabla_4^l\nabla\big(\psi\psi\psi\big)\]
We also recall the error term notation:
\[\mathcal{F}_{mlp}(\Psi)=\sum_{\substack{i+j+k\leq p \\ i\leq l,k\leq m}}\nabla^k\nabla_4^i\big(\psi^{j+1}\Psi\big).\]
\begin{lemma}\label{lemma to bound general F error terms}
For any $0\leq s\leq1,\ l\leq\frac{n-2}{2},\ m\leq M+\frac{n}{2},\ p=m+l$ we have the estimate:    \[\big\|\mathcal{F}_{mlp}(\Psi)\big\|^2_{H^{s}}\lesssim\big(1+|\log(v)|^2\big)\cdot\sum_{i\leq l}\sum_{k\leq m}\big\|\nabla^k\nabla_4^i\Psi\big\|^2_{H^{s}}+\big\|\nabla^k\nabla_4^i\psi^*\big\|^2_{H^{s}}.\]
Similarly, for any $0\leq s\leq1,\ l\leq\frac{n-4}{2},\ m\leq M+\frac{n}{2},\ p=m+l$ we have the estimate:    \[\big\|\mathcal{F}_{mlp}(\Psi)\big\|^2_{H^{s}}\lesssim\sum_{i\leq l}\sum_{k\leq m}\big\|\nabla^k\nabla_4^i\Psi\big\|^2_{H^{s}}+\big\|\nabla^k\nabla_4^i\psi^*\big\|^2_{H^{s}}.\]
\end{lemma}
\begin{proof} 
    Using the definition of $\mathcal{F}_{mlp}$, we get that for all $0\leq s\leq1,\ l\leq\frac{n-2}{2},\ m\leq M+\frac{n}{2},\ p=m+l$:
    \[\big\|\mathcal{F}_{mlp}(\Psi)\big\|^2_{H^{s}}\lesssim\sum_{\substack{i+j+k\leq p \\ i\leq l,k\leq m}}\big\|\nabla^k\nabla_4^i\big(\psi^{j+1}\Psi\big)\big\|^2_{H^{s}}\lesssim\sum_{\substack{|i|+j+|k|\leq p \\ |i|\leq l,|k|\leq m}}\Big\|\nabla^{k_0}\nabla_4^{i_0}\Psi\prod_{q=1}^{j+1}\nabla^{k_q}\nabla_4^{i_q}\psi\Big\|^2_{H^{s}}\]
    We use the fact that $\nabla\nabla_4^i\psi=\nabla\nabla_4^i\psi^*$ to get:
    \[\big\|\mathcal{F}_{mlp}(\Psi)\big\|^2_{H^{s}}\lesssim\sum_{\substack{|i|+j+|k|\leq p \\ |i|\leq l,|k|\leq m}}\Big\|\nabla^{k_0}\nabla_4^{i_0}\Psi\cdot\nabla^{k_{1}}\nabla_4^{i_{1}}\psi^*\prod_{q=2}^{j+1}\nabla^{k_q}\nabla_4^{i_q}\psi\Big\|^2_{H^{s}}+\sum_{\substack{|i|+j+|k|\leq p \\ |i|\leq l,|k|\leq m}}\Big\|\nabla^{k}\nabla_4^{i_0}\Psi\prod_{q=1}^{j+1}\nabla_4^{i_q}\psi\Big\|^2_{H^{s}}\]
    The second term can be bounded using the lower order pointwise norm and the mildly singular pointwise norm by:
    \[\sum_{\substack{|i|+j+|k|\leq p \\ |i|\leq l,|k|\leq m}}\Big\|\nabla^{k}\nabla_4^{i_0}\Psi\prod_{q=1}^{j+1}\nabla_4^{i_q}\psi\Big\|^2_{H^{s}}\lesssim\big(1+|\log(v)|^2\big)\cdot\sum_{i\leq l}\sum_{k\leq m}\big\|\nabla^k\nabla_4^i\Psi\big\|^2_{H^{s}}\]
    For the first term, we can assume that $|k_{1}|=\max(|k_{1}|,\ldots,|k_{j+1}|).$ In particular, we have $|k_{q}|<1+\frac{M}{2}+\frac{n}{4}$ for all $2\leq q\leq j+1,$ so we can control these factors using the lower order pointwise norm and the mildly singular pointwise norm. Thus, we get that the first term is bounded by:
    \[\big(1+|\log(v)|^2\big)\cdot\sum_{\substack{|i|+j+|k|\leq p \\ |i|\leq l,|k|\leq m}}\Big\|\nabla^{k_0}\nabla_4^{i_0}\Psi\cdot\nabla^{k_{q_1}}\nabla_4^{i_{q_1}}\psi^*\Big\|^2_{H^{s}}\lesssim\big(1+|\log(v)|^2\big)\cdot\sum_{i\leq l}\sum_{k\leq m}\big\|\nabla^k\nabla_4^i\Psi\big\|^2_{H^{s}}+\big\|\nabla^k\nabla_4^i\psi^*\big\|^2_{H^{s}}\]
    This completes the proof of the first statement. We remark that we can prove the second statement by following the exact same steps, but we only use the lower order pointwise norm when $l\leq\frac{n-4}{2}.$
\end{proof}
\begin{proposition}
    For any $0\leq l\leq\frac{n-4}{2}$, the error terms $Err_{M,l}^{\Psi}$ satisfy the estimate:
    \[\big\|Err_{M,l}^{\Psi}\big\|^2_{H^{1/2}}\lesssim\big(1+|\log(v)|^2\big)\big(\mathcal{T}+\mathcal{S}+\mathcal{L}+\mathcal{R}\big).\]
\end{proposition}
\begin{proof}
    The proof will follow from the fractional lower order energy estimates, once we prove the claim:
    \begin{equation}\label{preliminary fractional lower order estimate}
        \big\|Err_{M,l}^{\Psi}\big\|^2_{H^{1/2}}\lesssim\big(1+|\log(v)|^2\big)\big(\mathcal{T}+\mathcal{S}+\mathcal{L}+\mathcal{R}\big)+\mathcal{M}_{l-1}.
    \end{equation}
    Using the previous lemma, the Bianchi equations, and the null structure equations, we have the bound:
    \[\big\|v\mathcal{F}_{(M)(l+1)(l+M+1)}(\Psi)\big\|^2_{H^{1/2}}\lesssim\big(1+|\log(v)|^2\big)\cdot\sum_{i\leq l+1}\sum_{k\leq M}\big\|v\nabla^k\nabla_4^i\Psi\big\|^2_{H^{1/2}}+\big\|v\nabla^k\nabla_4^i\psi^*\big\|^2_{H^{1/2}}\]\[\lesssim\mathcal{T}+\mathcal{S}+\mathcal{L}+\mathcal{R}+\big(1+|\log(v)|^2\big)\cdot\Big(\sum_{k\leq M}\big\|v\nabla^k\nabla_4^{\frac{n-2}{2}}\psi^*\big\|^2_{H^{1/2}}+\big\|v\nabla^M\nabla_4^{\frac{n-2}{2}}\Psi\big\|^2_{H^{1/2}}+\big\|v\nabla_4^{\frac{n-2}{2}}\Psi\big\|^2_{H^{1}}\Big)\]
    \[\lesssim\big(1+|\log(v)|^2\big)\bigg(\mathcal{T}+\mathcal{S}+\mathcal{L}+\mathcal{R}+\sum_{k\leq M}\big\|v\nabla^k\nabla_4^{\frac{n-4}{2}}\big(\Psi+\psi\psi^*\big)\big\|^2_{H^{1/2}}+\big\|\nabla_4^{\frac{n-4}{2}}\big(\nabla\Psi+\psi\Psi+\Psi\big)\big\|^2_{H^{1}}\bigg)\]\[\lesssim\big(1+|\log(v)|^2\big)\big(\mathcal{T}+\mathcal{S}+\mathcal{L}+\mathcal{R}\big)\]
    Similarly, we use the previous lemma to get the bound:
    \[\big\|v\mathcal{F}_{(M+1)(l)(l+M+1)}(\Psi)\big\|^2_{H^{1/2}}\lesssim\sum_{i\leq l}\sum_{k\leq M+1}\big\|v\nabla^k\nabla_4^i\Psi\big\|^2_{H^{1/2}}+\big\|v\nabla^k\nabla_4^i\psi^*\big\|^2_{H^{1/2}}\]\[\lesssim\mathcal{L}+\mathcal{R}+\sum_{k\leq M}\big\|v\nabla^k\nabla_4^{\frac{n-4}{2}}\psi^*\big\|^2_{H^{3/2}}+\big\|v\nabla^k\nabla_4^{\frac{n-4}{2}}\Psi\big\|^2_{H^{3/2}}\lesssim\mathcal{T}+\mathcal{S}+\mathcal{L}+\mathcal{R}\]
    We use the fractional lower order energy to get the bound:
    \[\big\|\mathcal{F}_{(M+2)(l-1)(l+M+1)}(\Psi)\big\|^2_{H^{1/2}}\lesssim\sum_{i\leq l-1}\sum_{k\leq M+2}\big\|\nabla^k\nabla_4^i\Psi\big\|^2_{H^{1/2}}+\big\|\nabla^k\nabla_4^i\psi^*\big\|^2_{H^{1/2}}\]\[\lesssim\mathcal{L}+\mathcal{R}+\sum_{k\leq M}\big\|\nabla^k\nabla_4^{l-1}\psi^*\big\|^2_{H^{5/2}}+\sum_{k\leq M}\big\|\nabla^k\nabla_4^{l-1}\Psi\big\|^2_{H^{5/2}}\lesssim\mathcal{L}+\mathcal{R}+\mathcal{M}_{l-1}\]
    Next, we have the bound:
    \[\big\|\mathcal{F}_{(M)(l)(l+M)}(\Psi)\big\|^2_{H^{1/2}}\lesssim\sum_{i\leq l}\sum_{k\leq M}\big\|\nabla^k\nabla_4^i\Psi\big\|^2_{H^{1/2}}+\big\|\nabla^k\nabla_4^i\psi^*\big\|^2_{H^{1/2}}\lesssim\mathcal{T}+\mathcal{S}+\mathcal{L}+\mathcal{R}\]
    We can prove a similar result to the previous lemma for the error terms $\mathcal{F}^{lot}$, and we obtain:
    \[\big\|\mathcal{F}_{(M+1)(l)(l+M+1)}^{lot}(\Psi)\big\|^2_{H^{1/2}}\lesssim\sum_{\substack{i\leq l \\ k\leq M}}\big\|\nabla^k\nabla_4^i\Psi\big\|^2_{H^{1/2}}+\sum_{\substack{i\leq l-1 \\ k\leq M+1}}\big\|\nabla^k\nabla_4^i\Psi\big\|^2_{H^{1/2}}+\sum_{\substack{i\leq l \\ k\leq M+1}}\big\|\nabla^k\nabla_4^i\psi^*\big\|^2_{H^{1/2}}\]\[\lesssim\mathcal{T}+\mathcal{S}+\mathcal{L}+\mathcal{R}\]
    Next, we have the bound:
    \[\big\|\nabla^M\nabla_4^l\big(\Psi\cdot\Psi^G\big)\big\|^2_{H^{1/2}}\lesssim\big(1+|\log(v)|^2\big)\big(\mathcal{T}+\mathcal{S}+\mathcal{L}\big)\]
    Using this bound, the bound on $\mathcal{F}_{(M)(l)(l+M)}(\Psi)$ and the formula $\slashed{Riem}=\Psi+\psi\psi,$ we also obtain the bound:
    \[\sum_{i+2j=M}\big\|\nabla^i\big(\slashed{Riem}^{j+1}\cdot\nabla_4^l\Psi\big)\big\|^2_{H^{1/2}}\lesssim\mathcal{T}+\mathcal{S}+\mathcal{L}+\mathcal{R}\]
    Finally, we use the commutation formulas to get:
    \[\big\|\nabla^M\nabla_4^l\nabla\big(\psi\psi\psi\big)\big\|^2_{H^{1/2}}\lesssim\mathcal{R}+\sum_{k\leq M+1}\sum_{l\leq\frac{n-4}{2}}\big\|\nabla^k\nabla_4^l\psi^*\big\|^2_{H^{1/2}}\lesssim\mathcal{R}.\]
    Combining all our estimates, we established (\ref{preliminary fractional lower order estimate}). Using the fractional lower order energy estimate, we get:
    \[\big\|Err_{M,l}^{\Psi}\big\|^2_{H^{1/2}}\lesssim\big(1+|\log(v)|^2\big)\big(\mathcal{T}+\mathcal{S}+\mathcal{L}+\mathcal{R}\big)+\mathcal{M}_{l-1}\lesssim\big(1+|\log(v)|^2\big)\big(\mathcal{T}+\mathcal{S}+\mathcal{L}+\mathcal{R}\big)+\big\|Err_{M,l-1}^{\Psi}\big\|^2_{H^{1/2}}\]
    By induction we obtain the conclusion, since:
    \[\big\|Err_{M,l}^{\Psi}\big\|^2_{H^{1/2}}\lesssim\big(1+|\log(v)|^2\big)\big(\mathcal{T}+\mathcal{S}+\mathcal{L}+\mathcal{R}\big)+\big\|Err_{M,0}^{\Psi}\big\|^2_{H^{1/2}}\lesssim\big(1+|\log(v)|^2\big)\big(\mathcal{T}+\mathcal{S}+\mathcal{L}+\mathcal{R}\big).\]
\end{proof}
\begin{proposition}
    For any for any $0\leq l\leq\frac{n-4}{2}$ and $0\leq m\leq M+\frac{n-4}{2}-l$, the error terms $Err_{m,l}^{\Psi}$ satisfy the estimate:
    \[\big\|Err_{m,l}^{\Psi}\big\|^2_{L^2}\lesssim(1+|\log(v)|^2\big)\big(\mathcal{T}+\mathcal{S}+\mathcal{L}+\mathcal{R}\big).\]
\end{proposition}
\begin{proof}
    As before, we use the Lemma \ref{lemma to bound general F error terms}, the Bianchi equations, and the null structure equations to get the bound:
    \[\big\|v\mathcal{F}_{(m)(l+1)(l+m+1)}(\Psi)\big\|^2_{L^2}\lesssim(1+|\log(v)|^2\big)\sum_{i\leq l+1}\sum_{k\leq m}\big\|v\nabla^k\nabla_4^i\Psi\big\|^2_{L^2}+\big\|v\nabla^k\nabla_4^i\psi^*\big\|^2_{L^2}\lesssim\]\[\lesssim(1+|\log(v)|^2\big)\bigg(\mathcal{T}+\mathcal{S}+\mathcal{L}+\mathcal{R}+\big\|v\nabla_4^{\frac{n-2}{2}}\psi^*\big\|^2_{H^{M}}+\big\|v\nabla_4^{\frac{n-2}{2}}\Psi\big\|^2_{H^{M}}\bigg)\lesssim\]
    \[\lesssim(1+|\log(v)|^2\big)\Big(\mathcal{T}+\mathcal{S}+\mathcal{L}+\mathcal{R}+\big\|v\nabla_4^{\frac{n-4}{2}}\big(\Psi+\psi\psi^*\big)\big\|^2_{H^{M}}+\big\|\nabla_4^{\frac{n-4}{2}}\big(\nabla\Psi+\psi\Psi\big)\big\|^2_{L^2}+\big\|v\nabla_4\nabla^M\nabla_4^{\frac{n-4}{2}}\Psi\big\|^2_{L^2}\Big)\]
    \[\lesssim(1+|\log(v)|^2\big)\big(\mathcal{T}+\mathcal{S}+\mathcal{L}+\mathcal{R}\big)\]
    Similarly, we have the bound:
    \[\big\|v\mathcal{F}_{(m+1)(l)(l+m+1)}(\Psi)\big\|^2_{L^2}\lesssim\sum_{i\leq l}\sum_{k\leq m+1}\big\|v\nabla^k\nabla_4^i\Psi\big\|^2_{L^2}+\big\|v\nabla^k\nabla_4^i\psi^*\big\|^2_{L^2}\lesssim\mathcal{T}+\mathcal{S}+\mathcal{L}+\mathcal{R}\]
    We also have the bound:
    \[\big\|\mathcal{F}_{(m+2)(l-1)(l+m+1)}(\Psi)\big\|^2_{L^2}\lesssim\sum_{i\leq l-1}\sum_{k\leq m+2}\big\|\nabla^k\nabla_4^i\Psi\big\|^2_{L^2}+\big\|\nabla^k\nabla_4^i\psi^*\big\|^2_{L^2}\lesssim\mathcal{L}+\mathcal{R}\]
    Next, we have the bound:
    \[\big\|\mathcal{F}_{(m)(l)(l+m)}(\Psi)\big\|^2_{L^2}\lesssim\sum_{i\leq l}\sum_{k\leq m}\big\|\nabla^k\nabla_4^i\Psi\big\|^2_{L^2}+\big\|\nabla^k\nabla_4^i\psi^*\big\|^2_{L^2}\lesssim\mathcal{T}+\mathcal{S}+\mathcal{L}+\mathcal{R}\]
    Using the similar result to the lemma for the error terms $\mathcal{F}^{lot}$, we get:
    \[\big\|\mathcal{F}_{(m+1)(l)(l+m+1)}^{lot}(\Psi)\big\|^2_{L^2}\lesssim\sum_{\substack{i\leq l \\ k\leq m}}\big\|\nabla^k\nabla_4^i\Psi\big\|^2_{L^2}+\sum_{\substack{i\leq l-1 \\ k\leq m+1}}\big\|\nabla^k\nabla_4^i\Psi\big\|^2_{L^2}+\sum_{\substack{i\leq l \\ k\leq m+1}}\big\|\nabla^k\nabla_4^i\psi^*\big\|^2_{L^2}\lesssim\mathcal{T}+\mathcal{S}+\mathcal{L}+\mathcal{R}\]
    Next, we have the bound:
    \[\big\|\nabla^m\nabla_4^l\big(\Psi\cdot\Psi^G\big)\big\|^2_{L^2}\lesssim\big(1+|\log(v)|^2\big)\big(\mathcal{T}+\mathcal{S}+\mathcal{L}\big)\]
    As before, we use this bound, the bound on $\mathcal{F}_{(m)(l)(l+m)}(\Psi)$, and the formula $\slashed{Riem}=\Psi+\psi\psi,$ in order to obtain the bound:
    \[\sum_{i+2j=m}\big\|\nabla^i\big(\slashed{Riem}^{j+1}\cdot\nabla_4^l\Psi\big)\big\|^2_{L^2}\lesssim\mathcal{T}+\mathcal{S}+\mathcal{L}+\mathcal{R}\]
    Finally, to conclude our proof we use the commutation formulas to get:
    \[\big\|\nabla^m\nabla_4^l\nabla\big(\psi\psi\psi\big)\big\|^2_{L^2}\lesssim\mathcal{R}.\]
\end{proof}

We complete the proofs of Proposition \ref{main proposition forward direction full system} and Theorem \ref{main theorem forward direction full system}:

\textit{Proof of Proposition \ref{main proposition forward direction full system}.} Based on the estimates in this section, we have that for all $v\leq1$:
\[\mathcal{T}+\mathcal{L}+\mathcal{R}+\frac{1}{1+|\log v|^2}\mathcal{S}\lesssim\mathcal{D}+\epsilon\mathcal{R}+\int_0^v\big(1+|\log v'|^2\big)^2v'^{-\frac{1}{2}}\bigg(\mathcal{T}+\frac{1}{1+|\log v'|^2}\mathcal{S}+\mathcal{L}+\mathcal{R}\bigg)dv'\]
Taking $\epsilon>0$ small enough and then using Gronwall's inequality, we obtain that:
\[\mathcal{T}+\mathcal{L}+\mathcal{R}+\frac{1}{1+|\log v|^2}\mathcal{S}\lesssim\mathcal{D}.\]
\qed

\textit{Proof of Theorem \ref{main theorem forward direction full system}.}
Using Proposition \ref{main proposition forward direction full system} and self-similarity at $u=-1,v=1$ we obtain the desired estimates for the Ricci coefficients and curvature components. To complete the proof, we also show that $\big\|\slashed{g}^*\big\|_{H^{M+1}}^2\lesssim\mathcal{D}.$ We first need to prove this on the sphere $S_{-1,0}.$ As outlined in the beginning of the section, we use (\ref{Lie derivative in terms of covariant derivative}) to get:
\[\sum_{m\leq M+1}\big\|\mathcal{L}_{\theta}^m\slashed{g}_0^*\big\|_{L^2(S^n)}^2\lesssim\big\|\slashed{g}_0^*\big\|_{\mathring{H}^{M+1}}^2\lesssim\mathcal{D}.\]
We then use (\ref{Lie derivative in terms of covariant derivative}) for covariant derivatives with respect to $\slashed{g}_0$ and prove by induction on $m\leq M+1$ that:
\[\big\|\nabla^m\slashed{g}^*_0\big\|_{L^2}^2\lesssim\mathcal{D}.\]
The conclusion follows by using the metric equation $\mathcal{L}_v\slashed{g}^*=\psi^*$ and the estimate $\mathcal{R}\lesssim\mathcal{D}.$
\qed

\section{Estimates for the Second Model System}\label{model backward direction section}

An essential part of our argument is to prove estimates on the second model system (\ref{second model system definition}), in terms of the initial data at $\{\tau=1\}.$ This system includes the commuted Bianchi system, giving estimates on the induced asymptotic data at $\{v=0\}$ in terms of the initial data at finite times. In the present paper, we illustrate how to prove these estimates for a toy problem which models the singular top order quantity $\nabla^M\Phi_0,$ and captures the main difficulties of the problem. This will serve as a guideline for the general model systems treated in \cite{Cwave}. We encourage the reader to return to Section~\ref{scattering map section intro} of the introduction for an outline of the proof. We refer the reader to \cite[Section~4]{Cwave} for a complete proof for the system (\ref{second model system definition}), where we prove:
\begin{theorem}[{\cite[Theorem~1.2]{Cwave}}]\label{backward direction main result theorem general} For any $M>0$ large enough, the system (\ref{second model system definition}) satisfies the estimates for all $\tau\in(0,1)$:
    \[\sum_{m=0}^M\Big(\tau\big\|\nabla^m\Phi_0\big\|_{H^{1/2}}^2+\tau^2\big\|\nabla^m\Phi_0\big\|_{H^{3/2}}^2\Big)+\tau^2\big\|\nabla_{\tau}\nabla^M\Phi_0\big\|_{H^{1/2}}^2+\sum_{m=0}^{M-1}\tau^2\big\|\nabla_{\tau}\nabla^m\Phi_0\big\|_{L^2}^2+\int_{\tau}^1\tau'\big\|\Phi_0\big\|_{H^{M+1}}^2d\tau'\]\[+\sum_{i=1}^I\sum_{m=0}^M\big\|\nabla^m\Phi_i\big\|_{H^{3/2}}^2+\sum_{i=1}^I\sum_{m=0}^{M}\big\|\nabla_{\tau}\nabla^m\Phi_i\big\|_{H^{1/2}}^2+\sum_{i=1}^I\sum_{m=0}^{M}\int_{\tau}^1\frac{1}{\tau'}\big\|\nabla_{\tau}\nabla^m\Phi_i\big\|_{H^{1/2}}^2d\tau'\lesssim\]\[\lesssim \sum_{i=0}^I\sum_{m=0}^M\bigg(\big\|\nabla^m\Phi_i\big\|_{H^{3/2}}^2+\big\|\nabla^m\nabla_{\tau}\Phi_i\big\|_{H^{1/2}}^2\bigg)\bigg|_{\tau=1}+\sum_{i=0}^I\sum_{m=0}^{M}\int_{\tau}^1\tau'\big\|F_m^i\big\|_{H^{1/2}}^2d\tau'.\]
    In the above estimate, the implicit constant depends on $M>N$, the bound on $\|\psi\|_{H^{M+1}}\big|_{\tau=1},$ and the bounds satisfied by the background $\big(\mathcal{M},g\big)$ according to Theorem~\ref{stability of de sitter theorem in section}.
\end{theorem}

Additionally, in \cite[Section~4]{Cwave} we also prove estimates for the asymptotic quantities $\mathcal{O}$ and $\mathfrak{h}_M,$ where the tensor $h$ is renormalized as in Section~\ref{model forward direction section} in order to obtain a sharp scattering result:
\begin{theorem}[{\cite[Theorem~1.2]{Cwave}}]\label{asymptotic data estimates theorem general}
    For any $M>0$ and $x>0$ large enough we have the estimate:
    \[\big\|\mathcal{O}\big\|_{H^{M+1}}^2+\sum_{k\geq x}2^{2k}\big\|P_k\mathfrak{h}_M\big\|_{L^2}^2\lesssim \sum_{i=0}^I\sum_{m=0}^M\bigg(\big\|\nabla^m\Phi_i\big\|_{H^{3/2}}^2+\big\|\nabla^m\nabla_{\tau}\Phi_i\big\|_{H^{1/2}}^2\bigg)\bigg|_{\tau=1}+\sum_{i=0}^I\sum_{m=0}^{M}\int_{0}^1\tau\big\|F_m^i\big\|_{H^{1/2}}^2d\tau.\]
    As in Theorem~\ref{backward direction main result theorem general}, the implicit constant depends on $M,$ the bound on $\|\psi\|_{H^{M+1}}\big|_{\tau=1},$ and the bounds satisfied by the background $\big(\mathcal{M},g\big)$ according to Theorem~\ref{stability of de sitter theorem in section}.
\end{theorem}
\begin{remark}
    In practice, we use this result combined with the following bound for $\big\|\mathfrak{h}\big\|_{H^{M+1}}^2$:
    \[\big\|\mathfrak{h}\big\|_{H^{M+1}}^2\lesssim\big\|h\big\|_{L^2}^2+C\Big(\big\|\slashed{Riem}_0\big\|_{H^{M}}^2\Big)\big\|\mathcal{O}\big\|_{H^{M}}^2+\sum_{k\geq x}2^{2k}\big\|P_k\mathfrak{h}_M\big\|_{L^2}^2,\]
    which follows using the estimates for the LP projections in Section~\ref{LP Section} and \cite{geometricLP}.
\end{remark}

The main difficulties of the above results are already present in the following toy problem. We assume that the smooth horizontal tensor $\xi$ defined on the hypersurface $\{u=-1\}\times\{\tau\in(0,1)\}\times S^n$ satisfies the equations:
\begin{equation}\label{backward direction linear wave equation main}
    \nabla_{\tau}\big(\nabla_{\tau}\xi\big)+\frac{1}{\tau}\nabla_{\tau}\xi-4\Delta\xi=\psi\nabla\xi
\end{equation}
\[\xi=2\nabla^M\mathcal{O}\log\tau+2\big(\log\nabla\big)\nabla^M\mathcal{O}+\mathfrak{h}_M+O\big(\tau^2|\log\tau|^2\big),\ \nabla_{\tau}\xi=\frac{2\nabla^M\mathcal{O}}{\tau}+O\big(\tau|\log\tau|^2\big)\text{ in }C^{\infty}(S^n).\]
We notice that the asymptotic expansion at $\tau=0$ is the same as that of $\nabla^M\Phi_0,$ but in equation (\ref{backward direction linear wave equation main}) we only kept the terms depending on $\nabla\xi$ on the right hand side for simplicity.

In this section we prove the following result:
\begin{theorem}\label{backward direction main result theorem} For any $M>0$ large enough, we have for all $\tau\in(0,1)$:
\begin{equation}\label{backward direction main estimate one}
    \tau\big\|\xi\big\|_{H^{1/2}}^2+\tau^2\big\|\xi\big\|_{H^{3/2}}^2+\tau^2\big\|\nabla_{\tau}\xi\big\|_{H^{1/2}}^2+\int_{\tau}^1\tau'\big\|\xi\big\|_{H^{1}}^2d\tau'\lesssim \bigg(\big\|\xi\big\|_{H^{3/2}}^2+\big\|\nabla_{\tau}\xi\big\|_{H^{1/2}}^2\bigg)\bigg|_{\tau=1}
\end{equation}
\begin{equation}\label{backward direction asymptotic quantity estimate}
    \big\|\nabla^M\mathcal{O}\big\|_{H^{1}}^2+\sum_{k\geq 0}2^{2k}\big\|P_k\mathfrak{h}_M\big\|_{L^2}^2\lesssim \bigg(\big\|\xi\big\|_{H^{3/2}}^2+\big\|\nabla_{\tau}\xi\big\|_{H^{1/2}}^2\bigg)\bigg|_{\tau=1}
\end{equation}
\end{theorem}

We denote the initial data norm at $\tau=1$:
\[D:=\Big(\big\|\nabla_{\tau}\xi\big\|_{H^{1/2}}^2+\big\|\xi\big\|_{H^{3/2}}^2\Big)\Big|_{\tau=1}.\]
According to \cite[Proposition~4.2]{Cwave}, it is straightforward to establish the following preliminary estimate
\begin{equation}\label{backward direction basic estimate one}
    \big\|\tau\nabla_{\tau}\xi\big\|_{L^2}^2+\big\|\tau\nabla\xi\big\|_{L^2}^2+\tau\big\|\xi\big\|_{L^2}^2+\int_{\tau}^1\tau'\big\|\xi\big\|_{H^1}^2d\tau'\lesssim D.
\end{equation}
In order to prove optimal estimates for $\xi$, we split our analysis into the low frequency and high frequency regime. We consider $X=2^{x+1}$ to be a large constant, to be chosen later. We split the frequencies into the following regimes:
\begin{itemize}
    \item Low frequency regime $k<x$ for all $\tau\in[0,1]$. We bound the solution using the preliminary estimate \eqref{backward direction basic estimate one}.
    \item Low frequency regime $k\geq x$ for $\tau\in[0,X2^{-k-1}]$. We bound the solution in Proposition~\ref{preliminary low frequency estimate proposition} in Section~\ref{backward low freq estimates section}.
    \item High frequency regime $k\geq x$ for $\tau\in[X2^{-k-1},1]$. We bound the solution in Proposition~\ref{backward high freq proposition improved} in Section~\ref{high frequency estimates section}.
\end{itemize}
We notice that unlike the case of the first model system, in the current situation we had to introduce an additional parameter $X$ for technical reasons, which will become clear later. We also make the notation convention until otherwise stated that if the implicit constant in an inequality depends on $X,$ we write $\lesssim_X.$ On the other hand, if we only write $\lesssim$, then the implicit constant is independent of $X$.

Once we established the low frequency regime estimate in Proposition~\ref{preliminary low frequency estimate proposition} and the high frequency regime estimate in Proposition~\ref{backward high freq proposition improved}, we combine them to prove \eqref{backward direction main estimate one} in Section~\ref{backward main result estimates section}. We then prove the estimate \eqref{backward direction asymptotic quantity estimate} for the asymptotic quantities in Section~\ref{backward asymptotic quantities estimates section}, which completes the proof of Theorem~\ref{backward direction main result theorem}.

\subsection{Low Frequency Regime Estimates}\label{backward low freq estimates section}

We prove the following estimate for the low frequency regime:

\begin{proposition}\label{preliminary low frequency estimate proposition}
    For any $0\leq\tau<X2^{-k-1}\leq1$ we have the low frequency regime estimate:
    \[\big\|\tau\nabla P_k\xi\big\|_{L^2}^2+\big\|\tau P_k\nabla_{\tau}\xi\big\|_{L^2}^2\lesssim X^22^{-2k}\Big(\big\|P_k\nabla_{\tau}\xi\big\|^2_{L^2}+\big\|\nabla P_k\xi\big\|^2_{L^2}\Big)\Big|_{\tau=X2^{-k-1}}+C_X2^{-3k}D.\]
\end{proposition}
\begin{proof}
    For any $k\geq x$, we apply $P_k$ to equation (\ref{backward direction linear wave equation main}) to obtain:
    \[\nabla_{\tau}\big(\tau P_k\nabla_{\tau}\xi\big)-4\tau\Delta P_k\xi=\tau P_k\big(\psi\nabla\xi\big)+\tau[\nabla_{4},P_k]\tau\nabla_{\tau}\xi.\]
    We have the energy estimate:
    \[\big\|\tau\nabla P_k\xi\big\|_{L^2}^2+\big\|\tau P_k\nabla_{\tau}\xi\big\|_{L^2}^2\lesssim X^22^{-2k}\Big(\big\|P_k\nabla_{\tau}\xi\big\|^2_{L^2}+\big\|\nabla P_k\xi\big\|^2_{L^2}\Big)\Big|_{X2^{-k-1}}\]\[+\int_{\tau}^{X2^{-k-1}}(\tau')^3\big\|\nabla P_k\xi\big\|_{L^2}\cdot\big\|\nabla [P_k,\nabla_4]\xi\big\|_{L^2}+\int_{\tau}^{X2^{-k-1}}(\tau')^3\big\|\nabla P_k\xi\big\|_{L^2}\cdot\big\|[\nabla,\nabla_4]P_k\xi\big\|_{L^2}\]\[+\int_{\tau}^{X2^{-k-1}}(\tau')^2\big\|P_k\nabla_{\tau}\xi\big\|_{L^2}\cdot\big\|P_k\big(\psi\nabla\xi\big)\big\|_{L^2}+\int_{\tau}^{X2^{-k-1}}(\tau')^2\big\|P_k\nabla_{\tau}\xi\big\|_{L^2}\cdot\big\|[P_k,\nabla_4]\tau'\nabla_{\tau}\xi\big\|_{L^2}\]
    We use Lemma \ref{LP Projections lemma} to obtain:
    \[\big\|\tau\nabla P_k\xi\big\|_{L^2}^2+\big\|\tau P_k\nabla_{\tau}\xi\big\|_{L^2}^2\lesssim X^22^{-2k}\Big(\big\|P_k\nabla_{\tau}\xi\big\|^2_{L^2}+\big\|\nabla P_k\xi\big\|^2_{L^2}\Big)\Big|_{X2^{-k-1}}+\int_{\tau}^{X2^{-k-1}}(\tau')^3\big\|\nabla P_k\xi\big\|_{L^2}\cdot\big\|\xi\big\|_{H^1}\]\[+\int_{\tau}^{X2^{-k-1}}(\tau')^3\big\|\nabla P_k\xi\big\|_{L^2}\cdot\big\|P_k\xi\big\|_{H^1}+\int_{\tau}^{X2^{-k-1}}(\tau')^22^{-k}\big\|P_k\nabla_{\tau}\xi\big\|_{L^2}\cdot\big\|\xi\big\|_{H^1}+\int_{\tau}^{X2^{-k-1}}(\tau')^2\big\|P_k\nabla_{\tau}\xi\big\|_{L^2}\cdot\big\|\tau'\nabla_{\tau}\xi\big\|_{L^2}\]
    \[\lesssim X^22^{-2k}\Big(\big\|P_k\nabla_{\tau}\xi\big\|^2_{L^2}+\big\|\nabla P_k\xi\big\|^2_{L^2}\Big)\Big|_{X2^{-k-1}}+\int_{\tau}^{X2^{-k-1}}\frac{2^k}{X}\big\|\tau'\nabla P_k\xi\big\|_{L^2}^2+\int_{\tau}^{X2^{-k-1}}\frac{2^k}{X}\big\|\tau'P_k\nabla_{\tau}\xi\big\|_{L^2}^2\]
    \[+C_X\int_{\tau}^{X2^{-k-1}}2^{-3k}(\tau')^2\big\|\xi\big\|_{H^1}^2+C_X\int_{\tau}^{X2^{-k-1}}2^{-k}(\tau')^4\big\|\nabla_{\tau}\xi\big\|_{L^2}^2\]
    We use Gronwall for $\tau\leq X2^{-k-1}\leq1$ to obtain:
    \[\big\|\tau\nabla P_k\xi\big\|_{L^2}^2+\big\|\tau P_k\nabla_{\tau}\xi\big\|_{L^2}^2\lesssim X^22^{-2k}\Big(\big\|P_k\nabla_{\tau}\xi\big\|^2_{L^2}+\big\|\nabla P_k\xi\big\|^2_{L^2}\Big)\Big|_{X2^{-k-1}}+\]\[+C_X\int_{\tau}^{X2^{-k-1}}2^{-3k}(\tau')^2\big\|\xi\big\|_{H^1}^2+C_X\int_{\tau}^{X2^{-k-1}}2^{-k}(\tau')^4\big\|\nabla_{\tau}\xi\big\|_{L^2}^2.\]
    Using the preliminary estimate (\ref{backward direction basic estimate one}) we obtain the conclusion.
\end{proof}

\subsection{High Frequency Regime Estimates}\label{high frequency estimates section}

In this section, we prove an optimal high frequency regime estimate for the solution in Proposition~\ref{backward high freq proposition improved}. At first, we have the preliminary high frequency estimate proved in \cite[Section~4]{Cwave}:
\begin{proposition}[{\cite[Proposition~4.5]{Cwave}}]\label{high frequency backward estimate} For any $\tau\in\big[X2^{-k-1},1\big]$ we have that $\xi$ satisfy the estimate:
    \[2^k\tau\big\|P_k\nabla_{\tau}\xi\big\|^2_{L^2}+\frac{2^k}{\tau}\big\|P_k\xi\big\|^2_{L^2}+2^k\tau\big\|\nabla P_k\xi\big\|^2_{L^2}\lesssim\bigg(2^k\big\|P_k\nabla_{\tau}\xi\big\|^2_{L^2}+2^k\big\|P_k\xi\big\|^2_{L^2}+2^k\big\|\nabla P_k\xi\big\|^2_{L^2}\bigg)\bigg|_{\tau=1}\]\[+\int_{\tau}^1\frac{2^k}{(\tau')^2}\big\|P_k\xi\big\|_{L^2}^2+\int_{\tau}^12^k\tau'\big\|\underline{\widetilde{P}}_k\xi\big\|_{L^2}^2+\int_{\tau}^12^k(\tau')^3\big\|\underline{\widetilde{P}}_k\nabla\xi\big\|_{L^2}^2+\int_{\tau}^12^k(\tau')^3\big\|\underline{\widetilde{P}}_k\nabla_{\tau}\xi\big\|_{L^2}^2+\]\[+\int_{\tau}^1\frac{\tau'}{2^{k}}\big\|\xi\big\|_{H^1}^2+\int_{\tau}^1\frac{(\tau')^3}{2^{k}}\big\|\nabla_{\tau}\xi\big\|_{L^2}^2.\]
\end{proposition}
\begin{proof}
    The proof follows the same steps as Proposition \ref{high frequency forward estimate}, the only difference being the presence of the first bulk term which has a bad sign in this case. Moreover, we point out that the implicit constant is independent of the parameter $X.$ We refer the reader to \cite[Section~4]{Cwave} for the complete proof.
\end{proof}

\textbf{Notation.} We introduce the following notation: let $\mathbf{1}_{k,\tau}$ be the characteristic function of $\{1\geq\tau\geq X2^{-k-1}\};$ also we denote by $d_k$ any data terms at $\tau=1$ which satisfy $\sum d_k\lesssim D$; finally we denote:
\[a_k(\tau):=\tau\big\|P_k\nabla_{\tau}\xi\big\|^2_{L^2}+\frac{1}{\tau}\big\|P_k\xi\big\|^2_{L^2}+\tau\big\|\nabla P_k\xi\big\|^2_{L^2}.\]

For the rest of the section we prove the following result:
\begin{proposition}\label{backward high freq proposition improved}
    For any $\tau\in(0,1]$ we have the estimate:
    \begin{equation}\label{high freq prelim improved estimate}
    \sum_{\tau\geq X2^{-k-1}}2^ka_k(\tau)\mathbf{1}_{k,\tau}\lesssim D+\int_{\tau}^1(\tau')^3\Big(\big\|\xi\big\|_{H^{3/2}}^2+\big\|\nabla_{\tau}\xi\big\|_{H^{1/2}}^2\Big)d\tau'.
\end{equation}
\end{proposition}
We outline in more detail the plan of the proof that was also explained in Section~\ref{scattering map section intro}. To prove the optimal estimate in Proposition~\ref{backward high freq proposition improved}, we start with the preliminary estimate in Proposition~\ref{high frequency backward estimate}. The main challenge is caused by the first bulk error term, which was absent in the analogous estimate of Section~\ref{model forward direction section}. We can only bound this term using the refined Poincaré inequality \eqref{refined Poincare inequality}. This creates sums of low frequency regime and high frequency regime error terms. We bound the high frequency regime error terms using Gronwall inequality and the Gronwall-type Lemma~\ref{grownall like lemma}. We also use the estimates of Proposition~\ref{preliminary low frequency estimate proposition} to bound the low frequency regime error terms. This creates a sum of discrete error terms, which we bound using the discrete Gronwall inequality. Finally, the remaining error terms coming from the other terms in Proposition~\ref{high frequency backward estimate} can be summed to obtain the RHS of \eqref{high freq prelim improved estimate}.

\textit{Step 1. Using the preliminary estimate.} We rewrite the result of Proposition \ref{high frequency backward estimate}, for all $\tau\in[X2^{-k-1},1]$:
\begin{equation}\label{schematic inequality}
    2^ka_k(\tau)\mathbf{1}_{k,\tau}\lesssim2^ka_k(1)+ \mathbf{1}_{k,\tau}\int_{\tau}^1\frac{2^k}{(\tau')^2}\big\|P_k\xi\big\|_{L^2}^2+\mathbf{1}_{k,\tau}\int_{\tau}^1e_k(\tau')+\mathbf{1}_{k,\tau}\int_{\tau}^1\overline{e_k}(\tau'),
\end{equation}
where we denoted:
\[e_k(\tau)=2^k\tau^3\big\|\underline{\widetilde{P}}_k\nabla\xi\big\|_{L^2}^2+2^k\tau^3\big\|\underline{\widetilde{P}}_k\nabla_{\tau}\xi\big\|_{L^2}^2\]
\[\overline{e_k}(\tau)=2^k\tau\big\|\underline{\widetilde{P}}_k\xi\big\|_{L^2}^2+\frac{\tau}{2^{k}}\big\|\nabla\xi\big\|_{L^2}^2+\frac{\tau}{2^{k}}\big\|\xi\big\|_{L^2}^2+\frac{\tau^3}{2^{k}}\big\|\nabla_{\tau}\xi\big\|_{L^2}^2.\]
Using the bound $2^k\tau\big\|\underline{\widetilde{P}}_k\xi\big\|_{L^2}^2\lesssim2^{-k}\tau\big\|\xi\big\|_{H^1}^2$ and the preliminary estimate (\ref{backward direction basic estimate one}), we get:
\[\int_{\tau}^1\overline{e_k}(\tau')d\tau'\lesssim_X d_k.\]
So far, we obtain that for all $\tau\in[X2^{-k-1},1]$:
\begin{equation}\label{estimate in step 1}
    2^ka_k(\tau)\mathbf{1}_{k,\tau}\lesssim C_Xd_k+ \mathbf{1}_{k,\tau}\int_{\tau}^1\frac{2^k}{(\tau')^2}\big\|P_k\xi\big\|_{L^2}^2+\mathbf{1}_{k,\tau}\int_{\tau}^1e_k(\tau').
\end{equation}

\textit{Step 2. Applying the refined Poincarè inequality.} The last term in \eqref{estimate in step 1} is similar to the commutation error terms that we encountered in Section~\ref{model forward direction section} by the presence of the LP projection operator $\underline{\widetilde{P}}_k.$ We leave this term unchanged until the end of the argument when we sum all our estimates. On the other hand, as explained above, the second term is at top order and difficult to deal with. In order to bound it, we use the refined Poincaré inequality (\ref{refined Poincare inequality}) to get that for any $\delta>0$ and $\tau\in[X2^{-k-1},1]$:
\[\int_{\tau}^1\frac{2^k}{(\tau')^2}\big\|P_k\xi\big\|_{L^2}^2\lesssim \frac{1}{\delta}\int_{\tau}^1\frac{2^{-k}}{(\tau')^2}\big\|\nabla P_k\xi\big\|_{L^2}^2+\delta\int_{\tau}^1\frac{1}{(\tau')^2}\sum_{l=0}^{k-1}2^{-8k+7l}\big\|\nabla P_l\xi\big\|_{L^2}^2+\frac{1}{\delta}\int_{\tau}^1\frac{2^{-3k}}{(\tau')^2}\big\|\xi\big\|_{L^2}^2\]
The last term is bounded by $\delta^{-1}d_k$ using (\ref{backward direction basic estimate one}). We obtain that there exist constants $C_0,C_{X},C_{X,\delta}>0$ such that for all $k\geq x$ and $\tau\in[X2^{-k-1},1]$:
\[2^ka_k(\tau)\leq C_{X,\delta}d_k+\frac{C_0}{\delta}\int_{\tau}^1\frac{2^{-2k}}{(\tau')^3}\cdot2^ka_k(\tau')d\tau'+C_0\delta\int_{\tau}^1\frac{1}{(\tau')^2}\sum_{l=0}^{k-1}2^{-8k+7l}\big\|\nabla P_l\xi\big\|_{L^2}^2+C_{X,\delta}\int_{\tau}^1e_k(\tau').\]

\textit{Step 3. Gronwall inequality for high frequency regime terms.} We apply Gronwall for $\tau\in[X2^{-k-1},1]$, in order to get rid of the second term in the above estimate. The multiplication factor is:
\[\exp\bigg(\frac{C_0}{\delta}\int_{\tau}^1\frac{2^{-2k}}{(\tau')^3}d\tau'\bigg)\leq\exp\bigg(\frac{C_0}{\delta}\cdot\frac{1}{X^2}\bigg)\leq2,\]
where we fix $X$ large enough, depending on $C_0$ and $\delta$, such that:
\begin{equation}\label{fix X}
    \frac{C_0}{\delta}\cdot\frac{1}{X^2}\leq\log(2).
\end{equation}
We notice that we fix $\delta>0$ later, which will determine the value of $X$ as well. We have for all $\tau\in[X2^{-k-1},1]$:
\begin{equation}\label{some equation that I need to reference}
    2^ka_k(\tau)\leq C_{\delta}d_k+2C_0\delta\int_{\tau}^1\frac{1}{(\tau')^2}\sum_{l=0}^{k-1}2^{-8k+7l}\big\|\nabla P_l\xi\big\|_{L^2}^2d\tau'+C_{\delta}\int_{\tau}^1e_k(\tau')d\tau'.
\end{equation}

\textit{Step 4. Using the low frequency regime estimate.} The second term in \eqref{some equation that I need to reference} can be split into its low frequency regime and high frequency regime components. We want to bound its low frequency regime part next. We have:
\[\mathbf{1}_{k,\tau}\int_{\tau}^1\frac{1}{(\tau')^3}\sum_{l=0}^{k-1}2^{-8k+7l}\tau'\big\|\nabla P_l\xi\big\|_{L^2}^2\leq\mathbf{1}_{k,\tau}\int_{\tau}^1\frac{1}{(\tau')^4}\sum_{l=0}^{x-1}2^{-8k+7l}\big\|\tau'\nabla P_l\xi\big\|_{L^2}^2\]\[+\mathbf{1}_{k,\tau}\int_{\tau}^1\frac{1}{(\tau')^3}\sum_{l=x}^{k-1}2^{-8k+6l}\cdot 2^la_l(\tau')\mathbf{1}_{l,\tau'}+\mathbf{1}_{k,\tau}\int_{\tau}^1\frac{1}{(\tau')^4}\sum_{\tau'<X2^{-l-1}\leq1}2^{-8k+7l}\big\|\tau'\nabla P_l\xi\big\|_{L^2}^2.\]
In this inequality, the first and third terms on the RHS are in the low frequency regime. We use (\ref{backward direction basic estimate one}) to bound the first term by $C_Xd_k.$ For the third term we use Proposition \ref{preliminary low frequency estimate proposition} to obtain:
\[\mathbf{1}_{k,\tau}\int_{\tau}^1\frac{1}{(\tau')^4}\sum_{\tau'<X2^{-l-1}\leq1}2^{-8k+7l}\big\|\tau'\nabla P_l\xi\big\|_{L^2}^2=\mathbf{1}_{k,\tau}\sum_{\tau<X2^{-l-1}\leq1}\int_{\tau}^{X2^{-l-1}}\frac{1}{(\tau')^4}2^{-8k+7l}\big\|\tau'\nabla P_l\xi\big\|_{L^2}^2\]
\[\lesssim \mathbf{1}_{k,\tau}X\sum_{\tau<X2^{-l-1}\leq1}\frac{2^{-8k+6l}}{\tau^3}a_l(X2^{-l-1})+C_XD\cdot\mathbf{1}_{k,\tau}\frac{2^{-8k}}{\tau^7}\]
Using (\ref{fix X}), we get from \eqref{some equation that I need to reference} that there exist constants $C_1,C_{\delta}>0$ such that for all $\tau\in[X2^{-k-1},1]$:
\[2^ka_k(\tau)\leq C_{\delta}d_k+C_1\delta\int_{\tau}^1\frac{1}{(\tau')^3}\sum_{l=x}^{k-1}2^{-8k+6l}\cdot 2^la_l(\tau')\mathbf{1}_{l,\tau'}+C_1\delta X\sum_{\tau<X2^{-l-1}\leq1}\frac{2^{-8k+6l}}{\tau^3}a_l(X2^{-l-1})+C_{\delta}\int_{\tau}^1e_k(\tau').\]

\textit{Step 5. Using Lemma~\ref{grownall like lemma} for high frequency regime terms.} In the above estimate, we also have error terms that are discrete or continuous-discrete, unlike Section~\ref{model forward direction section}. We bound them using appropriate Gronwall-like inequalities. We will use Lemma~\ref{grownall like lemma} for the second term which is in the high frequency regime, and the discrete Gronwall lemma for the third term. For $\delta>0$ small enough, we have that:
\begin{equation}\label{fix delta part 1}
    C_1\delta<\frac{1}{10}.
\end{equation}
We introduce the following notation for the terms on the right hand side of the above inequality:
\[S_k(\tau)=\mathbf{1}_{k,\tau}C_1\delta X\sum_{\tau<X2^{-l-1}\leq1}\frac{2^{-8k+6l}}{\tau^3}a_l(X2^{-l-1}),\ E_k(\tau)=\mathbf{1}_{k,\tau}C_{\delta}\int_{\tau}^1e_k(\tau')d\tau',\ A_k(\tau)=C_{\delta}d_k+S_k(\tau)+E_k(\tau)\]

Thus, we proved that for all $k\geq x$ and $\tau\in[X2^{-k-1},1]$:
\begin{equation}\label{schematic equation 2}
    2^ka_k(\tau)\mathbf{1}_{k,\tau}\leq A_k(\tau)+\frac{1}{10}\cdot\mathbf{1}_{k,\tau}\int_{\tau}^1\frac{1}{(\tau')^3}\sum_{l=x}^{k-1}2^{-8k+6l}\cdot 2^la_l(\tau')\mathbf{1}_{l,\tau'}d\tau'.
\end{equation}
To deal with this estimate, we use the following Gronwall type lemma proved in \cite[Section~4]{Cwave}:
\begin{lemma}[{\cite[Lemma~4.1]{Cwave}}]\label{grownall like lemma}
    We consider the functions $u,A,b,c:\mathbb{N}\times[0,1]\rightarrow[0,\infty),$ satisfying for all $k\geq x,\ \tau\in(0,1]$:
    \[u(k,\tau)\leq A(k,\tau)+b(k)\int_{\tau}^1\sum_{l=x}^{k-1}c(l,\tau')u(l,\tau')d\tau'.\]
    Then, we have that for all $k\geq x,\ \tau\in(0,1]$:
    \[u(k,\tau)\leq A(k,\tau)+b(k)\int_{\tau}^1\sum_{l=x}^{k-1}c(l,\tau')A(l,\tau')\prod_{j=l+1}^{k-1}\bigg(1+\int_{\tau}^{\tau'}b(j)c(j,\tau'')d\tau''\bigg)d\tau'.\]
\end{lemma}
The proof of the following result is also contained in \cite[Section~4]{Cwave}. We reproduce it here for completeness:
\begin{corollary}[{\cite[Corollary~4.1]{Cwave}}]
    For all $k\geq x$ and $\tau\in[X2^{-k-1},1]$, we have the bound:
    \begin{equation}\label{schematic equation 3}
        2^ka_k(\tau)\mathbf{1}_{k,\tau}\lesssim A_k(\tau)+2^{-7k}\int_{\tau}^1\sum_{l=x}^{k-1}\frac{2^{5l}}{(\tau')^3}\cdot A_l(\tau')\mathbf{1}_{l,\tau'}d\tau'.
    \end{equation}
\end{corollary}
\begin{proof}
    We apply the lemma for (\ref{schematic equation 2}), with $u(k,\tau)=2^ka_k(\tau)\mathbf{1}_{k,\tau},b(k)=\frac{1}{10}\cdot2^{-8k},c(k,\tau)=\tau^{-3}2^{6k}\mathbf{1}_{k,\tau}$. As a result, we get for $\tau\in[X2^{-k-1},1]$:
    \[2^ka_k(\tau)\mathbf{1}_{k,\tau}\leq A_k(\tau)+2^{-8k}\int_{\tau}^1\sum_{l=x}^{k-1}A_l(\tau')\cdot\frac{1}{(\tau')^{3}}2^{6l}\mathbf{1}_{l,\tau'}\prod_{j=l+1}^{k-1}\bigg(1+\int_{\tau}^{\tau'}b(j)c(j,\tau'')d\tau''\bigg)d\tau'\]
    We compute:
    \[\int_{\tau}^{\tau'}b(j)c(j,\tau'')d\tau''=\frac{1}{10}\cdot2^{-2j}\int_{\tau}^{\tau'}\frac{\mathbf{1}_{j,\tau''}}{(\tau'')^3}d\tau''\leq\min\big(1,2^{-2j}\tau^{-2}\big).\]
    In the next bound it is essential that we had good control of the constants in the previous inequality, which was ensured using the smallness of $\delta.$ As a result, we obtain using the inequality $x+1\leq e^x$ for the terms with $2^{-j}<\tau$:
    \[\prod_{j=l+1}^{k-1}\bigg(1+\int_{\tau}^{\tau'}b(j)c(j,\tau'')d\tau''\bigg)\leq\prod_{j=l+1}^{k-1}\Big(\min\big(2,1+2^{-2j}\tau^{-2}\big)\Big)\lesssim1+2^{-l}\tau^{-1}.\]
    Finally, we notice that $2^{-8k}2^{6l}\cdot2^{-l}\tau^{-1}\leq2^{-7k}2^{5l}\cdot2^{-k}/\tau\leq2^{-7k}2^{5l}\cdot2/X\leq2^{-7k}2^{5l}.$
\end{proof}

\textit{Step 6. Rewriting estimate \eqref{schematic equation 3}.} We use the definition of $A_k(\tau)$ in order to bound the remaining terms on the RHS of (\ref{schematic equation 3}). We compute for the data terms:
\[d_k+2^{-7k}\int_{\tau}^1\sum_{l=x}^{k-1}\frac{2^{5l}}{(\tau')^3}d_l\mathbf{1}_{l,\tau'}d\tau'\lesssim_{\delta} d_k+\sum_{l=x}^{k-1}2^{5(l-k)}d_l\lesssim_{\delta} d_k,\]
where we notice that $\sum_k\sum_l2^{-5|k-l|}d_l\lesssim D,$ so we can write the term $\sum_l2^{-5|k-l|}d_l$ schematically as $d_k.$ Next, we compute for the discrete error terms the following bound:
\[S_k(\tau)+2^{-7k}\int_{\tau}^1\sum_{l=x}^{k-1}\frac{2^{5l}}{(\tau')^3}S_l(\tau')\mathbf{1}_{l,\tau'}d\tau'\lesssim S_k(\tau)+\delta X2^{-7k}\int_{\tau}^1\sum_{l=x}^{k-1}\sum_{\tau'<X2^{-m-1}\leq1}\frac{2^{-3l+6m}}{(\tau')^6}a_m(X2^{-m-1})\mathbf{1}_{l,\tau'}d\tau'\]
\begin{align*}
    &\lesssim S_k(\tau)+\delta X2^{-7k}\sum_{l=x}^{k-1}\sum_{\substack{\tau<X2^{-m-1} \\ x\leq m\leq l}}\int_{X2^{-l-1}}^{X2^{-m-1}}\frac{2^{-3l+6m}}{(\tau')^6}a_m(X2^{-m-1})d\tau'\\
    &\lesssim S_k(\tau)+\frac{\delta}{X^4}2^{-7k}\sum_{l=x}^{k-1}\sum_{\substack{\tau<X2^{-m-1} \\ x\leq m\leq l}}2^{2l+6m}a_m(X2^{-m-1})\lesssim\frac{\delta}{X^2}2^{-5k}\sum_{\tau<X2^{-m-1}\leq1}2^{6m}a_m(X2^{-m-1}).
\end{align*}
We also introduce the notation:
\[\widetilde{S}_k(\tau):=\frac{\delta}{X^2}2^{-5k}\sum_{\tau<X2^{-m-1}\leq1}2^{6m}a_m(X2^{-m-1}),\ \widetilde{E}_k(\tau):= E_k(\tau)+2^{-7k}\int_{\tau}^1\sum_{l=x}^{k-1}\frac{2^{5l}}{(\tau')^3}E_l(\tau')\mathbf{1}_{l,\tau'}d\tau'.\]
As a result, we get from (\ref{schematic equation 3}) that for all $\tau\in[X2^{-k-1},1]$:
\begin{equation}\label{intermediate estimate to be used later}
    2^ka_k(\tau)\mathbf{1}_{k,\tau}\lesssim C_{\delta}d_k+\widetilde{S}_k(\tau)+C_{\delta}\widetilde{E}_k(\tau).
\end{equation}

\textit{Step 7. Using the discrete Gronwall inequality.} Our next goal is to bound the discrete error terms from the above estimate using the discrete Gronwall inequality. We get that there exist constants $C_2,C_{\delta}>0$ such that for any $k\geq x:$
\[2^ka_k(X2^{-k-1})\leq C_{\delta}d_k+C_{\delta}\widetilde{E}_k(X2^{-k-1})+C_2\frac{\delta}{X^2}2^{-5k}\sum_{m=x}^{k-1}2^{6m}a_m(X2^{-m-1}).\]
To fix $X,\delta>0$ in addition to (\ref{fix X}) and (\ref{fix delta part 1}), we also ask that:
\begin{equation}\label{fix delta part 2}
    C_2\frac{\delta}{X^2}<\frac{1}{10}.
\end{equation}
Denoting $b_k=C_{\delta}d_k+C_{\delta}\widetilde{E}_k(X2^{-k-1}),$ the above inequality can be written as:
\[2^ka_k(X2^{-k-1})\leq b_k+\frac{1}{10}2^{-5k}\sum_{m=x}^{k-1}2^{5m}\cdot 2^ma_m(X2^{-m-1}).\]
We apply the discrete Gronwall inequality of \cite{gronwall}:
\[2^ka_k(X2^{-k-1})\leq b_k+\frac{1}{10}2^{-5k}\sum_{m=x}^{k-1}\bigg(2^{5m}b_m\cdot\prod_{j=m+1}^{k-1}\big(1+1/10\cdot2^{-5j}\cdot2^{5j}\big)\bigg)\lesssim b_k+2^{-4k}\sum_{m=x}^{k-1}2^{4m}b_m.\]
Therefore, we obtain that for all $k\geq x$:
\begin{equation}\label{intermediate equation name sth}
    2^ka_k(X2^{-k-1})\lesssim_{\delta} d_k+\widetilde{E}_k(X2^{-k-1})+\sum_{m=x}^{k-1}2^{-4k+4m}\widetilde{E}_m(X2^{-m-1}).
\end{equation}
\begin{remark}
    The constants $X>0$ and $\delta>0$ which satisfy (\ref{fix X}), (\ref{fix delta part 1}), and (\ref{fix delta part 2}), were introduced in order to allow us to apply Gronwall in the above inequalities. Once we established (\ref{intermediate equation name sth}), we do not need to track these constants anymore. We make the notation convention that from here onwards, whenever we write $\lesssim$ the implicit constant is allowed to depend on $X$ and $\delta$ as well.
\end{remark}

\textit{Step 8. Rewriting estimate \eqref{intermediate estimate to be used later}.} We use the bound \eqref{intermediate equation name sth} to get for all $\tau\in[X2^{-k-1},1]$:
\[\widetilde{S}_k(\tau)\lesssim\sum_{\tau<X2^{-m-1}\leq1}\Big(2^{-5k+5m}d_m+2^{-4k+4m}\widetilde{E}_m(X2^{-m-1})\Big)\lesssim\sum_{\tau<X2^{-m-1}\leq1}\Big(2^{-5k+5m}d_m+2^{-4k+4m}E_m(X2^{-m-1})\Big)\]
where the second inequality follows by:
\[\sum_{\tau<X2^{-m-1}\leq1}2^{-4k+4m}\cdot2^{-7m}\sum_{l=x}^{m-1}\int_{X2^{-l-1}}^1\frac{2^{5l}}{(\tau')^3}E_l(\tau')d\tau'\lesssim\sum_{\tau<X2^{-l-1}\leq1}2^{-4k+4l}E_l(X2^{-l-1})\]
In addition, we have the bound for all $\tau\in[X2^{-k-1},1]$:
\begin{align*}
    \widetilde{E}_k(\tau)&=E_k(\tau)+2^{-7k}\int_{\tau}^1\sum_{l=x}^{k-1}\frac{2^{5l}}{(\tau')^3}E_l(\tau')\mathbf{1}_{l,\tau'}d\tau'\lesssim E_k(\tau)+2^{-2k}\tau^{-2}\sum_{l\geq x}\int_{\tau}^1e_l(\tau')d\tau'\\
    &\lesssim E_k(\tau)+2^{-2k}\tau^{-2}\int_{\tau}^1(\tau')^3\Big(\big\|\xi\big\|_{H^{3/2}}^2+\big\|\nabla_{\tau}\xi\big\|_{H^{1/2}}^2\Big)d\tau'
\end{align*}
We get from (\ref{intermediate estimate to be used later}) that for all $\tau\in[X2^{-k-1},1]$:
\begin{align*}
    2^ka_k(\tau)\mathbf{1}_{k,\tau}&\lesssim d_k+\sum_{\tau<X2^{-m-1}\leq1}2^{-5k+5m}d_m+\sum_{\tau<X2^{-m-1}\leq1}2^{-4k+4m}E_m(X2^{-m-1})\\
    &+E_k(\tau)+2^{-2k}\tau^{-2}\int_{\tau}^1(\tau')^3\Big(\big\|\xi\big\|_{H^{3/2}}^2+\big\|\nabla_{\tau}\xi\big\|_{H^{1/2}}^2\Big)d\tau'.
\end{align*}

\textit{Step 9. Summing the high frequency regime estimates.} For each $\tau\in(0,1]$ we sum for all $k\geq x$ such that $X2^{-k-1}\leq\tau:$
\[\sum_{\tau\geq X2^{-k-1}}2^ka_k(\tau)\lesssim D+\sum_{\tau\geq X2^{-k-1}}E_k(\tau)+\sum_{\tau<X2^{-m-1}\leq1}E_m(X2^{-m-1})+\int_{\tau}^1(\tau')^3\Big(\big\|\xi\big\|_{H^{3/2}}^2+\big\|\nabla_{\tau}\xi\big\|_{H^{1/2}}^2\Big)d\tau'\]
We also have the estimate:
\[\sum_{\tau\geq X2^{-k-1}}E_k(\tau)+\sum_{\tau<X2^{-m-1}\leq1}E_m(X2^{-m-1})\lesssim\sum_{k\geq x}\int_{\tau}^1e_k(\tau')\lesssim\int_{\tau}^1(\tau')^3\Big(\big\|\xi\big\|_{H^{3/2}}^2+\big\|\nabla_{\tau}\xi\big\|_{H^{1/2}}^2\Big)d\tau'\]
In conclusion, we obtain the improved high frequency regime estimate \eqref{high freq prelim improved estimate}:
\[\sum_{\tau\geq X2^{-k-1}}2^ka_k(\tau)\lesssim D+\int_{\tau}^1(\tau')^3\Big(\big\|\xi\big\|_{H^{3/2}}^2+\big\|\nabla_{\tau}\xi\big\|_{H^{1/2}}^2\Big)d\tau'.\]

\subsection{Proof of the Main Result in Theorem \ref{backward direction main result theorem}}\label{backward main result estimates section}
In this section we use the previous estimates in order to prove (\ref{backward direction main estimate one}), which implies the desired bounds for the solution in Theorem \ref{backward direction main result theorem}. We use the preliminary estimates (\ref{backward direction basic estimate one}) to get for all $\tau\in(0,1]$:
\begin{align*}
&\tau^3\Big(\big\|\xi\big\|_{H^{3/2}}^2+\big\|\nabla_{\tau}\xi\big\|_{H^{1/2}}^2\Big)\lesssim\tau^3\Big(\big\|\xi\big\|_{H^{1}}^2+\big\|\nabla_{\tau}\xi\big\|_{L^2}^2\Big)+\sum_{k\geq x}\tau^32^k\Big(\big\|\nabla P_k\xi\big\|_{L^2}^2+\big\|P_k\nabla_{\tau}\xi\big\|_{L^2}^2\Big)\\
    &\lesssim\tau^3\Big(\big\|\xi\big\|_{H^{1}}^2+\big\|\nabla_{\tau}\xi\big\|_{L^2}^2\Big)+\sum_{\tau<X2^{-k-1}\leq1}\tau^32^k\Big(\big\|\nabla P_k\xi\big\|_{L^2}^2+\big\|P_k\nabla_{\tau}\xi\big\|_{L^2}^2\Big)+\tau^2\sum_{\tau\geq X2^{-k-1}}2^ka_k(\tau)\\
    &\lesssim\tau^2\Big(\big\|\xi\big\|_{H^{1}}^2+\big\|\nabla_{\tau}\xi\big\|_{L^2}^2\Big)+\tau^2\sum_{\tau\geq X2^{-k-1}}2^ka_k(\tau)\lesssim D+\tau^2\sum_{\tau\geq X2^{-k-1}}2^ka_k(\tau)
\end{align*}
Using the improved high frequency regime estimate (\ref{high freq prelim improved estimate}), we get:
\[\tau^3\Big(\big\|\xi\big\|_{H^{3/2}}^2+\big\|\nabla_{\tau}\xi\big\|_{H^{1/2}}^2\Big)\lesssim D+\int_{\tau}^1(\tau')^3\Big(\big\|\xi\big\|_{H^{3/2}}^2+\big\|\nabla_{\tau}\xi\big\|_{H^{1/2}}^2\Big)d\tau'\]
Using Gronwall, we proved that for all $\tau\in(0,1]$:
\begin{equation}\label{optimal high frequency estimate}
    \tau^3\Big(\big\|\xi\big\|_{H^{3/2}}^2+\big\|\nabla_{\tau}\xi\big\|_{H^{1/2}}^2\Big)+\sum_{\tau\geq X2^{-k-1}}2^ka_k(\tau)\lesssim D.
\end{equation}

As a result, we use this and Proposition \ref{preliminary low frequency estimate proposition} to get the bound for all $\tau\in(0,1]$:
\[\tau^2\Big(\big\|\xi\big\|_{H^{3/2}}^2+\big\|\nabla_{\tau}\xi\big\|_{H^{1/2}}^2\Big)\lesssim\tau^2\Big(\big\|\xi\big\|_{H^{1}}^2+\big\|\nabla_{\tau}\xi\big\|_{L^2}^2\Big)+\sum_{k\geq x}\tau^22^k\Big(\big\|\nabla P_k\xi\big\|_{L^2}^2+\big\|P_k\nabla_{\tau}\xi\big\|_{L^2}^2\Big)\]\[\lesssim D+\sum_{\tau<X2^{-k-1}\leq1}\tau^22^{2k}\Big(\big\|\nabla P_k\xi\big\|_{L^2}^2+\big\|P_k\nabla_{\tau}\xi\big\|_{L^2}^2\Big)+\tau\sum_{\tau\geq X2^{-k-1}}2^ka_k(\tau)\lesssim D+\sum_{\tau<X2^{-k-1}\leq1}2^ka_k(X2^{-k-1})\]
Note that using our previous bound (\ref{intermediate equation name sth}) for $2^ka_k(X2^{-k-1})$ we get:
\[\sum_{\tau<X2^{-k-1}\leq1}2^ka_k(X2^{-k-1})\lesssim D+\sum_{\tau<X2^{-k-1}\leq1}\widetilde{E}_k(X2^{-k-1})\lesssim D+\sum_{\tau<X2^{-k-1}\leq1}E_k(X2^{-k-1})\lesssim D\]
where we also used (\ref{optimal high frequency estimate}) in the last inequality. We proved that:
\[\tau^2\Big(\big\|\xi\big\|_{H^{3/2}}^2+\big\|\nabla_{\tau}\xi\big\|_{H^{1/2}}^2\Big)\lesssim D.\]
We also get for all $\tau\in(0,1]$:
\[\tau\big\|\xi\big\|_{H^{1/2}}^2\lesssim D+\int_{\tau}^1\tau'\big\|\xi\big\|_{H^{1/2}}\cdot\big\|\nabla_{\tau}\xi\big\|_{H^{1/2}}\lesssim D+\int_{\tau}^1\sqrt{\tau'}\big\|\xi\big\|_{H^{1/2}}^2+\int_{\tau}^1(\tau')^{3/2}\big\|\nabla_{\tau}\xi\big\|_{H^{1/2}}^2\lesssim D+\int_{\tau}^1\sqrt{\tau'}\big\|\xi\big\|_{H^{1/2}}^2\]
Applying Gronwall, we complete the proof of (\ref{backward direction main estimate one}).

\subsection{Estimates for the Asymptotic Quantities}\label{backward asymptotic quantities estimates section}
In this section we complete the proof of Theorem \ref{backward direction main result theorem} by proving estimates for the asymptotic quantities $\mathcal{O}$ and $\mathfrak{h}.$ We first prove the following result for $\mathcal{O}:$
\begin{proposition}
    We have the estimate $\big\|\nabla^M\mathcal{O}\big\|_{H^{1}}^2\lesssim D$.
\end{proposition}
\begin{proof}
    From (\ref{backward direction basic estimate one}), we have $\big\|\nabla^M\mathcal{O}\big\|_{L^2}^2\lesssim D$. We recall that for any $k\geq x$ and $0\leq\tau<X2^{-k-1}\leq1$, we have by the low frequency regime estimate in Proposition \ref{preliminary low frequency estimate proposition}:
    \[2^{2k}\big\|\tau P_k\nabla_{\tau}\xi\big\|_{L^2}^2\lesssim2^ka_k(X2^{-k-1})+2^{-k}D.\]
    As a result, we get that for all $k\geq x:$
    \[2^{2k}\big\|P_k\nabla^M\mathcal{O}\big\|_{L^2}^2\lesssim2^ka_k(X2^{-k-1})+2^{-k}D.\]
    Using the above estimate for $\big\|\nabla^{M}\mathcal{O}\big\|_{L^2}^2,$ and the estimates in the previous section, we have that:
    \[\big\|\nabla^{M}\mathcal{O}\big\|_{H^1}^2\lesssim D+\sum_{k\geq x}2^ka_k(X2^{-k-1})\lesssim D.\]
\end{proof}

The next goal is to derive estimates for $\mathfrak{h}_M.$ We consider the renormalized quantities for all $k\geq x$:
\[\overline{\xi}_k=P_k\xi-\tau\log(2^k\tau)P_k\nabla_{\tau}\xi.\]
We notice that using the above expansions for $\xi$ we get:
\[\lim_{\tau\rightarrow0}\overline{\xi}_k=\lim_{\tau\rightarrow0}\Big(P_k\xi-\tau\log(2^k\tau)P_k\nabla_{\tau}\xi\Big)=P_k\nabla^Mh-2\log(2^k)P_k\nabla^M\mathcal{O}=P_k\mathfrak{h}_M+R_k\nabla^M\mathcal{O}\]
We obtain using also (\ref{Rk 3}):
\[2^{2k}\big\|P_k\mathfrak{h}_M\big\|_{L^2}^2\lesssim2^{2k}\lim_{\tau\rightarrow0}\big\|\overline{\xi}_k\big\|_{L^2}^2+2^{2k}\big\|R_k\nabla^M\mathcal{O}\big\|_{L^2}^2\lesssim2^{2k}\lim_{\tau\rightarrow0}\big\|\overline{\xi}_k\big\|_{L^2}^2+\big\|\underline{P}_k\nabla^M\mathcal{O}\big\|_{H^1}^2\]
Summing for all $k\geq x$, we get that:
\[\sum_{k\geq x}2^{2k}\big\|P_k\mathfrak{h}_M\big\|_{L^2}^2\lesssim\big\|\nabla^M\mathcal{O}\big\|_{H^{1}}^2+\sum_{k\geq x}2^{2k}\lim_{\tau\rightarrow0}\big\|\overline{\xi}_k\big\|_{L^2}^2.\]
To complete the proof of (\ref{backward direction asymptotic quantity estimate}) we need to bound the second term by $D.$ We outline the proof of the following result, and refer the reader to \cite[Section~4]{Cwave} for a complete proof including the estimates of all the error terms:
\begin{proposition}
    We have the estimate $\sum_{k\geq x}\Big(2^{2k}\lim_{\tau\rightarrow0}\big\|\overline{\xi}_k\big\|_{L^2}^2\Big)\lesssim D$.
\end{proposition}
\begin{proof}
Using the equation satisfied by $\overline{\xi}_k$ and the preliminary estimate (\ref{backward direction basic estimate one}) we get:
\begin{equation}\label{inequality that gives 4 terms to deal with}
    2^{2k}\big\|\overline{\xi}_k\big\|_{L^2}^2\lesssim d_k+\int_{\tau}^{X2^{-k-1}}2^k(\tau')^2|\log(2^k\tau')|^2\big\|\Delta P_k\xi\big\|_{L^2}^2\lesssim d_k+\int_{\tau}^{X2^{-k-1}}2^{3k}(2^k\tau')^2|\log(2^k\tau')|^2\big\|\widetilde{P}_k\xi\big\|_{L^2}^2
\end{equation}
where $P_k=\widetilde{P}_k^2.$ The strategy is to decompose the RHS of \eqref{inequality that gives 4 terms to deal with} into its low frequency and high frequency regime parts in order to use our previous estimates. We bound the error term for each $\tau'\in[\tau,X2^{-k-1}]$ as follows:
\[\big\|\widetilde{P}_k\xi\big\|_{L^2}(\tau')\lesssim\sum_{l<k}\big\|P_l^2\widetilde{P}_k\xi\big\|_{L^2}(\tau')+\sum_{\substack{l\geq k \\ \tau'<X2^{-l-1}}}\big\|P_l^2\widetilde{P}_k\xi\big\|_{L^2}(\tau')+\sum_{l\geq k}\mathbf{1}_{l,\tau'}\big\|P_l^2\widetilde{P}_k\xi\big\|_{L^2}(\tau')\]
Using this in (\ref{inequality that gives 4 terms to deal with}), we obtain three error terms. The first one is in the low frequency regime, and it is bounded by $d_k$ using (\ref{backward direction basic estimate one}) and Proposition \ref{preliminary low frequency estimate proposition}. The second term is also in the low frequency regime. We use Cauchy-Schwarz and we consider a projection operator such that $\widetilde{P}_k=\underline{\widetilde{P}}_k^2$:
\[\bigg(\sum_{\substack{l\geq k \\ \tau'<X2^{-l-1}}}\big\|P_l^2\underline{\widetilde{P}}_k^2\xi\big\|_{L^2}(\tau')\bigg)^2\lesssim\sum_{\substack{l\geq k \\ \tau'<X2^{-l-1}}}2^{-5(l-k)}\big\|\underline{\widetilde{P}}_kP_l\xi\big\|_{L^2}^2(\tau')\lesssim\sum_{\substack{l\geq k \\ \tau'<X2^{-l-1}}}2^{-2k}2^{-5(l-k)}\big\|\nabla P_l\xi\big\|_{L^2}^2(\tau')\]
\[\lesssim\frac{1}{(\tau')^{2}}\sum_{\substack{l\geq k \\ \tau'<X2^{-l-1}}}2^{-2k}2^{-5(l-k)}\bigg(2^{-2l}\cdot2^la_l(X2^{-l-1})+2^{-3l}D\bigg)\]
As a result, the corresponding term in (\ref{inequality that gives 4 terms to deal with}) is bounded by:
\[\sum_{\substack{l\geq k \\ \tau<X2^{-l-1}}}\int_{\tau}^{X2^{-l-1}}2^{k}|\log(2^k\tau')|^22^{-7(l-k)}\bigg(2^la_l(X2^{-l-1})+2^{-l}D\bigg)d\tau'\lesssim2^{-k}D+\sum_{l\geq k}2^{-7(l-k)}\cdot2^la_l(X2^{-l-1})\]
For the third term which is in the high frequency regime, we have using our notation in Section~\ref{high frequency estimates section}:
\[\bigg(\sum_{l\geq k}\mathbf{1}_{l,\tau'}\big\|P_l^2\widetilde{P}_k\xi\big\|_{L^2}(\tau')\bigg)^2\lesssim\sum_{l\geq k}\mathbf{1}_{l,\tau'}2^{-5(l-k)}\big\|P_l\xi\big\|_{L^2}^2(\tau')\lesssim\tau'\sum_{l\geq k}\mathbf{1}_{l,\tau'}2^{-5(l-k)}a_l(\tau')\]
The corresponding term in (\ref{inequality that gives 4 terms to deal with}) can bounded using the estimates in Section~\ref{high frequency estimates section}:
\[\int_{\tau}^{X2^{-k-1}}2^{k}(2^k\tau')^3|\log(2^k\tau')|^2\sum_{l\geq k}\mathbf{1}_{l,\tau'}2^{-6(l-k)}\cdot2^la_l(\tau')d\tau'\lesssim\]\[\lesssim d_k+\tau^62^{6k}\cdot D+\sum_{\tau<X2^{-m-1}\leq1}2^{-4|k-m|}E_m(X2^{-m-1})+\widetilde{E}_k(X2^{-k-1})+\sum_{l\geq k}2^{-6(l-k)}\int_{\tau}^1e_l(\tau')d\tau'\]
Combining the previous bounds and taking the limit $\tau\rightarrow0$, we get from (\ref{inequality that gives 4 terms to deal with}) that:
\[2^{2k}\lim_{\tau\rightarrow0}\big\|\overline{\xi}_k\big\|_{L^2}^2\lesssim d_k+\sum_{l\geq x}2^{-3|k-l|}2^la_l(X2^{-l-1})+\widetilde{E}_k(X2^{-k-1})+\sum_{l\geq x}2^{-4|k-l|}E_l(X2^{-l-1})+\sum_{l\geq k}2^{-6(l-k)}\int_{0}^1e_l(\tau)\]
We sum for all $k\geq x$ and use the estimates in Section~\ref{high frequency estimates section} in order to obtain the conclusion.
\end{proof}

\section{Estimates from $\{v=-u\}$ to $\{v=0\}$}\label{backward direction full system section}

In this section we consider the smooth straight self-similar vacuum spacetime $(\mathcal{M},g)$ obtained from small initial data on the sphere $S_{(-1,1)}$ and we prove optimal estimates on the induced asymptotic data set $\Sigma(\slashed{g}_0,h)$ at $\{u=-1,\ v=0\}.$ Using the ambient metric formulation, in the original $(n+1)$-dimensional formulation these correspond to proving estimates on the asymptotic data at $\mathcal{I}^-$ in terms of initial data at a finite time.
\begin{theorem}\label{main theorem backward direction full system}
    For any $M>0$ large enough and $\epsilon>0$ small enough we consider the smooth straight initial data on the sphere $S_{(-1,1)}$, with the initial data norm:
    \begin{align*}
        \Xi_M^2&=\sum_{i+j=0}^{\frac{n-4}{2}}\sum_{m=0}^{M+\frac{n-4}{2}-i-j}\big\|\nabla_3^i\nabla_4^j\Psi\big\|^2_{H^{m+1}}+\sum_{i+j=0}^{\frac{n-4}{2}}\big\|\nabla^M\nabla_3^i\nabla_4^j\Psi\big\|^2_{H^{3/2}}+\sum_{i+j=0}^{\frac{n-2}{2}}\big\|\nabla^M\nabla_3^i\nabla_4^j\Psi\big\|^2_{H^{1/2}}\\
        &+\sum_{k=0}^1\sum_{i+j=\frac{n-4}{2}-k}\big\|\nabla^{M+1+k}\nabla_3^i\nabla_4^j\psi^*\big\|^2_{H^{1/2}}+\sum_{i+j=0}^{\frac{n-4}{2}}\sum_{m=0}^{M+\frac{n-4}{2}-i-j}\big\|\nabla_3^i\nabla_4^j\psi^*\big\|^2_{H^{m+1}}+\big\|\slashed{g}^*\big\|_{H^{M+1}}^2.
    \end{align*}
    We assume that the initial data satisfy the smallness assumption $\Xi_M\lesssim\epsilon$. We denote by $(\mathcal{M},g)$ the smooth vacuum spacetime obtained using Theorems \ref{stability of de sitter theorem in section}, \ref{propagation of regularity theorem}, and \ref{asymptotic completeness in section theorem}, with induced asymptotic data $(\slashed{g}_0,h)$ at $S_{(-1,0)}$. The corresponding asymptotic data set $\Sigma(\slashed{g}_0,h)$ satisfies the estimate:
    \[\big\|\Sigma(\slashed{g}_0,h)\big\|_{M}\lesssim \Xi_M\]
    where the asymptotic data norm of order $\big\|\Sigma\big\|^2_{M}$ is given in Definition \ref{asymptotic data set definition}.
\end{theorem}

We follow the strategy outlined in Section~\ref{scattering map section intro} of the introduction. Using self-similarity as in Section~\ref{forward direction full system section}, it suffices to work on the null hypersurface $\{u=-1\}.$ We define the following norms on $\big\{u=-1,\ 0\leq v\leq1\big\}:$
\begin{itemize}
    \item Top order energy $\mathcal{T}=\mathcal{T}(-1,v)$:
        \[\mathcal{T}=v^2\big\|\nabla_{4}\nabla^M\nabla_4^{\frac{n-4}{2}}\alpha\big\|^2_{H^{1/2}}+\sum_{m=0}^Mv\big\|\nabla^m\nabla_4^{\frac{n-4}{2}}\alpha\big\|^2_{H^{3/2}}+\sum_{m=0}^M\sqrt{v}\big\|\nabla^m\nabla_4^{\frac{n-4}{2}}\alpha\big\|^2_{H^{1/2}}+\]\[+\sum_{m=0}^M\big\|\nabla^m\nabla_4^{\frac{n-4}{2}}\Psi^G\big\|^2_{H^{3/2}}+\sum_{m=0}^Mv\big\|\nabla_4\nabla^m\nabla_4^{\frac{n-4}{2}}\Psi^G\big\|^2_{H^{1/2}}\]
    \item Lower order energy $\mathcal{L}=\mathcal{L}(-1,v)$:
    \[\mathcal{L}=\sum_{l=0}^{\frac{n-6}{2}}\sum_{m=0}^{M+\frac{n-4}{2}-l}v\big\|\nabla_4^{l+1}\Psi\big\|^2_{H^{m}}+\big\|\nabla_4^l\Psi\big\|^2_{H^{m+1}}\]
    \item Fractional lower order energy $\mathcal{M}_l=\mathcal{M}_l(-1,v)$ for any $0\leq l\leq\frac{n-6}{2}$:
    \[\mathcal{M}_l=\big\|\nabla^M\nabla_4^l\Psi\big\|^2_{H^{5/2}}\]
    \item Ricci coefficients norm $\mathcal{R}=\mathcal{R}(-1,v)$:
    \[\mathcal{R}=\sum_{k=0}^1\big\|\nabla^{M+1+k}\nabla_4^{\frac{n-4}{2}-k}\psi^*\big\|^2_{H^{1/2}}+\sum_{l=0}^{\frac{n-4}{2}}\sum_{m=0}^{M+\frac{n-4}{2}-l}\big\|\nabla_4^{l}\psi^*\big\|^2_{H^{m+1}}\]
    \item Lower order pointwise norm $\mathcal{P}=\mathcal{P}(-1,v)$ for $N'=\frac{M}{2}+\frac{n}{4}$:
    \[\mathcal{P}=\sum_{l=0}^{\frac{n-6}{2}}\sum_{m=0}^{N'}\big\|\nabla^m\nabla_4^{l}\Psi\big\|^2_{L^{\infty}}+\sum_{l=0}^{\frac{n-4}{2}}\sum_{m=0}^{N'}\big\|\nabla^m\nabla_4^{l}\psi^*\big\|^2_{L^{\infty}}\]
    \item Mildly singular pointwise norm $\mathcal{SP}=\mathcal{SP}(-1,v)$:
     \[\mathcal{SP}=\sum_{m=0}^{N'}\big\|\nabla^m\nabla_4^{\frac{n-4}{2}}\Psi\big\|^2_{L^{\infty}}+\sum_{m=0}^{N'}\big\|\nabla^m\nabla_4^{\frac{n-2}{2}}\psi^*\big\|^2_{L^{\infty}}\]
    \item Initial data norm $\mathcal{D}$ at $(u,v)=(-1,1)$:
    \[\mathcal{D}:=\sum_{l=0}^{\frac{n-4}{2}}\sum_{m=0}^{M+\frac{n-4}{2}-l}\big\|\nabla_4^l\Psi\big\|^2_{H^{m+1}}+\sum_{l=0}^{\frac{n-4}{2}}\big\|\nabla^M\nabla_4^l\Psi\big\|^2_{H^{3/2}}+\sum_{l=0}^{\frac{n-2}{2}}\big\|\nabla^M\nabla_4^l\Psi\big\|^2_{H^{1/2}}+\]
    \[+\sum_{k=0}^1\sum_{l=\frac{n-4}{2}-k}\big\|\nabla^{M+1+k}\nabla_4^l\psi^*\big\|^2_{H^{1/2}}+\sum_{l=0}^{\frac{n-4}{2}}\sum_{m=0}^{M+\frac{n-4}{2}-l}\big\|\nabla_4^l\psi^*\big\|^2_{H^{m+1}}+\big\|\slashed{g}^*\big\|_{H^{M+1}}^2\]
\end{itemize}

We remark that using self-similarity, we can replace $\nabla_3$ derivatives with $\nabla_4$ derivatives. Thus, we obtain that $\Xi_M^2\sim\mathcal{D},$ so it suffices to consider $\mathcal{D}$ as the initial data norm. We use the estimates of Section~\ref{stability dS section} to obtain preliminary estimates for our solution. Using Theorems \ref{stability of de sitter theorem in section}, \ref{asymptotic completeness in section theorem}, and \ref{propagation of regularity theorem}, we get that for $N'=\frac{M}{2}+\frac{n}{4}$ we have the estimates:
\[\mathcal{P}\leq\epsilon,\ \mathcal{SP}\leq\epsilon\big(1+|\log(v)|^2\big).\]

We prove Theorem~\ref{main theorem backward direction full system} at the end of the section. This will follow using Theorem \ref{asymptotic data estimates theorem general} and the following result:
\begin{proposition}\label{main proposition backward direction full system}
    The spacetime $(\mathcal{M},g)$ satisfies the following estimate on $\big\{u=-1,\ 0\leq v\leq1\big\}:$
    \[\mathcal{T}+\mathcal{L}+\mathcal{R}\lesssim\mathcal{D}.\]
\end{proposition}

The bound for the top order energy $\mathcal{T}$ follows from the refined estimates proved in Section~\ref{model backward direction section} and \cite[Theorem~1.2]{Cwave}. We can then bound the remaining norms using standard estimates. We note that as in Section~\ref{forward direction full system section}, the nonlinear error terms $Err^{\Psi}$ do not create significant difficulties.

As a consequence of Theorem \ref{backward direction main result theorem general}, we obtain the following estimate for the top order energy:
\begin{corollary}
    The top order energy $\mathcal{T}$ satisfies the estimate for $0\leq v\leq 1:$
    \[\mathcal{T}\lesssim\mathcal{D}+\sum_{m=0}^{M}\int_{v}^1\big\|Err_{m,\frac{n-4}{2}}^{\Psi}\big\|_{H^{1/2}}^2dv'.\]
\end{corollary}
\begin{proof}
    We recall that according to Section~\ref{model systemmm} and Remark~\ref{remark about model systems in section}: \[\Phi_0=\nabla_4^{\frac{n-4}{2}}\alpha,\ \Phi_i=\nabla_4^{\frac{n-4}{2}}\Psi^G,\ F_{m}^{0}=Err_{m,\frac{n-4}{2}}^{\Psi},\ F_{m}^{i}=Err_{m,\frac{n-4}{2}}^{\Psi}\]
    satisfy the second model system as defined in (\ref{second model system definition}) and also \cite[Definition~1.1]{Cwave}. The bounds on the background $\big(\mathcal{M},g\big)$ required in \cite[Theorem~1.2]{Cwave} follow by Theorem~\ref{stability of de sitter theorem in section}. Moreover, since $\mathcal{D}\lesssim\epsilon^2,$ we have in particular that $\|\psi\|_{H^{M+1}}(-1,1)\lesssim1,$ so the implicit constant in Theorem~\ref{backward direction main result theorem general} depends only on $M.$ Thus, we apply Theorem~\ref{backward direction main result theorem general} to obtain the following estimate:
    \[\sum_{m=0}^M\sqrt{v}\big\|\nabla^m\nabla_4^{\frac{n-4}{2}}\alpha\big\|^2_{H^{1/2}}+\sum_{m=0}^Mv\big\|\nabla^m\nabla_4^{\frac{n-4}{2}}\alpha\big\|^2_{H^{3/2}}+v^2\big\|\nabla_{4}\nabla^M\nabla_4^{\frac{n-4}{2}}\alpha\big\|^2_{H^{1/2}}+\int_v^1\big\|\nabla_4^{\frac{n-4}{2}}\alpha\big\|^2_{H^{M+1}}dv'\]\[+\sum_{m=0}^M\big\|\nabla^m\nabla_4^{\frac{n-4}{2}}\Psi^G\big\|^2_{H^{3/2}}+\sum_{m=0}^Mv\big\|\nabla_4\nabla^m\nabla_4^{\frac{n-4}{2}}\Psi^G\big\|^2_{H^{1/2}}+\sum_{m=0}^M\int_v^1\big\|\nabla_4\nabla^m\nabla_4^{\frac{n-4}{2}}\Psi^G\big\|^2_{H^{1/2}}dv'\lesssim\]\[\lesssim \sum_{m=0}^M\bigg(\big\|\nabla^m\nabla_4^{\frac{n-4}{2}}\Psi\big\|_{H^{3/2}}^2+\big\|\nabla^m\nabla_4^{\frac{n-2}{2}}\Psi\big\|_{H^{1/2}}^2\bigg)\bigg|_{v=1}+\sum_{m=0}^{M}\int_{v}^1\big\|Err_{m,\frac{n-4}{2}}^{\Psi}\big\|_{H^{1/2}}^2dv'.\]
    We bound the initial data term using $\mathcal{D}$ to obtain the conclusion.
\end{proof}

We prove the following result for the lower order energy:
\begin{proposition}
    The lower order energy $\mathcal{L}$ satisfies the estimate for $0\leq v\leq 1:$
    \[\mathcal{L}\lesssim\mathcal{D}+\epsilon\mathcal{R}+\int_v^1(v')^{-1/2}\mathcal{T}dv'+\sum_{l=0}^{\frac{n-6}{2}}\sum_{m=0}^{M+\frac{n-4}{2}-l}\int_{v}^1\Big\|Err^{\Psi}_{m,l}\Big\|_{L^2}^2dv'.\]
\end{proposition}
\begin{proof}
As in Section~\ref{forward direction full system section}, we contract (\ref{model wave system 1}) with $2\nabla_4\nabla^m\nabla_4^l\Psi$ and sum for all $0\leq l\leq\frac{n-6}{2}$, $0\leq m\leq M+\frac{n-4}{2}-l$, and every curvature component $\Psi$. We obtain a good bulk term since $n-2l-5\geq1:$
\[\sum_{l=0}^{\frac{n-6}{2}}\sum_{m=0}^{M+\frac{n-4}{2}-l}v\big\|\nabla_4\nabla^m\nabla_4^{l}\Psi\big\|^2_{L^2}+\big\|\nabla\nabla_4^l\Psi\big\|^2_{H^{m}}+\int_v^1\big\|\nabla_4\nabla^m\nabla_4^l\Psi\big\|^2_{L^2}dv'\]\[\lesssim\mathcal{D}+\sum_{l=0}^{\frac{n-6}{2}}\sum_{m=0}^{M+\frac{n-4}{2}-l}\int_v^1\big\|[\nabla,\nabla_4]\nabla^m\nabla_4^l\Psi\big\|^2_{L^2}+\big\|\nabla^{m+1}\nabla_4^l\Psi\big\|^2_{L^2}dv'+\sum_{l=0}^{\frac{n-6}{2}}\sum_{m=0}^{M+\frac{n-4}{2}-l}\int_v^1\big\|Err_{m,l}^{\Psi}\big\|^2_{L^2}dv'\]
We also have the estimate:
\[\big\|\nabla_4^l\Psi\big\|^2_{L^2}\lesssim\mathcal{D}+\int_v^1\big\|\nabla_4^{l+1}\Psi\big\|^2_{L^2}\lesssim\ldots\lesssim\mathcal{D}+\int_v^1\big\|\nabla_4^{\frac{n-4}{2}}\Psi\big\|^2_{L^2}\lesssim\mathcal{D}+\int_v^1(v')^{-1/2}\mathcal{T}dv'\]
We use the commutation formulas and Gronwall to obtain:
\[\sum_{l=0}^{\frac{n-6}{2}}\sum_{m=0}^{M+\frac{n-4}{2}-l}v\big\|\nabla_4\nabla^m\nabla_4^{l}\Psi\big\|^2_{L^2}+\big\|\nabla_4^l\Psi\big\|^2_{H^{m+1}}\lesssim\mathcal{D}+\int_v^1(v')^{-1/2}\mathcal{T}dv'+\sum_{l=0}^{\frac{n-6}{2}}\sum_{m=0}^{M+\frac{n-4}{2}-l}\int_v^1\big\|Err_{m,l}^{\Psi}\big\|^2_{L^2}\]
Next, we use the commutation formulas to obtain:
\[\mathcal{L}\lesssim\sum_{l=0}^{\frac{n-6}{2}}\sum_{m=0}^{M+\frac{n-4}{2}-l}v\big\|\psi\nabla_4^{l}\Psi\big\|^2_{H^m}+\mathcal{D}+\int_v^1(v')^{-1/2}\mathcal{T}dv'+\sum_{l=0}^{\frac{n-6}{2}}\sum_{m=0}^{M+\frac{n-4}{2}-l}\int_v^1\big\|Err_{m,l}^{\Psi}\big\|^2_{L^2}dv'\]
Similarly to Section~\ref{forward direction full system section}, we bound the first term in this relation by using the lower order pointwise bound and the previous inequality for $\big\|\nabla_4^l\Psi\big\|^2_{H^{m+1}}$ in order to get:
\[\mathcal{L}\lesssim\epsilon\mathcal{R}+\mathcal{D}+\int_v^1(v')^{-1/2}\mathcal{T}dv'+\sum_{l=0}^{\frac{n-6}{2}}\sum_{m=0}^{M+\frac{n-4}{2}-l}\int_v^1\big\|Err_{m,l}^{\Psi}\big\|^2_{L^2}dv'.\]
\end{proof}

The bounds for the remaining terms are similar to Section~\ref{forward direction full system section}. We notice that the norms defined above control the same terms as the norms of Section~\ref{forward direction full system section}, the only difference being the bound:
\[\sum_{m=0}^M\big\|\nabla^m\nabla_4^{\frac{n-4}{2}}\alpha\big\|^2_{H^{1/2}}\lesssim v^{-1/2}\mathcal{T}.\]
Due to the structure of the error terms it suffices to control only $M+1/2$ angular derivatives of $\nabla_4^{\frac{n-4}{2}}\alpha$ in the above bound, since the terms with more angular derivatives also have better $v$ weights. Thus, this bound replaces the use of the mildly singular top order energy in Section~\ref{forward direction full system section}. We briefly explain the proofs for the remaining terms:
\begin{proposition}
    The fractional lower order energy $\mathcal{M}_l$ satisfies the estimate for any $0\leq l\leq\frac{n-6}{2}$ and $0\leq v\leq 1:$
    \[\mathcal{M}_l\lesssim v^{-1/2}\mathcal{T}+\mathcal{L}+\mathcal{R}+\big\|Err_{M,l}^{\Psi}\big\|^2_{H^{1/2}}.\]
\end{proposition}
\begin{proof}
    As in Section~\ref{forward direction full system section}, we have the estimate for any $0\leq l\leq\frac{n-6}{2}$:
    \[\big\|\nabla^M\nabla_4^l\Psi\big\|^2_{H^{5/2}}\lesssim\mathcal{L}+\big\|v\nabla_4\nabla^M\nabla_4^{l+1}\Psi\big\|^2_{H^{1/2}}+\big\|v\nabla_4\nabla^M\big(\psi\nabla_4^l\Psi\big)\big\|^2_{H^{1/2}}+\]\[+\big\|\nabla^M\nabla_4^{l+1}\Psi\big\|^2_{H^{1/2}}+\big\|\nabla^M\big(\psi\nabla_4^l\Psi\big)\big\|^2_{H^{1/2}}+\big\|Err_{M,l}^{\Psi}\big\|^2_{H^{1/2}}\]
    \[\lesssim\mathcal{L}+v^{-1/2}\mathcal{T}+\mathcal{R}+\big\|v\nabla_4\nabla^M\nabla_4^{l+1}\Psi\big\|^2_{H^{1/2}}+\big\|v\nabla_4\nabla^M\big(\psi\nabla_4^l\Psi\big)\big\|^2_{H^{1/2}}+\big\|Err_{M,l}^{\Psi}\big\|^2_{H^{1/2}}\]
    Considering separately the cases $l=\frac{n-6}{2}$ and $0\leq l\leq\frac{n-8}{2}$, we also have as before:
    \[\big\|v\nabla_4\nabla^M\nabla_4^{l+1}\Psi\big\|^2_{H^{1/2}}\lesssim v^{-1/2}\mathcal{T}+\mathcal{L}+\mathcal{R}.\]
    Finally, we have the bound as in Section~\ref{forward direction full system section}:
    \[\big\|v\nabla_4\nabla^M\big(\psi\nabla_4^l\Psi\big)\big\|^2_{H^{1/2}}\lesssim\big\|v\nabla^M\big(\psi\nabla_4^{l+1}\Psi\big)\big\|^2_{H^{1/2}}+\big\|v\nabla^M\big((\Psi+\psi\psi)\nabla_4^l\Psi\big)\big\|^2_{H^{1/2}}+\mathcal{L}+\mathcal{R}\lesssim v^{-1/2}\mathcal{T}+\mathcal{L}+\mathcal{R}.\]
\end{proof}

Next, we follow the same steps as in Section~\ref{forward direction full system section} to prove the result for the Ricci coefficients:
\begin{proposition}
    The Ricci coefficients norm $\mathcal{R}$ satisfies the estimate for any $0\leq v\leq 1:$
    \[\mathcal{R}\lesssim\mathcal{D}+\int_v^1\big((v')^{-1/2}\mathcal{T}+\mathcal{L}+\mathcal{R}+(v')^{1/2}\mathcal{M}_{\frac{n-6}{2}}\big)dv'.\]
\end{proposition}
\begin{proof}
    Using the commuted equations satisfied by the Ricci coefficients $\psi^*$, we obtain the estimate:
    \[\sum_{l=0}^{\frac{n-4}{2}}\sum_{m=0}^{M+\frac{n-4}{2}-l}\big\|\nabla_4^{l}\psi^*\big\|^2_{H^{m+1}}\lesssim\mathcal{D}+\sum_{l=0}^{\frac{n-4}{2}}\sum_{m=0}^{M+\frac{n-4}{2}-l}\int_v^1(v')^{1/2}\big\|\nabla_4^l\big(\psi\psi^*\big)\big\|^2_{H^{m+1}}+(v')^{1/2}\big\|\psi\nabla_4^l\big(\psi\psi^*\big)\big\|^2_{H^{m+1}}dv'+\]\[+\sum_{l=0}^{\frac{n-4}{2}}\sum_{m=0}^{M+\frac{n-4}{2}-l}\int_v^1(v')^{1/2}\big\|\nabla_4^l\Psi\big\|^2_{H^{m+1}}dv'\lesssim\mathcal{D}+\int_v^1\big((v')^{-1/2}\mathcal{T}+\mathcal{L}+\mathcal{R}\big)dv'\]
    Using the LP projections and Gronwall as before, we also obtain the following fractional estimate:
    \[\big\|\nabla^{M+1}\nabla_4^{\frac{n-4}{2}}\psi^*\big\|^2_{H^{1/2}}\lesssim\mathcal{D}+\int_v^1\big((v')^{-1/2}\mathcal{T}+\mathcal{L}+\mathcal{R}\big)dv'+\int_0^v(v')^{1/2}\big\|\nabla_4\nabla^{M+1}\nabla_4^{\frac{n-4}{2}}\psi^*\big\|_{H^{1/2}}^2\]
    \[\lesssim\mathcal{D}+\int_v^1\big((v')^{-1/2}\mathcal{T}+\mathcal{L}+\mathcal{R}\big)dv'+\sum_{m=0}^{M+1}\int_v^1\big\|\nabla^{m}\nabla_4^{\frac{n-4}{2}}\big(\psi\psi^*\big)\big\|^2_{H^{1/2}}+(v')^{1/2}\big\|\nabla^{m}\nabla_4^{\frac{n-4}{2}}\Psi\big\|_{H^{1/2}}^2\]\[\lesssim\mathcal{D}+\int_v^1\big((v')^{-1/2}\mathcal{T}+\mathcal{L}+\mathcal{R}\big)dv'\]
    Similarly, we also get the estimate in the case when $n\geq6$:
    \[\big\|\nabla^{M+2}\nabla_4^{\frac{n-6}{2}}\psi^*\big\|^2_{H^{1/2}}\lesssim\mathcal{D}+\int_v^1\big((v')^{-1/2}\mathcal{T}+\mathcal{L}+\mathcal{R}\big)dv'+\int_0^v(v')^{1/2}\big\|\nabla_4\nabla^{M+2}\nabla_4^{\frac{n-6}{2}}\psi^*\big\|_{H^{1/2}}^2\]
    \[\lesssim\mathcal{D}+\int_v^1\big((v')^{-1/2}\mathcal{T}+\mathcal{L}+\mathcal{R}\big)dv'+\sum_{m=0}^{M+2}\int_0^v\big\|\nabla^{m}\nabla_4^{\frac{n-6}{2}}\big(\psi\psi^*\big)\big\|^2_{H^{1/2}}+(v')^{1/2}\big\|\nabla^{m}\nabla_4^{\frac{n-6}{2}}\Psi\big\|_{H^{1/2}}^2\]
    \[\lesssim\mathcal{D}+\int_v^1(v')^{-1/2}\mathcal{T}+\mathcal{L}+\mathcal{R}+(v')^{1/2}\mathcal{M}_{\frac{n-6}{2}}.\]
\end{proof}

The remaining step in proving Proposition \ref{main proposition backward direction full system} is controlling the error terms from the above estimates. Since the lower order pointwise norms $\mathcal{P}$ and $\mathcal{SP}$ satisfy the same bounds as in Section~\ref{forward direction full system section}, we get that Lemma \ref{lemma to bound general F error terms} applies in the current situation. We adapt the proof in Section~\ref{forward direction full system section} to prove the following result for the error terms:
\begin{proposition}
    For any $0\leq l\leq\frac{n-4}{2},\ 0\leq m\leq M$, the error terms $Err_{m,l}^{\Psi}$ satisfy the estimate:
    \[\big\|Err_{m,l}^{\Psi}\big\|^2_{H^{1/2}}\lesssim \big(1+|\log(v)|^2\big)\big(v^{-1/2}\mathcal{T}+\mathcal{L}+\mathcal{R}\big).\]
\end{proposition}
\begin{proof}
    As in Section~\ref{forward direction full system section}, the proof follows from the fractional lower order energy estimates, once we prove the claim:
    \[\big\|Err_{m,l}^{\Psi}\big\|^2_{H^{1/2}}\lesssim\big(1+|\log(v)|^2\big)\big(v^{-1/2}\mathcal{T}+\mathcal{L}+\mathcal{R}\big)+\mathcal{M}_{l-1}.\]
    However, the proof of this estimate follows the exact same steps as in Section~\ref{forward direction full system section} once we replace all the mildly singular norms $\mathcal{S}$ by $v^{-1/2}\mathcal{T},$ so we omit repeating the proof here. We then use the fractional lower order energy estimate to obtain:
    \[\big\|Err_{m,l}^{\Psi}\big\|^2_{H^{1/2}}\lesssim\big(1+|\log(v)|^2\big)\big(v^{-1/2}\mathcal{T}+\mathcal{L}+\mathcal{R}\big)+\mathcal{M}_{l-1}\lesssim\big(1+|\log(v)|^2\big)\big(v^{-1/2}\mathcal{T}+\mathcal{L}+\mathcal{R}\big)+\big\|Err_{m,l-1}^{\Psi}\big\|^2_{H^{1/2}}\]
    By induction we obtain that:
    \[\big\|Err_{m,l}^{\Psi}\big\|^2_{H^{1/2}}\lesssim\big(1+|\log(v)|^2\big)\big(v^{-1/2}\mathcal{T}+\mathcal{L}+\mathcal{R}\big)+\big\|Err_{m,0}^{\Psi}\big\|^2_{H^{1/2}}\lesssim\big(1+|\log(v)|^2\big)\big(v^{-1/2}\mathcal{T}+\mathcal{L}+\mathcal{R}\big).\]
\end{proof}

Similarly, we can repeat the proof in Section~\ref{forward direction full system section} and replace the mildly singular norms $\mathcal{S}$ by $v^{-1/2}\mathcal{T},$ in order to establish the following result for the remaining error terms:
\begin{proposition}
    For any for any $0\leq l\leq\frac{n-4}{2}$ and $0\leq m\leq M+\frac{n-4}{2}-l$, the error terms $Err_{m,l}^{\Psi}$ satisfy the estimate:
    \[\big\|Err_{m,l}^{\Psi}\big\|^2_{L^2}\lesssim(1+|\log(v)|^2\big)\big(v^{-1/2}\mathcal{T}+\mathcal{L}+\mathcal{R}\big).\]
\end{proposition}

We complete the proofs of Proposition \ref{main proposition backward direction full system} and Theorem \ref{main theorem backward direction full system}:

\textit{Proof of Proposition \ref{main proposition backward direction full system}.} Using the estimates in this section, we have that for all $0\leq v\leq1$:
\[\mathcal{T}+\mathcal{L}+\mathcal{R}\lesssim\mathcal{D}+\epsilon\mathcal{R}+\int_v^1\big(1+|\log v'|^2\big)\big((v')^{-1/2}\mathcal{T}+\mathcal{L}+\mathcal{R}\big)dv'.\]
Taking $\epsilon>0$ small enough and using Gronwall's inequality, we conclude that  $\mathcal{T}+\mathcal{L}+\mathcal{R}\lesssim\mathcal{D}.$
\qed

\textit{Proof of Theorem \ref{main theorem backward direction full system}.} In order to complete the proof of Theorem \ref{main theorem backward direction full system}, we prove the bound:
\begin{equation}\label{claim for backward direction theorem}
    \big\|\slashed{g}_0^*\big\|_{\mathring{H}^{M+1}}^2+\big\|\mathcal{O}\big\|_{H^{M+1}}^2+\big\|\mathfrak{h}\big\|_{H^{M+1}}^2\lesssim\mathcal{D},
\end{equation}
since the other terms in the asymptotic data norm $\big\|\Sigma\big\|_M^2$ are already bounded by $\mathcal{D}$ using Proposition \ref{main proposition backward direction full system}. Using the estimates on the asymptotic data in Theorem \ref{asymptotic data estimates theorem general} we obtain:
\[\big\|\mathcal{O}\big\|_{H^{M+1}}^2\lesssim \mathcal{D}+\sum_{m=0}^{M}\int_{0}^1\Big\|Err_{m,\frac{n-4}{2}}^{\Psi}\Big\|^2_{H^{1/2}}dv.\]
We also use the above estimates on the error terms and Proposition \ref{main proposition backward direction full system}:
\[\big\|\mathcal{O}\big\|_{H^{M+1}}^2\lesssim \mathcal{D}+\int_{0}^1\big(1+|\log(v)|^2\big)\big(v^{-1/2}\mathcal{T}+\mathcal{L}+\mathcal{R}\big)dv\lesssim \mathcal{D}.\]

Using the metric equation $\mathcal{L}_v\slashed{g}^*=\psi^*$ and the estimate $\mathcal{R}\lesssim\mathcal{D},$ we obtain that $\big\|\slashed{g}^*\big\|_{H^{M+1}}^2\lesssim\mathcal{D}.$ We use this, together with the estimates in Proposition \ref{region III bounds proposition} and (\ref{Lie derivative in terms of covariant derivative}) in order to prove by induction on $m\leq M+1$ that $\big\|\mathcal{L}_{\theta}^m\slashed{g}_0^*\big\|_{L^2(S^n)}^2\lesssim\epsilon.$ In particular, this allows us to bound the Christoffel symbols terms in (\ref{Lie derivative in terms of covariant derivative}) and prove:
\[\sum_{m\leq M+1}\big\|\mathcal{L}_{\theta}^m\slashed{g}_0^*\big\|_{L^2(S^n)}^2\lesssim\mathcal{D}\]
We can then use (\ref{Lie derivative in terms of covariant derivative}) for covariant derivatives with respect to $\slashed{g}_{S^n}$ and prove by induction on $m\leq M+1$ that:
\[\big\|\mathring{\nabla}^m\slashed{g}^*_0\big\|_{L^2}^2\lesssim\mathcal{D}.\]

In order to complete the proof of (\ref{claim for backward direction theorem}), we must bound $\big\|\mathfrak{h}\big\|_{H^{M+1}}.$ By Theorem \ref{asymptotic data estimates theorem general} we get:
\[\big\|\mathfrak{h}\big\|_{H^{M+1}}^2\lesssim\mathcal{D}+\big\|h\big\|_{L^2}^2+C\Big(\big\|\slashed{Riem}_0\big\|_{H^{M}}^2\Big)\big\|\mathcal{O}\big\|_{H^{M}}^2+\sum_{m=0}^{M}\int_{0}^1\Big\|Err_{m,\frac{n-4}{2}}^{\Psi}\Big\|^2_{H^{1/2}}dv\lesssim\mathcal{D}+\big\|h\big\|_{L^2}^2,\]
where we used the constraint $\slashed{Riem}=R+\psi\psi$ and the previously established bounds for $\mathcal{O}$ and the error terms. For the rest of the proof we prove a suitable lower order estimate for $\big\|h\big\|_{L^2}^2.$ As in the proof of Theorem \ref{asymptotic data estimates theorem general}, we define $\xi=\nabla_4^{\frac{n-4}{2}}\alpha$ and $\overline{\xi}=\xi-v\log v\nabla_4\xi.$ We compute using (\ref{model wave system 2}):
\[\nabla_4\overline{\xi}=-\log v\cdot\nabla_4\big(v\nabla_4\xi\big)=\log v\cdot\Big(\Delta\nabla_4^{\frac{n-4}{2}}\alpha+\psi\nabla\nabla_4^{\frac{n-4}{2}}\Psi+Err^{\Psi}_{0,\frac{n-4}{2}}\Big).\]
Therefore, we get using our previous estimates:
\[\big\|h\big\|_{L^2}^2=\lim_{v\rightarrow0}\big\|\overline{\xi}\big\|_{L^2}^2\lesssim\mathcal{D}+\int_0^1(\log v)^2\cdot\Big(\big\|\nabla_4^{\frac{n-4}{2}}\Psi\big\|_{H^2}^2+\big\|Err^{\Psi}_{0,\frac{n-4}{2}}\big\|_{L^2}^2\Big)\lesssim\mathcal{D}.\]
This completes the proof of Theorem \ref{main theorem backward direction full system}.
\qed

\section{The Scattering Map}\label{scattering map section}
In this section, we follow the outline in Section \ref{scattering map section intro} of the introduction and we put together our previous results in order to complete the proof of the third statement in Theorem~\ref{main theorem of the paper ambient}. For any $M>0$ large enough and $\epsilon>0$ small enough we consider smooth straight initial data $\big(\slashed{g}_0,h\big)$ such that the corresponding asymptotic data set satisfies $\Sigma\big(\slashed{g}_0,h\big)\in B_{\epsilon}^M\big(\Sigma_{\mathrm{Minkowski}}\big).$ By Theorem \ref{stability of de sitter theorem in section} there exists a unique smooth straight self-similar vacuum solution $(\mathcal{M},g)$ in double null coordinates defined in $\{u<0,\ v>0\}$ with asymptotic initial data at $\{v=0\}$ given by $\big(\slashed{g}_0,h\big).$ Moreover, this induces smooth asymptotic data $\big(\underline{\slashed{g}_{0}},\underline{h}\big)$ at $\{u=0\}$ by Theorem \ref{asymptotic completeness in section theorem}. The estimates in Theorem \ref{main theorem forward direction full system} imply that:
\[\Xi_M\lesssim\Big\|\Sigma\big(\slashed{g}_0,h\big)\Big\|_M.\]
We also have by Theorem \ref{main theorem backward direction full system} that the reverse inequality holds. We use this inequality for the spacetime obtained by reversing the time orientation, namely by replacing $(u,v)$ with $(-v,-u)$. We obtain the estimate:
\[\Big\|\Sigma\big(\underline{\slashed{g}_{0}},\underline{h}\big)\Big\|_M\lesssim\Xi_M,\]
where $\Sigma\big(\underline{\slashed{g}_{0}},\underline{h}\big)$ is defined according to Remark \ref{remark about def of sigma}. Therefore, there exists a constant $C_M>0$ depending only on $M$ such that:
\[\Big\|\Sigma\big(\underline{\slashed{g}_{0}},\underline{h}\big)\Big\|_M\leq C_M\Big\|\Sigma\big(\slashed{g}_0,h\big)\Big\|_M.\]
We obtain that $\Sigma\big(\underline{\slashed{g}_{0}},\underline{h}\big)\in B_{C_M\epsilon}^M\big(\Sigma_{\mathrm{Minkowski}}\big),$ so we can define the scattering map: \[\mathcal{S}:B_{\epsilon}^M\big(\Sigma_{\mathrm{Minkowski}}\big)\rightarrow B_{C_M\epsilon}^M\big(\Sigma_{\mathrm{Minkowski}}\big),\ \mathcal{S}\big(\Sigma(\slashed{g}_0,h)\big)=\Sigma(\underline{\slashed{g}_{0}},\underline{h}).\]
Using the uniqueness of scattering states statement from Theorem \ref{stability of de sitter theorem in section}, we obtain that $\mathcal{S}$ is injective. Moreover, we can apply the existence and uniqueness results, together with the estimate \eqref{main estimate for scattering map intro ambient} in the reverse time direction to obtain that $B_{\epsilon/C_M}^M\big(\Sigma_{\mathrm{Minkowski}}\big)\subset \mathcal{S}\big(B_{\epsilon}^M(\Sigma_{\mathrm{Minkowski}})\big).$ Therefore, $\mathcal{S}$ is locally invertible at $\Sigma_{\mathrm{Minkowski}}.$ Finally, the above estimate implies that $\mathcal{S}$ is locally Lipschitz at $\Sigma_{\mathrm{Minkowski}}.$

\bibliographystyle{amsalpha}
\bibliography{refs}

\end{document}